\documentclass[10pt]{article}

\usepackage{latexsym}
\usepackage{amsmath}
\usepackage{amssymb}
\usepackage{amsthm}
\usepackage{amscd}
\usepackage[mathscr]{eucal}
\usepackage{fullpage}
\usepackage{graphicx}
\usepackage{subfigure}
\usepackage{float}
\usepackage{psfrag}
\usepackage{rotating}
\usepackage{color}
\usepackage{hyperref}

\definecolor{ggreen}{cmyk}{0.7,     0,      0.9,      0}
\definecolor{viol}{cmyk}{0.3,1,0,0}
\definecolor{myred}{cmyk}{0.1, 1, 0.5, 0}
\definecolor{bblue}{rgb}{0.2, 0.29996, 0.8 }

\DeclareMathOperator{\arctanh}{arctanh}

\theoremstyle{plain}

\newtheorem{theorem}{Theorem}

\newtheorem{lemma}{Lemma}
\newtheorem{remark}{Remark}
\newtheorem{proposition}{Proposition}

\newtheorem{definition}{Definition}

\newcounter{mnotecount}[section]

\renewcommand{\themnotecount}{\thesection.\arabic{mnotecount}}

\newcommand{\mnote}[1]
{\protect{\stepcounter{mnotecount}}$^{\mbox{\footnotesize
			$
			\bullet$\themnotecount}}$ \marginpar{
		\raggedright\tiny\em
		$\!\!\!\!\!\!\,\bullet$\themnotecount: #1} }
\begin{document}
	\title{\bf \huge{Dynamics of interacting monomial scalar field potentials and perfect fluids}}
	
	\author{Artur Alho$^{(1)}$, Vitor Bessa$^{(1,2,3)}$ and Filipe C.
		Mena$^{(1,2)}$\\\\
		{\small $^{(1)}$Centro de An\'alise Matem\'atica, Geometria e Sistemas
			Din\^amicos,}
		\\
		{\small	Instituto Superior T\'ecnico, Universidade de Lisboa, Av. Rovisco Pais
			1, 1049-001 Lisboa, Portugal}\\
		{\small $^{(2)}$Centro de Matem\'atica, Universidade do Minho, 4710-057 Braga,
			Portugal}
		\\
		{\small $^{(3)}$Faculdade de Ci\^encias, Universidade do Porto, R. Campo Alegre,
			4169-007 Porto, Portugal}
	}	
\maketitle
\begin{abstract}
	Motivated by cosmological models of the early universe we analyse the dynamics of the Einstein equations with a minimally coupled scalar field with monomial potentials $V(\phi)=\frac{(\lambda\phi)^{2n}}{2n}$, $\lambda>0$, $n\in\mathbb{N}$, interacting with a perfect fluid with linear equation of state $p_\mathrm{pf}=(\gamma_\mathrm{pf}-1)\rho_\mathrm{pf}$, $\gamma_\mathrm{pf}\in(0,2)$, in flat Robertson-Walker spacetimes. The interaction is a friction-like term of the form $\Gamma(\phi)=\mu \phi^{2p}$, $\mu>0$, $p\in\mathbb{N}\cup\{0\}$. The analysis relies on the introduction of a new regular 3-dimensional dynamical systems' formulation of the Einstein equations on a compact state space, and the use of dynamical systems' tools such as quasi-homogeneous blow-ups and averaging methods involving a time-dependent perturbation parameter. 
	We find a bifurcation at $p=n/2$ due to the influence of the interaction term. In general, this term has more impact on the future (past) asymptotics for $p<n/2$ ($p>n/2$). For $p<n/2$ we find a complexity of possible future attractors, which depends on whether $p=(n-1)/2$ or $p<(n-1)/2$. In the first case the future dynamics is governed by Li\'enard systems. On the other hand when $p=(n-2)/2$ the generic future attractor consists of new solutions previously unknown in the literature which can drive future acceleration whereas the case $p<(n-2)/2$ has a generic future attractor de-Sitter solution. For $p=n/2$ the future asymptotics can be either fluid dominated or have an oscillatory behaviour where neither the fluid nor the scalar field dominates. For $p>n/2$ the future asymptotics is similar to the case with no interaction. Finally, we show that irrespective of the parameters, an inflationary quasi-de-Sitter solution always exists towards the past, and therefore the cases with $p\leq(n-2)/2$ may provide new cosmological models of quintessential inflation. 
\end{abstract}

\newpage
\tableofcontents
\newpage
\section{Introduction}
The most successful theory of gravity is General Relativity which is based on the Einstein field equations. These equations form a non-linear system of partial differential equations for the dynamics of the metric of a four-dimensional Lorentzian manifold called spacetime. In the case of the standard cosmological models of very large scales, the spacetime metric is commonly assumed to be spatially homogeneous which reduces the Einstein equations to a system of ordinary differential equations. It is remarkable that many recent rigorous results about the dynamics of cosmological models result from the application of methods of dynamical systems,  see e.g. \cite{Bre17, HLU20} as well as \cite{Ringstrom-book} and references therein. Part of the challenge in applying those methods is the necessity to use conformal rescalings and appropriate normalizations that lead to dimensionless and regular autonomous finite-dimensional dynamical system on compact state spaces. 
For a review on the theory of dynamical systems in cosmology see \cite{WainElis,Col03,Boehmer}. 

According to the current cosmological paradigm, the universe went through a period of accelerated expansion soon after the initial (big bang) singularity, called the inflationary epoch~\cite{Guth,Linde}.
Cosmological Inflation, has the ability to solve some of the puzzles of the big bang model, such as the horizon and flatness problems, and also provides probably the best-known mechanism to predict the observable nearly scale-invariant primordial scalar and tensor power-spectrum~\cite{Alexei}. The most popular physical mechanism driving inflation in the early universe consists of a scalar field (called inflaton) with monomial type potentials~\cite{Muk05}. A drawback of the original (cold) inflation models is the separation that exists between the inflationary and the reheating periods. The process of reheating is crucial in inflationary cosmology and a natural answer to the problem is the so-called warm inflation first introduced by Berera \cite{Berera1}. In this scenario, the production of radiation occurs simultaneously with the inflationary expansion, through interactions between the inflaton and other fields on a thermal bath. The presence of radiation allows having a smooth transition between the inflationary phase and a radiation-dominated phase without a reheating separation phase. After the inflationary epoch, the standard early universe scenario then involves a period of radiation dominance until a time of decoupling between radiation and matter, after which the universe is dominated by matter, and then galaxies start to form as well as other cosmic structures. The radiation and matter-dominated eras are usually modeled by a perfect fluid with a linear equation of state.  
It is our purpose here to investigate this dynamical interaction between scalar fields and perfect fluids. 

We will consider the Einstein equations for a spatially homogeneous and isotropic metric (the so-called Robertson-Walker metric) having a scalar field with monomial potentials interacting with perfect fluids with linear equation of state. Our main goal is to obtain a global dynamical picture of the resulting system of non-linear ordinary differential equations (ODEs) and in particular of its past and future asymptotics. Our analysis relies on the introduction of a new set of dimensionless bounded variables which results in a regular dynamical system on a compact state-space consisting of a 3-dimensional cylinder. This allow us to describe the global evolution of these cosmological models identifying all possible past and future attractor sets which, as we will see, in many situations, can be non-hyperbolic fixed points, partially hyperbolic lines of fixed points, or even bands of periodic orbits. So, our analysis will require in one hand blow-up techniques and center manifold theory around the non-hyperbolic fixed points and, on the other hand, averaging methods involving a time dependent perturbation parameter. 

Our paper is mostly self-contained and is organized as follows: In Section \ref{sec2} we explain how the non-linear system of ODEs is obtained from physical principles.  In Section \ref{sec3} we find the appropriate dimensionless variables that transform the ODE system into a 3-dimensional regular dynamical system on a compact state space. This construction naturally shows that a significant bifurcation occurs for $p=n/2$. We therefore split the analysis into three cases according to the different exponents of the scalar field potential and the interaction term: Section~\ref{case1} treats the case $p<n/2$, where the future asymptotic analysis is further split in two distinct sub-cases corresponding to $p=(n-1)/2$ and $p<(n-1)/2$. The case $p=n/2$ is treated in Section~\ref{case2} and 
the case $p>n/2$ in Section~\ref{case3}. Interestingly, in some situations we encounter non-hyperbolic fixed points whose exceptional divisor of the blow-up space consists of generalised Li\'enard systems. We provide proofs as well as conjectures about the global dynamics complemented by numerical pictures of representative cases. We present our conclusions in Section~\ref{conclusion} where we also mention some physically interesting consequences of our results.

\section{Non-linear ODE system}
\label{sec2}
A spacetime $(M,g)$ is a 4-dimensional Lorentzian manifold with metric $g$ whose evolution is given by the Einstein equations
\begin{equation}
\label{EFEs}
\mathrm{Ric}-\frac{1}{2}R g={\cal T},
\end{equation}
where $\mathrm{Ric}$ denotes the Ricci curvature tensor, $R$ the scalar curvature of the spacetime, while ${\cal T}$ is the energy-momentum tensor which encodes the spacetime physical content. When the spacetime is spatially homogeneous, the Einstein equations became ODEs and can be analysed using dynamical systems' methods. The set of solutions depends crucially on the right-hand-side of the equations, i.e. on the energy-momentum tensor. In turn, this depends on the physical scenario under consideration.

Motivated by the warm inflation scenario of the early universe, here we assume a minimally coupled scalar field $\phi$ with self-interaction potential $V(\phi)$, interacting with a perfect fluid. The evolution equations can be derived from an action principle and the most general action for this case is given by
\begin{equation}
S=\int_M \Big(\frac{R}{2}+\frac{1}{2}(\partial \phi)^2-V(\phi) + \mathcal{L}_{\mathrm{pf}} + \mathcal{L}_{int}\Big) \sqrt{-\det{(g)}} d^4x^\mu,
\end{equation} 
where as standard, we use greek indices $\mu,\nu,...=0,1,2,3$ for each coordinate in spacetime. Here  $(\partial \phi)^2:=(\partial_\mu\phi)( \partial^\mu \phi)$ with $\partial^{\mu}\phi=g^{\mu\nu}\partial_\nu\phi$ whereas $\mathcal{L}_{\mathrm{pf}}$ is the Lagrangian density of the perfect fluid and $\mathcal{L}_{int}$ describes the interaction between the scalar field and the thermal bath.
By varying the Lagrangian density with respect to the metric we obtain the Einstein equations \eqref{EFEs} with stress-energy tensor components
\begin{equation}
\label{Tmunu}
{\cal T}_{\mu\nu}={\cal T}_{\mu\nu}^{(\phi)}+{\cal T}_{\mu \nu}^{(\mathrm{pf})},
\end{equation}
where
\begin{subequations}
	\begin{align}
{\cal T}_{\mu \nu }^{({\mathrm{pf}})} &=(\rho_\mathrm{pf}+p_\mathrm{pf})u_\mu u_\nu + p_\mathrm{pf} g_{\mu \nu}+g_{\mu\nu}\mathcal{L}_{int}, \label{IndConsFluid}\\
	{\cal T}_{\mu\nu}^{(\phi)} &= \partial_\mu \phi \partial_\nu \phi -(\frac{1}{2}(\partial \phi)^2-V(\phi)) g_{\mu \nu}-g_{\mu\nu}\mathcal{L}_{int}, \label{IndConsScalar}
	\end{align}
\end{subequations}
and $u^\mu$ denotes the unit $4$-velocity vector field of the perfect fluid, with $\rho_\mathrm{pf}> 0$ and $p_{\mathrm{pf}}$ being the fluid energy density and pressure, respectively. 

The stress-energy tensor for the scalar field can be written in a perfect fluid form with the identifications $u^\mu_{(\phi)}=(\partial^\mu\phi)/\sqrt{-(\partial\phi)^2}$,
\begin{equation}
\rho_\phi = \frac{1}{2}(\partial\phi)^2+V(\phi),\qquad p_\phi=\frac{1}{2}(\partial\phi)^2 - V(\phi).
\end{equation}
The total stress-energy tensor ${\cal T}_{\mu\nu}$ obeys the conservation law $\nabla_\nu {\cal T}^{\nu}_{\mu}=0$. However each component of the total stress-energy tensor, ${\cal T}_{\mu\nu}^{(\phi)}$ and ${\cal T}_{\mu \nu}^{({\mathrm{pf}})}$, is not conserved, in contrast to the case when the scalar field does not interact with the thermal bath. In the presence of interactions
\begin{equation}
	\nabla^\nu  {\cal T}_{\mu\nu}^{(\phi)} = Q_{\mu}^{(\phi)},\qquad
\nabla^\nu {\cal T}_{\mu\nu}^{({\mathrm{pf}})} = Q_{\mu}^{({\mathrm{pf}})},
\end{equation}
where $Q_{\mu}^{(\phi)}$ and $Q_{\mu}^{({\mathrm{pf}})}$ describe the energy exchange between the scalar field and the perfect fluid. It follows from the energy-momentum conservation equation that
\begin{equation}
\nabla^\nu {\cal T}_{\mu\nu} = Q_{\mu}^{(\phi)}+Q_{\mu}^{({\mathrm{pf}})}=0.
\end{equation}
In this work we consider a phenomenological friction-like interaction term for which
\begin{equation}\label{exchange}
Q_{\mu}^{(\phi)}=-Q_{\mu}^{({\mathrm{pf}})}=-\Gamma(\phi) u^\nu \partial_\mu \phi \partial_\nu \phi,
\end{equation}
where we assume that $\Gamma=\Gamma(\phi)$ is a function of the scalar field $\phi$ only. In more general warm inflationary models, the function $\Gamma$ can also depend on the thermal bath temperature \cite{Moss,BC00,Ventura,Berg,Setare}, although recent studies suggest that temperature dependence is redundant~\cite{Visinelli}. Equation~\eqref{IndConsFluid} then gives the modified energy "conservation" equation and the Euler equation
\begin{subequations}
	\begin{align}
	-u^\mu\nabla_\mu \rho_{\mathrm{pf}} +(\rho_{\mathrm{pf}}+p_{\mathrm{pf}})\nabla_\mu u^\mu &=  \Gamma(\phi)(\partial\phi)^2, \label{consrho}\\	
	(\rho_{\mathrm{pf}}+p_{\mathrm{pf}})u^\mu\nabla_\mu u^\nu +u^\nu u^\mu\nabla_\mu p_{\mathrm{pf}}+\nabla^\nu p_{\mathrm{pf}} &= 0.
	\end{align}
\end{subequations}
The above system is closed once an equation of state relating the pressure and the energy density is given. Here we assume that the fluid obeys a linear equation of state
\begin{equation}
p_{\mathrm{pf}}=(\gamma_{\mathrm{pf}}-1)\rho_\mathrm{pf},\qquad\gamma_{\mathrm{pf}}\in(0,2),
\end{equation}
where for example, $\gamma_{\mathrm{pf}}=1$ corresponds to a dust fluid, and  $\gamma_{\mathrm{pf}}=4/3$ to a radiation fluid. The value $\gamma_{\mathrm{pf}}=0$ corresponds to the case of a positive cosmological constant and $\gamma_{\mathrm{pf}}=2$ to a stiff fluid, both yielding significant dynamical bifurcations. Equation~\eqref{IndConsScalar} yields the wave equation
\begin{equation}
\label{KG}
\square_g \phi=-\frac{dV}{d\phi}+\Gamma(\phi) u^\mu\partial_\mu \phi,
\end{equation}
where $\square_g$ is the usual D'Alembertian operator associated with the metric $g$. Motivated by the current cosmological models, as mentioned in the Introduction, we will use a flat spatially homogeneous and isotropic metric $g$, called Robertson-Walker (RW) metric, that in the Cartesian coordinates $(t,x,y,z)\in (t_-,t_+) \times \mathbb{R}^{3}$ takes the form
\begin{equation}
g=-dt^2+a(t)^2(dx^2+dy^2+dz^2),
\end{equation}
where $a:\,(t_-,t_+)\rightarrow \mathbb{R}^+$ is a $C^2$ positive function of time $t$ called scale-factor whose evolution and maximal existence interval will be determined by the Einstein equations. A solution is said to be global to the past (future) if $t_-=-\infty$ ($t_+=+\infty$). 

The Einstein equations coupled to the nonlinear scalar field equation \eqref{KG} and the energy conservation equation for the fluid component~\eqref{consrho}, form the following non-linear ODE system for the unknowns $\{a,H,\phi,\rho_\mathrm{pf}\}$:
\begin{subequations}\label{sistema1}
	\begin{align}
	\dot{a} &= aH, \\
	\dot{H} &=-\frac{1}{2}\gamma_{\mathrm{pf}} \rho_\mathrm{pf}-\frac{\dot{\phi}^2}{2},\\
	\ddot{\phi} &= -(3H+\Gamma(\phi))\dot{\phi}-\frac{dV}{d\phi}, \label{scalar} \\ 
	\dot{\rho}_\mathrm{pf} &= -3\gamma_{\mathrm{pf}} H\rho_\mathrm{pf} + \Gamma(\phi) \dot{\phi}^2, \label{rho_m}  
	\end{align}
\end{subequations}
together with the Gauss (Hamiltonian) constraint
\begin{equation}
H^{2} = \frac{\rho_{\mathrm{pf}}}{3}+\frac{\dot{\phi}^2}{6}+\frac{V(\phi)}{3}, 
\label{constraint}
\end{equation}
where
\begin{equation}
H:=\frac{\dot a }{a}
\end{equation}
is called Hubble function and a dot denotes differentiation with respect to time $t$. For expanding cosmologies $H>0$. Note also that the equation for the scale factor $a(t)$ decouples leaving a reduced system of equations for the unknowns $\{H,\phi,\rho_\mathrm{pf}\}$. The scale-factor can then be obtained by quadrature $a=a_0e^{\int H dt}$. 
The $\Gamma(\phi)$ term appearing in~\eqref{scalar} and~\eqref{rho_m} acts as a friction term  which describes the decay of the scalar field $\phi$ due to the interactions encoded in the Lagrangian $\mathcal{L}_{int}$. 
Here we assume monomial scalar field potentials which are popular examples of inflaton models
\begin{equation}\label{potential}
V(\phi)=\frac{(\lambda\phi)^{2n}}{2n}\quad,\quad \lambda>0 \quad,\quad n=1,2,3,...
\end{equation}
and a monomial scalar field interaction term
\begin{equation} \label{dissipation}
\Gamma(\phi)= \mu \phi^{2p} \quad,\quad \mu>0\quad,\quad p=0,1,2,3,...
\end{equation}
The exponent $2p$ reflects the parity invariance of the interaction term, and the condition $\mu>0$ ensures that the second law of thermodynamics is satisfied (see e.g.~\cite{Bastero,Oliveira,Bin}). 

To summarise, we will analyse the system \eqref{sistema1} for the unknowns $\{H,\phi,\rho_\mathrm{pf}\}$ having the free parameters $\{n,p, \lambda,\mu,\gamma_\mathrm{pf}\}$ besides the initial conditions $\{\rho_0, \phi_0, \dot \phi_0,H_0\}$. Note that $n$, $p$ and $\gamma_\mathrm{pf}$ are dimensionless parameters while $\lambda$ and $\mu$ have dimensions. We shall see ahead that it is the dimensionless ratio (see~\eqref{NuDef} ahead) of these two quantities that plays an important role on the qualitative behaviour of solutions.

\section{Dynamical systems' formulation}
\label{sec3}
In order to obtain a regular dynamical system on a compact  state-space, we start by introducing dimensionless variables normalized by the Hubble function $H$ (which is positive for ever expanding models)
\begin{equation} \label{variables}
\Omega_{\mathrm{pf}}:=\frac{\rho_\mathrm{pf}}{3H^2}>0,\qquad\Sigma_{\mathrm{\phi}}:=\frac{\dot{\phi}}{\sqrt{6}H},\qquad X:=\frac{\lambda\phi}{(6nH^2)^{\frac{1}{2n}}},\qquad\tilde{T}:=\frac{c}{H^{\frac{1}{n}}}>0,
\end{equation}
where $c=\left(\frac{6^{n-1}}{n}\right)^{\frac{1}{2n}}\lambda$ is a positive constant. We also introduce a new time variable $\tilde{N}$ defined by
\begin{equation}\label{tempo1}
\frac{d}{d\tilde{N}}:=\frac{\left(\frac{c}{H^{\frac{1}{n}}}\right)^{\delta}}{H}\frac{d}{dt},
\end{equation}
where
\[\delta=\left\{
\begin{aligned}
0 \quad&\quad\text{if}\quad  p\leq \frac{n}{2} ,\\
2p-n \quad&\quad\text{if}\quad p>\frac{n}{2}.
\end{aligned}
\right.\]
When $\delta=0$, i.e. $p\leq \frac{n}{2}$, then $\tilde{N}=N=\ln(a/a_0)$ is the number of $e$-folds $N$ from some reference epoch at which $a=a_0$, i.e. $N=0$. When written in terms of the new variables, the system \eqref{sistema1} reduces to a regular $local$ 3-dimensional dynamical system

\begin{subequations}\label{sistemalocal}
	\begin{align}
	\frac{dX}{d\tilde{N}} &= \frac{1}{n}(1+q)\tilde{T}^{\delta}X+\tilde{T}^{1+\delta}\Sigma_{\mathrm{\phi}} ,\label{local1}\\
	\frac{d\Sigma_{\mathrm{\phi}}}{d\tilde{N}} &= -\Bigg[(2-q)\tilde{T}^{\delta}+\nu \tilde{T}^{\delta+n-2p}X^{2p}\Bigg]\Sigma_{\mathrm{\phi}}-n \tilde{T}^{1+\delta}X^{2n-1} ,\label{local2} \\
	\frac{d\tilde{T}}{d\tilde{N}}&= \frac{1}{n}(1+q)\tilde{T}^{1+\delta},\label{local3}
	\end{align}
\end{subequations}
where the constraint equation
\begin{equation}\label{constraintOmega}
\Omega_{\mathrm{pf}}=1-\Sigma_{\mathrm{\phi}}^2-X^{2n},
\end{equation}
is used to globally solve for $\Omega_\mathrm{pf}$. Since $\Omega_\mathrm{pf}>0$, the above equation implies that $\Omega_\mathrm{pf}$ is bounded as $\Omega_\mathrm{pf}\in(0,1)$, while $\Sigma_\phi\in(-1,1)$ and $X\in(-1,1)$. The positive dimensionless constant $\nu>0$, is given explicitly by
\begin{equation}\label{NuDef}
\nu=6^p\mu c^{-n}=\sqrt{n 6^{2p-(n-1)}}\frac{\mu}{\lambda^{n}},
\end{equation}
and $q$ is the usual \emph{deceleration parameter} defined via $\dot{H}=-(1+q)H^2$, i.e.
\begin{eqnarray}\label{deceleration}
q=  -1+3\Sigma_{\mathrm{\phi}}^2+\frac{3}{2}\gamma_{\mathrm{pf}}\Omega_{\mathrm{pf}}
= -1+\frac{3}{2}\left(\gamma_{\mathrm{\phi}}\Omega_{\mathrm{\phi}}+\gamma_{\mathrm{pf}}\Omega_{\mathrm{pf}}\right),
\end{eqnarray}
where we introduced the \emph{Hubble normalized scalar-field energy density}
\begin{equation}
\Omega_\phi = \frac{\rho_\phi}{3H^2}= \Sigma^2_\phi+X^{2n}=1-\Omega_\mathrm{pf}, \quad \Omega_\phi\in(0,1),
\end{equation}
and the \emph{scalar-field effective equation of state} $\gamma_{\mathrm{\phi}}$, defined by
\begin{equation} \label{gamma}
\gamma_{\mathrm{\phi}}:=1+\frac{p_{\mathrm{\phi}}}{\rho_{\mathrm{\phi}}}=\frac{\dot{\phi}^2}{\frac{1}{2}\dot{\phi}^2+\frac{1}{2n}(\lambda\phi)^{2n}}=\frac{2\Sigma_{\mathrm{\phi}}^2}{\Omega_\mathrm{\phi}}.
\end{equation}
Moreover, since $\gamma_{\mathrm{pf}} \in (0,2)$, it follows from~\eqref{constraintOmega} and  \eqref{deceleration} that 
\begin{equation}
-1<q<2,
\end{equation}
with limits $q=-1$ when $\Sigma_\phi=0$ and $\Omega_\mathrm{pf}=0$; $q=2$ when $X=0$ and $\Omega_\mathrm{pf}=0$; and $q=\frac{1}{2}(3\gamma_\mathrm{pf}-2)$ when $X=0$ and $\Sigma_\phi=0$. These special constant values of $q$ correspond to well-known solutions: (quasi-)de-Sitter (dS) solution when $q=-1$, kinaton or massless scalar field self-similar solution when $q=2$ and whose scale factor is given by $a(t)=t^{1/3}$, and the flat Friedmann-Lema\^itre (FL) self-similar solution when $q=\frac{1}{2}(3\gamma_\mathrm{pf}-2)$ with scale factor given by  $a(t)=t^{\frac{2}{3\gamma_\mathrm{pf}}}$.

Although the constraint \eqref{constraintOmega} is used to solve for $\Omega_\mathrm{pf}$, it is nevertheless useful to consider the auxiliary equation for $\Omega_{\mathrm{pf}}$ (equivalently $\Omega_\phi=1-\Omega_\mathrm{pf}$) given by
\begin{eqnarray} \label{localomega}
\frac{d\Omega_{\mathrm{pf}}}{d\tilde{N}} &=& 2\left(1+q-\frac{3}{2}\gamma_{\mathrm{pf}}\right)\tilde{T}^{\delta}\Omega_{\mathrm{pf}}+2\nu\tilde{T}^{\delta+n-2p}X^{2p}\Sigma_{\mathrm{\phi}}^2 \nonumber \\
&=& 3(\gamma_{\mathrm{\phi}}-\gamma_{\mathrm{pf}})\Omega_{\mathrm{pf}}(1-\Omega_{\mathrm{pf}})\tilde{T}^{\delta}+2\nu \tilde{T}^{\delta+n-2p}X^{2p}\Sigma_{\mathrm{\phi}}^2.
\end{eqnarray}
While the variables $(X,\Sigma_\phi)$ are bounded, the variable $\tilde{T}$ becomes unbounded ($\tilde{T}\rightarrow +\infty$) when $H\rightarrow0$. In order to obtain a \textit{regular} and \textit{global} $3$-dimensional dynamical system, we introduce
\begin{equation}\label{T}
T=\frac{\tilde{T}}{1+\tilde{T}},
\end{equation}
so that $T\in(0,1)$ with $T\rightarrow 0$ as $\tilde{T} \rightarrow 0$, and $T\rightarrow 1$ as $\tilde{T} \rightarrow +\infty$, as well as a new independent time variable $\tau$ defined by
\begin{equation}\label{newtime}
\frac{d}{d\tau}=(1-T)^k\frac{d}{d\tilde{N}}=\frac{T^{\delta}(1-T)^{k-\delta}}{H}\frac{d}{dt},
\end{equation}
where
\[k=\left\{
\begin{aligned}
n-2p+\delta  \quad&\quad \text{if} \quad  p<\frac{n}{2},  \\
1+\delta \quad&\quad\text{if}\quad  p\geq \frac{n}{2} .
\end{aligned}
\right.\]
This leads to a \emph{regular} and \emph{global} $3$-dimensional dynamical system for the state-vector $(X,\Sigma_\phi,T)$
\begin{subequations}\label{globalsys}
	\begin{align}
	\frac{dX}{d\tau} &= \frac{1}{n}(1+q)T^{\delta}(1-T)^{k-\delta}X+T^{1+\delta}(1-T)^{k-(1+\delta)}\Sigma_{\mathrm{\phi}}, \\
	\frac{d\Sigma_{\mathrm{\phi}}}{d\tau} &= -\Bigg[(2-q)T^{\delta}(1-T)^{k-\delta}+\nu T^{\delta+n-2p}(1-T)^{k-(\delta+n-2p)}X^{2p}\Bigg]\Sigma_{\mathrm{\phi}}-nT^{1+\delta}(1-T)^{k-(1+\delta)}X^{2n-1},\\
	\frac{dT}{d\tau} &= \frac{1}{n}(1+q)T^{1+\delta}(1-T)^{1+k-\delta}.
	\end{align}
\end{subequations}
where $q$ is given by \eqref{deceleration}. The auxiliary equation \eqref{localomega} written in terms of the new time variable $\tau$ becomes
\begin{equation}\label{MatEq}
\frac{d\Omega_{\mathrm{pf}}}{d\tau}=2\left(1+q-\frac{3}{2}\gamma_{\mathrm{pf}}\right)T^{\delta}(1-T)^{k-\delta}\Omega_{\mathrm{pf}}+2\nu T^{\delta+n-2p}(1-T)^{k-(\delta+n-2p)}X^{2p}\Sigma_{\mathrm{\phi}}^2.
\end{equation}  
The state space $\mathbf{S}$ is thus a 3-dimensional space consisting of a (deformed when $n>1$) open and bounded solid cylinder without its axis
\begin{equation}
{\mathbf{S}}=\{(X,\Sigma_\phi,T)\in\mathbb{R}^3\,: 0< X^{2n}+\Sigma^2_\phi< 1 ,\quad 0<T<1\}.
\end{equation}
The state space $\mathbf{S}$ can be regularly extended to include the axis of the cylinder with $\Omega_\mathrm{pf}=1$ ($\Omega_\phi=X^{2n}+\Sigma^2_\phi=0$) which is an invariant boundary subset as follows from~\eqref{MatEq}, and the outer shell of the cylinder which consists of the pure {\em scalar field boundary subset}, $\Omega_{\mathrm{pf}}=0$ ($\Omega_\phi=X^{2n}+\Sigma^2_\phi=1$). Due to the interaction term when $\nu\neq0$, the $\Omega_\mathrm{pf}=0$ boundary subset is  not invariant for the flow. Furthermore at $\Omega_\mathrm{pf}=0$ it follows that
\begin{subequations}
	\begin{align}
	\left.\frac{d\Omega_{\mathrm{pf}}}{d\tau}\right|_{\Omega_{\mathrm{pf}}=0}&= 2\nu T^{\delta+n-2p}(1-T)^{k-(\delta+n-2p)}X^{2p}\Sigma_{\mathrm{\phi}}^2 \geq 0,\qquad
	\left.\frac{d^2\Omega_{\mathrm{pf}}}{d\tau^2}\right|_{\Omega_{\mathrm{pf}}=0}=
	\left.\frac{d^3\Omega_{\mathrm{pf}}}{d\tau^3}\right|_{\Omega_{\mathrm{pf}}=0}=0, \\
	\left.\frac{d^4\Omega_{\mathrm{pf}}}{d\tau^4}\right|_{\Omega_{\mathrm{pf}}=0}&=6n^2(1-T)^{k-\delta-1}(3(1-T)^{3k-3\delta-1}T^{1+3\delta}\gamma_{\mathrm{pf}}^2+ \nonumber\\
	&+(1-T)^{2k-2\delta-1}T^{1+2\delta}\gamma_{\mathrm{pf}}(3(1-T)^{k-\delta}T^{\delta}+(1-T)^{k-n+2p-\delta}T^{n-2p-\delta}\nu))>0.
	\end{align}
\end{subequations}
Since $\nu>0$, this shows that the surface $\Omega_\mathrm{pf}=0$ not being invariant, it is future-invariant, which motivates the following definition:
\begin{definition}
\label{definition1}
	The orbits in $\mathbf{S}$ with initial data $\Omega_\mathrm{pf}(\tau_0)>0$ are said to be of class $\mathrm{B}$ if there is a finite $\tau_*<\tau_0$ such that $\Omega_\mathrm{pf}(\tau_*)=0$. The complement of such orbits in $\mathbf{S}$ are said to be of class $\mathrm{A}$.
\end{definition}
Class $\mathrm{B}$ orbits enter the state-space $\mathbf{S}$ by crossing the outer cylindrical shell with $\Omega_\phi=\Sigma^2_\phi+X^{2n}=1$. Moreover $\mathbf{S}$ can be regularly extended to include the invariant boundaries $T=0$ and $T=1$, which is essential, since all possible past attractor sets for class $\mathrm{A}$ orbits are located at $T=0$ and all possible future attractors for both class $\mathrm{A}$ and $\mathrm{B}$ orbits are located at $T=1$ as follows from the following lemma:
\begin{lemma}
	\label{lemma1}
	The $\alpha$-limit set of class $\mathrm{A}$ interior orbits in $\mathbf{S}$ is located at $\{T=0\}$, while the $\omega$-limit set of all interior orbits in $\mathbf{S}$ is located at $\{T=1\}$.
\end{lemma}
\begin{proof} 
	Since $1+q\leq0$, then $T$ is strictly monotonically increasing in the interval $(0,1)$, except when $q=-1$ in which case
	\begin{eqnarray}	
&&\left.\frac{dT}{d\tau}\right|_{q=-1}= 0, \qquad \left.\frac{d^2T}{d\tau^2}\right|_{q=-1} = 0, \qquad \nonumber \\ 
&&\left.\frac{d^3T}{d\tau^3}\right|_{q=-1}
= \frac{nT^{2+3\delta}}{(1-T)^{3(1+\delta-k)}}\Big(3n(k-\delta)+3(1-T)^{2-k}T^{1-\delta}\gamma_{\mathrm{pf}}+\nu nk(1-T)^{1+\delta-k}T^{k-\delta-1} \Big)>0. \nonumber
\end{eqnarray}
Therefore the points in $\mathbf{S}$ with $q=-1$ are just inflection points in the graph of $T(\tau)$. By the \emph{monotonicity principle}, it follows that there are no fixed points, recurrent or periodic orbits in the interior of the state space ${\bf S}$, and the $\alpha$-limit sets of class $\mathrm{A}$ orbits are contained at $\{T=0\}$ and $\omega$-limit sets of all orbits in ${\bf S}$ are located at $\{T=1\}$. 
\end{proof}
Thus the global behavior of both classes of orbits can be inferred by a complete detailed description of the invariant subsets $T=0$ and $T=1$, which are associated with the past and future limits $H \rightarrow + \infty$ and $H \rightarrow 0$, respectively. Due to their distinct properties, we split our analysis into three cases: $p<n/2$, $p=n/2$ and $p>n/2$.

\section{Dynamical systems' analysis when $p<\frac{n}{2}$} 
\label{case1}
When $p<\frac{n}{2}$ the global dynamical system \eqref{globalsys} takes the form
\begin{subequations}\label{globalng2p}
	\begin{align}
	\frac{dX}{d\tau} &= \left[\frac{1}{n}(1+q)(1-T)X+T\Sigma_{\mathrm{\phi}}\right](1-T)^{n-2p-1} ,\label{gl1} \\
	\frac{d\Sigma_{\mathrm{\phi}}}{d\tau} &=-\nu T^{n-2p}X^{2p}\Sigma_{\mathrm{\phi}} -\Bigg[(2-q)(1-T)\Sigma_{\mathrm{\phi}}+n TX^{2n-1}\Bigg](1-T)^{n-2p-1}, \label{gl2}\\
	\frac{dT}{d\tau} &= \frac{1}{n}(1+q)T(1-T)^{n-2p+1}, \label{gl3}
	\end{align}
\end{subequations}
where we recall  $q=-1+3\Sigma_{\mathrm{\phi}}^2+\frac{3}{2}\gamma_\mathrm{pf}\Omega_{\mathrm{pf}}$ and the auxiliary equation~\eqref{MatEq} becomes
\begin{equation}\label{Aux1}
\frac{d\Omega_{\mathrm{pf}}}{d\tau} = 2(1-T)^{n-2p}(1+q-\frac{3}{2}\gamma_{\mathrm{pf}})\Omega_{\mathrm{pf}}+2\nu T^{n-2p}X^{2p}\Sigma_{\mathrm{\phi}}^2.
\end{equation}
%
%
%
\subsection{Invariant boundary $T=0$}\label{T0,n-2p>1}
The induced flow on the $T=0$ invariant boundary is given by
\begin{equation}
	\frac{dX}{d\tau} = \frac{1}{n}(1+q)X ,\qquad 
	\frac{d\Sigma_{\mathrm{\phi}}}{d\tau} = -(2-q)\Sigma_{\mathrm{\phi}},
\end{equation}
where $q=-1+3\Sigma_{\mathrm{\phi}}^2+\frac{3}{2}\gamma_\mathrm{pf}\Omega_{\mathrm{pf}}$ and $\Omega_{\mathrm{pf}}=1-\Omega_{\mathrm{\phi}}=1-X^{2n}-\Sigma_{\mathrm{\phi}}^2$ satisfies 
\begin{equation}
\frac{d\Omega_{\mathrm{pf}}}{d\tau} = 2\left(1+q-\frac{3}{2}\gamma_{\mathrm{pf}}\right)\Omega_{\mathrm{pf}}.
\end{equation}
Thus
\begin{equation}
\left.\frac{d\Omega_{\mathrm{pf}}}{d\tau}\right|_{\Omega_{\mathrm{pf}}=0}=0, \qquad \left.\frac{d\Omega_{\mathrm{pf}}}{d\tau}\right|_{\Omega_{\mathrm{pf}}=1}=0,
\end{equation}
and the intersection of the $T=0$ invariant boundary with the pure perfect fluid subset $\Omega_{\mathrm{pf}}=1$, i.e. the axis of the cylinder, consists of the fixed point
\begin{equation}
\mathrm{FL}_{\mathrm{0}}:\quad X=0,\quad\quad \Sigma_{\mathrm{\phi}}=0,\quad\quad T=0,
\end{equation}
where $q=\frac{1}{2}(3\gamma_{\mathrm{pf}}-2)$, corresponding to the flat FL self-similar solution. The linearisation around $\mathrm{FL}_0$ yields the eigenvalues $\frac{3}{2n}\gamma_{\mathrm{pf}}$, $-\frac{3}{2}(2-\gamma_{\mathrm{pf}})$ and $\frac{3}{2n}\gamma_{\mathrm{pf}}$ with eigenvectors the canonical basis of $\mathbb{R}^3$. Since $0<\gamma_\mathrm{pf}<2$, $\mathrm{FL}_{\mathrm{0}}$ has two positive real eigenvalues and a negative real eigenvalue, being a hyperbolic saddle, and the $\alpha$-limit point of a 1-parameter set of class A orbits in $\mathbf{S}$. 

On the intersection of the invariant boundary $T=0$ with the  subset $\Omega_\mathrm{pf}=0$ there are two equivalent fixed points
\begin{equation}
\mathrm{K}^{\mathrm{\pm}}: \quad X=0,\quad\quad \Sigma_{\mathrm{\phi}}=\pm 1,\quad\quad T=0,
\end{equation}
with $q=2$ corresponding to the self-similar massless scalar field or kinaton solution. The linearisation of the full system around these fixed points yields the eigenvalues $\frac{3}{n}$, $3(2-\gamma_{\mathrm{pf}})$, and $\frac{3}{n}$, with generalised eigenvectors $(1,0,0)$, $(0,1,0)$ and $(\mp1,0,1)$. 
Since $\gamma_{\mathrm{pf}} \in (0,2)$, then $\mathrm{K}^{\mathrm{\pm}}$ are hyperbolic sources and the $\alpha$-limit points of a 2-parameter set of class A orbits in $\mathbf{S}$. 

Finally there are other two equivalent fixed points given by
\begin{equation}
\mathrm{dS^{\pm}_0}: \quad X=\pm 1,\quad\quad \Sigma_{\mathrm{\phi}}=0, \quad\quad T=0,
\end{equation}
that correspond to quasi-de-Sitter states with $q=-1$. The linearisation around $\mathrm{dS^{\pm}_0}$ yields the eigenvalues $-3\gamma_{\mathrm{pf}}$, $-3$ and $0$ with eigenvectors $(1,0,0)$, $(0,1,0)$ and $(0,\mp \frac{n}{3},1)$ respectively. The fixed points $\mathrm{dS^{\pm}_0}$ have two negative real eigenvalues (since $\gamma_\mathrm{pf}>0$) and a zero eigenvalue corresponding to a center manifold. Due to the monotonicity of $T$ it is clear that a single class A  orbit originates from each $\mathrm{dS}^\pm_0$ into $\mathbf{S}$ corresponding to a 1-dimensional center manifold. This center manifold corresponds to what is usually called in the physics literature the  \emph{inflationary attractor solution}, see e.g.~\cite{Alho,Alho2} and references therein. In order to simplify the analysis of the center manifold we use instead system \eqref{sistemalocal} for the unbounded  variable $\tilde{T}$, and introduce the adapted variables
\begin{equation}
\bar{X} = X \mp 1, \qquad \bar{\Sigma}_\phi= \Sigma_{\mathrm{\phi}} \pm\frac{n}{3}\tilde{T}.
\end{equation}
This leads to the transformed adapted system
\begin{equation}
\frac{d\bar{X}}{dN} = -3\gamma_\mathrm{pf} \bar{X}+F(\bar{X},\bar{\Sigma}_\phi,\tilde{T}),\qquad 
\frac{d\bar{\Sigma}_\phi}{dN} = -3\bar{\Sigma}_\phi+ G(\bar{X},\bar{\Sigma}_\phi,\tilde{T}),\qquad 
\frac{d\tilde{T}}{dN} = N(\bar{X},\bar{\Sigma}_\phi,\tilde{T}),
\end{equation}
where the fixed points $\mathrm{dS}^\pm_0$ are now located  at the origin of coordinates $(\bar{X},\bar{\Sigma}_\phi,\tilde{T})=(0,0,0)$, and $F$, $G$ and $N$ are functions of higher order terms.  
The 1-dimensional center manifold $W^{\mathrm{c}}$ at $\mathrm{dS}_\pm$ can be locally represented as the graph  $h\,: \,E^c\rightarrow E^s$, i.e. $(\bar{X},\bar{\Sigma}_\phi)=(h_1(\tilde{T}),h_2(\tilde{T}))$, satisfying the fixed point $h(0)=0$ and  the tangency $\frac{dh(0)}{d\tilde{T}}=0$ conditions (see e.g. \cite{Carr}). Using this in the above equation and using $\tilde{T}$ as an independent variable, we get
\begin{subequations}
	\begin{align}
	&\frac{1}{n}(1+q)\Big(\frac{dh_1}{d\tilde{T}}(\tilde{T})\tilde{T}-h_1(\tilde{T})\mp 1\Big)-\tilde{T}\Big(h_2(\tilde{T})\mp \frac{n}{3}\tilde{T}\Big)=0, \label{Xcentereq}\\
	&\frac{1}{n}(1+q)\tilde{T}\Big(\frac{dh_2}{d\tilde{T}}(\tilde{T})\mp \frac{n}{3}\Big)+(2-q)\Big(h_2(\tilde{T})\mp \frac{n}{3}\tilde{T}\Big)+ \nonumber \\
	&\qquad\qquad+\nu \tilde{T}^{n-2p}\Big(h_1(\tilde{T})\pm 1\Big)^{2p}\Big(h_2(\tilde{T})\mp \frac{n}{3}\tilde{T}\Big)+n \tilde{T}\Big(h_1(\tilde{T})\pm 1\Big)^{2n-1}=0,\label{Sigcentereq}
	\end{align}
\end{subequations}
where $q=-1+3\Big(h_2(\tilde{T})\mp \frac{n}{3}\tilde{T}\Big)^2+\frac{3}{2}\gamma_{\mathrm{pf}}\Big(1-(h_1\big(\tilde{T})\pm 1\big)^{2n}-\big(h_2(\tilde{T})\mp \frac{n}{3}\tilde{T}\big)^{2}\Big)$. The problem of finding the inflationary  attractor solution amounts to solving the previous system of non-linear ordinary differential equations. Although in general the existence of an explicit solution for the above system is not expected, it is possible to approximate the solution by a formal truncated power series expansion in $\tilde{T}$: 
\begin{equation}\label{powerseries}
	h_1(\tilde{T}) = \sum_{i=1}^{N}a_i \tilde{T}^{i} , \qquad
h_2(\tilde{T}) = \sum_{i=1}^{N}b_i \tilde{T}^{i}, \qquad \text{as}\quad \tilde{T}\rightarrow 0,
\end{equation}
with $a_i, b_i\in\mathbb{R}$. Plugging \eqref{powerseries} into~\eqref{Xcentereq}-\eqref{Sigcentereq}, using the expansions $(\bar{X}\pm 1)^{2n}= 1 \pm 2n \bar{X} +\binom{2n}{2n-2}\bar{X}^2+\dots$, $\tilde{T}^{n-2p}=\tilde{T} \delta^{n-2p}_1+\tilde{T}^2 \delta^{n-2p}_2 +\dots$, where $\delta^{a}_{b}$ is the Kronecker delta symbol, and solving the resulting linear system of equations for the coefficients of the expansions yields, as $\tilde{T}\rightarrow0$,
\begin{subequations}
	\begin{align}
	X&=\pm 1 \mp \frac{n}{18}\tilde{T}^2\pm \left(\frac{n^2}{648}(5-2n)+\frac{\nu n}{27 \gamma_\mathrm{pf}}\delta^{n-2p}_1\right)\tilde{T}^4+\mathcal{O}(\tilde{T}^6),\\
	\Sigma_\mathrm{\phi}&=\mp \frac{n}{3}\tilde{T}\Big[1 \mp \left(\frac{n}{18}+\frac{\nu}{3}\delta^{n-2p}_1\right)\tilde{T}^2 \pm \\
	&\pm\left(\frac{n^2}{648}(17-6n)-\frac{\nu}{3}\left(\delta^{n-2p}_2-\frac{n(2-4n+(7+2p)\gamma_\mathrm{pf})+6\gamma_\mathrm{pf}}{18\gamma_\mathrm{pf}}\delta^{n-2p}_1\right)\right)\tilde{T}^4+\mathcal{O}(\tilde{T}^6)\Big],\\
	\Omega_\mathrm{pf}&= \frac{2\nu n^2}{27 \gamma_\mathrm{pf}} \delta^{n-2p}_1 \tilde{T}^4+\mathcal{O}(\tilde{T}^6). \label{CMOmega}
	\end{align}
\end{subequations}
Therefore, it follows that to leading order on the center manifold
\begin{equation}
\frac{d\tilde{T}}{dN}=\frac{n}{3}\tilde{T}^3+\mathcal{O}(\tilde{T}^4)\quad \text{as}\quad \tilde{T}\rightarrow 0,
\end{equation}
which shows explicitly that $\mathrm{dS}^\pm_0$ are center saddles with a unique class A center manifold orbit originating from each fixed point into the interior of $\mathbf{S}$. 

We now show that on $T=0$ the above fixed points are the only possible $\alpha$-limit sets for class A orbits in $\mathbf{S}$, and that the orbit structure on $T=0$ is very simple:

\begin{lemma}\label{LemT0_ng2p}
Let $p<\frac{n}{2}$. Then the $\{T=0\}$ invariant boundary consists of heteroclinic orbits connecting the fixed points as depicted in Figure~\ref{fig:T0_n-2p>0}.
\end{lemma}
\begin{proof}
	It is straightforward to check that $\{\Sigma_\phi=0\}$ and $\{X=0\}$ are invariant 1-dimensional subsets consisting of heteroclinic orbits $\mathrm{FL}_0\rightarrow\mathrm{dS}^\pm_0$ and $\mathrm{K}^\pm\rightarrow\mathrm{FL}_0$ respectively.  Therefore the two axis divide the (deformed) circle with boundary $X^{2n}+\Sigma^2_\phi=1$, consisting of the heteroclinic orbits $\mathrm{K^+}\rightarrow\mathrm{dS}^\pm_0$ and $\mathrm{K^-}\rightarrow\mathrm{dS}^\pm_0$, into $4$-invariant quadrants. On each quadrant there are no interior fixed points and hence, by the \emph{index theorem}, no periodic orbits. Since closed saddle connections do not exist, it follows by the \emph{Poincar\'e-Bendixson theorem} that each quadrant consists of heteroclinic orbits connecting the fixed points. Moreover, in this case, the $\{T=0\}$ invariant boundary admits the following conserved quantity
	\begin{equation}\label{conserved1}
	\Sigma_{\mathrm{\phi}}^{\gamma_{\mathrm{pf}}}X^{(2-\gamma_{\mathrm{pf}})n}\Omega_{\mathrm{pf}}^{-1}=\text{const.},\quad \Omega_\mathrm{pf}=1-\Sigma^2_\phi-X^{2n},
	\end{equation}
which determines the solution trajectories on $T=0$.
\end{proof}
\begin{figure}[ht!]
	\begin{center}
		\subfigure[$(p,n)=(0,1)$.]{\label{fig:T=0,n=1}
		\includegraphics[width=0.30\textwidth]{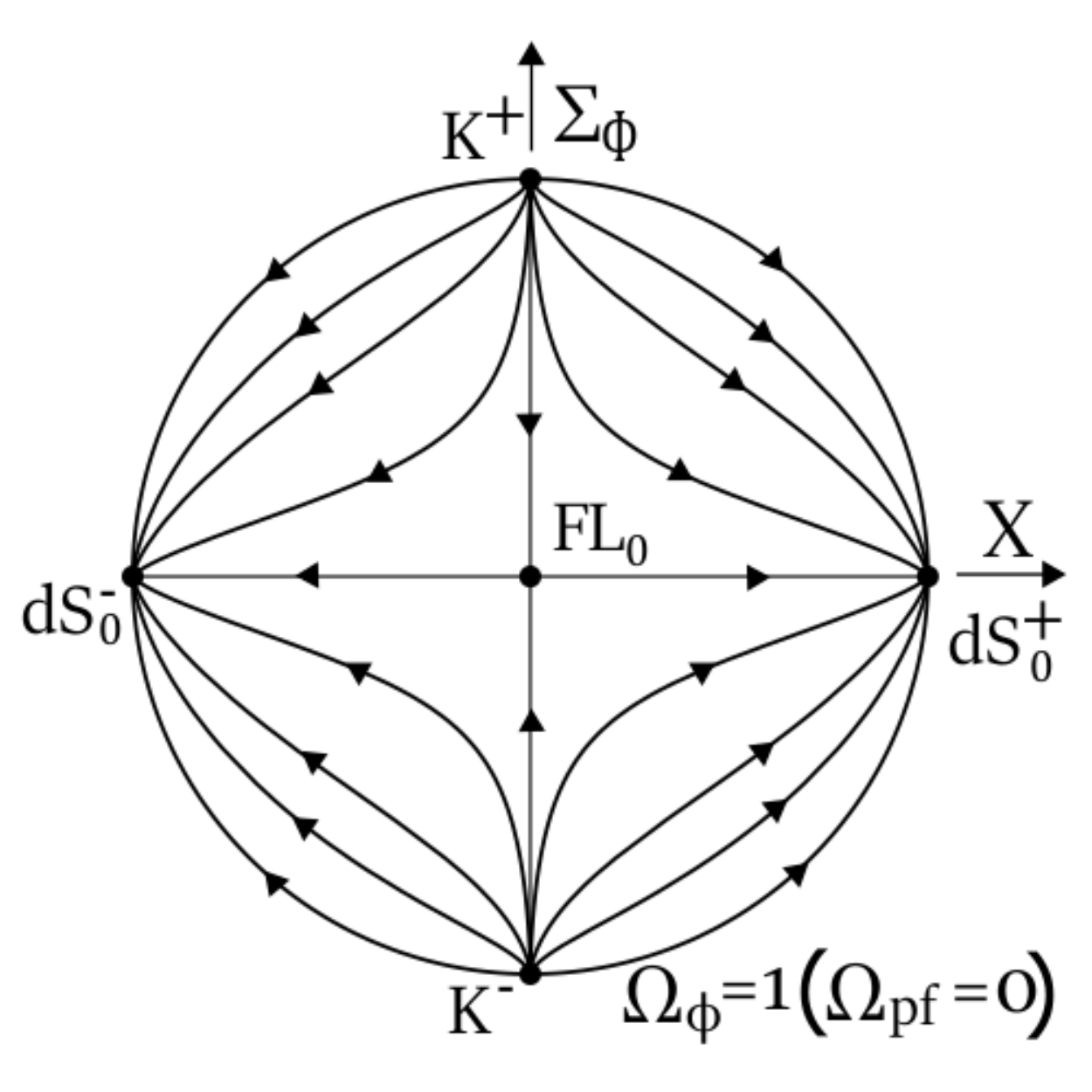}}
		\hspace{2cm}
		\subfigure[$(p,n)=(1,3)$.]{\label{fig:T=0,n=3}
		\includegraphics[width=0.30\textwidth]{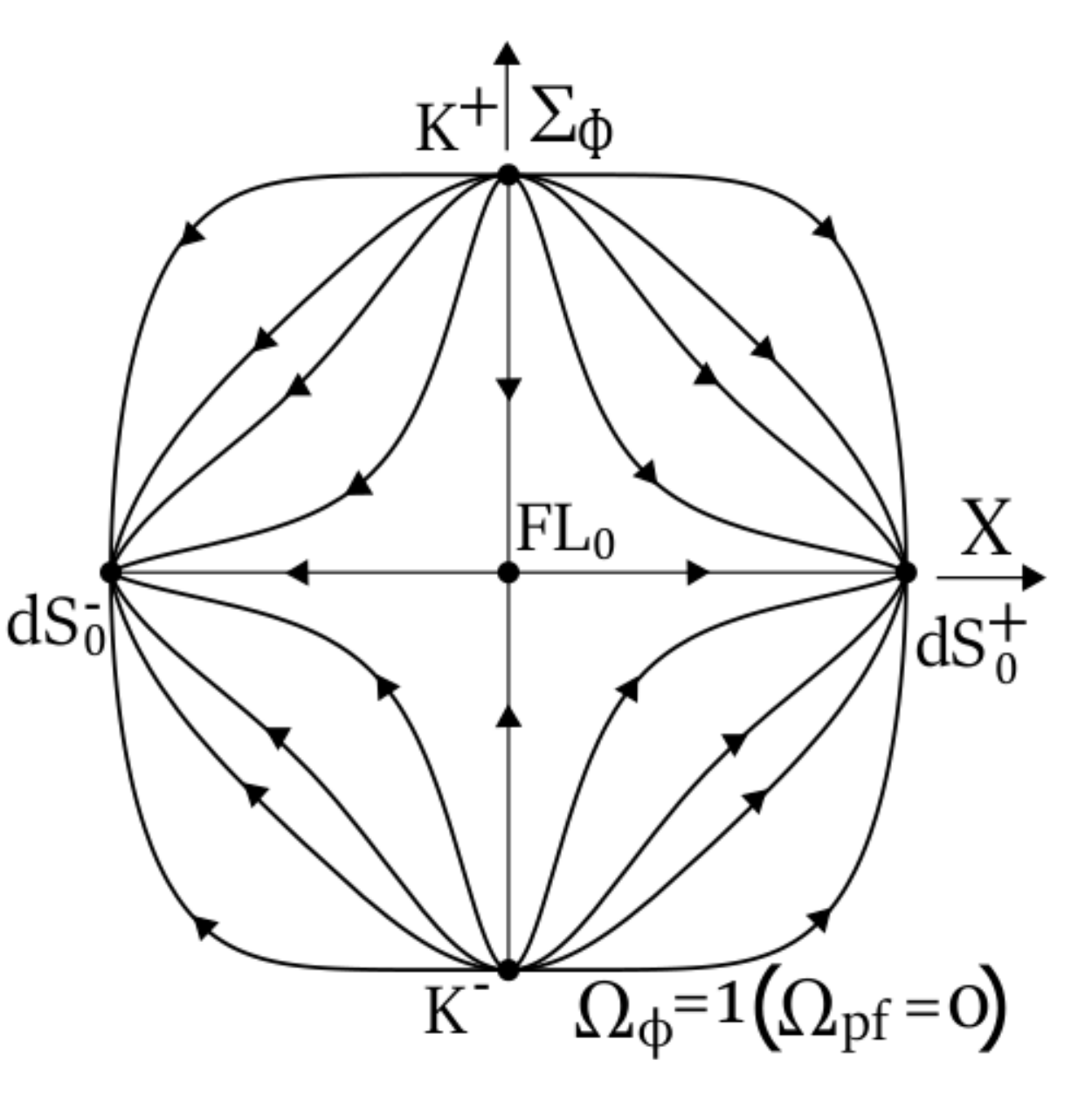}}
	\end{center}
	\vspace{-0.5cm}
	\caption{Invariant boundary $\{T=0\}$ when $p<\frac{n}{2}$.}
	\label{fig:T0_n-2p>0}
\end{figure}
\begin{theorem}
	Let $p<\frac{n}{2}$. Then the $\alpha$-limit set for class $\mathrm{A}$ orbits in $\mathbf{S}$, consists of fixed points on $\{T=0\}$. In particular as $\tau\rightarrow-\infty$ ($N\rightarrow-\infty$), a 2-parameter set of orbits converge to each $\mathrm{K}^\pm$, with asymptotics
	\begin{equation}
	X(N)=(C_X\pm C_T N) e^{\frac{3}{n}N},\qquad \Sigma_\phi(N) =\pm 1\mp C_\Sigma e^{3(2-\gamma_\mathrm{pf})N},\qquad \tilde{T}(N)=C_T e^{\frac{3}{n}N}
	\end{equation}
	with $C_X$, $C_\Sigma>0$, and $C_T>0$ constants. A 1-parameter set of orbits converges to $\mathrm{FL}_0$ with asymptotics
	 \begin{equation}
	 X(N)=C_X e^{\frac{3\gamma_\mathrm{pf}}{2n}N},\qquad \Sigma_\phi(N)=0,\qquad 
	 \tilde{T}(N)=C_T e^{\frac{3\gamma_\mathrm{pf}}{2n}N}
	 \end{equation}
	 with $C_X$ and $C_T>0$ constants, and a unique center manifold orbit converges to each $\mathrm{dS}^\pm_0$ with asymptotics
	 \begin{equation}
	 \label{asymp-var}
	  X  = \pm 1 \mp \frac{n}{18}\left(1-\frac{2n}{3} N\right)^{-1},\qquad 
	  \Sigma_{\mathrm{\phi}}= \mp \frac{n}{3}\left(1-\frac{2n}{3} N\right)^{-1/2}, \qquad \tilde{T}(N)=\left(1-\frac{2n}{3} N\right)^{-1/2}.
	 \end{equation}
	 When $p=\frac{1}{2}(n-1)$ we also get from~\eqref{CMOmega} the asymptotics $\Omega_\mathrm{pf} = \frac{2\nu n^2 }{27\gamma_{\mathrm{pf}}}\left(1-\frac{2n}{3} N\right)^{-2}$.
	 For $p<\frac{1}{2}(n-1)$ one needs to go higher order on the center manifold of $\Omega_\mathrm{pf}$.
\end{theorem}
\begin{proof}
	The proof follows from Lemmas~\ref{lemma1} and~\ref{LemT0_ng2p}, and the above local analysis of the fixed points.
\end{proof}
\begin{remark}\label{PastAs1}
	The solutions of class $\mathrm{A}$ which approach $\mathrm{K}^{\pm}$ behave asymptotically as the self-similar massless scalar field or kinaton solution, and the ones approaching $\mathrm{FL}_0$ as the self-similar Friedmann-Lema\^itre solution whose asymptotics towards the past exhibit well-known (big-bang) singularities. In the context of early cosmological inflation the physical interesting solution is the center manifold originating from each $\mathrm{dS}^\pm_0$ whose asymptotics for the variables $(H,\phi,\rho_\mathrm{pf})$ of the original system~\eqref{sistema1}, are given by
	\begin{subequations}	
		\begin{align*}
		&n=1\,:\qquad H\sim -t\quad,\quad \phi\sim-t\quad,\quad \rho_\mathrm{pf} \sim (-t)^{-2}, \quad \quad\text{as}\quad t\rightarrow -\infty \\
		&n=2\,:\qquad H\sim e^{-\frac{2}{3}t}\quad,\quad \phi\sim e^{-\frac{t}{3}}\quad,\quad \rho_\mathrm{pf} \sim e^{\frac{4}{3}t}, \quad \quad\text{as}\quad t\rightarrow -\infty \\
		&n\geq 3\,:\qquad H\sim (-t)^{\frac{n}{n-2}}\quad,\quad \phi\sim (-t)^{\frac{2}{n-2}}\quad,\quad \rho_\mathrm{pf} \sim (-t)^{-\frac{2n(2p+1)}{n-2}}, \quad \quad\text{as}\quad t\rightarrow -\infty.
		\end{align*}
	\end{subequations}
 with $p<\frac{n}{2}$.
\end{remark}
%
\subsection{Invariant boundary $T=1$}
On the $T=1$ invariant boundary, the system \eqref{globalng2p} reduces to
\begin{equation}\label{ng2p1}
\frac{dX}{d\tau}=\Sigma_{\mathrm{\phi}}\delta_{1}^{n-2p},\qquad
\frac{d\Sigma_{\mathrm{\phi}}}{d\tau}=-nX^{2n-1}\delta_{1}^{n-2p}-\nu X^{2p}\Sigma_{\mathrm{\phi}},
\end{equation}
and the auxiliary equation~\eqref{Aux1} for $\Omega_\mathrm{pf}$, satisfies
\begin{equation}
\frac{d\Omega_{\mathrm{pf}}}{d\tau}=2\nu X^{2p}\Sigma^2_\phi.
\end{equation}
The analysis can be divided into two particular sub-cases: $p<\frac{1}{2}(n-1)$, i.e. $(p,n)=(0,2),(0,3),...$, $(p,n)=(1,4),(1,5),...$, $etc.$,
and $p=\frac{1}{2}(n-1)$, i.e. $(p,n)=(0,1),(1,3),(2,5),...$.  
%
\subsubsection{Case $p<\frac{1}{2}(n-1)$}
In this case for all $p\geq0$ there is a line of fixed points
\begin{equation}
\mathrm{L}_{1}: \quad X=X_0, \quad\quad \Sigma_{\mathrm{\phi}}=0, \quad\quad T=1,
\end{equation}
parameterised by $X_0\in [-1,1]$. In addition to $\mathrm{L}_1$ there exists another line of fixed points when $p>0$,
\begin{equation}
\mathrm{L}_{2}: \quad X=0, \quad\quad \Sigma_{\mathrm{\phi}}=\Sigma_{\phi 0}, \quad\quad T=1,
\end{equation}
parameterised by $\Sigma_{\phi 0}\in [-1,1]$.
We shall refer to the non-isolated fixed point at the origin of the $T=1$ invariant set as $\mathrm{FL}_1$, and the end points of $\mathrm{L}_1$ with $X=\pm1$ as $\mathrm{dS}^{\pm}_{1}$. The description of the induced flow on $T=1$ when $p<\frac{1}{2}(n-1)$ is given by the following simple lemma:
\begin{lemma}\label{T1:p<(n-1)/2}
	When $p<\frac{1}{2}(n-1)$, the set $\{T=1\}\setminus \mathrm{L}_1$ for $p=0$, and the set $\{T=1\}\setminus \{\mathrm{L}_1\cup\mathrm{L}_2\}$ for $p>0$ are foliated by invariant subsets $X=\text{const.}$ consisting of regular orbits which enter the region $\Omega_{\mathrm{pf}}>0$ by crossing the set $\Omega_\mathrm{pf}=0$ and converging to the line of fixed points $\mathrm{L}_\mathrm{1}$, as depicted in Figure~\ref{figT1:n-2p-1>0}.
\end{lemma}
\begin{proof}
	When $p<\frac{1}{2}(n-1)$, the system \eqref{ng2p1} admits the following conserved quantity 
	\begin{equation}
	X=\text{const.}
	\end{equation}
	which determines the solution trajectories on the $T=1$ invariant boundary. The remaining properties of the flow follow from the fact that on $\{T=1\}\setminus \mathrm{L}_1$ for $p=0$, and on $\{T=1\}\setminus \{ \mathrm{L}_1\cup\mathrm{L}_2\}$ for $p>0$, we have $d\Sigma_\phi/d\tau<0$ and $d\Omega_\mathrm{pf}/d\tau<0$.
\end{proof}
\begin{figure}[ht!]
	\begin{center}
		\subfigure[$p=0$.]{\label{fig:T=1,n-2p>1,p=0}	\includegraphics[width=0.30\textwidth]{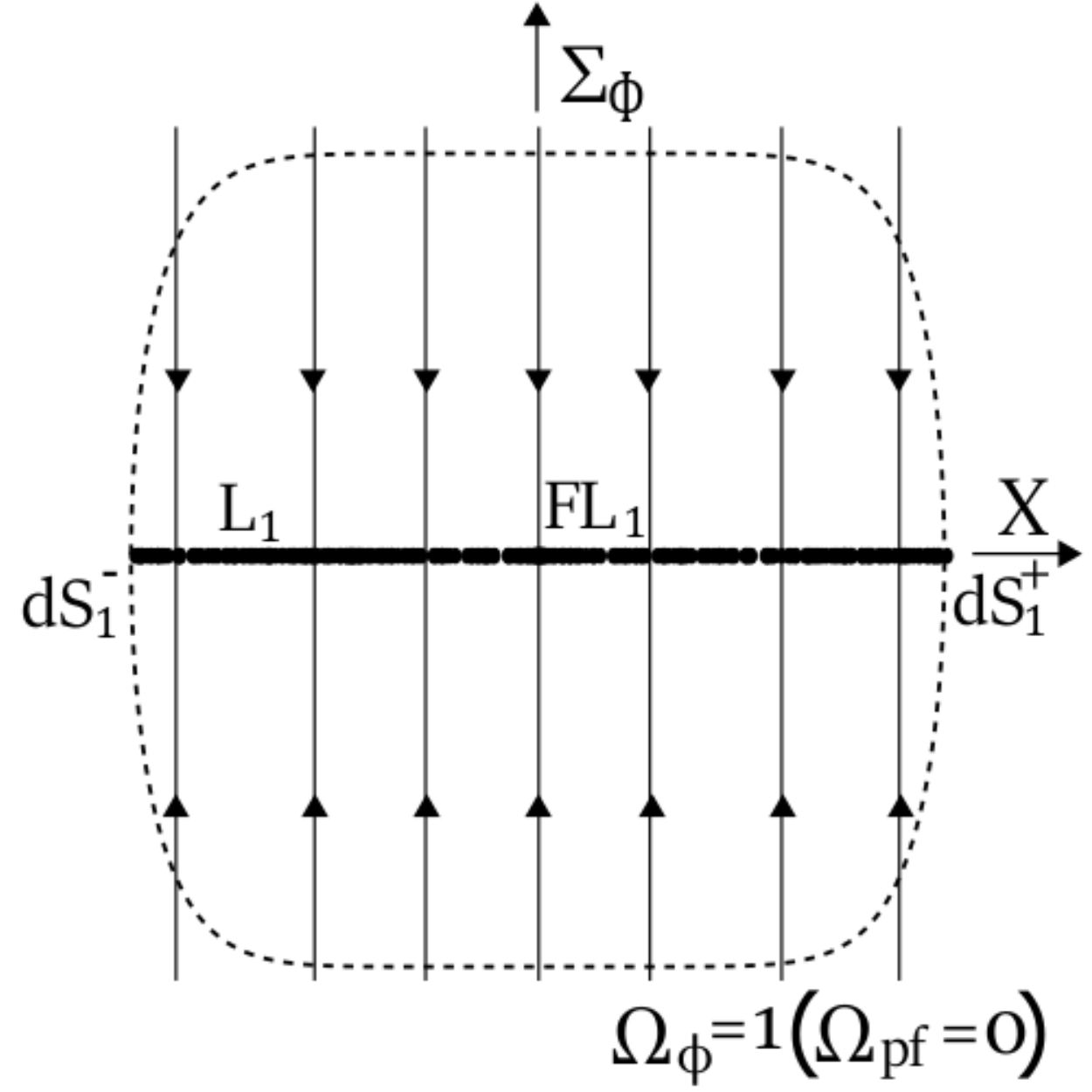}}
		\hspace{2cm}
		\subfigure[$p>0$.]{\label{fig:T=1,n-2p>1,p>0}	\includegraphics[width=0.30\textwidth]{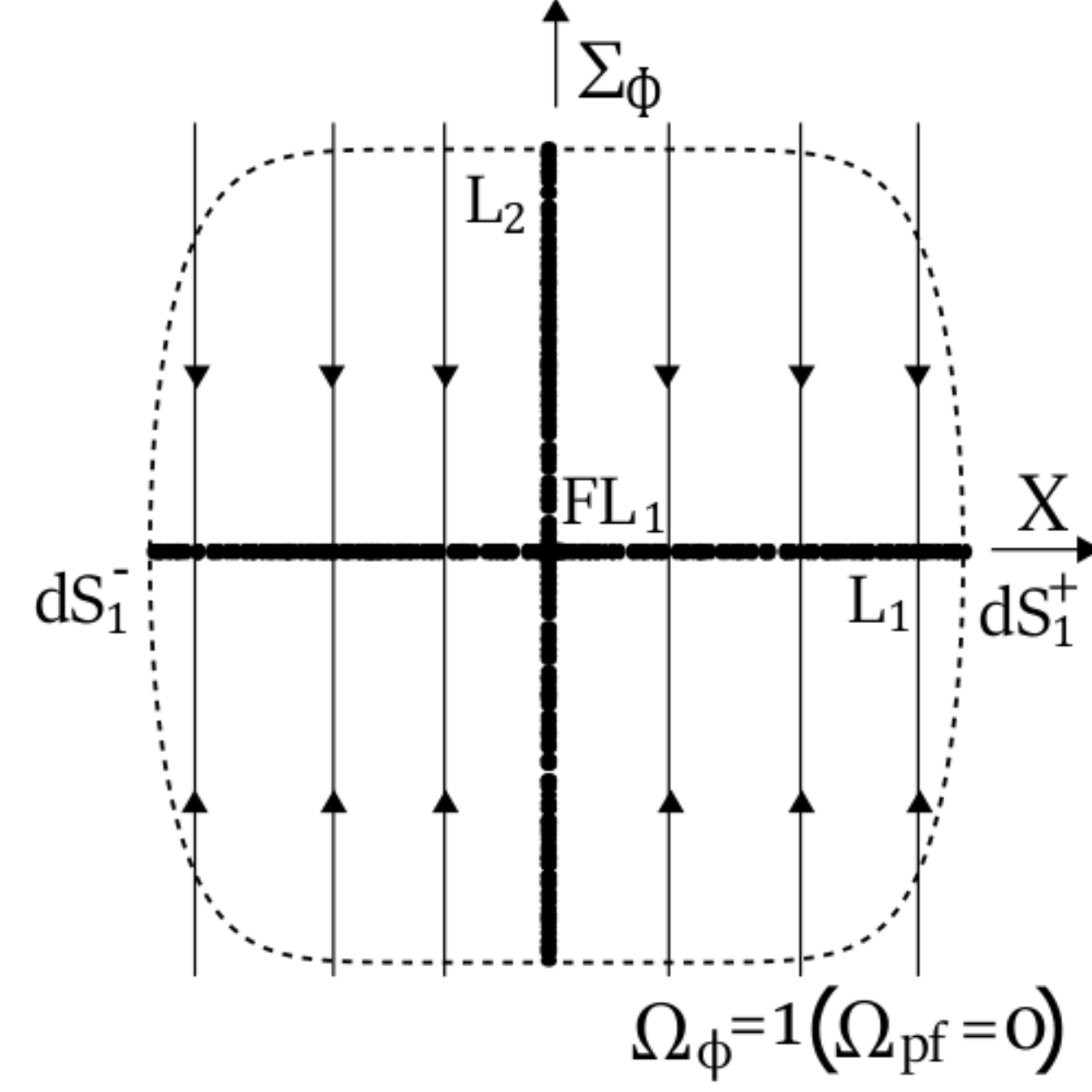}}
	\end{center}
	\vspace{-0.5cm}
	\caption{Invariant boundary $\{T=1\}$ when $p<\frac{1}{2}(n-1)$.}
	\label{figT1:n-2p-1>0}
\end{figure}
\begin{theorem}
	Let $p<\frac{1}{2}(n-1)$. Then the $\omega$-limit set of all orbits in $\mathbf{S}$ is contained on $\mathrm{L}_1$. In particular:
	\begin{itemize}
		\item[i)] If $p<\frac{1}{2}(n-2)$, then as $\tau\rightarrow+\infty$, a 2-parameter set of orbits converges to each of the two fixed points $\mathrm{dS}^{\pm}_1$ on the line $\mathrm{L}_1$ with $X_0 = \pm 1$, and a 1-parameter set of orbits converges to $\mathrm{FL}_1$ with $X_0=0$. 
		\item[ii)] If $p=\frac{1}{2}(n-2)$, then as $\tau\rightarrow+\infty$, a 2-parameter set of orbits converges to each of the two fixed points $\mathrm{S}^{\pm}$ on the line $\mathrm{L}_1$ with
		\begin{equation*}
		\mathrm{S}^{\pm}: \quad X_0 = \pm  \left(\frac{n^2}{3\gamma_{\mathrm{pf}}\nu}\left(-1+\sqrt{1+\left(\frac{3\gamma_{\mathrm{pf}}\nu}{n^2}\right)^2}\right)\right)^{1/n},
		\end{equation*}
		and a 1-parameter set of orbits converges to $\mathrm{FL}_1$ with $X_0=0$.
	\end{itemize}
\end{theorem}

\begin{proof}
	The first statement follows from lemmas~\ref{lemma1} and~\ref{T1:p<(n-1)/2} if $p=0$, while for $p>0$, it is shown in Lemma~\ref{L2BU}, Subsection~\ref{BUP1st} (by doing a cylindrical blow-up of $\mathrm{L}_2$ on top of the blow-up of $\mathrm{FL}_1$), that no solution trajectories converge to the set $\mathrm{L}_2\setminus \mathrm{FL}_1$. 

	The linearised system around $\mathrm{L}_{1}$ has eigenvalues $0$, $-\nu X_{0}^{2p}$ and $0$, with associated eigenvectors $(1,0,0)$, $(0,1,0)$ and $(0,-\frac{2}{\nu}(p-\frac{1}{2}(n-1)) X_{\mathrm{0}}^{2(n-p)-1}\delta_{2}^{n-2p},1)$. On the $\{T=1\}$ invariant boundary the line of fixed points $\mathrm{L}_{1}$ is normally hyperbolic, i.e. the linearisation yields one negative eigenvalue for all $X_0\in[-1,1]$, except at $X_0=0$ when $p>0$ (where the two lines intersect), and one zero eigenvalue with eigenvector tangent to the line itself, see e.g.~\cite{Aul84}. 
	On the closure of $\mathbf{S}$, the line $\mathrm{L}_1$ is said to be \emph{partially hyperbolic}~\cite{Tak71}. 
	Each fixed point on the line, including the  point at the center when $p=0$, has a 1-dimensional stable manifold and a 2-dimensional center manifold, while the point with $X_0=0$ is non-hyperbolic for $p>0$. In this case the blow-up of $\mathrm{FL}_1$ is done in Section~\ref{BUP1st}.
	
	 To analyse the 2-dimensional center manifold of each partially hyperbolic fixed point on the line, we start by making the change of coordinates given by
	\begin{equation}
	\bar{X}=X-X_0,\qquad \bar{\Sigma}_\phi=\Sigma_\phi+ \frac{2n}{\nu}(p-\frac{1}{2}(n-1))X_0^{2(n-p)-1}(1-T)\delta^{n-2p}_2 ,\qquad \bar{T}=1-T,
	\end{equation}
	which takes a point in the line $\mathrm{L_1}$ to the origin $(\bar{X},\bar{\Sigma}_{\phi},\bar{T})=(0,0,0)$ with $\bar{T}\geq0$. The resulting system of equations takes the form
	\begin{equation}\label{centereq}
	\frac{d\bar{X}}{d\tau} = F(\bar{X},\bar{\Sigma}_{\mathrm{\phi}},\bar{T}),\qquad 
	\frac{d\bar{\Sigma}_{\mathrm{\phi}}}{d\tau} = -\nu X_{\mathrm{0}}^{2p}\bar{\Sigma}_{\mathrm{\phi}}+G(\bar{X},\bar{\Sigma}_{\mathrm{\phi}},\bar{T}),\qquad
	\frac{d\bar{T}}{d\tau} =  N(\bar{X},\bar{\Sigma}_{\mathrm{\phi}},\bar{T}),
	\end{equation}
	where $F$, $G$ and $N$ are functions of higher order. The center manifold reduction theorem yields that the above system is locally topological equivalent to a decoupled system on the 2-dimensional center manifold, which can be locally represented as the graph $h:\, E^{c}\rightarrow E^{s}$ with $\bar{\Sigma}_{\mathrm{\phi}}=h(\bar{X},\bar{T})$ which solves
	the nonlinear partial differential equation
	\begin{equation}\label{FlowCMLine}
	F(\bar{X},h(\bar{X},\bar{T}),\bar{T})\partial_{\bar{X}}  h(\bar{X},\bar{T})+N(\bar{X},h(\bar{X},\bar{T}),\bar{T}) \partial_{\bar{T}} h(\bar{X},\bar{T})= -\nu X_0^{2p} h(\bar{X},\bar{T}) + G(\bar{X},h(\bar{X},\bar{T}),\bar{T})
	\end{equation}
	subject to the fixed point and tangency conditions $h(0,0)=0$, and $\nabla h(0,0)=0$, respectively. A quick look at the nonlinear terms suggests that we approximate the center manifold at $(\bar{X},\bar{T})=(0,0)$, by making a formal multi-power series expansion for $h$ of the form
	\begin{equation}
	h(\bar{X},\bar{T})=\bar{T}^{n-2p}\sum^{N}_{i,j=0} \tilde{a}_{ij}\bar{X}^{i}\bar{T}^{j},\qquad \tilde{a}_{ij}\in\mathbb{R}.
	\end{equation}
	Solving for the coefficients it is easy to verify that $\tilde{a}_{i0}$ are identically zero, so that $h$ can be written as a series expansion in $\bar{T}$ with coefficients depending on $\bar{X}$, i.e.	
	\begin{equation}
	h(\bar{X},\bar{T})=\bar{T}^{n-2p}\sum^{N}_{j=1} \bar{a}_{j}(\bar{X})\bar{T}^{j}, \qquad \bar{a}_j(\bar{X})=\sum^{N}_{i=0}a_{ij}\bar{X}^{i},\qquad a_{ij}\in\mathbb{R},
	\end{equation}
	where for example
	\begin{subequations}
		\begin{align}
		a_{01} &= 0,\qquad a_{11}=-\frac{(2(n-p)-1)!}{\nu(n-1)!(n-2p-1)!}X_0^{n}, \nonumber\\
		a_{02} &=\frac{n}{\nu}X_0^{2(n-p)-1} ,\qquad a_{12}= -\frac{n(2n-2p-1)}{\nu}X_0^{2(n-p-2)}, \nonumber \\
		a_{03} &=-\frac{n(n-2p-1)}{\nu}X_0^{2(n-p)-1}+\frac{\left(3(n-1)!\nu-(2(n-p)-1)!X_0^{2(n-p-1)}\right)}{(n-1)!\nu^3}X_0^{n-2p+1}\delta^{n-2p}_2. \nonumber
		\end{align}
	\end{subequations}
	\begin{figure}[ht!]
		\begin{center}
			\subfigure[$b_{01}>0$. For $b_{01}<0$ the direction of the flow is reversed.]{\label{center4}
				\includegraphics[width=0.30\textwidth]{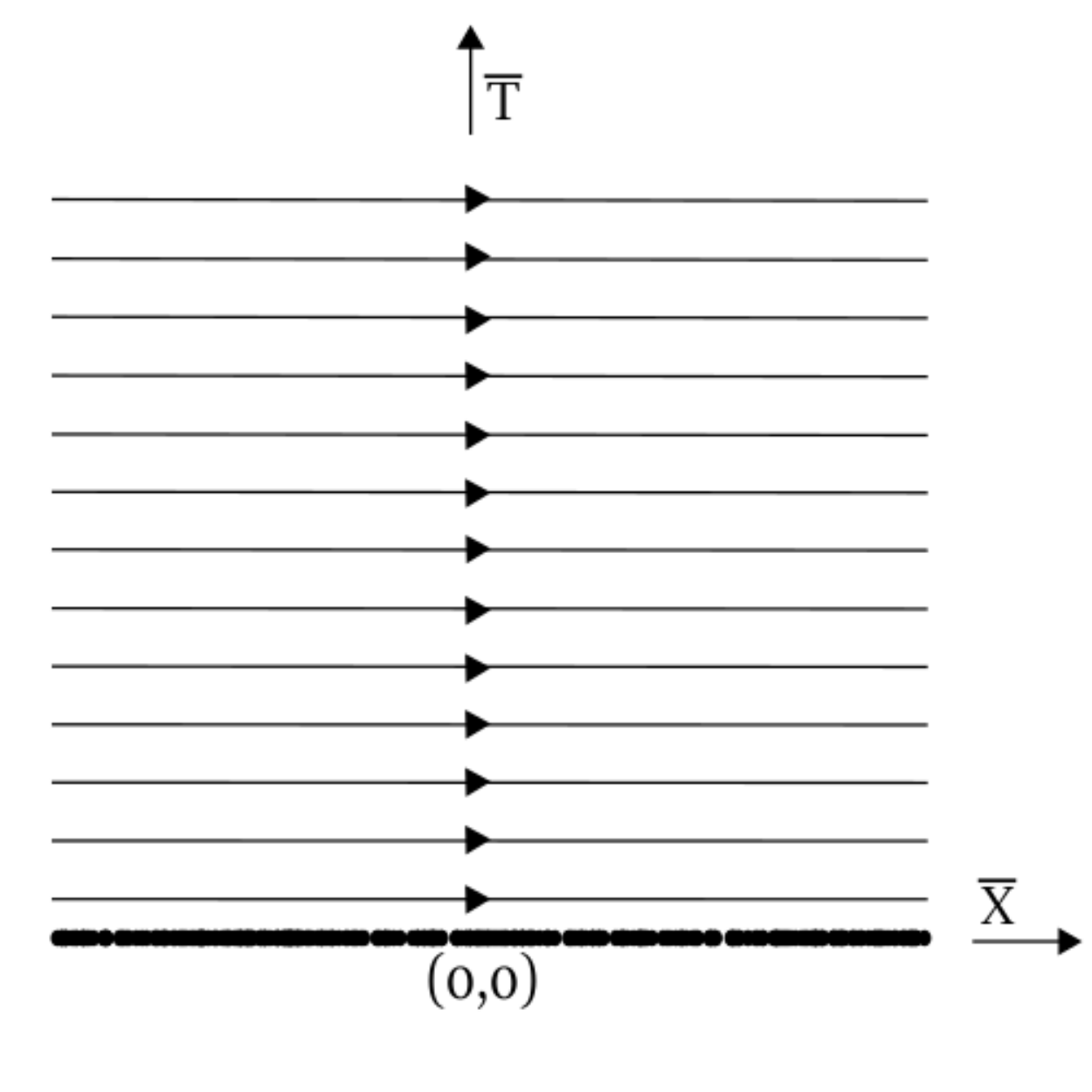}}
			\hspace{2cm}
			\subfigure[$X_0=0$ when $p=0$.]{\label{center1}
				\includegraphics[width=0.30\textwidth]{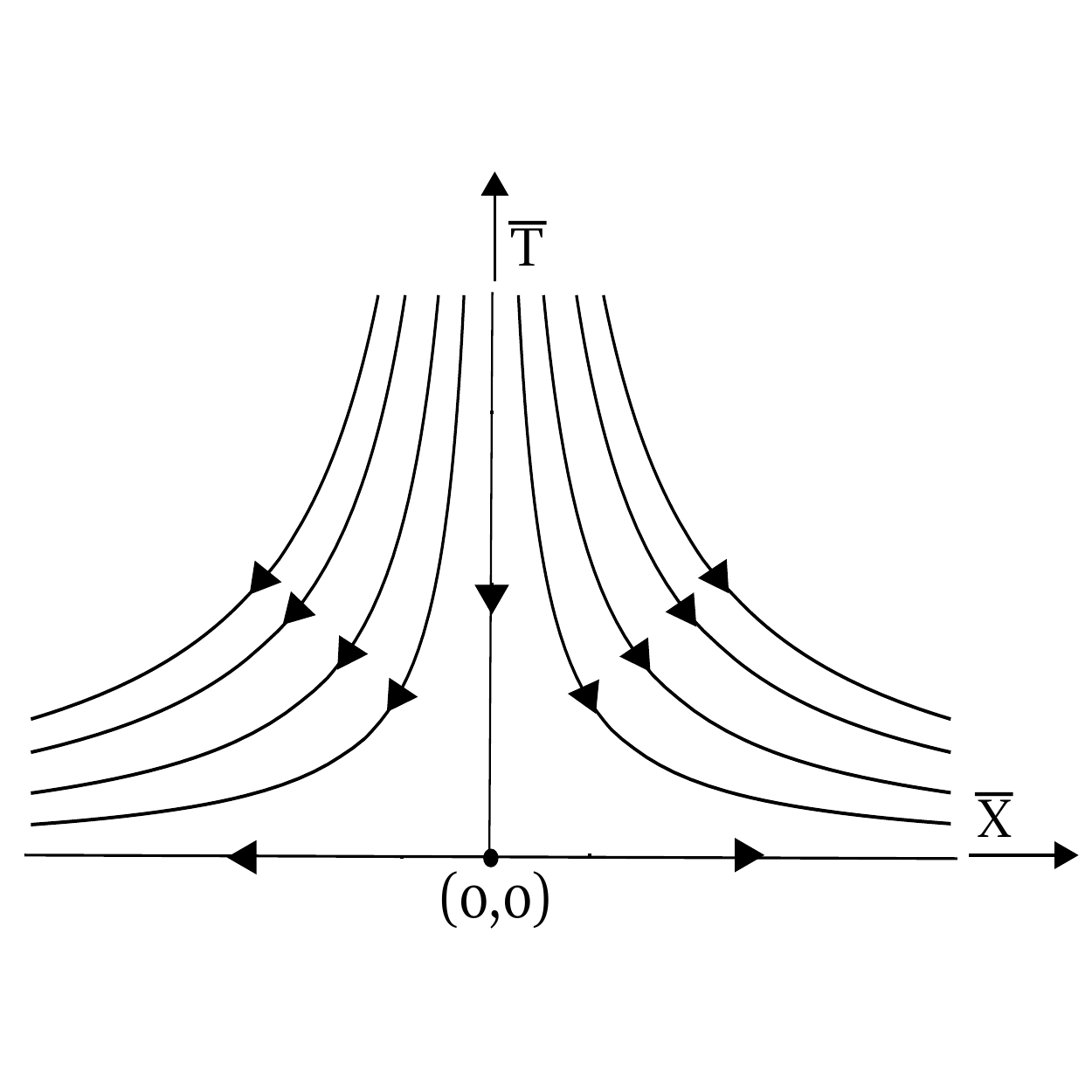}} \\
			\subfigure[$X_0=X_{\pm}$ when $p=\frac{1}{2}(n-2)$.]{\label{center2}
				\includegraphics[width=0.30\textwidth]{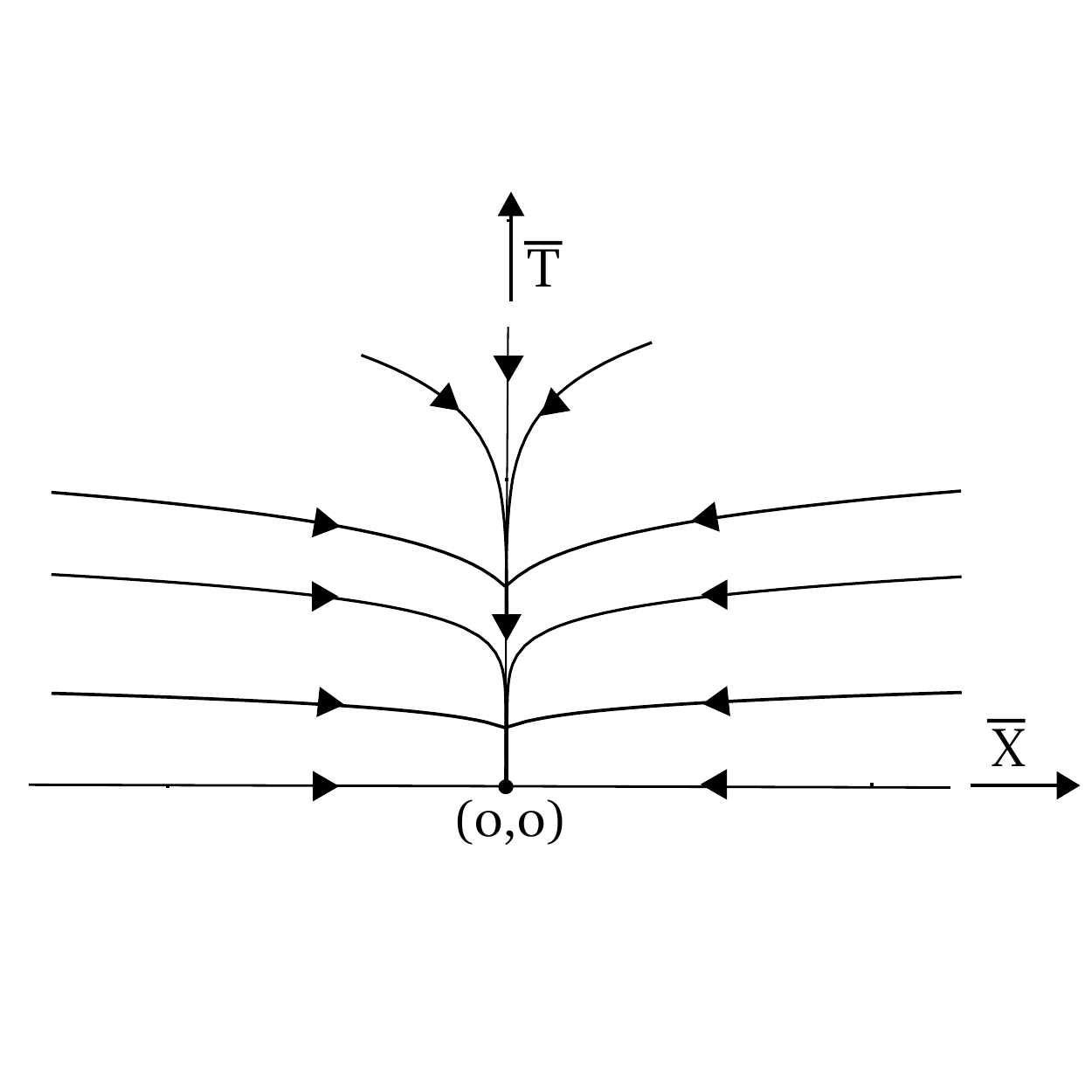}}
			\hspace{2cm}
			\subfigure[$X_0=\pm 1$ when $p<\frac{1}{2}(n-2)$.]{\label{center3}
				\includegraphics[width=0.30\textwidth]{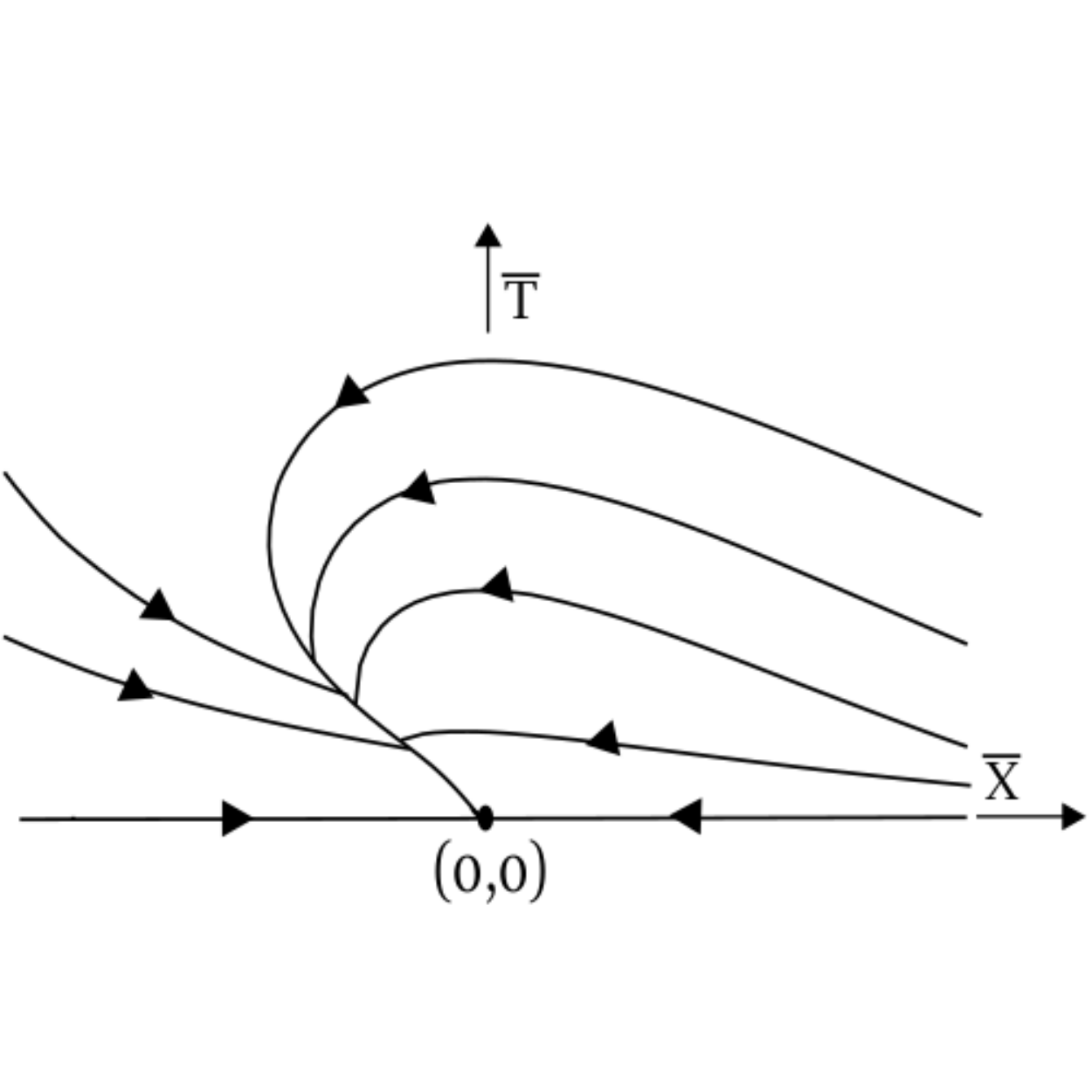}}
		\end{center}
		\vspace{-0.5cm}
		\caption{Flow on the 2-dimensional center manifold of each point on $\mathrm{L}_1$.}
		\label{fig:2D_CM_L1}
	\end{figure}
	After a change of time variable $d/d\tau=\bar{T}^{n-2p-1}d/d\bar{\tau}$, the flow on the 2-dimensional center manifold is given by 
	\begin{subequations}
		\begin{align}
		\frac{d\bar{X}}{d\bar{\tau}} &= \sum^{N}_{j=1} \bar{b}_{j}(\bar{X})\bar{T}^{j}, \qquad \bar{b}_j(\bar{X})=\sum^{N}_{i=0}b_{ij}\bar{X}^{i},\qquad b_{ij}\in\mathbb{R},\\
		\frac{d\bar{T}}{d\bar{\tau}} &=\bar{T}\sum^{N}_{j=1} \bar{c}_{j}(\bar{X})\bar{T}^{j}\qquad \bar{c}_j(\bar{X})=\sum^{N}_{i=0}c_{ij}\bar{X}^{i},\qquad c_{ij}\in\mathbb{R},
		\end{align}
	\end{subequations}
	with
	\begin{subequations}
		\begin{align*}
		b_{01} &=X_0\left(\frac{3\gamma_{\mathrm{pf}}}{2n}(1-X_0^{2n})-\frac{n(n-2p-1)}{\nu}X_0^{2(n-p-1)}\delta^{n-2p}_2\right), \\
		b_{11} &= \left(\frac{3\gamma_{\mathrm{pf}}}{2n}\left(1-(1+2n)X_0^{2n}\right)-\frac{(2(n-p)-1)!}{\nu(n-1)!(n-2p-1)!}X_0^{2(n-p-1)}\delta^{n-2p}_2\right), \\
		b_{02} &=\frac{n}{\nu}X_0^{2(n-p)}, \\
		c_{01}&=-\frac{3\gamma_{\mathrm{pf}}}{2n}(1-X_0^{2n}),\qquad c_{11}=0, \\
		c_{02}&=\frac{3\gamma_{\mathrm{pf}}}{2n}(1-X_0^{2n}),\qquad c_{12}= -3\gamma_{\mathrm{pf}}X_0^{2n}.
		\end{align*}
	\end{subequations}
	For $p=\frac{1}{2}(n-2)$, with $n$ even, the coefficient $b_{01}$ vanishes at $X_0=0$ and $X_0=X_{\pm}$, where
	\begin{equation}
	X_{\pm}=\pm  \left[\frac{n^2}{3\gamma_{\mathrm{pf}}\nu}\left(-1+\sqrt{1+\left(\frac{3\gamma_{\mathrm{pf}}\nu}{n^2}\right)^2}\right)\right]^{1/n}.
	\end{equation}
	Note that $X_-\in(-1,0)$ and $X_+\in(0,1)$. Moreover $b_{01}<0$ for $X_0\in (X_{-},0)\cup(X_+,1]$ and $b_{01}>0$ for $X_0\in [-1,X_{-})\cup(0,X_+)$, with the origin $(0,0)$ being a nilpotent singularity. Since the coefficient $\bar{c}_{01}(\bar{X})\neq 0$ for all $X_0$, then the formal normal form is zero with
	\begin{equation}
	\frac{d\bar{X}_*}{d\bar{\tau}_*} = \text{sign}(b_{01})\bar{T}_*, \qquad \frac{d\bar{T}_*}{d\bar{\tau}_*} =\bar{T}^2_* \Phi(\bar{X}_*,\bar{T}_*) ,
	\end{equation}
	and $\Phi$ an analytic function. The phase-space is the flow-box multiplied by the function $\bar{T}_*$, with the direction of the flow given by the sign of $b_{01}$, see Figure~\ref{center4}. For $X_0=0$ (when $p=0$), then $b_{11}=\frac{3\gamma_{\mathrm{pf}}}{2n}>0$ and $c_{01}=-\frac{3\gamma_{\mathrm{pf}}}{2n}<0$ which after changing the time variable to $d/d\tilde{\tau}=\bar{T}^{-1}d/d\bar{\tau}$ yields a hyperbolic saddle, see Figure~\ref{center1}. For $X_0=X_{\pm}$, we have
	\begin{subequations}
		\begin{align*}
		b_{11}(X_{\pm}) &=-\left(\frac{n^2}{\nu \sqrt{3 \gamma_{\mathrm{pf}}}}\right)^2\left(1+\left(\frac{3\gamma_{\mathrm{pf}}\nu}{n^2}\right)^2-\sqrt{1+\left(\frac{3\gamma_{\mathrm{pf}}\nu}{n^2}\right)^2}\right)<0,\\ 
		c_{01}(X_{\pm}) &=-\frac{n^3}{6 \gamma_{\mathrm{pf}} \nu^2}\left(-1+\sqrt{1+\left(\frac{3\gamma_{\mathrm{pf}} \nu}{n^2}\right)^2}\right)<0,
		\end{align*}
	\end{subequations}
	and after changing the time variable to $d/d\tilde{\tau}=\bar{T}^{-1}d/d\bar{\tau}$ the origin is a hyperbolic sink, see Figure~\ref{center2}. 
	
	For $p<\frac{1}{2}(n-2)$, the coefficient $b_{01}$ vanishes at $X_0=0$ and $X_0=\pm1$, being negative for $X_0\in(-1,0)$ and positive for $X_0\in(0,1)$. For $b_{01}\neq0$ the phase-space is again as depicted in Figure~\ref{center4} with the direction of the flow given by the sign of $b_{01}$, i.e. of $X_0$. 
	
	When $X_0=0$ (and restricting to $p=0$), $b_{01}=0$ and $b_{11}=\frac{3\gamma_{\mathrm{pf}}}{2n}>0$ which after changing the time variable to $d/d\tilde{\tau}=\bar{T}^{-1}d/d\bar{\tau}$ yields that $\mathrm{FL}_1$ is a hyperbolic saddle, see Figure~\ref{center1}. For $p>0$ the blow-up of $\mathrm{FL}_1$ can be found in Section~\ref{BUP1st}, where it is shown that a 1-parameter set of interior orbits in $\mathbf{S}$ end at $\mathrm{FL}_1$, see also Remark~\ref{AsFL1}.
	
	For $X_0=\pm1$, we have that $b_{11}=-3\gamma_\mathrm{pf}<0$, $c_{01}=0$, $c_{02}=0$ and $c_{12}=- 3\gamma_{\mathrm{pf}}<0$ after changing the time variable to $d/d\tilde{\tau}=\bar{T}^{-1}d/d\bar{\tau}$, then 
	\begin{subequations}
		\begin{align*}
		\frac{d\bar{X}}{d\tilde{\tau}} & =-3\gamma_{\mathrm{pf}}\bar{X}+\frac{n}{\nu}\bar{T}+\frac{n}{\nu}\bar{T}^2+\frac{n}{2}(n-2p)\bar{X}\bar{T}-\frac{3}{2}(2n+1)\bar{X}^2+\mathcal{O}(\|(\bar{X},\bar{T})\|^3),\\ \frac{d\bar{T}}{d\tilde{\tau}} & =-3\gamma_{\mathrm{pf}} \bar{X}\bar{T}+\mathcal{O}(\|(\bar{X},\bar{T})\|^3),
		\end{align*}
\end{subequations}	
	and the origin is a semi-hyperbolic fixed point with eigenvalues $-3\gamma_{\mathrm{pf}}$, $0$ and associated eigenvectors $(1,0)$ and $(-\frac{n}{3\gamma_{\mathrm{pf}}\nu},1)$. To analyse the 1-dimensional center manifold $W^{c}$ at $(\tilde{X},\bar{T})=(0,0)$ we introduce the adapted variable $\tilde{X}=\bar{X}+\frac{n}{3\gamma_{\mathrm{pf}}\nu}\bar{T}$. 
	The $1$-dimensional center manifold can be locally represented as the graph $h:E^{c}\rightarrow E^{s}$, i.e. $\tilde{X}=h(\bar{T})$, satisfying the fixed point $h(0)=0$ and tangency $\frac{dh(0)}{d\bar{T}}=0$ conditions, using $\bar{T}$ as an independent variable. Approximating the solution by a formal truncated power series expansion $h(\bar{T})=\sum_{i=2}^{N}a_{i}\bar{T}^{i}$, $a_i\in\mathbb{R}$,	and solving for the coefficients  yields to leading order on the center manifold
	\begin{equation}
	\frac{d\bar{T}}{d\tilde{\tau}}=-\frac{n}{\nu}\bar{T}^2+\mathcal{O}(\bar{T}^3)\quad \text{as}\quad \bar{T}\rightarrow 0.
	\end{equation}
	Therefore for $X_0=\pm 1$, the origin is the $\omega$-limit set of a 1-parameter set of orbits on the 2-dimensional center manifold, see Figure~\ref{center3}.
\end{proof}
\begin{remark}
	Let $p<(n-1)/2$. The asymptotics for solutions of \eqref{sistema1} converging to $\mathrm{dS}^{\pm}_1$ are given by
	\begin{equation*}
	H\sim t^{-\frac{n}{2(n-p)}}\quad,\quad \phi\sim\text{const.}+t^{-\frac{1}{3+4p-n}}\quad,\quad \rho_\mathrm{pf} \sim t^{-\frac{n}{n-p}}, \quad \quad\text{as}\quad t\rightarrow +\infty,
	\end{equation*}
	 while those converging to $\mathrm{S}^{\pm}$ are given by 
	 \begin{equation*}
	 H\sim t^{-\frac{n}{3+n}}\quad,\quad \phi\sim\text{const.}+t^{-\frac{1}{n-1}}\quad,\quad \rho_\mathrm{pf} \sim t^{-\frac{2n}{3+n}}, \quad \quad\text{as}\quad t\rightarrow +\infty.
	 \end{equation*}
	 It is possible to deduce the asymptotics for the 1-parameter set of solutions converging towards the fixed point $\mathrm{FL}_1$ when $p=0$, while for $p>0$, the asymptotics can be obtained by the analysis of the blow-up of $\mathrm{FL}_1$ done in Section~\ref{BUP1st}, see Remark~\ref{AsFL1}.
\end{remark}
The global state-space picture for solutions of the dynamical system~\eqref{globalng2p} when $p<\frac{1}{2}(n-1)$ is shown in Figure~\ref{fig:n-1>2p}. The solid numerical curves correspond to the center manifold of $\mathrm{dS}^{\pm}_0$, and the dashed numerical curves to solution curves  originating from the source $\mathrm{K}^{-}$. All these solutions end at the generic future attractors $\mathrm{dS}^{\pm}_{1}$ when $p<\frac{1}{2}(n-2)$ or $\mathrm{S}^{\pm}$ when $p=\frac{1}{2}(n-2)$.
\begin{figure}[ht!]
	\begin{center}
		\subfigure[$p<\frac{1}{2}(n-2)$. 
		]{\label{fig:}
			\includegraphics[width=0.30\textwidth]{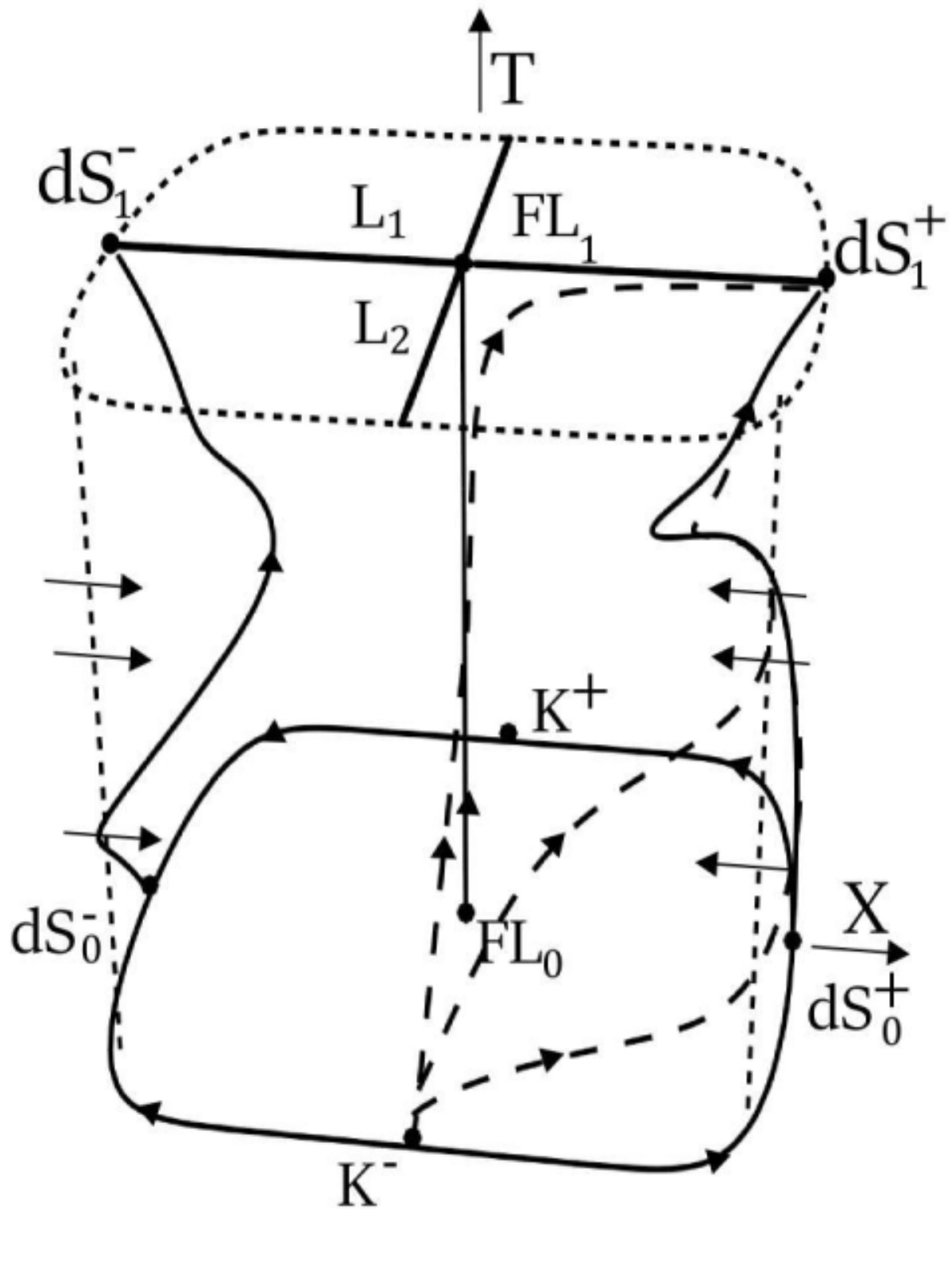}}
		\hspace{2cm}
		\subfigure[$p=\frac{1}{2}(n-2)$. 
		]{\label{fig:}
			\includegraphics[width=0.30\textwidth]{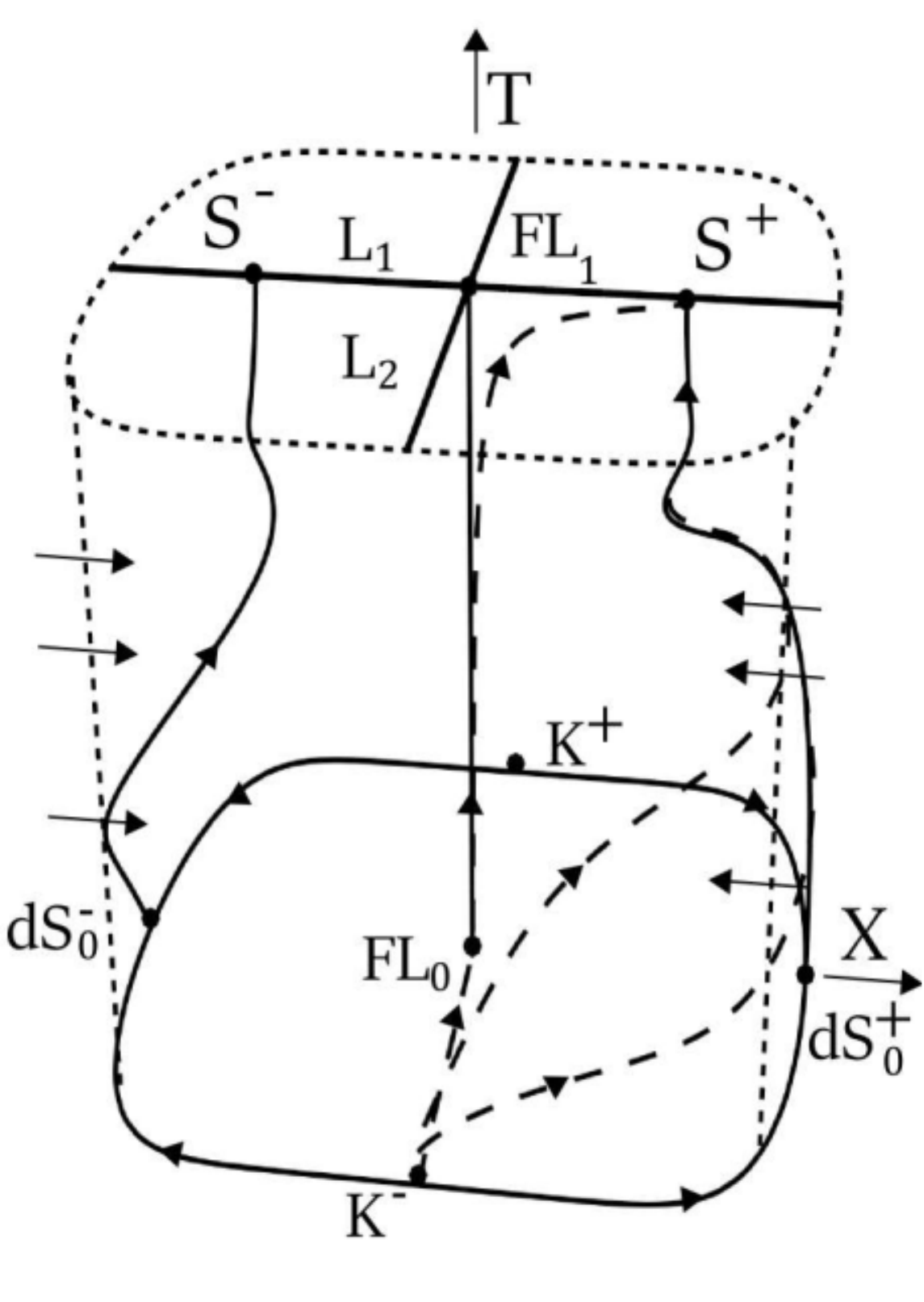}}
	\end{center}
	\vspace{-0.5cm}
	\caption{Global state-space $\mathbf{S}$ when $p<\frac{1}{2}(n-1)$.}
	\label{fig:n-1>2p}
\end{figure}
\subsubsection{Case $p=\frac{1}{2}(n-1)$}
In this case there is a single fixed point lying in the intersection of the $T=1$ invariant boundary with the pure perfect fluid invariant subset $\Omega_{\mathrm{pf}}=1$ given by
\begin{equation}
\mathrm{FL_1}: \quad X=0,\quad\quad \Sigma_{\mathrm{\phi}}=0, \quad\quad T=1.
\end{equation}
The linearisation around this fixed point on the $T=1$ invariant boundary yields the characteristic polynomial $\lambda^2+\nu\delta^{n}_{1}\lambda+n\delta^{n}_{1}=0$. Since $\nu>0$, then for $n=1$ $(p=0)$, $\mathrm{FL}_1$  has two eigenvalues with negative real part being a hyperbolic sink on $T=1$ (stable node if $\nu\geq2$, and a stable focus if $0<\nu<2$), while on the full state space $\mathrm{FL}_1$ has a 1-dimensional center manifold with center tangent space $E^c=\langle (0,0,1)\rangle$, i.e. consisting of the $\Omega_\mathrm{pf}=1$ invariant subset. 
Therefore $\mathrm{FL}_1$ is the $\omega$-limit point of a 2-parameter set of orbits, converging to $\mathrm{FL}_1$ tangentially to the center manifold when $\nu\geq2$, or spiraling around the center manifold when $0<\nu<2$. When $p\neq0$ and $n>1$ all eigenvalues of the fixed point $\mathrm{FL}_1$ are zero. The blow-up of the fixed point $\mathrm{FL}_1$ when $p>0$ is done in Subsection~\ref{BUP2nd}. 
\begin{lemma}\label{Lem2}
	Let $p=\frac{1}{2}(n-1)$. Then the $\{T=1\}$ invariant boundary consists of orbits which enter the region $\Omega_{\mathrm{pf}}>0$ by crossing the set $\Omega_\mathrm{pf}=0$ and converging to the fixed point $\mathrm{FL}_1$ at the origin as $\tau\rightarrow+\infty$, see Figure~\ref{figT1:n-1-2p=0}.
\end{lemma}
\begin{proof}
	It suffices to note that the bounded function $\Omega_{\mathrm{pf}}$ is strictly monotonically increasing along the orbits, except at the axis of coordinates $\Sigma_\phi=0$ or $X=0$ when $p>0$. However since $d\Sigma_\phi/d\tau\neq0$ on $\Sigma_\phi=0$ and $dX/d\tau\neq0$ on $X=0$, except at origin where the axis intersect, it follows by \emph{LaSalle's invariance principle} that $(\Sigma_\phi,X)\rightarrow(0,0)$ and $\Omega_\mathrm{pf}\rightarrow 1$. In fact, when $p=\frac{1}{2}(n-1)$, the system \eqref{ng2p1} admits the following conserved quantity on $T=1$:
	\begin{subequations}
		\begin{align}
		-\frac{ \arctan \left[\frac{\frac{2n}{\nu} \frac{\Sigma_{\mathrm{\phi}}}{X^{n}} +1}{\sqrt{\left(\frac{2n}{\nu}\right)^2-1}}\right]}{\sqrt{\left(\frac{2n}{\nu}\right)^2-1}}+\frac{1}{2}\log\left[\nu \Sigma_{\mathrm{\phi}} X^{n}+n (1-\Omega_{\mathrm{pf}})\right]=\text{const.},\qquad &\text{if}\quad 0<\nu<2n, \\
		\log\left[\Sigma_{\mathrm{\phi}}+X^{n}\right]+\frac{X^{n}}{\Sigma_{\mathrm{\phi}}+X^{n}}=\text{const.},\qquad &\text{if}\quad \nu=2n, \\
		\frac{\arctanh\left[\frac{\frac{2n}{\nu}\frac{\Sigma_{\mathrm{\phi}}}{X^{n}}+1}{\sqrt{1-\left(\frac{2n}{\nu}\right)^2}}\right]}{\sqrt{1-\left(\frac{2n}{\nu}\right)^2}}+\frac{1}{2}\log\left[\nu \Sigma_{\mathrm{\phi}} X^{n}+n(1-\Omega_{\mathrm{pf}})\right]=\text{const.},\qquad &\text{if}\quad \nu>2n,
		\end{align}
	\end{subequations}
	which determine the solution trajectories on the $T=1$ invariant boundary. %
\end{proof}
\begin{figure}[ht!]
	\begin{center}
		\subfigure[$(p,n)=(0,1)$.]{\label{fig:n=1 futuro}	\includegraphics[width=0.30\textwidth]{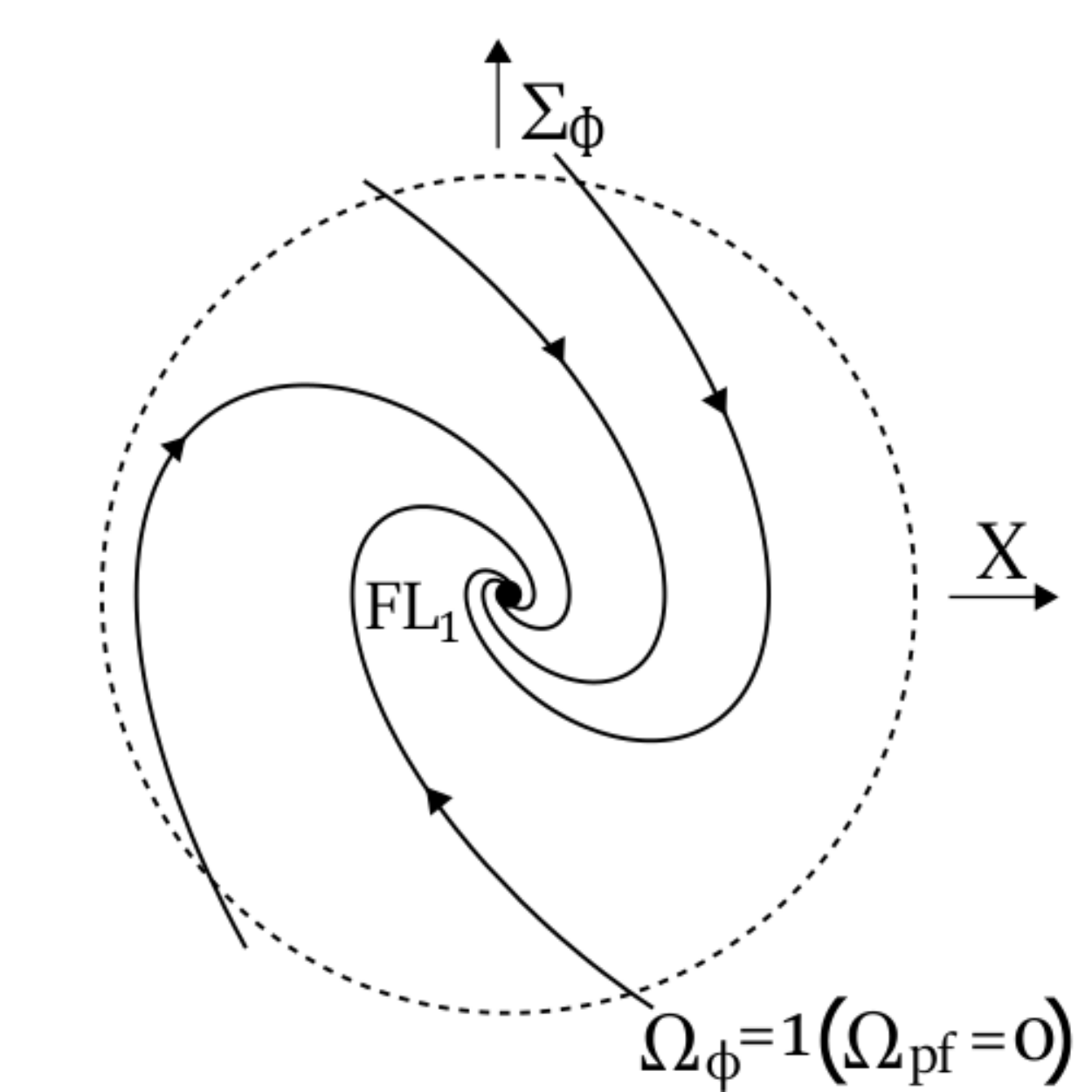}}
		\hspace{2cm}
		\subfigure[$p>0$ with $(p,n)=(1,3)$.]{\label{fig:T=1,n-2p=1}
			\includegraphics[width=0.30\textwidth]{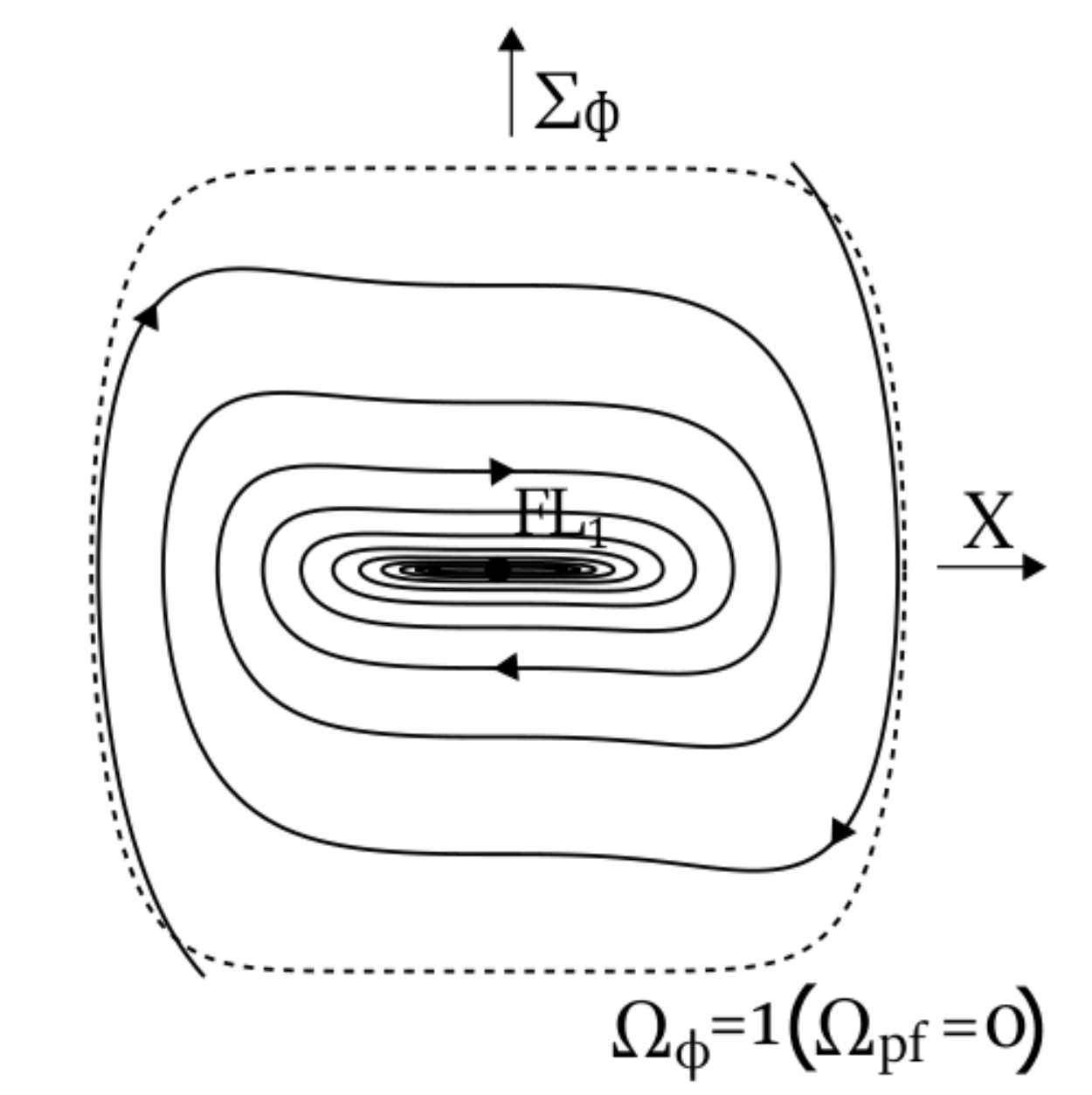}}
	\end{center}
	\vspace{-0.5cm}
	\caption{Invariant boundary $\{T=1\}$ when $p=\frac{1}{2}(n-1)$.}
	\label{figT1:n-1-2p=0}
\end{figure}
\begin{theorem}
     Let $p=\frac{1}{2}(n-1)$. Then as $\tau\rightarrow+\infty$, all orbits in $\mathbf{S}$ converge to the fixed point $\mathrm{FL}_1$. 
\end{theorem}
\begin{proof}
	The proof follows from Lemmas~\ref{lemma1} and~\ref{Lem2}.
\end{proof}
\begin{remark}
	It is possible to deduce the asymptotics towards the fixed point $\mathrm{FL}_1$. For $(n,p)=(1,0)$,  and all $\gamma_\mathrm{pf}\in(0,2)$ the preceding analysis yields to leading order on the center manifold
	\begin{equation}
	T(\tau) = 1-\left(1+\frac{3\gamma_\mathrm{pf}}{2}C_T \tau\right)^{-1}\qquad\text{as}\qquad \tau\rightarrow+\infty.
	\end{equation}
	Moreover if  $0<\nu<2$, then
	     	 \begin{subequations}
	     	 	\begin{align}
	     	 	X(\tau)&=e^{-\frac{\nu}{2}\tau}\Big(C_X \cos(\frac{1}{2}\sqrt{4-\nu^2}\tau)+\frac{(\nu C_X+2 C_\Sigma)\sin(\frac{1}{2}\sqrt{4-\nu^2}\tau)}{\sqrt{4-\nu^2}}\Big), \\
	     	 	\Sigma_\phi(\tau) &= e^{-\frac{\nu}{2}\tau}\Big(C_\Sigma \cos(\frac{1}{2}\sqrt{4-\nu^2}\tau)-\frac{(2\nu C_X+\nu C_\Sigma)\sin(\frac{1}{2}\sqrt{4-\nu^2}\tau)}{\sqrt{4-\nu^2}}\Big),
	     	 	\end{align}
	     	 \end{subequations}
as $\tau\rightarrow+\infty$. If $\nu=2$,
	     	 \begin{subequations}
	     	 	\begin{align}
	     	 	X(\tau)&=e^{-\tau}\Big((1+\tau)C_X+\tau C_\Sigma\Big),\\
	     	 	\Sigma_\phi(\tau) &= e^{-\tau}\Big(C_\Sigma-\tau(C_X+C_\Sigma)\Big),\end{align}
	     	 \end{subequations}
as $\tau\rightarrow+\infty$, and if $\nu> 2$, then
	     	 \begin{subequations}
	     	 	\begin{align}
	     	 	X(\tau)&=e^{-\frac{\nu}{2}\tau}\Big(C_X \cosh(\frac{1}{2}\sqrt{\nu^2-4}\tau)+\frac{(\nu C_X+2 C_\Sigma)\sinh(\frac{1}{2}\sqrt{\nu^2-4}\tau)}{\sqrt{\nu^2-4}}\Big),\\
	     	 	\Sigma_\phi(\tau) &= e^{-\frac{\nu}{2}\tau}\Big(C_\Sigma \cosh(\frac{1}{2}\sqrt{\nu^2-4}\tau)-\frac{(2\nu C_X+\nu C_\Sigma)\sinh(\frac{1}{2}\sqrt{\nu^2-4}\tau)}{\sqrt{\nu^2-4}}\Big)
	     	 	\end{align}
	     	 \end{subequations}
	     	 as $\tau\rightarrow+\infty$.
	For $p>0$, i.e. $(p,n)=(1,3),(2,5),(3,7),etc.$, the asymptotics can be obtained by the analysis of the blow-up of $\mathrm{FL}_1$ done in Section~\ref{BUP2nd}, see Remark~\ref{ASFL1_pEQ}.
\end{remark}
The global state-space picture for solutions of the dynamical system~\eqref{globalng2p} when $p=\frac{1}{2}(n-1)$ is shown in Figure~\ref{fig:n-2p=1}. The solid numerical curves correspond to the center manifold of $\mathrm{dS}^{\pm}_0$ and the dashed numerical curves to solution curves  originating from the source $\mathrm{K}^{-}$. All these solutions end up at the generic future attractor $\mathrm{FL}_{1}$.
\begin{figure}[ht!]
	\begin{center}
		\subfigure[$(p,n)=(0,1)$.]{\label{fig:n=1}
			\includegraphics[width=0.30\textwidth]{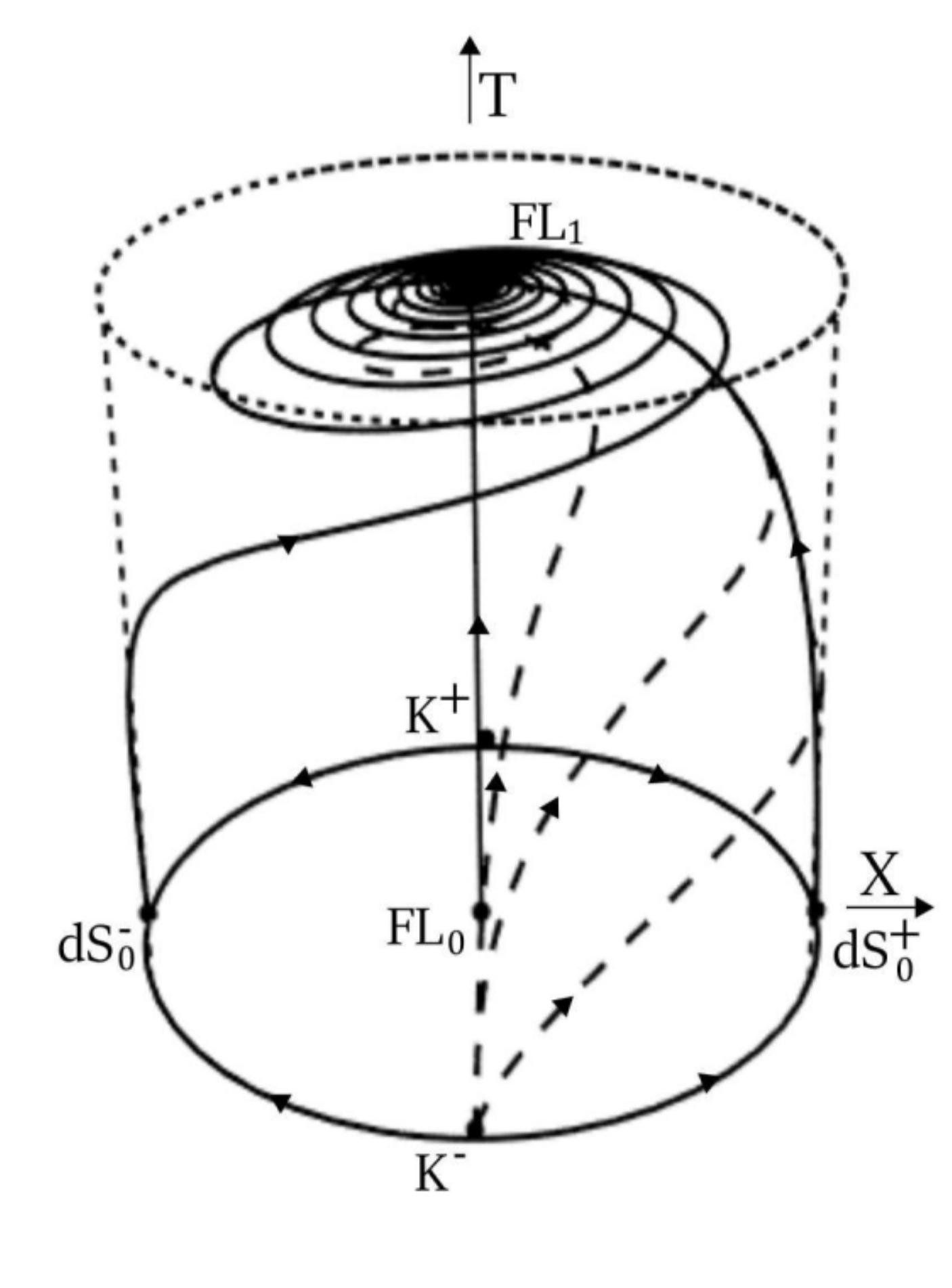}}
		\hspace{2cm}
		\subfigure[$p>0$ with $(p,n)=(1,3)$.]{\label{fig:n>1}
			\includegraphics[width=0.30\textwidth]{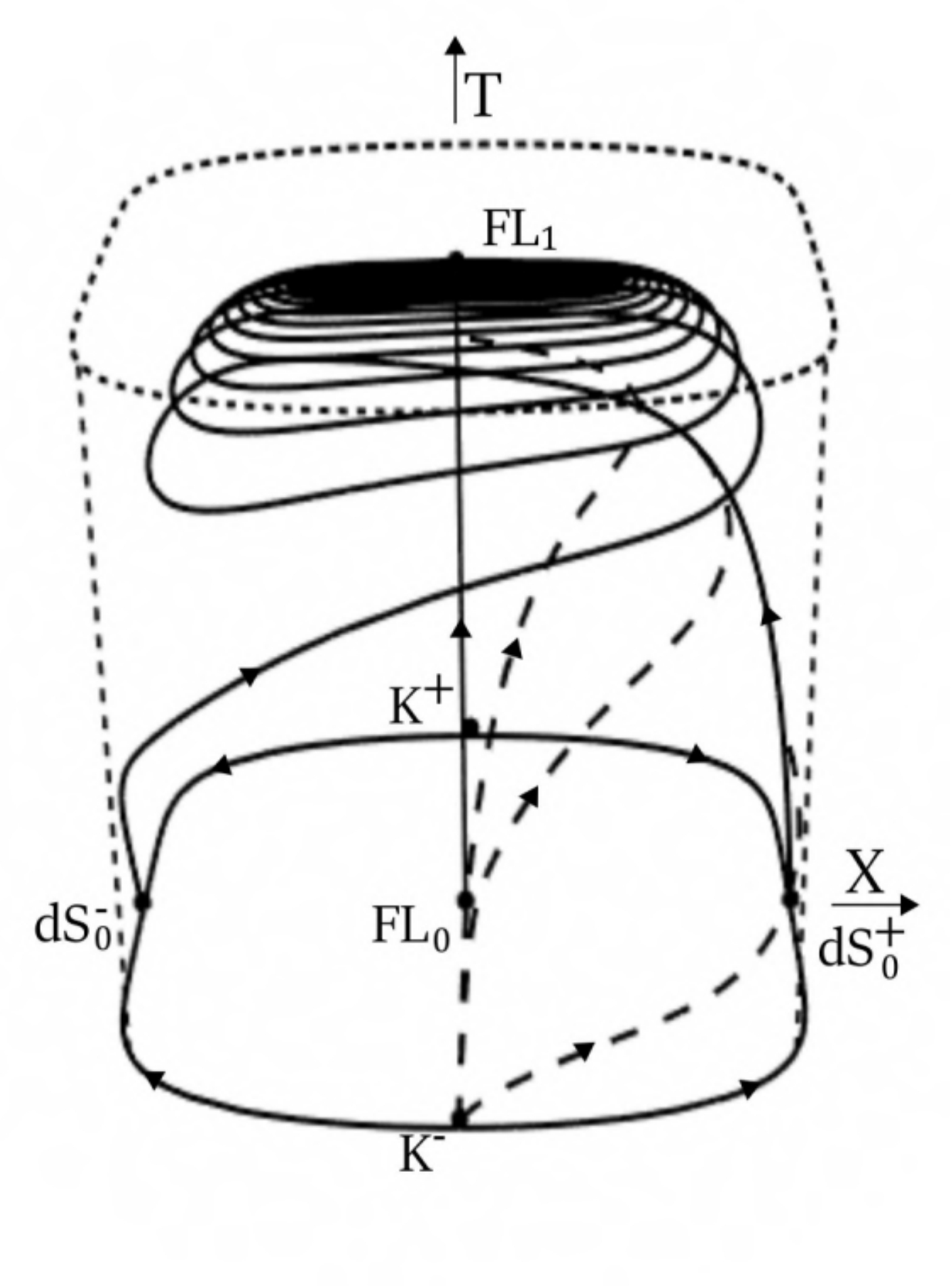}}
	\end{center}
	\vspace{-0.5cm}
	\caption{Global state-space 
		$\mathbf{S}$ when $p=\frac{1}{2}(n-1)$.}
	\label{fig:n-2p=1}
\end{figure}
%
%
%
\subsection{Blow-up of $\mathrm{FL}_1$ when $p>0$}
\label{ApA}
To analyse the non-hyperbolic fixed point $\mathrm{FL}_1$ of the dynamical system \eqref{globalng2p} when $p>0$, we start by relocating $\mathrm{FL}_1$ at the origin, i.e. we introduce $\bar{T}=1-T$, to obtain the dynamical system
\begin{subequations}
	\begin{align}
	\frac{dX}{d\tau} &=\frac{1}{n}(1+q)\bar{T}^{n-2p}X+(1-\bar{T})\bar{T}^{n-2p-1}\Sigma_\phi , \\
	\frac{d\Sigma_\phi}{d\tau} &= -\left[(2-q)\bar{T}^{n-2p}+\nu(1-\bar{T})^{n-2p}X^{2p}\right]\Sigma_\phi-n(1-\bar{T})\bar{T}^{n-2p-1}X^{2n-1} , \\
	\frac{d\bar{T}}{d\tau} &= -\frac{1}{n}(1+q)(1-\bar{T})\bar{T}^{n-2p+1},
	\end{align}
\end{subequations}
where we recall
\begin{equation}
q=-1+3\Sigma^2_\phi+\frac{3}{2}\gamma_\mathrm{pf}(1-\Sigma^2_\phi-X^{2n}).
\end{equation}
In order to understand the dynamics near the origin $(X,\Sigma_\phi,\bar{T})=(0,0,0)$, which is a non-hyperbolic fixed point for $p>0$, we employ the spherical blow-up method, see e.g.~\cite{Dum93,Perko,BookPlanar}. I.e. we transform the fixed point at the origin to the unit 2-sphere $\mathbb{S}^2=\{(x,y,z)\in\mathbb{R}^3:x^2+y^2+z^2=1\}$, and define the blow-up space manifold as $\mathcal{B}:=\mathbb{S}^2\times[0,u_0]$ for some fixed $u_0>0$. We further define the quasi-homogeneous blow-up map
\begin{equation}\label{BUPMAP}
\Psi\,: \mathcal{B}\rightarrow \mathbb{R}^3,\qquad \Psi(x,y,z,u)=(u^{n-2p}x, u^{n}y, u^{2p}z),
\end{equation}
which after cancelling a common factor $u^{2p(n-2p)}$ (i.e. by changing the time variable to $\bar{\tau}$ defined by $d/d\bar{\tau}=u^{-2p(n-2p)}d/d\tau$, where recall $p<\frac{n}{2}$, with $p>0$) leads to a desingularisation of the non-hyperbolic fixed point on the blow-up locus $\{u=0\}$. Since $\Psi$ is a diffeomorphism outside of the sphere $\mathbb{S}^2\times\{u=0\}$, which corresponds to the fixed point  $(0,0,0)$, the dynamics on the blow-up space $\mathcal{B}\setminus\{\mathbb{S}^2\times\{u=0\}\}$ is topological conjugate to $\mathbb{R}^3\setminus\{0,0,0\}$.

It usually simplifies the computations if instead of standard spherical coordinates on $\mathcal{B}$, one uses different local charts $\kappa_i:\mathcal{B}\rightarrow\mathbb{R}^3$ such that $\psi_i:\Psi\circ \kappa^{-1}_i$ and the resulting state vector(-fields) are simpler to analyze. We choose six charts $\kappa_i$ such that
\begin{subequations}
	\begin{align}
	\psi_{1\pm}&=(\pm u_{1\pm}^{n-2p}, u_{1\pm}^{n}y_{1\pm},u_{1\pm}^{2p}z_{1\pm}),\\
	\psi_{2\pm} &= (u_{2\pm}^{n-2p}x_{2\pm},\pm u_{2\pm}^{n},u_{2\pm}^{2p}z_{2\pm}),\\
	\psi_{3\pm} &= (u_{3\pm}^{n-2p}x_{3\pm},u_{3\pm}^n y_{3\pm},\pm u_{3\pm}^{2p}),
	\end{align}
\end{subequations}
where $\psi_{1\pm}$, $\psi_{2\pm}$, and $\psi_{3\pm}$ are called the directional blow-ups in the positive/negative $x$, $y$, and $z$-directions respectively.
It is easy to check that the different charts are given explicitly by
\begin{subequations}
	\begin{align}
	\kappa_{1\pm}&:\quad(u_{1\pm},y_{1\pm},z_{1\pm})=(\pm u x^{\frac{1}{n-2p}},\pm y x^{-n},\pm z x^{-2p}),\\
	\kappa_{2\pm}&:\quad(x_{2\pm},u_{2\pm},z_{2\pm})=(\pm x y^{-\frac{n-2p}{n}},\pm u y^{\frac{1}{n}},\pm z y^{-\frac{2p}{n}}),\\
	\kappa_{3\pm}&:\quad (x_{3\pm},u_{3\pm},z_{3\pm})=(\pm x z^{-\frac{n-2p}{2p}},\pm y z^{-\frac{n}{2p}}, \pm u z^{\frac{1}{2p}}).
	\end{align}
\end{subequations}
The transition maps $\kappa_{ij}=\kappa_j\circ \kappa^{-1}_i$ then allow us to identify fixed points and special invariant manifolds on different charts, and to deduce all the dynamics on the blow-up space. For example, in this case, we will need some of the following transition charts
\begin{subequations}
	\begin{align}
	\kappa_{1+2+}\quad &:\quad(x_{2+},u_{2+},z_{2+})=(y_{1+}^{-\frac{n-2p}{n}},u_{1+}y_{1+}^{\frac{1}{n}},y_{1+}^{-\frac{2p}{n}}z_{1+}), \quad y_{1+}>0;\\
	\kappa_{2+1+}\quad &: \quad(u_{1+},y_{1+},z_{1+})=(u_{2+}x_{2+}^{\frac{1}{n-2p}},x_{2+}^{-n}, z_{2+} x_{2+}^{-2p}),\quad x_{2+}>0;
	\end{align}
\end{subequations}
\begin{subequations}
	\begin{align}
	\kappa_{1+3+}\quad &: \quad (x_{3+},y_{3+},u_{3+})=(z_{1+}^{-\frac{n-2p}{2p}},y_{1+}z_{1+}^{-\frac{n}{2p}},u_{1+}z_{1+}^{\frac{1}{2p}}),\quad 	z_{1+}>0;\\
	\kappa_{3+1+}\quad &: \quad (u_{1+},y_{1+},z_{1+})=(u_{3+}x_{3+}^{\frac{1}{n-2p}},y_{3+}, y_{3+}x_{3+}^{-n},x_{3+}^{-2p}),\quad x_{3+}>0;
	\end{align}
\end{subequations}
\begin{subequations}
	\begin{align}
	\kappa_{2+3+}\quad &: \quad (x_{3+},y_{3+},u_{3+})=(x_{2+} z_{2+}^{-\frac{n-2p}{2p}},z_{2+}^{-\frac{n}{2p}}, u_{2+}z_{2+}^{\frac{1}{2p}}),\quad z_{2+}>0;\\
	\kappa_{3+2+}\quad &:\quad (x_{2+},u_{2+},z_{2+})=(x_{3+}y_{3+}^{-\frac{n-2p}{n}},u_{3+}y_{3+}^{\frac{1}{n}},y_{3+}^{-\frac{2p}{n}}),\quad y_{3+}>0.
	\end{align}
\end{subequations}
Since the physical state-space has $\bar{T}\geq0$, we are only interested in the region $\{z\geq0\}$, i.e. the union of the upper hemisphere of the unit sphere $\mathbb{S}^2$ with the equator of the sphere $\{z=0\}$ which constitutes an invariant boundary. This motivates that we start the analysis by using chart $\kappa_{3+}$, i.e. the directional blow-up map in the positive $z$-direction, on which the northern hemisphere is mapped into the invariant plane of coordinates $(x_3,y_3)=(x,y,1)$. After cancelling a common factor $u_3^{2p(n-2p)}$ (i.e. by changing the time variable $d/d\tau=u_3^{2p(n-2p)}d/d\bar{\tau}_{3}$) we obtain the regular dynamical system
\begin{subequations}\label{PosZ}
	\begin{align}
	\frac{dx_3}{d\bar{\tau}_3} =&\frac{1}{2np}(1+q)(2p+(n-2p)(1-u_3^{2p}))x_3 +(1-u_3^{2p})y_3, \\
	\frac{dy_3}{d\bar{\tau}_3} =& \frac{1}{2p}(1+q)(1-u^{2p})y_3-\big(2-q+\nu(1-u_3^{2p})^{n-2p}x_3^{2p} \big)y_3-n(1-u_3^{2p})x_3^{2n-1}u_3^{2n(n-1-2p)},\\
	\frac{du_3}{d\bar{\tau}_3} =& -\frac{1}{2np}(1+q)(1-u_3^{2p})u_3	,
	\end{align}
\end{subequations}
where
$$
q=-1+\frac{3\gamma_\mathrm{pf}}{2}\left(1+\frac{2-\gamma_\mathrm{pf}}{\gamma_\mathrm{pf}}y_3^2u_3^{2n}-x_3^{2n}u_3^{2n(n-2p)}\right).
$$
In these coordinates the equator of the sphere is at infinity and it is better analysed using charts $\kappa_{1\pm}$ and $\kappa_{2\pm}$. Moreover, the above system is symmetric under the transformation $(x_3,y_3)\rightarrow -(x_3,y_3)$ and, therefore, it suffices to consider the charts in the positive directions.
To study the points at infinity, we notice that both the directional blow-ups in the positive $x$ and $y$ directions already tell how such local chart must be given. To study the region where $x_3$ becomes infinite, we use the transition chart $\kappa_{3+1+}$:
\begin{equation}
\left(y_1,z_1,u_1\right)=\left(y_3 x_3^{-n}, x_3^{-2p},u_3 x_3^{\frac{1}{n-2p}}\right),
\end{equation}
together with the change of time variable $d/d\bar{\tau}_1=z_1d/d\bar{\tau}_3$, i.e. $d/d\tau=u^{2p(n-2p)}_1d/d\bar{\tau}_1$, which leads to the regular system of equations
\begin{subequations}\label{PosX}
	\begin{align}
	\frac{dy_1}{d\bar{\tau}_1}=&-n(1-u_{1}^{2p}z_{1})z_{1}^{n-2p-1}u_{1}^{2n(n-2p-1)}-\left(2-q+\frac{1+q}{n-2p}\right)y_1 z_1^{n-2p}-\nu y_{1}(1-u_1^{2p})^{n-2p}  \nonumber \\
	&-\frac{n}{n-2p}\left(1-u_{1}^{2p}z_{1}\right)y_1^{2}z_1^{n-2p-1},\\
	\frac{dz_1}{d\bar{\tau}_1}=& -\left(\frac{1+q}{n}\left(\frac{n}{n-2p}-u_{1}^{2p}z_1\right)z_{1}^{n-2p}+\frac{2p}{n-2p}(1-u_1^{2p}z_1)y_{1}z_{1}^{n-2p-1}\right)z_1, \\
	\frac{du_1}{d\bar{\tau}_1}=& \frac{1}{n-2p}\left(\frac{1+q}{n}z_1+\left(1-u_1^{2p}z_1\right)y_1\right)u_1 z_1^{n-2p-1},
	\end{align}
\end{subequations}
where
\begin{equation*}
q=-1+\frac{3\gamma_\mathrm{pf}}{2}\left(1+\frac{2-\gamma_\mathrm{pf}}{\gamma_\mathrm{pf}}y_{1}^{2} u_{1}^{2n}-u_1^{2n(n-2p)}\right).
\end{equation*}
To study the region where $y_3$ becomes infinite, we use the chart $\kappa_{3+2+}$:
\begin{equation}
\left(x_2,z_2,u_2\right)=\left(x_3 y_3^{\frac{2p-n}{n}},y_3^{\frac{2p}{n}},u_3 y_3^{\frac{1}{n}}\right),
\end{equation}
and change time variable $d/\bar{\tau}_2=z_2d/d\bar{\tau}_3$, i.e. $d/d\tau=u^{2p(n-2p)}_2d/\bar{\tau}_2$, to obtain the regular dynamical system
\begin{subequations}\label{BU_PosYDir}
	\begin{align}
	\frac{dx_2}{d\bar{\tau}_2}&=\frac{1}{n}\left((2(n-2p)+1)-(n-2p-1)q\right)x_2 z_2^{n-2p}+\left(1+(n-2p)x_2^{2n}\right)(1-u_2^{2p}z_2)z_2^{n-2p-1}\nonumber\\&+\nu \frac{n-2p}{n}(1-u_2^{2p}z_2)^{n-2p}x_2^{2p+1}, \\
	\frac{dz_2}{d\bar{\tau}_2}&=\frac{2p}{n}\left(\left(2-q+\frac{1+q}{2p}(1-u_2^{2p}z_2)\right)z_2^{n-2p}+\nu(1-u_2^{2p}z_2)^{n-2p}x_2^{2p}\right)z_2 \nonumber \\ &+2p(1-u_2^{2p}z_2)x_2^{2n-1}z_2^{n-2p}u_2^{2n(n-2p-1)}, \\
	\frac{du_2}{d\bar{\tau}_2}&= -\left((2-q)z_{2}^{n-2p}+\frac{\nu}{n} (1-u_{2}^{2p}z_{2})^{n-2p}x_2^{2p}+\left(1-u_2^{2p}z_2\right)z_2^{n-2p-1}x_2^{2n-1}u_2^{2n(n-2p-1)}\right)u_2,
	\end{align}
\end{subequations}
where
\begin{equation}
q=-1+\frac{3\gamma_\mathrm{pf}}{2}\left(1+\frac{2-\gamma_\mathrm{pf}}{\gamma_\mathrm{pf}}u_2^{2n}-x_2^{2n}u_2^{2n(n-2p)}\right).
\end{equation}
The general structure of the blow-up space $\mathcal{B}$ for the two different cases with $p<\frac{1}{2}(n-1)$ and  $p=\frac{1}{2}(n-1)$ is shown in Figure~\ref{fig:BUFL1}. In the first case, we shall see ahead that the points $\mathrm{R}^{\pm}$ are still non-hyperbolic and, therefore, we need a further blow-up, while in the second case the number of fixed points on the equator depends on whether $\nu<2n$, $\nu=2n$ or $\nu>2n$. The lines of fixed points $\mathrm{L}^{\pm}_1$ and $\mathrm{L}^{\pm}_2$ correspond to the lines of fixed points $\mathrm{L}_1$ and $\mathrm{L}_2$ on the cylinder state space $\mathbf{S}$ with $X_0<0$, or $X_0>0$, and $\Sigma_{\phi0}<0$, or $\Sigma_{\phi0}<0$, respectively, i.e. $\mathrm{L}_1=\mathrm{L}^{-}_1\cup\{\mathbf{0}\}\cup\mathrm{L}^{+}_1$ and $\mathrm{L}_2=\mathrm{L}^{-}_2\cup\{\mathbf{0}\}\cup\mathrm{L}^{+}_2$.
\begin{figure}[ht!]
	\begin{center}
		\subfigure[Case $p<\frac{1}{2}(n-1)$.]{\label{fig:BUPa}
			\includegraphics[width=0.450\textwidth]{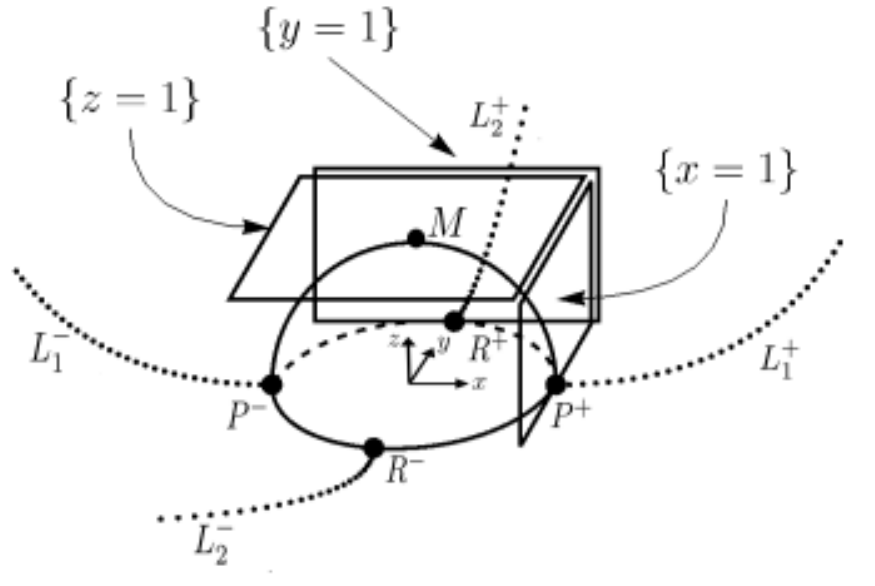}}
		\hspace{0cm}
		\subfigure[Case $p=\frac{1}{2}(n-1)$.]{\label{fig:BUPb}
			\includegraphics[width=0.450\textwidth]{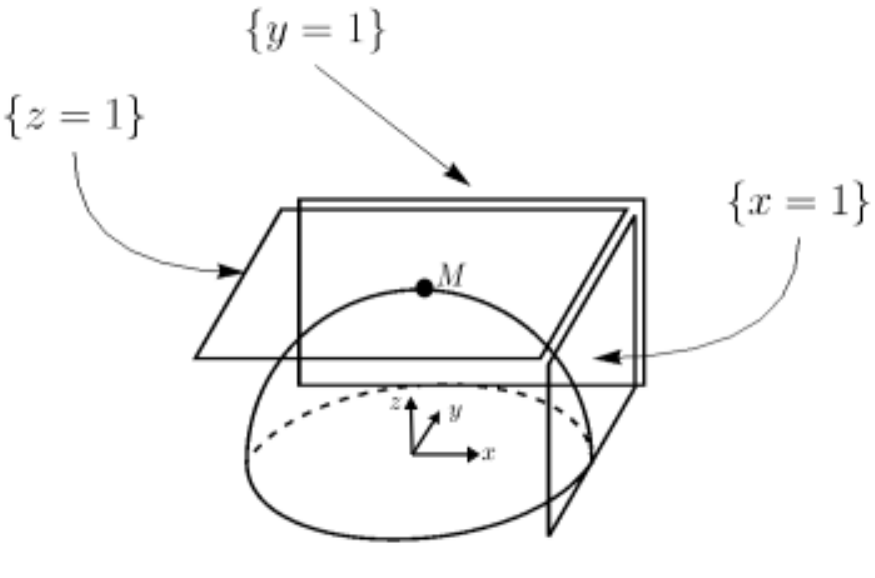}}
	\end{center}
	\vspace{-0.5cm}
	\caption{Blow-up space $\mathcal{B}$ and its model independent structure. The structure of the equator when $p=\frac{1}{2}(n-1)$ depends on whether $\nu<2n$, $\nu=2n$ or $\nu>2n$.}
	\label{fig:BUFL1}
\end{figure}

To obtain a global phase-space description, instead of projecting the upper-half of the unit $2$-sphere on the $z=1$ plane, we shall project it into the open unit disk  $x^2+y^2< 1$ which can be joined to the equator (the unit circle on $\{z=0\}$). In this way we obtain a global understanding of the flow on the Poincar\'e-Lyapunov unit cylinder $\mathbb{D}^2\times [0,\bar{u}_0]$ for some fixed $\bar{u}_0>0$. Usually, a generalised angular variable is used on the invariant subset $\{u_3=0\}$, see e.g.~\cite{BRU90,Lia66}. Here we use a different type of transformation, based on~\cite{Alho2,Alho3}, which makes the analysis somewhat simpler: 
\begin{equation}\label{PLCylinder}
\left(x_3,y_3,u_3\right)=\left(\left(\frac{r}{1-r}\right)^{\frac{1}{2p}}\cos \theta ,\left(\frac{r}{1-r}\right)^{\frac{2p+1}{2p}}F(\theta)\sin\theta,(1-r)^{\frac{1}{2p}}\bar{u}\right),
\end{equation}
where
\begin{equation}
F(\theta)=\sqrt{\frac{1-\cos^{2(2p+1)}\theta}{1-\cos^2 \theta}}=\sqrt{\sum^{2p}_{k=0}\cos^{2k}\theta},
\end{equation}
is bounded and analytic for $\theta\in[0,2\pi)$, satisfying $F(\theta)\geq1$ (with  $F\equiv1$ when $p=0$) and $F(0)=\sqrt{2p+1}$. The above transformation leads to
\begin{equation}
x_3^{2(2p+1)}+y_3^{2}=\left(\frac{r}{1-r}\right)^{\frac{2p+1}{p}}.
\end{equation}
Making a further change of time variable from $\bar{\tau}_3$ to $\xi$ defined by
\begin{equation}
\frac{d}{d\xi}=(1-r)\frac{d}{d\bar{\tau}_3},
\end{equation}
we obtain the dynamical system
\begin{subequations}\label{PL_Cylinder}
	\begin{align}
	\frac{dr}{d\xi}&=2p (1-r)r\Bigg(\frac{(1+q)(1-r)}{2pn(2p+1)}\left((2p+1)\left(n-(n-2p)(1-r)\bar{u}^{2p}\right)-n(1-\bar{u}^{2p})F^{2}(\theta)\sin^{2}\theta\right)\nonumber\\
	&+\frac{F^{2}(\theta)\sin^{2}\theta}{2p+1}\left((2-q)(1-r)+r(1-(1-r)\bar{u}^{2p})^{n-2p}\nu\cos^{2p}\theta\right)+r(1-(1-r)\bar{u}^{2p})F(\theta)\sin \theta\nonumber\\
	&-\frac{n(1-\bar{u}^{2p})}{2p+1}\bar{u}^{2n(n-2p-1)}(1-r)^{\frac{(n-1)(n-2p-1)}{p}}r^{\frac{n-1}{p}-1}F(\theta)\sin \theta\Bigg),\\
	\frac{d\theta}{d\xi}&=-\frac{F(\theta)^2}{2np}\left((1+q)(1-r)(n-(1-r)(n-2p)\bar{u}^{2p})\cos\theta+2pn r(1+(1-r)\bar{u}^{2p})F(\theta)\sin\theta\right)\sin \theta\nonumber\\&-\frac{4n}{4(2p+1)}(1-r)^{\frac{(n-2p-1)(n-1)}{p}}r^{\frac{n-p-1}{p}}\bar{u}^{2n(n-2p-1)}(1-\bar{u}^{2p})\cos^{2n}\theta F(\theta)\nonumber\\&+\frac{(1-r)}{4p(2p+1)}(1+q)(1-\bar{u}^{2p}-2p(2-q))F(\theta)^2 \sin 2\theta +2r(1+(1-r)\bar{u}^{2p})^{n-2p}\nu F(\theta)\cos^{2p}\theta \sin 2\theta ,\\
	\frac{d\bar{u}}{d\xi}&= -\frac{(1+q)(1-(1-r)\bar{u}^{2p})}{2np}(1-r)\bar{u}-\frac{\sin^2\theta}{4p(2p+1)}
	\Big(2(1-r)r\left(2p(2-q)-(1+q)(1-\bar{u}^{2p})\right)F^{2}(\theta)\nonumber\\&+4p\nu r^2(1-(1-r)\bar{u}^{2p})^{n-2p}F^2(\theta)+4pn(1-r)^{\frac{(n-2p-1)(n-1)}{p}}r^{\frac{(n-1)}{p}}\bar{u}^{2n(n-2p-1)}(1-\bar{u}^{2p})\cos^{n-1}\theta F(\theta)\Big)\bar{u}	\nonumber\\
	&+\frac{r\cos^{4p+1}\theta}{2np}\left((1+q)(1-r)(n-(1-r)(n-2p)\bar{u}^{2p})\cos\theta +2pnr(1+(1-r\bar{u}^{2p}))F(\theta)\right)\bar{u},
	\end{align}
\end{subequations}
where
$$
q=-1+\frac{3\gamma_\mathrm{pf}}{2}\left(1-(1-r)^{\frac{n(n-2p-1)}{p}}r^{\frac{n}{p}}\bar{u}^{2n(n-2p)}\cos^{2n}\theta +\frac{2-\gamma_\mathrm{pf}}{\gamma_\mathrm{pf}}(1-\cos^{2p(2p+1)}\theta)(1-r)^{\frac{n-2p-1}{p}}r^{2+\frac{1}{p}}\bar{u}^{2n}\right).
$$
On the $\{\bar{u}=0\}$ invariant boundary, the right-hand-side of the above dynamical system is regular and can be regularly extended up to $\{r=0\}$ and $\{r=1\}$, while for $\bar{u}>0$ it can be extended to the invariant boundaries $\{r=0\}$ and $\{r=1\}$ at least in a $C^1$ manner. The general structure of the Poincar\'e-Lyapunov cylinders in both cases when $p<\frac{1}{2}(n-1)$ and $p=\frac{1}{2}(n-1)$ is shown in Figure~\ref{fig:BUFL1_Lyapunov}. All fixed points are thus hyperbolic or semi-hyperbolic and, in particular, when $p<\frac{1}{2}(n-1)$ the line $\mathrm{L}_2$ no longer exists, see Figure~\ref{fig:BUFL1_Lyapunov}.
\begin{figure}[ht!]
	\begin{center}
		\subfigure[Case $p<\frac{1}{2}(n-1)$.]{\label{fig:BUP_PLa}
			\includegraphics[width=0.35\textwidth]{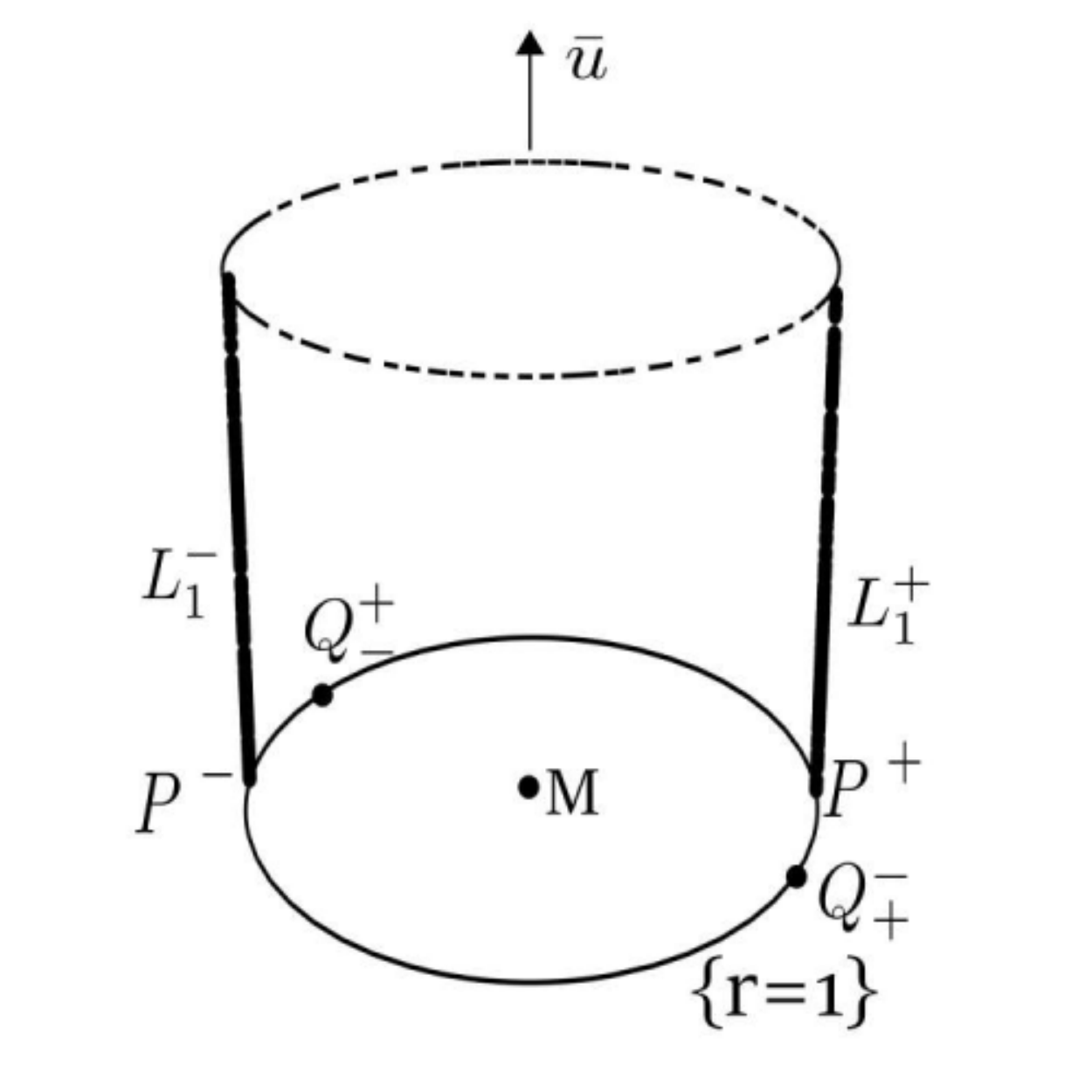}}
		\hspace{0cm}
		\subfigure[Case $p=\frac{1}{2}(n-1)$.]{\label{fig:BUP_PLb}
			\includegraphics[width=0.35\textwidth]{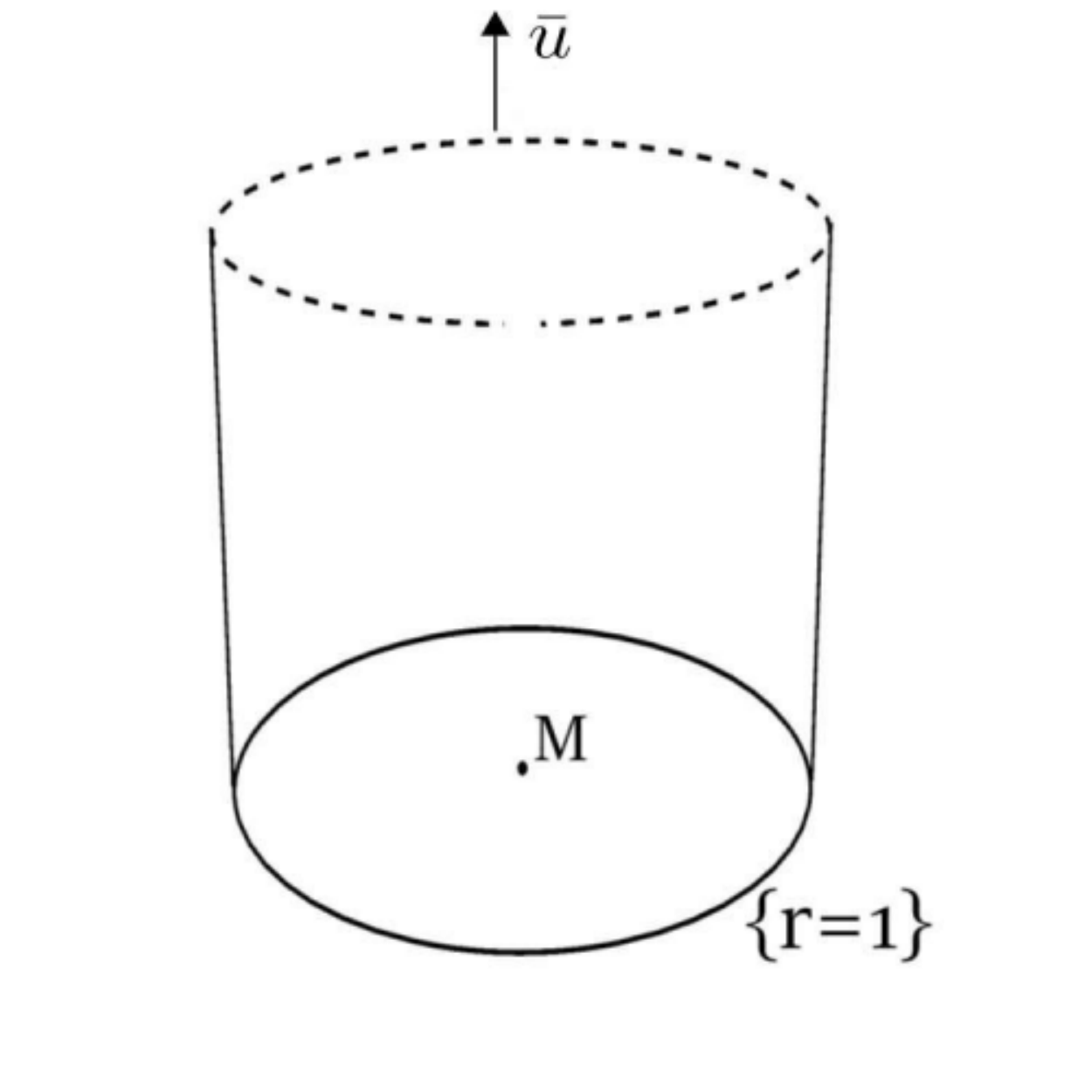}}
	\end{center}
	\vspace{-0.5cm}
	\caption{Blow-up space on the Poincar\'e-Lyapunov cylinder and its model independent structure. The structure of the equator when $p=\frac{1}{2}(n-1)$ depends on whether $\nu<2n$, $\nu=2n$ or $\nu>2n$.}
	\label{fig:BUFL1_Lyapunov}
\end{figure}
%
\subsubsection{Case $p<\frac{1}{2}(n-1)$}\label{BUP1st}
Consider the case $p<\frac{1}{2}(n-1)$ with $p>0$, i.e. $n>2$, for example $(p,n)=(1,4),(1,5),(1,6),...$, $(p,n)=(2,6),(2,7),..$ , etc. 
\subsubsection*{Positive $z$-direction}
In this case all fixed points of~\eqref{PosZ} are located at the invariant subset $\{u_3=0\}$, where the induced flow is given by
\begin{equation}\label{PosZSMA}
\frac{dx_3}{d\bar{\tau}_3}=\frac{3}{(1+2p)(1-K)}x_3+y_3,\quad \frac{dy_3}{d\bar{\tau}_3}=\left(\frac{3K}{1-K}-\nu x^{2p}_3\right)y_3,
\end{equation}
and where we introduced the notation
\begin{equation}\label{DefK}
K(\gamma_\mathrm{pf},p)=1-\frac{4p}{(1+2p)\gamma_\mathrm{pf}}.
\end{equation}
Since $\gamma_\mathrm{pf}\in(0,2)$, it follows that $K\in\left(-\infty,\frac{1}{1+2p}\right)$, i.e. $K<1$ for all $p>0$.
\begin{remark}
	By changing coordinates to $\left(\bar{x},\bar{y}\right)=\left(x_3,\frac{3x_3}{(1+2p)(1-K)}+y_3\right)$
	the system of equations~\eqref{PosZSMA} is transformed into an equivalent Li\'enard-type system
	\begin{equation}\label{LS}
	\frac{d\bar{x}}{d\bar{\tau}_3}=\bar{y},\quad \frac{d\bar{y}}{d\bar{\tau}_3}=-f(\bar{x})\bar{y}-g(\bar{x}),
	\end{equation}
	where
	\begin{subequations}
	\begin{align}
	f(\bar{x}) &= -\frac{3(1+p)K}{2p(1+2p)(1-K)}+\nu \bar{x}^{2p}, \\
	g(\bar{x}) &= \frac{3}{(1+2p)(1-K)}\left(\frac{3K}{1-K}-\nu \bar{x}^{2p}\right)\bar{x},
	\end{align}
\end{subequations}
	which arises from the second-order Li\'enard-type differential equation
	\begin{equation}
	\frac{d^2\bar{x}}{d\bar{\tau}^2_3}+f(\bar{x})\frac{d\bar{x}}{d\bar{\tau}_3}+g(\bar{x})=0.
	\end{equation}
	Introducing the functions
	\begin{equation}
	F(\bar{x})=\int^{\bar{x}}_{0} f(s)ds,\qquad G(\bar{x})=\int^{\bar{x}}_{0} g(s)ds,
	\end{equation}
	the energy of the system is $E=\frac{1}{2}\bar{y}^2+G(\bar{x})$,  
	and making a further change of variable $Y=\bar{y}+F(\bar{x})$ leads to the Li\'enard plane
	\begin{equation}\label{LIenardPlane}
	\frac{d\bar{x}}{d\bar{\tau}_3}= Y-F(\bar{x}),\qquad \frac{dY}{d\bar{\tau}_3}=-g(\bar{x}).
	\end{equation}
	There is vast amount of literature on Li\'enard-type systems, see e.g.,~\cite{Perko,BookPlanar,LMP76,DH99} and references therein. The most difficult problem in this context concerns the existence, number, relative position and bifurcations of limit cycles arising in the Li\'enard equations.
\end{remark}
The fixed points of~\eqref{PosZSMA} are the real solutions to
\begin{equation}
\left(\frac{3K}{1-K}-\nu x^{2p}_3\right)y_3=0,\quad y_3=-\frac{3}{(1+2p)(1-K)}x_3.
\end{equation}
In this case there are at most three fixed points. The fixed point at the origin of coordinates
\begin{equation}
\mathrm{M}\,:\quad x_3=0,\quad y_3=0,
\end{equation}
whose linearised system has eigenvalues
$\lambda_1 =\frac{3}{(1+2p)(1-K)}$, $\lambda_2=\frac{3K}{(1-K)}$, $\lambda_3=-\frac{3}{n(1+2p)(1-K)}$ and associated eigenvectors $v_1=(1,0,0)$, $v_2=\left(-\frac{(1+2p)(1-K)}{3\left(1-(1+2p)K\right)},1,0\right)$, $v_3=(0,0,1)$. Therefore on $\{u_3=0\}$, $\mathrm{M}$ is a hyperbolic fixed point if and only if $K\neq0$, being a saddle if $K<0$ ($0<\gamma_\mathrm{pf} < \frac{4p}{2p+1}$) and a source if $K>0$ ($\frac{4p}{2p+1}<\gamma_\mathrm{pf}$). 
When $K=0$ ($\gamma_\mathrm{pf}=\frac{4p}{2p+1}$), there is a bifurcation leading to a center manifold associated with a zero eigenvalue. 

To analyse the center manifold we introduce the adapted variable $\bar{y}_3=y_3-\frac{(1+2p)(1-K)}{3\left(1-(1+2p)K\right)}x_3$. The center manifold reduction theorem yields that the resulting system is locally topological equivalent to the 1-dimensional decoupled equation on the center manifold, which can be locally represented as graph $h: E^c\rightarrow E^u$, i.e. $x_3=h(\bar{y}_3)$, solving the nonlinear differential equation
\begin{equation*}
\begin{split}
\frac{3h(\bar{y}_3)}{(1+2p)(1-K)}-\frac{\nu(1+2p)(1-K)}{3\left(1-(1+2p)K\right)}&\left(h(\bar{y}_3)-\frac{(1+2p)(1-K)}{3\left(1-(1+2p)K\right)}\bar{y}_3\right)\bar{y}_3= \\
&=\bar{y}_3\left(\frac{3}{1-K}-\nu \left(h(\bar{y}_3)-\frac{(1+2p)(1-K)}{3\left(1-(1+2p)K\right)}\bar{y}_3\right)^{2p}\right)\frac{dh}{d\bar{y}_3},
\end{split}
\end{equation*}
subject to the fixed point $h(0)=0$ and tangency $\frac{dh}{d\bar{y}_3}(0)=0$ conditions, respectively. Approximating the solution by a formal truncated power series expansion, and solving for the coefficients, yields to leading order
\begin{equation}
\frac{d\bar{y}_{3}}{d\bar{\tau}_3}=-\nu \left(\frac{2p+1}{3}\right)^{2p}\bar{y}_{3}^{2p+1}+\mathcal{O}(\bar{y}_{3}^{2p+3}),\qquad \bar{y}_{3} \rightarrow 0,
\end{equation} 
and therefore the one dimensional center manifold is stable. In addition to $\mathrm{M}$ there are two more fixed points when $K>0$, whereas no additional fixed points exist when $K\leq0$. When $K>0$, the fixed points are
\begin{equation}
\mathrm{S}^\pm:\quad x_3=\pm \left(\frac{3K}{\nu(1-K)}\right)^{\frac{1}{2p}},\quad y_3=\mp \left(\frac{3}{(1+2p)(1-K)}\right)\left(\frac{3K}{\nu(1-K)}\right)^{\frac{1}{2p}},
\end{equation}
The linearisation around these fixed points yields the eigenvalues
\begin{equation}
\lambda_{1,2}=\frac{3\left(1\mp\sqrt{1+8p(1+2p)K}\right)}{2(1+2p)(1-K)},\quad 
\lambda_3= -\frac{3}{n(1+2p)(1-K)},
\end{equation}
with associated eigenvectors
\begin{equation}
v_{1,2}=\left(\frac{(1-K)\left(1\mp\sqrt{1+8p(1+2p)K}\right)}{12pK},1,0\right),\quad 
v_3 = (0,0,1).
\end{equation}
It follows that $\mathrm{S}^\pm$ are hyperbolic saddles.
%
\subsubsection*{Fixed points at infinity}
On the positive $x$-direction, for $p<\frac{1}{2}(n-1)$, the system~\eqref{PosX} when restricted to the invariant subset $\{u_1=0\}$ reduces to
\begin{equation}
\frac{dy_1}{d\bar{\tau}_1}=-\left[3\left(1-\frac{n-2p-1}{n-2p}\gamma_\mathrm{pf}\right) z_1^{n-2p}+\nu \right]y_1,\quad \frac{dz_{1}}{d\bar{\tau}_1}=-\frac{1}{n-2p}\left(\frac{3\gamma_\mathrm{pf}}{2}z_1-2p y_1\right)z_1^{n-2p}.
\end{equation}
Furthermore, we are interested in the invariant subset $\{z_1=0\}$ on $\{u_1=0\}$, resulting in
\begin{equation}\label{2p<n-1:EQPosX}
\frac{dy_1}{d\bar{\tau}_1}=-\nu y_1,
\end{equation}
which has only one fixed point at the origin
\begin{equation}
\mathrm{P}^{+}\quad:\quad y_1=0,\quad z_1=0,\quad u_1=0.
\end{equation}
The linearisation yields the eigenvalues $\lambda_1=-\nu$, $\lambda_2=0$, and $\lambda_3=0$ with associated eigenvectors $v_1=(1,0,0)$, $v_2=(0,1,0)$, and $v_3=(0,0,1)$. The zero eigenvalue in the $u_1$-direction is associated with a line of fixed points parameterized by constant values of $u_1=u_0>0$, and which corresponds to the half of the line of fixed points $\mathrm{L}_1$ with $X_0>0$, and denoted by $\mathrm{L}^{+}_1$, see Figure~\ref{fig:BUPa}. Thus on the  $\{u_1=0\}$ invariant set, the fixed point $\mathrm{P}^{+}$ is semi-hyperbolic, with the center manifold being the invariant subset $\{y_1=0\}$, where the flow is given by
\begin{equation}
\frac{dz_1}{d\bar{\tau}_1}=-\frac{3\gamma_\mathrm{pf}}{2(n-2p)}z^{n-2p+1}_1, \qquad \text{as}\quad z_1\rightarrow 0,
\end{equation}
and so, on $\{u_1=0\}$, $\mathrm{P}^+$ is the $\omega$-limit point of a 1-parameter set of orbits. By the symmetry of the system on the $(x_3,y_3)$ plane, an equivalent fixed point $\mathrm{P}^-$ exists in the negative $x$-direction.

On the positive $y$-direction, when $p<\frac{1}{2}(n-1)$, the flow induced on the invariant subset $\{u_2=0\}$ is given by
\begin{subequations}
	\begin{align}
	\frac{dx_2}{d\bar{\tau}_2}&=\nu \frac{n-2p}{n}x_2^{2p+1}+\left(\frac{1}{n}\left((n-2p)-\frac{3(n-2p-1)\gamma_\mathrm{pf}}{2}\right)x_2 z_2+(1+(n-2p)x_2^{2n})\right)z_2^{n-2p-1},\\
	\frac{dz_2}{d\bar{\tau}_2}&=\frac{2p}{n}\left(3\left(1-\frac{(2p+1)\gamma_\mathrm{pf}}{2}\right)z_2^{n-2p}+\nu x_2^{2p}\right)z_2.
	\end{align}
\end{subequations}
We now focus on the invariant subset $\{z_2=0\}$ on $\{u_2=0\}$, which results in
\begin{equation}\label{EQPosY}
\frac{dx_2}{d\bar{\tau}_2}=\nu \frac{n-2p}{n}x_2^{2p+1},
\end{equation}
and has only one fixed point at the origin
\begin{equation}
\mathrm{R}^{+}\quad:\quad x_2=0,\quad z_2=0,\quad u_2=0,
\end{equation}
whose linearised system has all eigenvalues zero. The zero eigenvalue in the $u_2$ direction is due to the line of fixed points parameterized by constant values of $u_2=u_0>0$, and which corresponds to the half of the line of fixed points $\mathrm{L}_2$ with $\Sigma_{\phi 0}>0$, and denoted by $\mathrm{L}^{+}_2$, see Figure~\ref{fig:BUPa}. Nevertheless it follows from~\eqref{2p<n-1:EQPosX}, and~\eqref{EQPosY}, that the equator of the Poincar\'e sphere consists of heteroclinic orbits $\mathrm{R}_{+}\rightarrow\mathrm{Q}_{\pm}$, and $\mathrm{R}_{-}\rightarrow\mathrm{Q}_{\pm}$.

To blow-up $\mathrm{R}^{+}$ and, more generally, the complete line $\mathrm{L}^{+}_2\cup\mathrm{R}^+$, we perform a cylindrical blow-up, i.e. we transform each point on the line to a circle $\mathbb{S}^1=\{(v,w)\in\mathbb{R}^2:v^2+w^2=1\}$. The blow-up space is $\bar{\mathcal{B}}=\mathbb{S}^1\times[0,u_{20})\times[0,s_0)$ for some fixed $u_{20}>0$ and $s_0>0$. We further define the quasi-homogeneous blow-up map
\begin{equation*}
\bar{\Psi}:\bar{\mathcal{B}}\rightarrow\mathbb{R}^3,\quad \bar{\Psi}(v,w,u_2,s)=(s^{n-2p-1}v,s^{2p+1}w,u_2),
\end{equation*}
and choose four charts such that
\begin{subequations}
	\begin{align}
	\bar{\psi}_{1\pm}&=\left(\pm s_{1\pm}^{n-2p-1},s_{1\pm}^{2p+1}w_{1\pm},u_2\right),\\
	\bar{\psi}_{2\pm} &= \left(s^{n-2p-1}_{2\pm}v_{2\pm},\pm s^{2p+1}_{2\pm},u_2\right).
	\end{align}
\end{subequations}
In fact we only consider the semi-circle with $w\geq0$ since $z_2\geq0$, which means that we only need to consider the blow-up in the positive $w$-direction, i.e. the directional blow-up defined by $\bar{\psi}_{2+}$. 

We start with the $v$-direction $\{v=\pm 1\}$ which after cancelling a common factor $s_{1\pm}^{2p(n-2p-1)}$ (i.e by changing the time variable $d/d\bar{\tau}_2=s_{1\pm}^{2p(n-2p-1)}d/d\tilde{\tau}_{1\pm}$) leads to the regular dynamical system
\begin{subequations}
	\begin{align}
	\frac{dw_{1\pm}}{d\tilde{\tau}_{1\pm}}&=\frac{1}{2p-n+1}\left(1-s_{1\pm}^{2p+1}u^{2p}_2w_{1\pm}\right)w_{1\pm}-\frac{1}{(n-2p-1)}\left(n(2-q)+(1+2p)(1+q)\right)s_{1\pm}^{2n}w_{1\pm}^{n-2p+1}\nonumber\\
	&+\left(\frac{1+q}{n}s_{1\pm}^{2n}w_{1\pm}\mp \frac{1+2p}{n-2p-1}\left(1+(n-2p)s_{1\pm}^{2n(n-2p-1)}\right)\right)\left(1-s_{1\pm}^{2p}u^{2p}_2w_{1\pm}\right)w_{1\pm}^{n-2p}\nonumber\\
	&\mp 2p \left(1-s_{1\pm}^{2p+1}u^{2p}_2 w_{1\pm}\right)\left(s_{1\pm}u_2\right)^{2n(n-2p-1)}w^{n-2p}_{1\pm},\nonumber \\
	\frac{ds_{1\pm}}{d\tilde{\tau}_{1\pm}}&=\frac{1}{n(n-2p-1)}\left(\frac{1}{n}\left((n-2p)(2-q)+1+q\right)s_{1\pm}^{n}w_{1\pm}^{n-2p}+(n-2p)\nu \left(1-s_{1\pm}^{2p+1}w_{1\pm}u_2\right)^{n-2p}\right)s_{1\pm}\nonumber\\
	&\pm \frac{1}{n-2p-1}\left(1+(n-2p)s_{1\pm}^{2n(n-2p-1)}\right)\left(1-s_{1\pm}w_{1\pm}u^{2p}_2\right)w^{n-2p-1}_{1\pm}s_{1\pm}, \nonumber\\
	\frac{du_2}{d\tilde{\tau}_{1\pm}}&=\left(-(2-q)s_{1\pm}^n w_{1\pm}^{n-2p}-\frac{\nu}{n}(1-s_{1\pm}^{2p+1}w_{1\pm}u^{2p}_2)^{n-2p}\mp \left(1-s_{1\pm}^{2p+1}w_{1\pm}u^{2p}_2\right)\left(s_{1\pm}u_2\right)^{2n(n-2p-1)}w_{1\pm}^{n-2p-1}\right)u_2,\nonumber
	\end{align}
\end{subequations}
where
\begin{equation}
q=-1+\frac{3\gamma_\mathrm{pf}}{2}\left(1+\frac{2-\gamma_\mathrm{pf}}{\gamma_\mathrm{pf}}u^{2n}_2-u^{2n(n-2p)}_2s_{1\pm}^{2n(n-2p-1)}\right).
\end{equation}
The above system has the fixed points
\begin{equation}
\mathrm{T}^{+}_{\pm}:\quad w_{1\pm}=0,\quad s_{1\pm}=0,\quad u_2=0,
\end{equation}
whose linearised system has eigenvalues $-\frac{\nu}{n-2p-1}$, $-\frac{\nu}{n}$, and $\frac{(n-2p)}{n(n-2p-1)}\nu$ with eigenvectors the canonical basis of $\mathbb{R}^3$, and the fixed points
\begin{equation}
\mathrm{Q}^{+}_{\pm}:\quad w_{1\pm}=\mp \left(\frac{\nu}{2p+1}\right)^{\frac{1}{n-2p-1}} ,\quad s_{1\pm}=0,\quad u_2=0,
\end{equation}
where only $\mathrm{Q}^{+}_{-}$ exists in the region $w_{1\pm}\geq0$. The eigenvalues of the linearised system around $\mathrm{Q}^{+}_{-}$ are $\frac{2p \nu}{n(2p+1)}$, $\nu$, and $-\frac{\nu}{n}$ with associated eigenvectors the canonical basis of $\mathbb{R}^3$.

In the positive $w$-direction $\{w=1\}$, and after cancelling a common factor $s_{2+}^{2p(n-2p-1)}$ (i.e. by changing the time variable $d/d\bar{\tau}_2=s_{2+}^{2p(n-2p-1)}d/d\tilde{\tau}_{2+}$), we get the regular dynamical system
\begin{subequations}
	\begin{align}
	\frac{dv_{2+}}{d\tilde{\tau}_{2+}}&=\frac{n(2-q)+(2p+1)(1+q)}{n(2p+1)}s^n_{2+}v_{2+}+\frac{\nu}{1+2p}\left(1-s_{2+}u^{2p}_2\right)^{n-2p}v_{2+}^{2p+1}\nonumber \\
&+\left(1+(n-2p)s_{2+}^{2n(n-2p-1)}v_{2+}^{2n}\right)\left(1-s_{2+}^{2p+1}u^{2p}_2\right)\nonumber \\
&-\frac{n-2p-1}{2p+1}\left(\frac{1+q}{n}s_{2+}^{2n}+\left(s_{2+}u_2\right)^{2n(n-2p-1)}2pv_{2+}^{2n-1}\right)\left(1-s_{2+}^{2p+1}u^{2p}_2\right)v_{2+}, \nonumber \\
	\frac{ds_{2+}}{d\tilde{\tau}_{2+}}&=\frac{1}{2p+1}\left(\frac{1+q}{n}s_{2+}^{n+1}+2p\left(s_{2+}u_2\right)^{2n(n-2p-1)}v_{2+}^{2n-1}\right)\left(1-s_{2+}^{2p+1}u^{2p}_2\right)s_{2+}\nonumber \\
	&+\frac{1}{n(2p+1)}\left(2p(2-q)s_{2+}^{n}+\nu\left(1-s_{2+}^{2p+1}u^{2p}_2\right)^{n-2p}v_{2+}^{2p}\right),\nonumber \\
	\frac{du_2}{d\tilde{\tau}_{2+}}&= \left(-(2-q)s_{2+}^n-\frac{\nu}{n}(1-s_{2+}^{2p+1}u^{2p}_2)^{n-2p}v_{2+}^{2p}- \left(1-s_{2+}^{2p+1}u^{2p}_2\right)v_{2+}^{2n-1}\left(s_{2+}u_2\right)^{2n(n-2p-1)}\right)u_2,\nonumber
	\end{align}
\end{subequations}
where
\begin{equation}
q=-1+\frac{3\gamma_\mathrm{pf}}{2}\left(1+\frac{2-\gamma_\mathrm{pf}}{\gamma_\mathrm{pf}}u^{2n}_2-s_{2+}^{2n(n-2p-1)}u^{2n(n-2p)}_2v_{2+}^{2n}\right).
\end{equation}
This system only has the fixed point $\mathrm{Q}^{+}_{-}$ which is located at $v_{2+}=-\left(\frac{2p+1}{\nu}\right)^{\frac{1}{2p+1}}$, $s_{2+}=0$, and $u_2=0$. In turn, the linearised system around $\mathrm{Q}^{+}_{-}$ has eigenvalues $\left(\frac{2p+1}{\nu}\right)^{\frac{2p}{2p+1}}$, $\frac{2p \nu}{n(2p+1)}\left(\frac{2p+1}{\nu}\right)^{\frac{2p}{2p+1}}$, and $-\frac{\nu}{n}\left(\frac{2p+1}{\nu}\right)^{\frac{2p}{2p+1}}$ with associated eigenvectors the canonical basis of $\mathbb{R}^3$. The blow-up of $\mathrm{L}^+_2\cup\mathrm{R}^{+}$ is shown in Figure~\ref{fig:BUP_PL}. 

The blow-up of the equivalent non-hyperbolic set $\mathrm{L}^{-}_{2}\cup\mathrm{R}^{-}$ follows by symmetry considerations from which we deduce the existence of the equivalent fixed points $\mathrm{T}^{-}_{\pm}$, and $\mathrm{Q}^{-}_{+}$.
\begin{figure}[ht!]
	\begin{center}
		\includegraphics[width=0.75\textwidth]{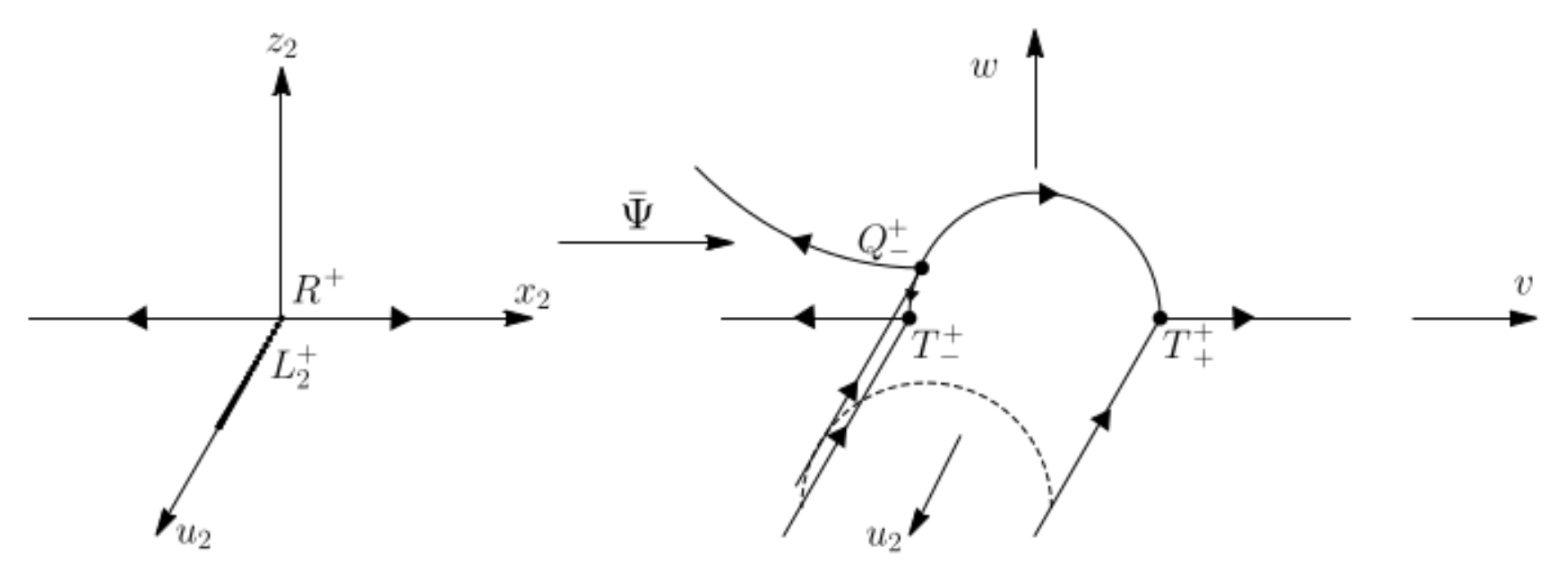}
	\end{center}
	\vspace{-0.5cm}
	\caption{Blow-up of the non-hyperbolic line of fixed points $\mathrm{L}^{+}_2\cup\mathrm{R}^{+}$.}
	\label{fig:BUP_PL}
\end{figure}
Moreover, since all fixed points are located on the invariant subset $\{u_2=0\}$, we have the following result concerning the line $\mathrm{L}_2$ on the cylinder state-space $\mathbf{S}$:
\begin{lemma}\label{L2BU}
	No interior orbit in $\mathbf{S}$ converges to the points on the set $\mathrm{L}_2\setminus\mathrm{FL}_1$.
\end{lemma}
\subsubsection*{Global phase-space on the Poincar\'e-Lyapunov disk}
The previous results can be collected in a global phase-space by employing the Poincar\'e-Lyapunov compactification. This compactification has the advantage that all fixed points are hyperbolic or semi-hyperbolic and, in particular, on the Poincar\'e-Lyapunov cylinder the line $\mathrm{L}_2$ is absent. 

When $p<\frac{1}{2}(n-1)$, we obtain from~\eqref{PL_Cylinder} that the induced flow on the $\{\bar{u}=0\}$ invariant subset is given by
\begin{subequations}
	\begin{align*}
	\frac{dr}{d\xi}&=2p(1-r)r\left(-\frac{\nu r}{2p+1}\left(1-\cos^{2(2p+1)}\theta\right)\cos^{2p}\theta+\frac{12p (1-r)}{(2p+1)^2(1-K)}\left(\frac{K}{1-K}+\cos^{2(2p+1)}\theta\right)\right)\nonumber\\&+2p(1-r)r^2 F(\theta)\sin \theta, \\
	\frac{d\theta}{d\xi} &= -F(\theta)\left(\frac{3}{2p+1}F(\theta)\sin 2\theta (1-r)+\left(1+\frac{\nu}{2(2p+1)}\sin2 \theta F(\theta)\cos^{2p}\theta-\cos^{2(2p+1)}\theta\right)r \right).
	\end{align*}
\end{subequations}
At $\{r=0\}$ lies the fixed point $\mathrm{M}$ which is the origin of the $(x_3,y_3)$ plane. The fixed point $\mathrm{M}$ is a hyperbolic saddle for $K<0$, a center-saddle when $K=0$ and a hyperbolic source when $K>0$. When $K>0$ we have two additional saddle fixed points $\mathrm{S}^\pm$ that are located at
\begin{subequations}
	\begin{align*}
	\frac{r_{\mathrm{S}^\pm}}{1-r_{\mathrm{S}^\pm}}&=\left(\frac{9K}{\nu(1+2p)(1-K)^2}\right)^{\frac{1}{2p+1}}\left[1+\left(\frac{(2p+1)K}{\nu}\right)^2\right]^{\frac{p}{2p+1}}, \\
	\theta_{\mathrm{S}^\pm} &= \arccos\left(\pm \left(1+\left(\frac{(2p+1)K}{\nu}\right)^{-2}\right)^{-\frac{1}{2(2p+1)}}\right).
	\end{align*}
\end{subequations}
The points at infinity in the $(x_3,y_3)$ plane are now located in the $\{r=1\}$ invariant set. The hyperbolic sinks $\mathrm{P}^{\pm}$ and the hyperbolic sources $\mathrm{Q}^{+}_{-}$ and $\mathrm{Q}^{-}_{+}$ are given by
\begin{subequations}
	\begin{align*}
	\theta_{\mathrm{P}^+} &=0 ,\qquad \theta_{\mathrm{P}^-}= \pi, \\
	\theta_{\mathrm{Q}^{+}_{-}} &=\arccos\left(-\left(\frac{(2p+1)^2}{(2p+1)^2+4\nu^2}\right)^{\frac{1}{2(2p+1)}} \right) ,\qquad 	\theta_{\mathrm{Q}^{-}_{+}} =\arccos\left(\left(\frac{(2p+1)^2}{(2p+1)^2+4\nu^2}\right)^{\frac{1}{2(2p+1)}} \right).
	\end{align*}
\end{subequations}
\begin{figure}[ht!]
	\begin{center}
		\subfigure[Case $0<\gamma_\mathrm{pf} < \frac{4p}{2p+1}$, i.e. $K\in(-\infty,0)$.]{\label{}
			\includegraphics[width=0.30\textwidth]{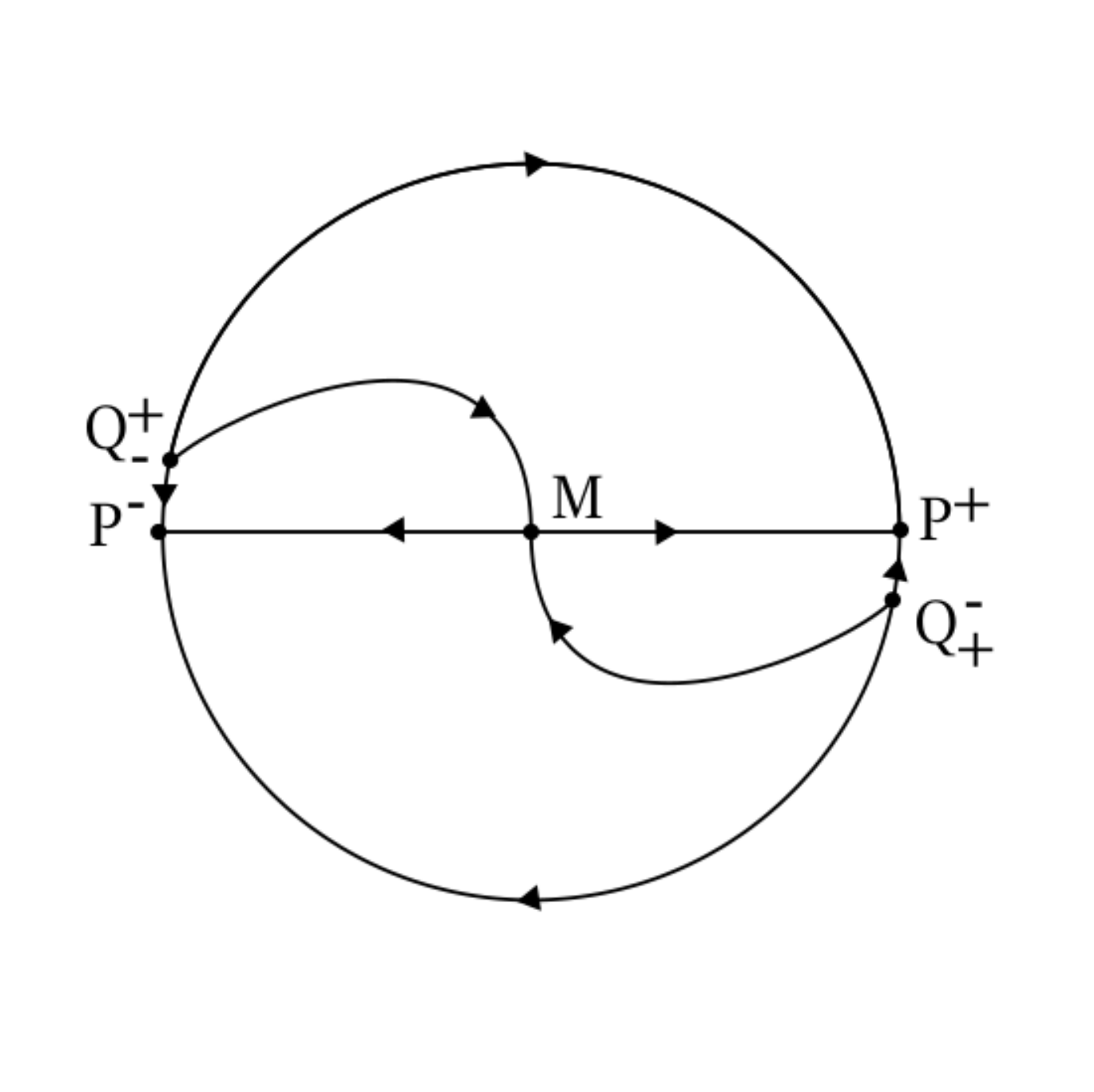}}
				\subfigure[Case $\gamma_\mathrm{pf}=\frac{4p}{2p+1}$, i.e.  $K=0$.]{\label{}
			\includegraphics[width=0.30\textwidth]{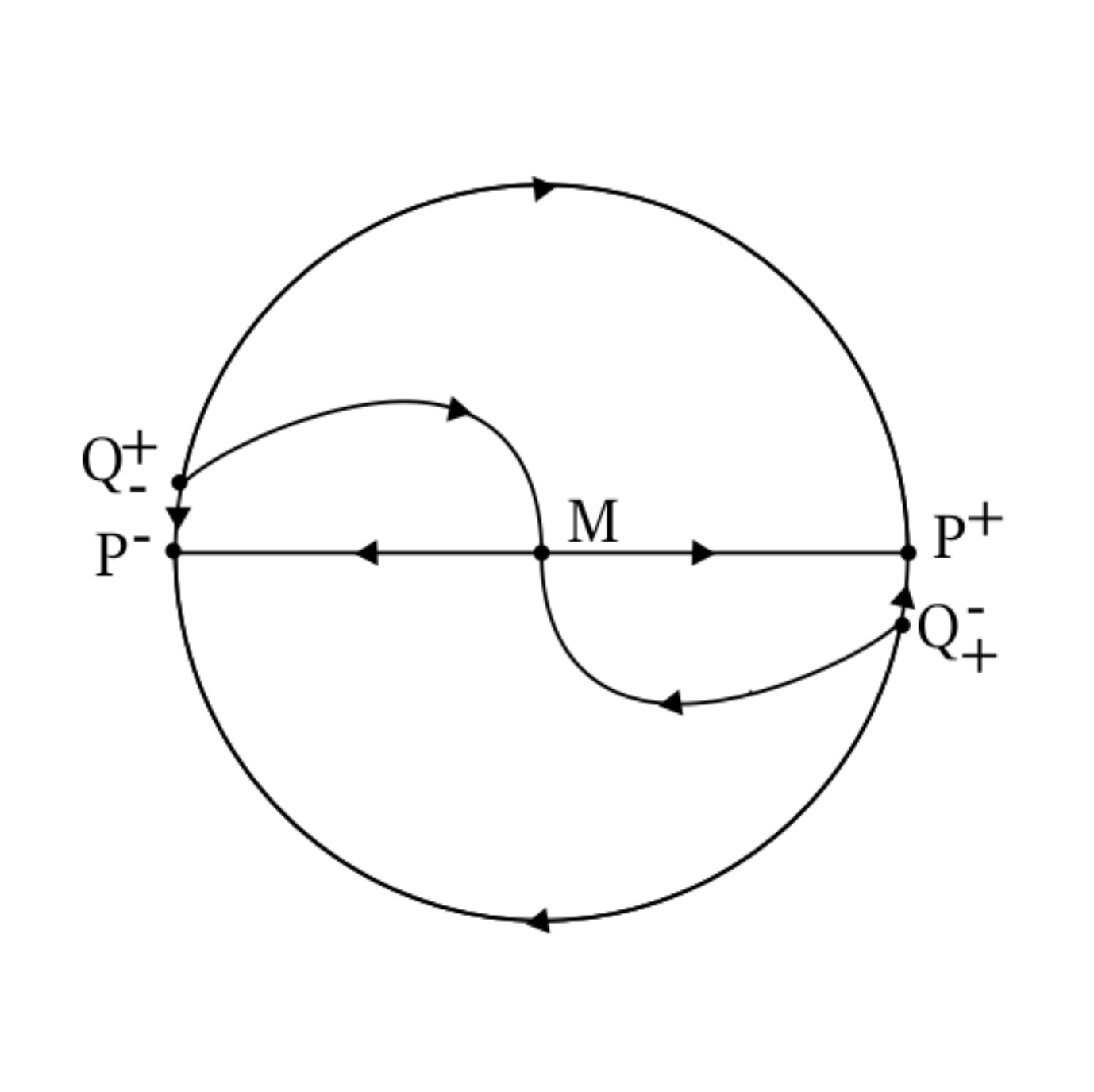}}
		\subfigure[Case  $\frac{4p}{2p+1}<\gamma_\mathrm{pf}<2$, i.e. $K\in(0,\frac{1}{1+2p})$.]{\label{}
			\includegraphics[width=0.30\textwidth]{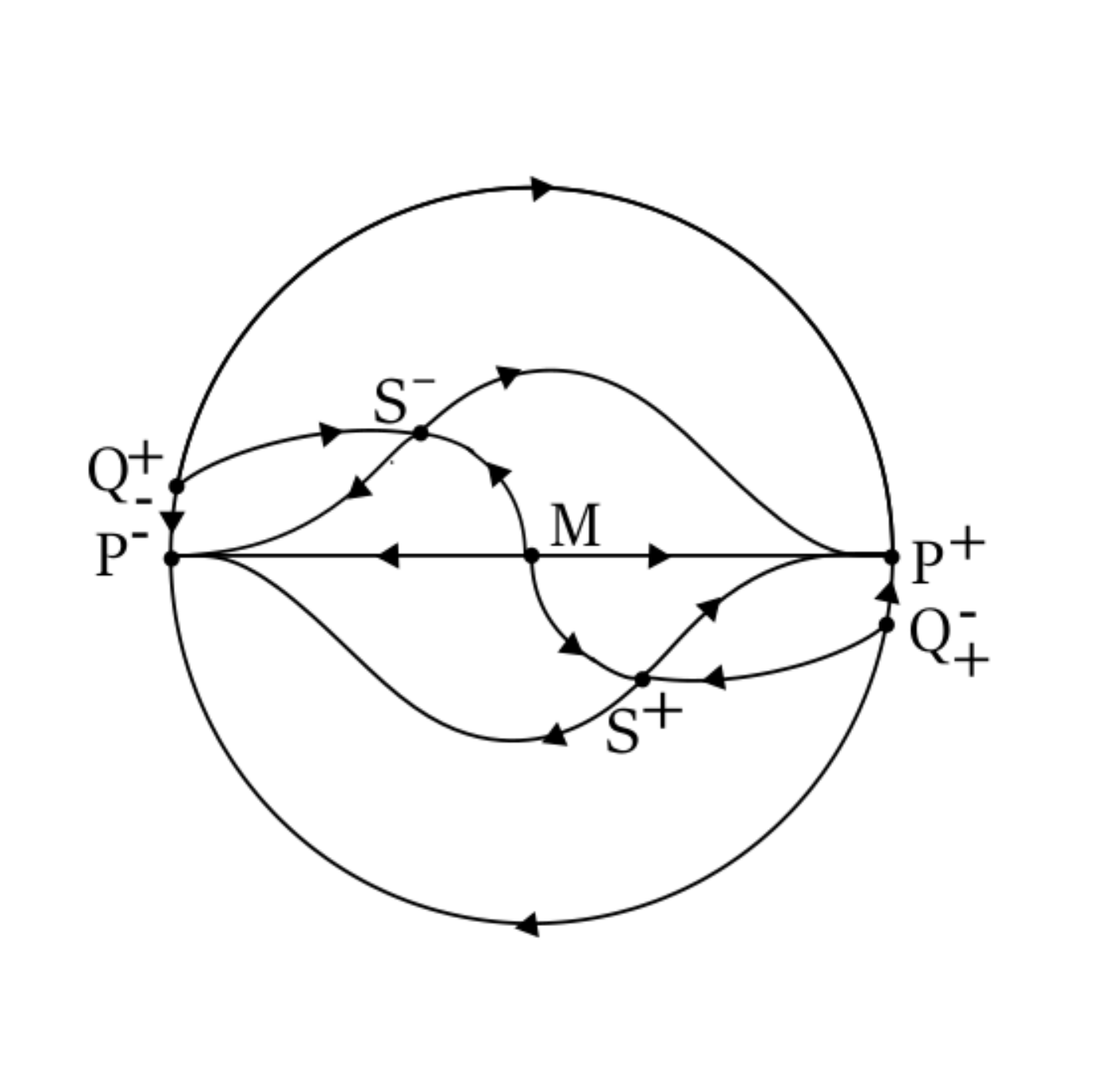}}
	\end{center}
	\vspace{-0.5cm}
	\caption{Poincar\'e-Lyapunov disk when $p<\frac{1}{2}(n-1)$ with $p>0$.}
	\label{fig:PP0}
\end{figure}
\begin{theorem}
	Let $p<\frac{1}{2}(n-1)$ with $p>0$. Then for all $\nu>0$, the Poincar\'e-Lyapunov disk  consists of heteroclinic orbits connecting the fixed points $\mathrm{M}$, $\mathrm{P}^\pm$, $\mathrm{Q}^{+}_{-}$, $\mathrm{Q}^{-}_{+}$, and $\mathrm{S}^\pm$ when they exist, with the separatrix skeleton as depicted in Figure~\ref{fig:PP0}.
\end{theorem}
\begin{proof}
	First notice that $\{y_3=0\}$ is an invariant subset consisting of heteroclinic orbits $\mathrm{M}\rightarrow \mathrm{P}^\mathrm{\pm}$ which splits the phase-space into two invariant sets $\{y_3>0\}$ and $\{y_3<0\}$. On each of these invariant sets there are no fixed points when $K\leq0$, and if $K>0$ there is a single fixed point which is a saddle. Therefore by the index theorem there are no periodic orbits in each of these regions. Since closed saddle connections cannot exist either, then by the Poincar\'e-Bendixson theorem the $\omega$ and $\alpha$-limit sets of all orbits in the Poincar\'e-Lyapunov disk are the fixed points $\mathrm{M}$, $\mathrm{P}^\pm$, $\mathrm{Q}^{+}_{-}$, $\mathrm{Q}^{-}_{+}$, and $\mathrm{S}^\pm$ when they exist and the phase-space consists of heteroclinic orbits connecting these fixed points. In particular, from the local stability properties of the fixed points, it is straightforward to show that when $K\leq0$ there are two separatrices $\mathrm{Q}^{+}_{-}\rightarrow\mathrm{M}$ and $\mathrm{Q}^{-}_{+}\rightarrow\mathrm{M}$ which further split the regions $y_3>0$ and $y_3<0$ into two invariant subsets. The flow on these subsets is trivial. When $0<K<\frac{1}{1+2p}$ there are four separatrices $\mathrm{Q}^{+}_{-}\rightarrow \mathrm{S}^-$, $\mathrm{M}\rightarrow\mathrm{S}^-$ and $\mathrm{Q}^{-}_{+}\rightarrow \mathrm{S}^+$, $\mathrm{M}\rightarrow\mathrm{S}^+$ which further split the invariant regions $y_3>0$ and $y_3<0$ into four invariant subsets where the flow is also trivial.
\end{proof}

\begin{remark}\label{AsFL1}
	It is interesting to obtain the asymptotics for the orbits on the cylinder $\mathbf{S}$ towards $\mathrm{FL}_1$ when $p<(n-1)/2$. For example when $0<\gamma_\mathrm{pf}<\frac{2n}{n+1}$, there exists a 1-parameter family of orbits in $\mathbf{S}$ with the following asymptotics towards $\mathrm{FL}_1$ as $\tau\rightarrow+\infty$
	\begin{subequations}
		\begin{align*}
		X(\tau) &=  -\frac{2C}{3(2-\gamma_{\mathrm {pf}})}\left(1+\frac{3\gamma_{\mathrm{pf}}}{2n}(n-2p)\tau\right)^{-\frac{n+1}{(n-2p)\gamma_{\mathrm{pf}}}\left(\frac{2n}{n+1}-\gamma_{\mathrm{pf}}\right)}, \\
		\Sigma_\phi(\tau) &= C\left(1+\frac{3\gamma_{\mathrm{pf} }}{2n}(n-2p)\tau\right)^{-\frac{n\gamma_{\mathrm{pf}}}{n-2p}(2-\gamma_{\mathrm{pf}})}, \\
		T(\tau) &=1-\left(1+\frac{3\gamma_{\mathrm {pf}}}{2n}(n-2p)\tau\right)^{-\frac{1}{n-2p}},
		\end{align*}
	\end{subequations}
	with $C\in\mathbb{R}$ a constant, which is obtained via the linearised solution at $\mathrm{M}$ restricted to the 2-dimensional stable manifold when $K<0$. Similarly, one can obtain the asymptotics towards $\mathrm{FL}_1$ when $K=0$, by making use of the flow on the center manifold of $\mathrm{M}$, and when $0<K<\frac{1}{1+2p}$ via the linearised solution at $\mathrm{S}^{\pm}$ restricted to the 2-dimensional stable manifold. 
%
\end{remark}

\subsubsection{Case $p=\frac{1}{2}(n-1)$}\label{BUP2nd}
Consider now the case $p=\frac{1}{2}(n-1)$ with $p>0$, $n>1$ with $n$ odd, i.e. $(p,n)=(1,3),(2,5),(3,7),etc.$.
\subsubsection*{Positive $z$-direction}
Setting $p=\frac{1}{2}(n-1)$ in~\eqref{PosZ} leads to
\begin{subequations}\label{X3EQ}
	\begin{align}
		\frac{dx_3}{d\bar{\tau}_3} &= \frac{1}{n-1}\left(1-\frac{u^{n-1}}{n}\right)(1+q)x_3+(1-u_3^{n-1})y_3,\\
\frac{dy_3}{d\bar{\tau}_3} &= \frac{1}{n-1}(1+q)(1-u^{n-1})y_3-\big(2-q+\nu(1-u_3^{n-1})x_3^{n-1} \big)y_3-n(1-u_3^{n-1})x_3^{2n-1},\\
\frac{du_3}{d\bar{\tau}_3}&= -\frac{1}{n(n-1)}(1+q)(1-u_3^{n-1})u_3	,
	\end{align}
\end{subequations}
where
\begin{equation}
q=-1+\frac{3\gamma_\mathrm{pf}}{2}+\frac{3\gamma_\mathrm{pf}}{2}\left(\frac{(2-\gamma_\mathrm{pf})}{\gamma_\mathrm{pf}}y_3^2-x_3^{2n}\right)u_3^{2n}.
\end{equation}
Since $n$ is odd, the system is symmetric under the transformation $(x_3,y_3,u_3)\rightarrow (-x_3,-y_3,u_3)$, and all fixed points lie on the invariant subset $\{u_3=0\}$ where the induced flow is given by
\begin{equation}\label{PosZEQ}
	\frac{dx_3}{d\bar{\tau}_3} = \frac{3}{n(1-K)}x_3+y_3,\qquad 
	\frac{dy_3}{d\bar{\tau}_3} = \frac{3}{n(1-K)} n K y_3-n x^{2n-1}_3-\nu y_3 x^{n-1}_3,
\end{equation}
and $K$ defined in \eqref{DefK} reduces to
\begin{equation}
K(\gamma_\mathrm{pf},n)=1-\frac{2(n-1)}{n\gamma_\mathrm{pf}},
\end{equation}
with $K\in\left(-\infty,\frac{1}{n}\right)$. Since $n>1$, it follows that $K<1$ for all $n$. 
\begin{remark}
	The system~\eqref{PosZEQ} can be transformed into the Li\'enard type system~\eqref{LS}, where now the functions $f(\bar{x})$ and $g(\bar{x})$ are given by
		\begin{subequations}\label{Lienard2}
		\begin{align}
		f(\bar{x}) &= \nu \bar{x}^{n-1}-\frac{3(1+n K)}{n(1-K)}, \\
		g(\bar{x}) &= n \bar{x}^{2n-1}-\frac{3}{n(1-K)}\nu \bar{x}^{n}+\left(\frac{3}{n(1-K)}\right)^2 n K \bar{x}.
		\end{align}
	\end{subequations}
\end{remark}
The fixed points of~\eqref{PosZEQ} are the real solutions to 
\begin{equation}
-x_3\left(n x^{2(n-1)}_3-\frac{3}{n(1-K)}\nu x^{n-1}_3+\left(\frac{3}{n(1-K)}\right)^2 nK\right)=0,\qquad y_3=-\frac{3}{n(1-K)}x_3.
\end{equation}
The first equation admits at most five real solutions. The origin of coordinates is always a fixed point
\begin{equation}
\mathrm{M}:\quad x_3=0,\quad y_3=0,
\end{equation}
and the linearised system at $\mathrm{M}$ has eigenvalues $\frac{3}{n^2(1-K)}$, $\frac{3}{n(1-K)}nK$ and $-\frac{3}{n^2(1-K)}$, with associated eigenvectors $(1,0,0)$, $(-\frac{n(1-K)}{3(1-nK)},1,0)$ and $(0,0,1)$.
On $\{u_3=0\}$, $\mathrm{M}$ is a hyperbolic fixed point if and only if $K\neq0$, being a saddle if $K<0$ (i.e. $0<\gamma_\mathrm{pf} < \frac{2(n-1)}{n}$) and a source if $K>0$  (i.e. $\frac{2(n-1)}{n}<\gamma_\mathrm{pf}<2$). 

When $K=0$ ($\gamma_\mathrm{pf}=\frac{2(n-1)}{n}$) there is a bifurcation leading to a center manifold associated with a zero eigenvalue. To analyse the center manifold we introduce the adapted variables $\bar{y}_3=y_{3}-\frac{n(1-K)}{3(1-nK)}x_{3}$. The center manifold can then be locally represented as the graph $h: E^c\rightarrow E^u$ i.e. $x_3=h(\bar{y}_3)$, which solves the differential equation
\begin{equation*}
	\aligned
	&	\frac{3h(\bar{y}_3)}{n(1-K)}+\frac{n(1-K)}{1-nK}\Bigg[\frac{1}{3}\left(h(\bar{y}_3)-\frac{n(1-K)}{3(1-nK)}\bar{y}_3\right)^{2n-1}+\left(1-\frac{2}{n}-\frac{2}{3}\left(h(\bar{y}_3)-\frac{n(1-K)}{3(1-nK)}\bar{y}_3\right)^{n-1}\right)\bar{y}_3\Bigg]= \\
	&=\left(\frac{3K \bar{y}_3}{n(1-K)}-\nu \bar{y}_3\left(h(\bar{y}_3)-\frac{n(1-K)}{3(1-nK)}\bar{y}_3\right)^{n-1}-n\left(h(\bar{y}_3)-\frac{n(1-K)}{3(1-nK)}\bar{y}_3\right)^{2n-1}\right)\frac{dh}{d\bar{y}_3}
	\endaligned
\end{equation*}
subject to the fixed point $h(0)=0$ and tangency $\frac{dh(0)}{d\bar{y}_3}=0$ conditions respectively. Approximating the solution by a formal truncated power series expansion and solving for the coefficients yields
\begin{equation}\label{CM_M1}
\frac{dy_3}{d\bar{\tau}_3}=-\nu\left(\frac{n}{3}\right)^{n-1}y^{n}_3\left(1-\frac{n}{\nu}\left(\frac{n}{3}\right)^{n}y^{n-1}_3\right),\quad y_3\rightarrow 0,
\end{equation}
and therefore $\mathrm{M}$ is a center-saddle in this case.  
The remaining four fixed points that may or not exist, depending on the parameters range, are 
\begin{equation}
\mathrm{S}^{\pm}_\pm\,:\quad	x_3=\pm \left(\frac{3A_\pm}{n(1-K)}\right)^{\frac{1}{n-1}},\qquad
y_3=\mp \left(\frac{3}{n(1-K)}\right)^{\frac{n}{n-1}}A^{\frac{1}{n-1}}_\pm ,
\end{equation}
with
\begin{equation}
A_\pm = \frac{\nu}{2n}\pm \sqrt{\left(\frac{\nu}{2n}\right)^2-K},
\end{equation}
where the subscripts $\pm$ stand for $A_\pm$ and the superscripts $\pm$ stand for the sign of $x_3$ at the fixed point. 

To study the existence of the fixed points $\mathrm{S}^{\pm}_{\pm}$ we note that since $n>1$, with $n$ odd, then $A_\pm$ must be positive constants. For all $K\leq0$ only $A_+$ is positive, while for $K>0$, $A_\pm$ are real and positive if and only if $\nu\geq2n\sqrt{K}$. The linearisation around the fixed points $\mathrm{S}^{\pm}_\pm$ yields the eigenvalues

\begin{subequations}
	\begin{align}
	&\lambda_{1_{\pm}}= \frac{3\left(1+nK-\nu A_\pm\right)}{2n(1-K)}+\sqrt{\frac{4}{9}(1-nK)^2+2\left(n-1+n(1-K)\right)A_\pm +(\nu^2-8n^2+4n)A_\pm},\nonumber\\
	&\lambda_{2_{\pm}}= \frac{3\left(1+nK-\nu A_\pm\right)}{2n(1-K)}-\sqrt{\frac{4}{9}(1-nK)^2+2\left(n-1+n(1-K)\right)A_\pm +(\nu^2-8n^2+4n)A_\pm},\nonumber \\
	&\lambda_3=-\frac{3}{n^2(1-K)}, \nonumber
	\end{align}
\end{subequations}
with associated eigenvectors
\begin{equation*}
v_1=(a_1,1,0),\quad v_2=(a_2,1,0),\quad v_3=(0,0,1)
\end{equation*}
and
\begin{equation*}
a_i=\frac{\lambda_i n^2 (1-K)^2-6n(n-1)(1-K)}  {9A_\pm \left( n(2n-1)A_\pm+\nu(n-1)\right)}.
\end{equation*}
To study the character and stability of the fixed points $\mathrm{S}^{\pm}_+$ it is helpful to introduce the quantities
\begin{equation}
\begin{split}
f_\pm(K,n)=&\sqrt{n}\Bigg[(2n-1)\left(1-K\right)\frac{ 1-nK\left(n\frac{10n-9}{2n-1}-\frac{(3n-1)^2}{n\left(1-K\right)}\right)}{1+4n(n-1)K} \nonumber \\
& \qquad\qquad\qquad\qquad\qquad\pm\sqrt{\frac{4n(n-1)}{1-K}}\left(1-K\right)^2 \frac{\Big|-1-nK\left(n+\frac{(n^2-1)}{n(1-K)}\right)\Big|}{1+4n(n-1)K}\Bigg]^{1/2},
\end{split}
\end{equation}
for $K\neq -\frac{1}{4n(n-1)}$. Moreover, and for future purposes, let
\begin{equation*}
	g_{\pm}^{\mp}(n)= \frac{3}{n}-4\mp4\sqrt{\frac{n-1}{n}}\pm 2 \sqrt{10+\frac{2\mp 12\sqrt{n(n+1)}(3-5n)-11n}{n^2}}
\end{equation*}
which satisfy $g^{\pm}_{+}(n)>\frac{1}{n}$ and $g^{\pm}_{-}(n)<-\frac{1}{n}$.

Just as in the stability properties of the fixed point $\mathrm{M}$, the sign of $K$ plays a prominent role in the qualitative properties of the phase-space. Hence we shall split the analysis into two subcases $K\leq0$ and $K>0$.
\begin{itemize}
	\item[a)] {\bf Subcase} $-\infty<K\leq0$, i.e. $0<\gamma_\mathrm{pf}\leq\frac{2(n-1)}{n}$: 
	
	Since for $K\leq0$, only $A_+$ is positive, then, and in addition to $\mathrm{M}$, only $\mathrm{S}^{\pm}_{+}$ exist.
	
    For $K<-\frac{1}{4n(n-1)}$, $f_{-}$ is real and satisfies $f_{-}(K,n)>\sqrt{n}(1+n K)$, while $f_+$ is imaginary. In this case the pair of eigenvalues $\lambda_{1+},\lambda_{2+}$ are real if $\nu\geq f_{-}(K,n)$ and complex otherwise. Moreover, if $0<\nu < (1+n K)\sqrt{n}$, then $\lambda_{1+},\lambda_{2+}$ have positive real part, if $\nu>(1+n K)\sqrt{n}$ then $\lambda_{1+},\lambda_{2+}$ have negative real part, and  if $\nu=(1+n K)\sqrt{n}>0$ both eigenvalues are purely imaginary. Therefore on the  $\{u_3=0\}$ invariant subset, $\mathrm{S}^{\pm}_{+}$ are unstable strong focus if $\nu <(1+n K)\sqrt{n}$, a center or weak focus if $\nu=(1+n K)\sqrt{n}$, a stable strong focus if $(1+n K)\sqrt{n}<\nu<f_{-}(K,n)$ and a stable node if $\nu \geq f_-(K,n)$.
    Notice that the cases $\nu \leq (1+n K)\sqrt{n}$ only exists when $-\frac{1}{n}<K<-\frac{1}{4n(n-1)}$, i.e. when $\frac{2(n-1)}{n+1}<\gamma_\mathrm{pf}\leq\frac{8(n-1)^2}{(2n-1)^2}$ holds.
	Finally we also note that $f_{-}(K,n)<2n$ for $K>g^{-}_{-}(n)$, $f_{-}(K,n)=2n$ for $K=g^{-}_{-}(n)$ and $f_{-}(K,n)>2n$ for $K<g^{-}_{-}(n)$.

	For $K=-\frac{1}{4n(n-1)}$. the pair of eigenvalues $\lambda_{1+}$, $\lambda_{2+}$ are real if $\nu\geq \frac{7+8n(2n-3)}{4(n-1)}\sqrt{\frac{n}{2(2n-1)}}$ and complex otherwise. Moreover if $0<\nu<\frac{\sqrt{n}(4n-5)}{4(n-1)}$, then $\lambda_{1+}$, $\lambda_{2+}$ have positive real part, if $\nu>\frac{\sqrt{n}(4n-5)}{4(n-1)}$ then $\lambda_{1+}$, $\lambda_{2+}$ have negative real part, and if $\nu=\frac{\sqrt{n}(4n-5)}{4(n-1)}$, then  both eigenvalues are purely imaginary. Therefore on the  $\{u_3=0\}$ invariant subset, $\mathrm{S}^{\pm}_{+}$ are unstable focus if $\nu<\frac{\sqrt{n}(4n-5)}{4(n-1)}$, a center or weak focus if $\nu=	\frac{\sqrt{n}(4n-5)}{4(n-1)}$, a stable strong focus if $\frac{\sqrt{n}(4n-5)}{4(n-1)}<\nu<\frac{7+8n(2n-3)}{4(n-1)}\sqrt{\frac{n}{2(2n-1)}}$ and a stable node if $\nu \geq \frac{7+8n(2n-3)}{4(n-1)}\sqrt{\frac{n}{2(2n-1)}}$.
    Moreover notice that $\frac{7+8n(2n-3)}{4(n-1)}\sqrt{\frac{n}{2(2n-1)}}<2n$.
	
	For $-\frac{1}{4n(n-1)}<K\leq0$, both $f_\pm$ are real with $f_{-}(K,n)<\sqrt{n}(1+n K)<f_{+}(K,n)$. In this case, the pair of eigenvalues $\lambda_{1+},\lambda_{2+}$ are real if $0<\nu\leq f_{-}(K,n)$ or $\nu \geq f_{+}(K,n)$. Moreover, if $0<\nu < (1+n K)\sqrt{n}$, then $\lambda_{1+},\lambda_{2+}$ have positive real part, if $\nu>(1+n K)\sqrt{n}$ then $\lambda_{1+},\lambda_{2+}$ have negative real part, and  if $\nu=(1+n K)\sqrt{n}>0$ both eigenvalues are purely imaginary. Therefore on the  $\{u_3=0\}$ invariant subset, $\mathrm{S}^{\pm}_{+}$ are unstable nodes if $0<\nu \leq f_{-}(K,n)$, unstable strong focus if $f_{-}(K,n)<\nu<\sqrt{n}(1+n K)$, a center or weak focus if $\nu=(1+n K)\sqrt{n}$, a stable strong focus if $(1+n K)\sqrt{n}<\nu<f_{+}(K,n)$, and a stable node if $\nu \geq f_{+}(K,n)$.
	Finally we also note that for this range of $K$ it holds that $f_{-}(K,n)<2n$, and $f_{+}(K,n)<2n$ for $g^{+}_{-}(n)<K$, $f_{+}(K,n)=2n$ for $K=g^{+}_{-}(n)$ and $f_{+}(K,n)>2n$ for $K<g^{+}_{-}(n)$
	.
	
	The cases when $K\in(-\frac{1}{n},0)$ with $\nu=(1+n K)\sqrt{n}>0$ consist of bifurcations where stability is changed, from an unstable focus to a stable focus leading to a center-focus problem, which will not be treated here, although as we will see ahead, numerical results indicate that they are centers. 
	
	As a final remark, note that when $K=0$ the expression for the fixed points reduce to $x_3=\pm \left(\frac{3\nu}{n^2}\right)^{\frac{1}{n-1}}$, $y_3=\mp \left(\frac{\nu}{n}\right)^{\frac{1}{n-1}}\left(\frac{3}{n}\right)^{\frac{n}{n-1}}$, and for $\nu<\sqrt{n}$, the solution of Equation~\eqref{CM_M1} describes globally the two heteroclinic orbits $\mathrm{S}^{\pm}_{+}\rightarrow \mathrm{M}$.  
\end{itemize}
\begin{itemize}
	\item[b)] {\bf Subcase} $0<K<1/n$, i.e. $\frac{2(n-1)}{n}<\gamma_\mathrm{pf}<2$: 
	
	In addition to $\mathrm{M}$, there are fixed points only for $\nu\geq 2n\sqrt{K}$, while no additional fixed points exists for $0<\nu<2n\sqrt{K}$. When $\nu > 2n\sqrt{K}$ there are four fixed points $\mathrm{S}^{\pm}_{\pm}$, which  merge into two fixed points $\mathrm{S}^{\pm}_{0}$ when $\nu=2n\sqrt{K}$,
	\begin{equation}
		\mathrm{S}^{\pm}_0\,:\quad	x_3=\pm \left(\frac{6K}{n(1-K)}\right)^{\frac{1}{n-1}},\qquad
		y_3=\mp \left(\frac{6K}{n(1-K)}\right)^{\frac{1}{n-1}}.
	\end{equation}
For $K\in\left(0,\frac{1}{n}\right)$ both $f_\pm$ are real, and the following inequalities hold $2n\sqrt{K}<f_+(K,n)<(1+nK)\sqrt{n}<f_-(K,n)<2n$. The eigenvalues $\lambda_{1+}$ and $\lambda_{2+}$ of the linearisation around $\mathrm{S}^{\pm}_+$ are real if $2n\sqrt{K}<\nu\leq f_{+}(K,n)$ or $\nu\geq f_{-}(K,n)$, and complex if $f_+(K,n)<\nu<f_{-}(n,K)$. Moreover if $\nu<(1+n K)\sqrt{n}$ then $\lambda_{1_{+}},\lambda_{2_{+}}$ have positive real part, and if $\nu>(1+nK)\sqrt{n}$ then $\lambda_{1_{+}},\lambda_{2_{+}}$ have negative real part. Finally if $\nu=(1+nK)\sqrt{n}$ both eigenvalues are purely imaginary. Therefore on the invariant subset $\{u_3=0\}$ the fixed points $\mathrm{S}^{\pm}_{+}$ unstable nodes if $2n\sqrt{K}<\nu\leq f_+(K,n)$, unstable strong focus if $f_+(K,n)<\nu<(1+n K)\sqrt{n}$, a center or a weak focus if $\nu=(1+n K)\sqrt{n}$, a stable strong focus if $(1+nK)\sqrt{n}<\nu<f_{-}(n,K)$, and a stable node if $\nu\geq f_-(n,K)$. Notice that when $\nu=(1+n K)\sqrt{n}$ leads to a center-focus problem, which will not be treated here, although numerical simulations indicate that they are centers. 

The eigenvalues $\lambda_{1-}$ and $\lambda_{2-}$ of the linearisation around $\mathrm{S}^{\pm}_-$ are real for all $\nu>2n\sqrt{K}$. Moreover, $\lambda_{1_{-}}$ is negative and $\lambda_{2_{-}}$ is positive, so that $\mathrm{S}^{\pm}_-$ are hyperbolic saddles.

When $\nu=2n\sqrt{K}$, with $0<K<1/n$, it follows that $\lambda_1$ and $\lambda_2$ are always real and positive, so that the fixed points $\mathrm{S}^{\pm}_{0}$ are unstable nodes.
\end{itemize}
%
%
\subsubsection*{Fixed points at infinity}
We start with the positive $x$-direction. Setting $p=\frac{1}{2}(n-1)$, the flow induced on the $\{u_1=0\}$ invariant subset is given by
\begin{equation}
	\frac{dy_1}{d\bar{\tau}_1} =-n-\nu y_1-n y^2_1-3y_1z_1 ,\qquad 
	\frac{dz_1}{d\bar{\tau}_1} = -(n-1)\left(\frac{3\gamma_\mathrm{pf}}{2(n-1)}z_1+y_1\right)z_1.
\end{equation}
We are interested in the invariant subset $\{z_1=0\}$ on $\{u_1=0\}$, where the system reduces to
\begin{equation}
z_1\frac{dy_1}{d\bar{\tau}_3} =-n-\nu y_1-n y^2_1,
\end{equation}
and yields two fixed points when $\nu>2n$: 
\begin{equation}
\mathrm{P}^+_\pm\quad:\quad y_1=B_\pm,\quad z_1=0,\quad u_1=0,
\end{equation}
where we introduced the notation
\begin{equation}
B_\pm = -\frac{\nu}{2n}\pm\sqrt{\left(\frac{\nu}{2n}\right)^2-1}<0.
\end{equation}
The linearised system at $\mathrm{P}^+_\pm$ has eigenvalues
$\lambda_1=-(n-1)B_\pm$, $\lambda_2=\mp 2n \sqrt{\left(\frac{\nu}{2n}\right)^2-1}$ and $\lambda_3=B_\pm$, with associated eigenvectors $v_1=\left(1,a,0\right)$, $v_2=(0,1,0)$ and $v_3=(0,0,1)$,
where $a=-\frac{3}{n-1}\frac{B_\pm+\frac{\lambda_2}{2n}}{B_\pm+\lambda_2}$.
Since $\mathrm{P}^+_\mathrm{+}$ has $\lambda_2,\lambda_3<0$ and $\lambda_1>0$, and $\mathrm{P}^+_\mathrm{-}$ has $\lambda_3<0$ and $\lambda_1,\lambda_2>0$, they are both hyperbolic saddles. But on $\{u_1=0\}$, $\mathrm{P}^+_-$ is a source and $\mathrm{P}^+_+$ is a saddle. In particular, a one-parameter set of orbits originates from $\mathrm{P}^+_-$ and a single orbit from $\mathrm{P}^+_+$ into the region $\{z_1>0\}$. 
%
%
When $\nu=2n$, the fixed points $\mathrm{P}^+_\pm$ merge into a single fixed point 
\begin{equation}
\mathrm{P}^+_0\quad:\quad y_1=-1,\quad z_1=0,\quad u_1=0,
\end{equation}
where the eigenvalues reduce to $\lambda_1=n-1$, $\lambda_2=0$ and $\lambda_3=-1$. Hence the $z_1=0$ axis is the center manifold of $\mathrm{P}^+_0$. As before, symmetry considerations allow us to deduce the existence of equivalent fixed points $\mathrm{P}^{-}_\pm$ and $\mathrm{P}^{-}_0$ when blowing-up in the negative $x$-direction.
Finally, when $0<\nu<2n$ there are no fixed points on the $\{z_1=0\}$ axis. 

Using the directional blow-up in the positive $y$-direction, the flow induced on $\{u_2=0\}$ is given by
\begin{equation}
\aligned
&\frac{dx_2}{d\bar{\tau}_2}=1+\frac{\nu}{n}x_2^{n}+x_2^{2n}+\frac{3}{2}x_2 z_2\\
&\frac{dz_2}{d\bar{\tau}_2}=\frac{(n-1)}{n}\left(n x_2^{2n-1}+\nu x_2^{n-1}+ 3\left(1+\frac{(2-n)\gamma_\mathrm{pf}}{2(n-1)}\right)z_2\right)z_2.
\endaligned
\end{equation}
Further restricting to the invariant subset $\{z_2=0\}$ results in
\begin{equation}
\frac{dx_2}{d\bar{\tau}_2}=1+\frac{\nu}{n}x_2^{n}+x_2^{2n}.
\end{equation}
Since $n$ is odd, there are two fixed points when $\nu>2n$, corresponding to $\mathrm{P}^-_\pm$ studied above which are located at $(x_2,z_2,u_2)=\left(B^{1/n}_{\pm},0,0\right)$ and, for $\nu=2n$, merge into a single fixed point $\mathrm{P}^-_{0}$. When $0<\nu<2n$ there are no fixed points on $\{z_2=0\}\cap\{u_2=0\}$. A similar analysis in the negative $y$-direction shows that the only fixed points are $\mathrm{P}^{+}_\pm$ and $\mathrm{P}^{+}_{0}$.
Therefore the equator of the Poincar\'e sphere has fixed points when $\nu\geq 2n$ on the second and fourth quadrants, which are $\alpha$-limit points for orbits in the northern hemisphere. When $0<\nu<2n$ the equator consists of a periodic orbit. In order to study the stability of the periodic orbit at infinity  in the $(x_3,y_3)$ plane when $0<\nu<2n$, we shall employ a Poincar\'e-Lyapunov compactification on the unit disk.
%
%
\subsubsection*{Global phase-space on the Poincar\'e-Lyapunov disk}
When $p=\frac{1}{2}(n-1)$, we obtain from~\eqref{PL_Cylinder} that the induced flow on the $\{\bar{u}=0\}$ invariant subset is given by
\begin{subequations}
	\begin{align}
	\frac{dr}{d\xi} &= \frac{(n-1)(1-r)r}{n}\left[-\nu r (1-\cos^{2n}\theta)\cos^{n-1}\theta+3(1-r)\left(\frac{K}{1-K}+\cos^{2n}\theta\right)\right], \\
	\frac{d\theta}{d\xi} &=-F(\theta)\left[ \frac{3}{2n} (1-r)F(\theta)\sin{2\theta}+\left(1+\frac{\nu}{2n}F(\theta)\cos^{n-1}{\theta}\sin{2\theta}\right)r\right].
	\end{align}
\end{subequations}
At $\{r=0\}$ lies the fixed point $\mathrm{M}$ which is the origin of the $(x_3,y_3)$ plane. The previous analysis showed that $\mathrm{M}$ is a hyperbolic saddle if $K<0$, a center-saddle if $K=0$, and a hyperbolic source if $K>0$. The fixed points $\mathrm{S}^\pm_\pm$ are located at
\begin{equation}
\frac{r_{\mathrm{S}^\pm_\pm}}{1-r_{\mathrm{S}^\pm_\pm}}=\left(\frac{3A_\pm}{n(1-K)}\right)\left(1+A^{-2}_\pm\right)^{\frac{n-1}{2n}}, \quad \theta_{\mathrm{S}^\pm_\pm} = \arccos{\left(\pm\frac{1}{\left(1+A^{-2}_\pm\right)^{\frac{1}{2n}}}\right)},
\end{equation}
while the infinity of the $(x_3,y_3)$ plane is now located at $\{r=1\}$. When $\nu>2n$ there are four fixed points corresponding to $\mathrm{P}^\pm_\pm$ and given by
\begin{equation}
\theta_{\mathrm{P}^+_\pm} = \arccos{\left(\pm\left(\frac{1-\sqrt{1-4(n/\nu)^2}}{\sqrt{2}}\right)^{\frac{1}{2n}}\right)},\quad \theta_{\mathrm{P}^-_\pm} =\pi+ \arccos{\left(\pm\left(\frac{1-\sqrt{1-4(n/\nu)^2}}{\sqrt{2}}\right)^{\frac{1}{2n}}\right)},
\end{equation}
which merge into two fixed points $\mathrm{P}^{\pm}_0$ when $\nu=2n$, with
\begin{equation}
\theta_{\mathrm{P}^{+}_0} = \arccos\left(-\left(\frac{1}{\sqrt{2}}\right)^{\frac{1}{n}}\right),\quad \theta_{\mathrm{P}^{-}_0} =\pi+ \arccos\left(-\left(\frac{1}{\sqrt{2}}\right)^{\frac{1}{n}}\right),
\end{equation}
which we have seen to be $\alpha$-limit points for interior orbits on the disk. Finally, when $\nu\in(0,2n)$ there exists a periodic orbit:
\begin{equation}
\Gamma_\infty:\quad\frac{d\theta}{d\xi}= -F(\theta)\left(1+\frac{\nu}{2n}F(\theta)\cos^{n-1}{\theta}\sin{2\theta}\right)<0.
\end{equation}
\begin{lemma}\label{StabGamma}
	For all $0<\nu<2n$, with $n>1$, $n$ odd, $\Gamma_\infty$ is an unstable limit cycle.
\end{lemma}
\begin{proof}
	To analyse the periodic orbit we introduce the variable $s=(1-r)/r$, so that the periodic orbit is now located at $s=0$, and a new time variable $\xi$ by $d/d\bar{\tau}_3=\frac{(1-r)}{r}d/d\xi$ which leads to a regular ODE that, close to $s=0$, is given by
	\begin{equation}
	\frac{ds}{d\theta}=S_1(\theta) s+S_2(\theta) s^2+o(s^2),
	\end{equation}
	where
	\begin{subequations}
		\begin{align}
		S_1 &= -\frac{\nu(n-1) F(\theta)\sin^2\theta\cos^{n-1}\theta}{n\left(1+\frac{\nu}{2n}F(\theta)\cos^{n-1}{\theta}\sin{2\theta}\right)}, \\
		S_2 &=-\frac{3(n-1)\left(\frac{K}{1-K}+\cos^{2n}\theta+\frac{\nu}{n(1-K)}F(\theta)\cos^{n}\theta\sin\theta\right)}{F(\theta)\left(1+\frac{\nu}{2n}F(\theta)\cos^{n-1}\theta\sin 2\theta\right)^2}
		\end{align}
	\end{subequations}
	Denoting by $s(\theta,s_0)$ the solution of the above differential equation such that $s(0,s_0)=s_0$, then close to $s=0$ we have
	\begin{equation}
	S(\theta)=\beta_1(\theta)s_0+\beta_2(\theta)s^2_0+o(s^2_0),
	\end{equation}
	where $\beta_1$ and $\beta_2$ solve the initial value problem
	\begin{subequations}
		\begin{align}
		\frac{d\beta_1}{d\theta} &= S_1(\theta)\beta_1(\theta), \qquad \beta_1(0)=1, \\
		\frac{d\beta_2}{d\theta} &= S_1(\theta)\beta_2(\theta)+S_2(\theta)\beta^2_1(\theta),\qquad \beta_2(0) =0.
		\end{align}
	\end{subequations}
	The solutions are
	\begin{equation}
	\beta_1(\theta) = e^{\alpha(\theta)},\qquad 
	\beta_2(\theta) = -e^{\alpha(\theta)}\int_{0}^{\theta}e^{\alpha(\psi)}S_2(\psi)d\psi,
	\end{equation}
	with
	\begin{equation}
	\alpha(\theta)=-\frac{\nu}{n}(n-1)\int^{\theta}_{0} \frac{F(\tilde{\theta})\sin^2{\tilde{\theta}}\cos^{n-1}\tilde{\theta}}{1+\frac{\nu}{2n}F(\tilde{\theta})\sin{2\tilde{\theta}}\cos^{n-1}\tilde{\theta}}d\tilde{\theta}.
	\end{equation}
	The Poincar\'e return map near $s=0$ is $P(s_0)=s(2\pi,s_0)$. Since $P(0)=0$, and for $n$ odd, $P^{\prime}(0)=e^{\alpha(2\pi)}<1$ and $\theta$ is strictly monotonically decreasing, then the periodic orbit $\Gamma_\infty$ is an unstable limit cycle for all $0<\nu<2n$ and $n>1$.
\end{proof}
\begin{proposition}\label{EqRep}
	Let $p=\frac{1}{2}(n-1)$ with $p>0$ ($n>1$ with $n$ odd). Then the infinity of the $(x_3,y_3)$ plane (corresponding to the $\{r=1\}$ invariant boundary of the Poincar\'e-Lyapunov unit disk) is a repeller.
\end{proposition}
\begin{proof}
	The proof follows by the local analysis of the fixed points $\mathrm{P}^\pm_\pm$, $\mathrm{P}^\pm_0$ and Lemma~\ref{StabGamma}.
\end{proof}
Figure~\ref{fig:PP1}, shows the three different types of orbit structure at the invariant boundary $\{r=1\}$ of the unit disk.
\begin{figure}[ht!]
	\begin{center}
		\subfigure[$0<\nu<2n$.]{\label{d1}
			\includegraphics[width=0.30\textwidth]{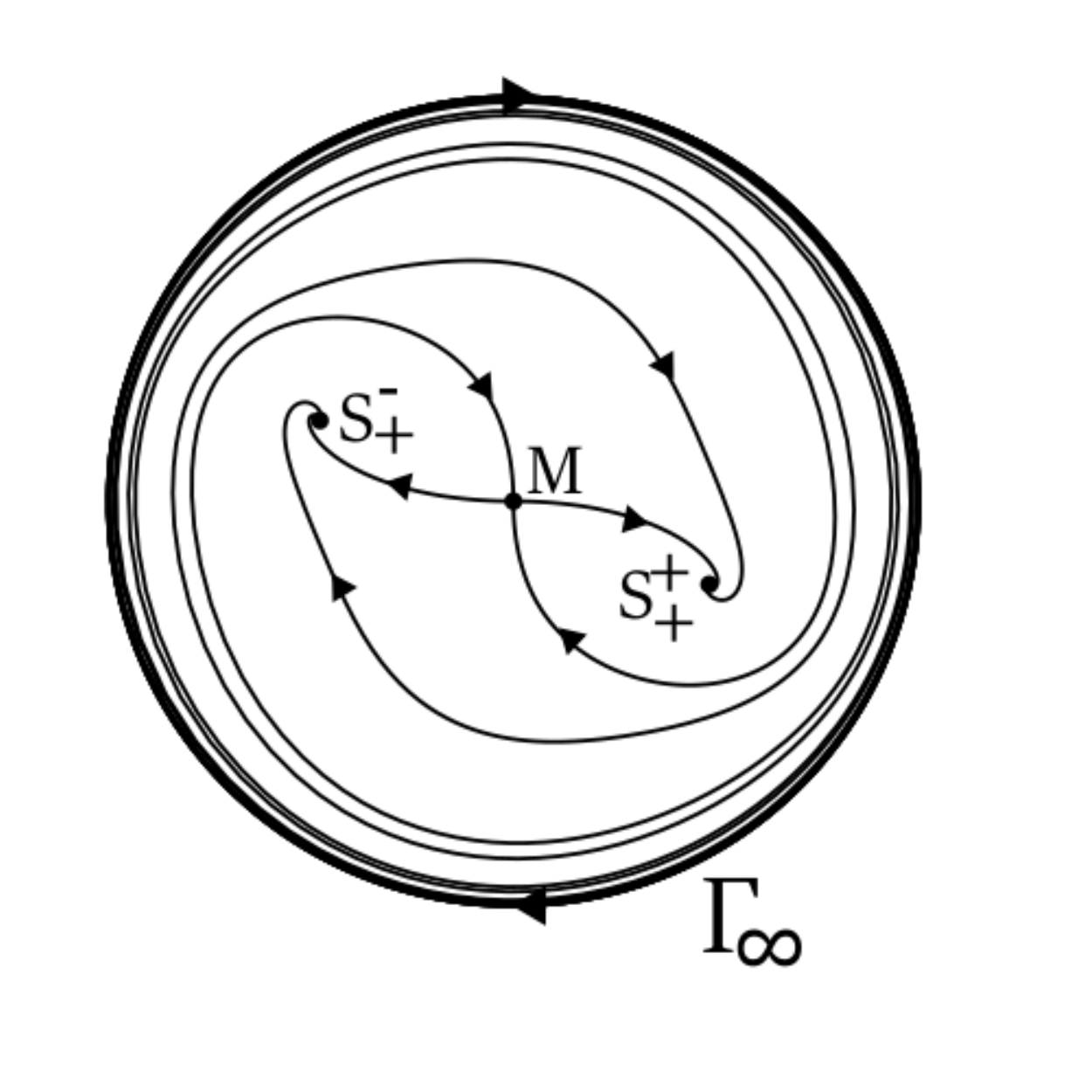}}
		\subfigure[$\nu=2n$.]{\label{d2}
			\includegraphics[width=0.30\textwidth]{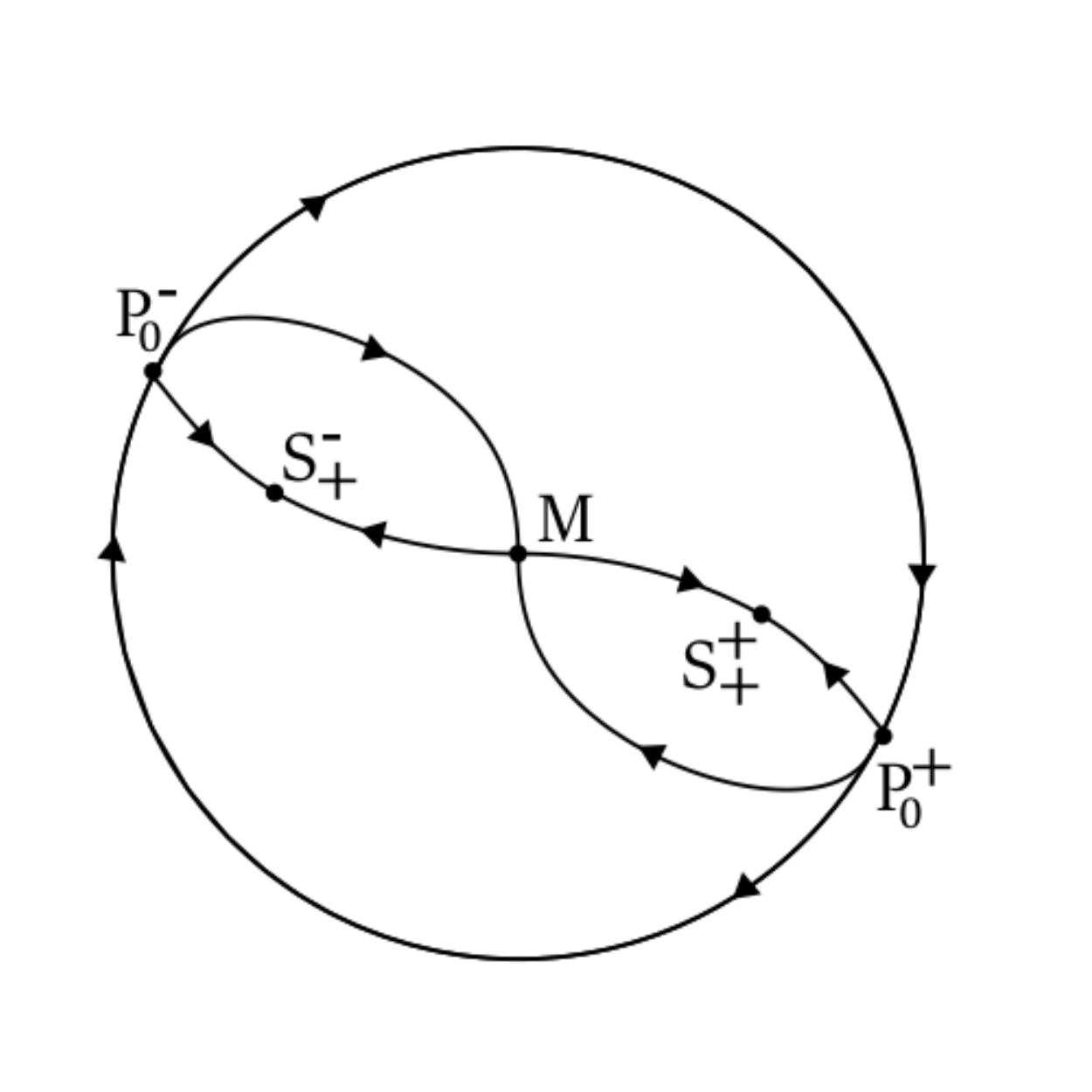}}
		\subfigure[$\nu>2n$.]{\label{d3}
			\includegraphics[width=0.30\textwidth]{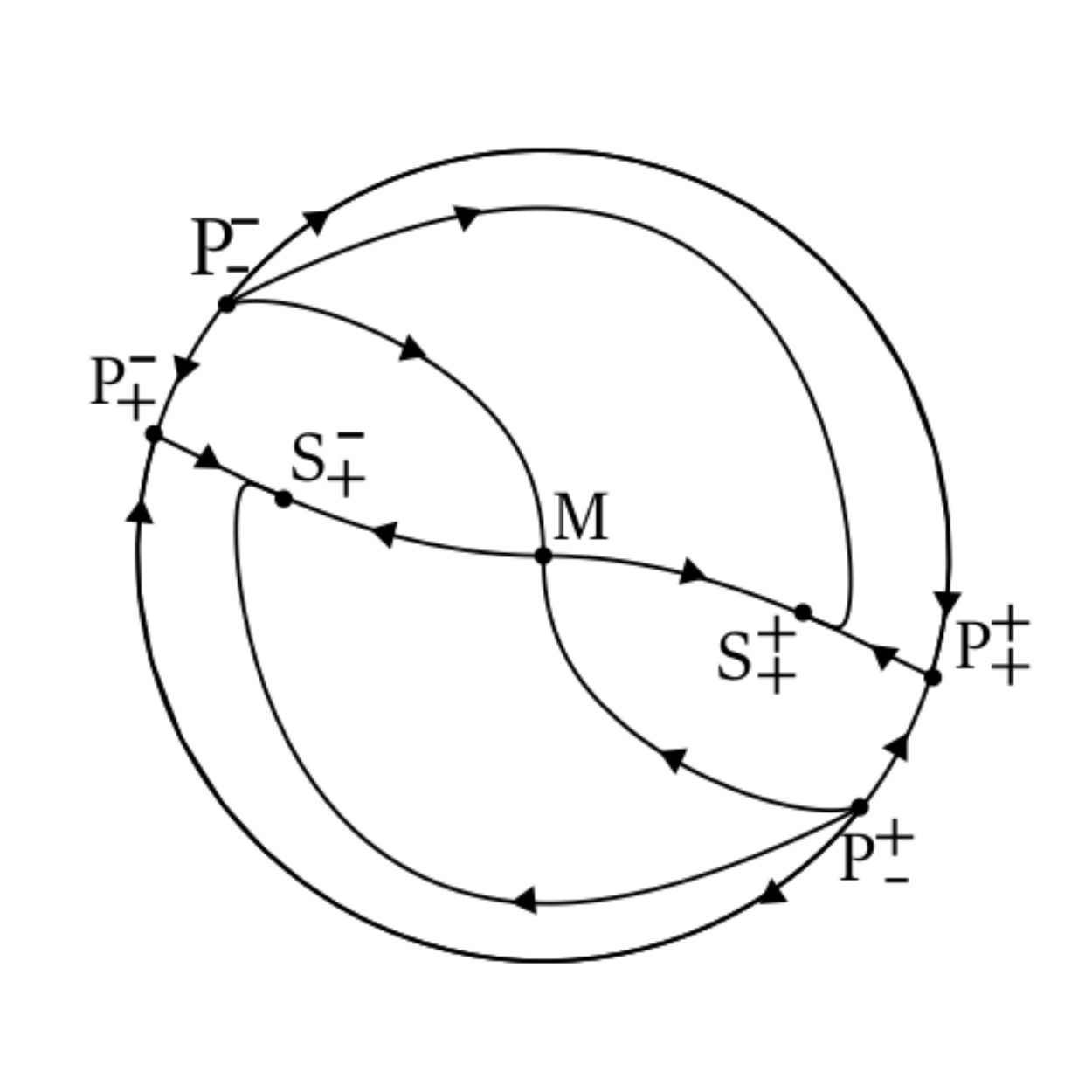}}
	\end{center}
	\vspace{-0.5cm}
	\caption{Poincar\'e-Lyapunov disk when $p=\frac{1}{2}(n-1)$ and $K\in(-\infty,-\frac{1}{n}]$, for $0<\nu<2n$, $\nu=2n$ and $\nu>2n$. These cases are numerically represented with the values $(p,n)=(1,3)$ and $\gamma_\mathrm{pf}=1$ i.e. $K=-1/3$, which corresponds to the borderline case $K=-\frac{1}{n}>g^{-}_{-}(n)$ for which $f_-(K,n)<2n$. In particular Fig.(a) has $\nu<f_-(K,n)<2n$, so that $\mathrm{S}^{\pm}_+$ are strong focus, while in Figs. (b) and (c) $\mathrm{S}^{\pm}_+$ are nodes.}
	\label{fig:PP1}
\end{figure}

\begin{theorem}
	Let $p=\frac{1}{2}(n-1)$ with $p>0$ ($n>1$ with $n$ odd). If $K\in(-\infty,-\frac{1}{n}]$, i.e. $\gamma_\mathrm{pf}\in(0,\frac{2(n-1)}{n+1})$, then for all $\nu>0$ the $\omega$-limit set of all interior orbits on the Poincar\'e-Lyapunov disk is contained on the set $\mathrm{M}\cup\mathrm{S}^{\pm}_+$. In particular, as $\xi\rightarrow+\infty$, exactly 2 orbits converge to the fixed point $\mathrm{M}$ and a 1-parameter family of orbits converges to each fixed point $\mathrm{S}^{\pm}_+$, the separatrix skeleton being as depicted in Figure~\ref{fig:PP1}.
\end{theorem}
\begin{proof}
	From Proposition~\ref{EqRep} every regular orbits on the $(x_3,y_3)$ plane remains bounded for all future times. The divergence of the vector field~\eqref{PosZEQ} yields
	\begin{equation}\label{DivF}
	\frac{3\gamma_\mathrm{pf}}{2(n-1)}(1+nK)-\nu x^{n-1}_3,
	\end{equation}
	which for $K\leq-1/n$ and $\nu>0$ does not change sign and vanishes at a set of measure zero. Therefore by the \emph{Bendixson-Dulac criteria} there are no interior periodic orbits. Moreover the origin $\mathrm{M}$ is a saddle and $\mathrm{S}^\pm_+$ are sinks for all $\nu>0$ (focus if $\nu < f_-(K,n)$, or nodes if $\nu \geq f_-(K,n)$). Since closed saddle connections cannot exist, it follows by the Poincar\'e-Bendixson theorem that the only possible $\omega$-limit sets in this case are the fixed points $\mathrm{S}^{\pm}_+$ and $\mathrm{M}$. The last statement follows from the local stability properties of the fixed points.
\end{proof}
%
%
\begin{remark}\label{Missing1}
	We now discuss briefly the case when $K\in(-\frac{1}{n},0]$. Numerical results indicate that in this case a unique stable interior limit cycle $\Gamma_\mathrm{in}$ exists if $\nu\leq (1+nK)\sqrt{n}$, with $\mathrm{S}^{\pm}_+$ being sources (unstable nodes or focus) for $\nu< (1+nK)\sqrt{n}$, and centers with an unstable outer periodic orbit when $\nu= (1+nK)\sqrt{n}$. Furthermore, no interior limit cycle exists when $\nu>(1+nK)\sqrt{n}$ and $\mathrm{S}^{\pm}_{+}$ are sinks (stable focus or stable nodes), see Figure~\ref{fig:PL2}, for some representative examples.
\end{remark}
\begin{figure}[ht!]
	\begin{center}
		\subfigure[$0<\nu<(1+nK)\sqrt{n}$.]{\label{Sim1}
			\includegraphics[width=0.30\textwidth]{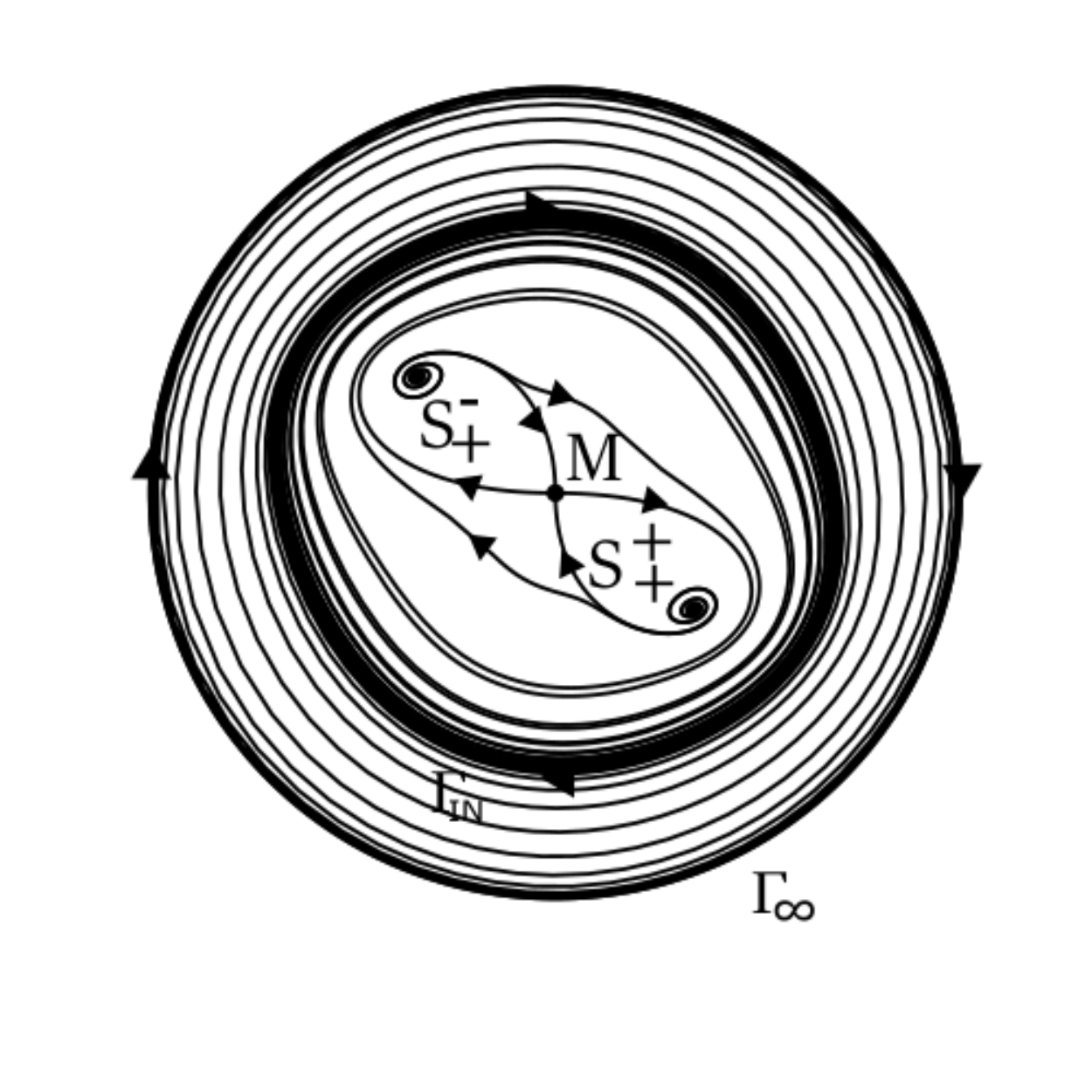}}
		\subfigure[$\nu=(1+nK)\sqrt{n}>0$.]{\label{Sim2}
			\includegraphics[width=0.30\textwidth]{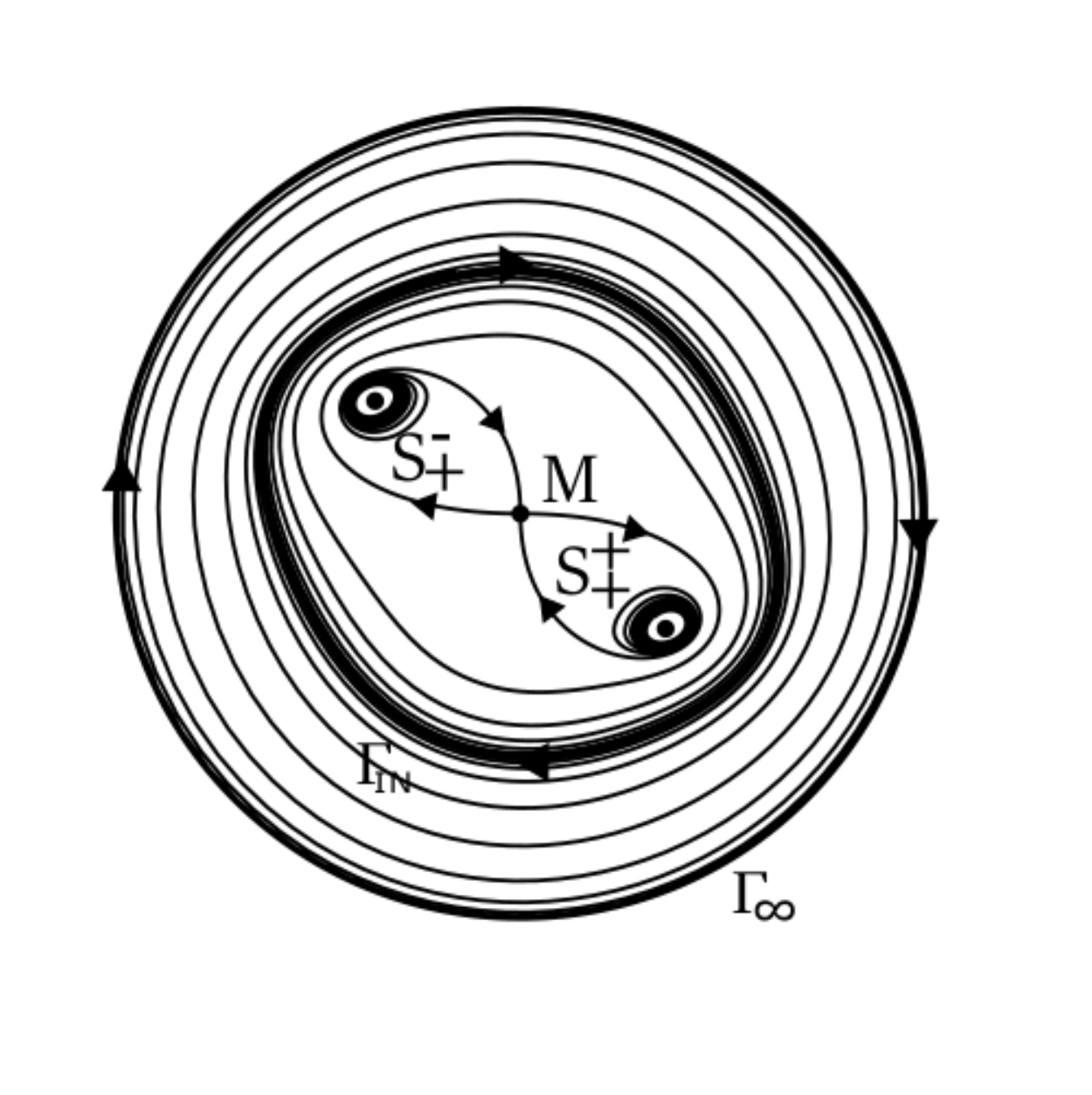}}
		\subfigure[$(1+nK)\sqrt{n}<\nu< f_{-}(K,n)<2n$.]{\label{Sim3}
			\includegraphics[width=0.30\textwidth]{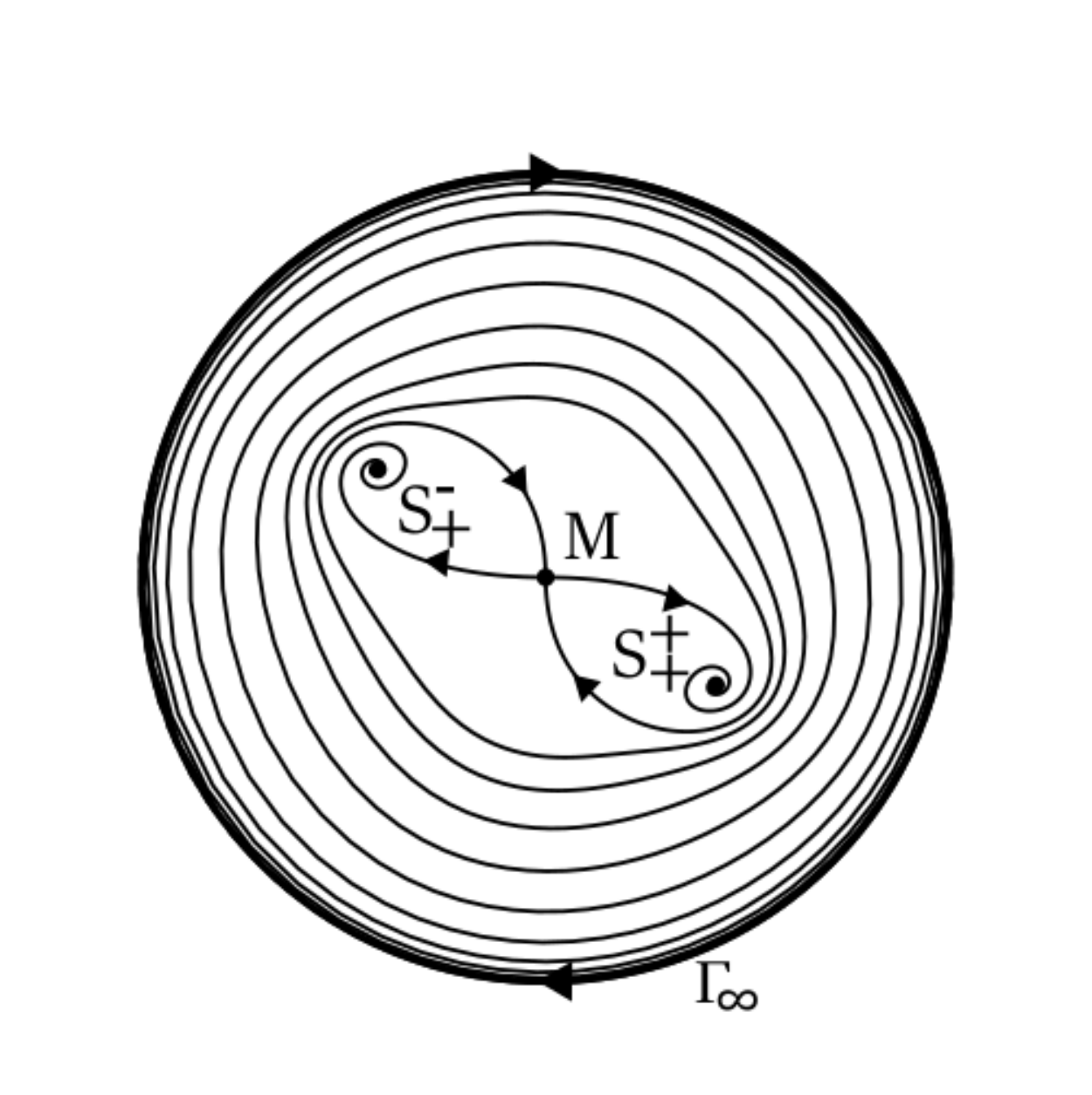}}
		\subfigure[$f_{-}(K,n)\leq \nu<2n$.]{\label{Sim4}
			\includegraphics[width=0.30\textwidth]{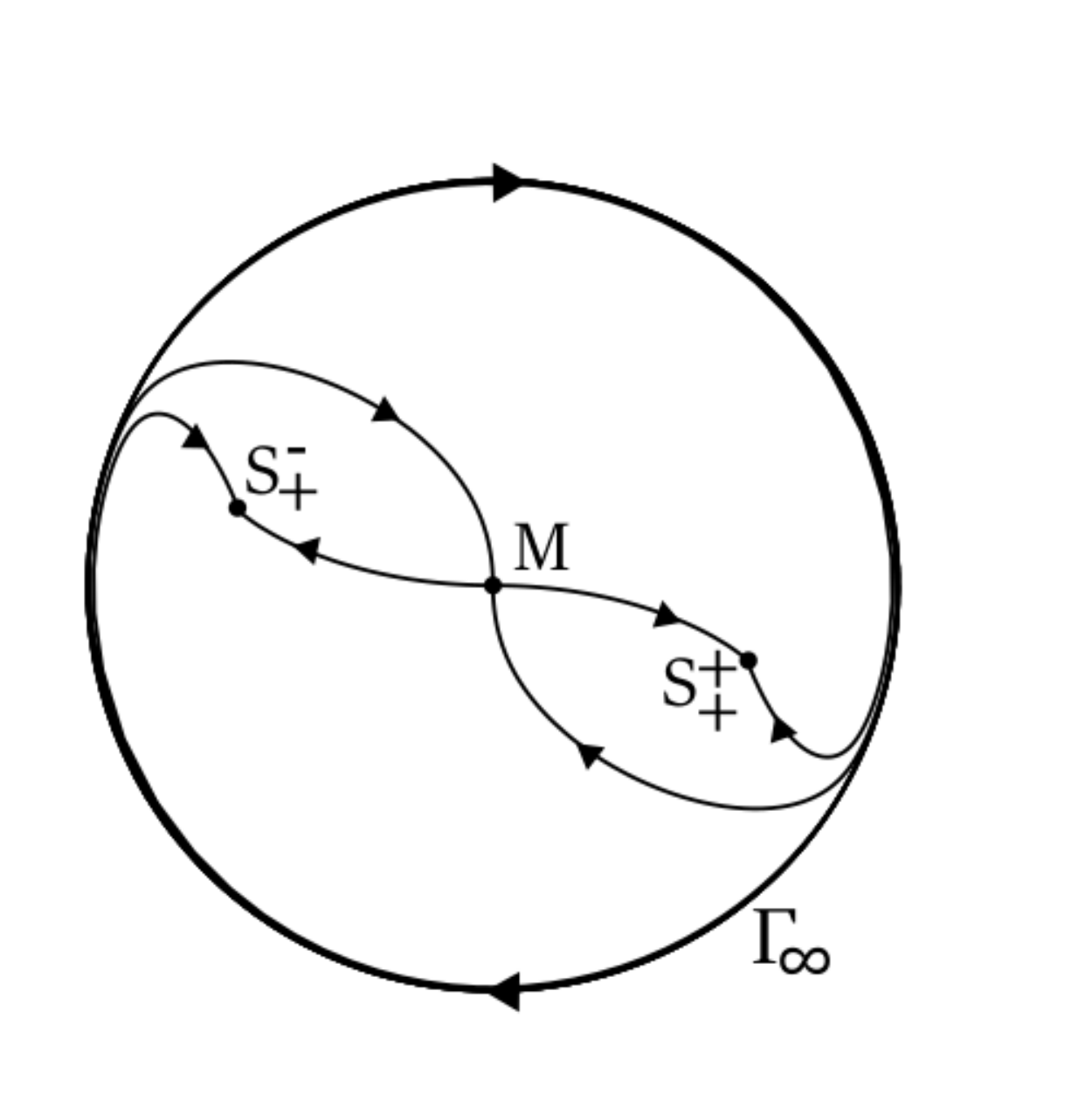}}
			\subfigure[$f_{-}(K,n)< \nu=2n$.]{\label{Sim5}
			\includegraphics[width=0.30\textwidth]{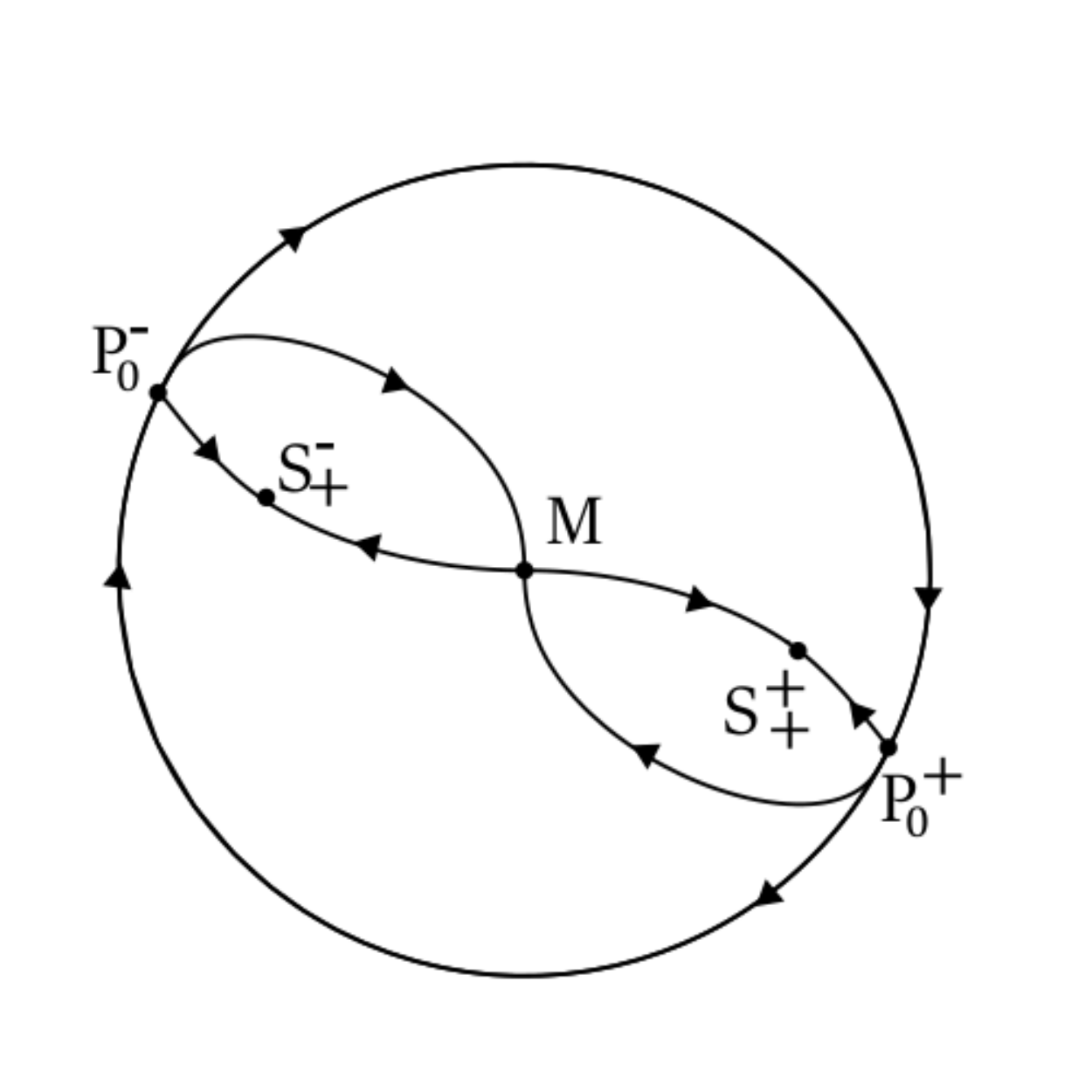}}
			\subfigure[$f_{-}(K,n)<2n<\nu$.]{\label{Sim6}
			\includegraphics[width=0.30\textwidth]{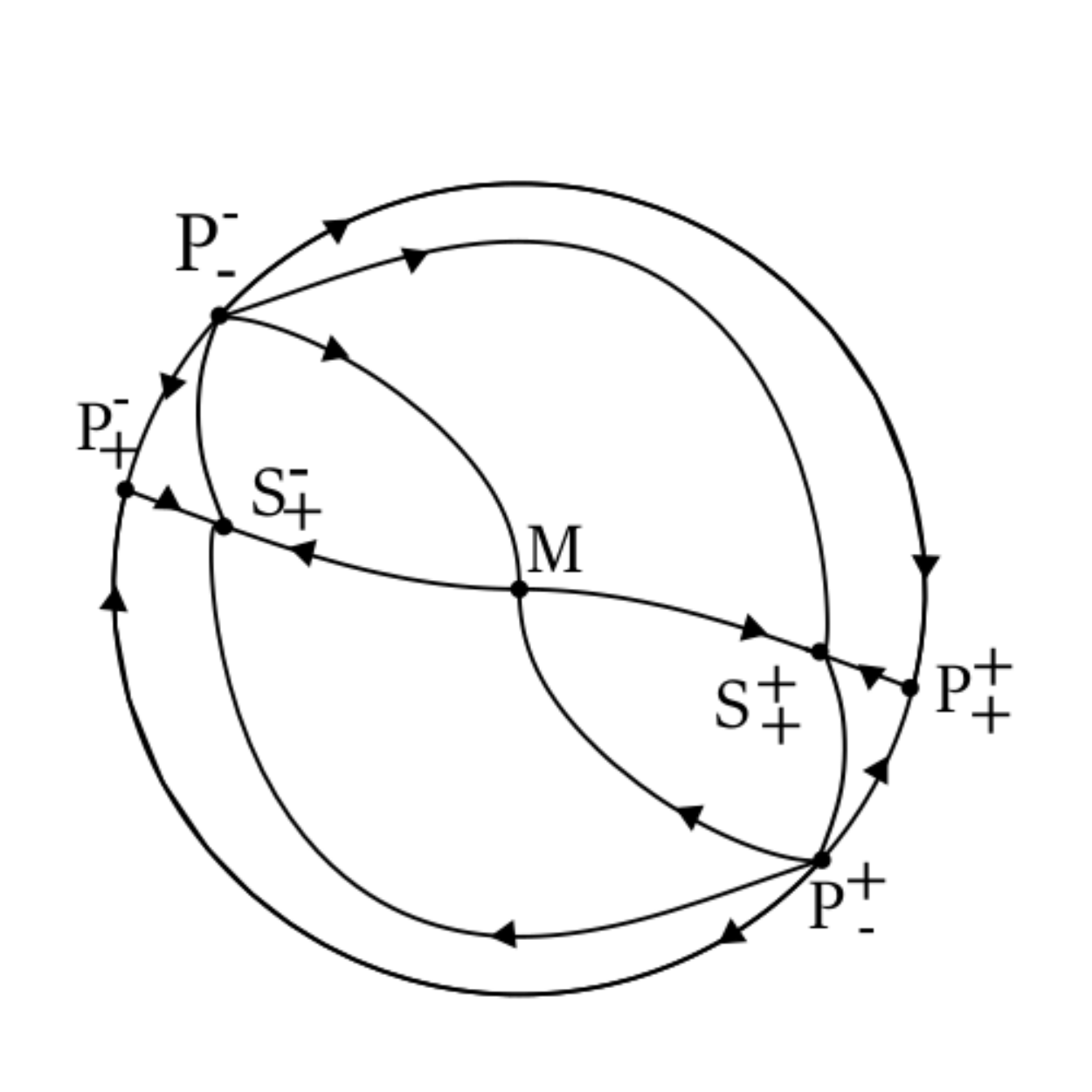}}
	\end{center}
	\vspace{-0.5cm}
	\caption{Poincar\'e-Lyapunov disks when $p=\frac{1}{2}(n-1)$ and $K\in(-\frac{1}{n},0]$. These cases are numerically exemplified with $n=3$ and $\gamma_{\mathrm{pf}}=5/4$, i.e. $K=-\frac{1}{15}$ and $f_{-}(-\frac{1}{15},3)=\frac{8}{5}\sqrt{5+2\sqrt{10}}$, so that $K<-\frac{1}{4n(n-1)}$ and $f_-(K,n)<2n$ since $K>g^{-}_{-}(n)$.
	}
	\label{fig:PL2}
\end{figure}
\begin{theorem}
	Let $p=\frac{1}{2}(n-1)$ with $p>0$ ($n>1$ with $n$ odd).  If $K\in\left(0,\frac{1}{n}\right)$, i.e. $\gamma_\mathrm{pf}\in\left(\frac{2(n-1)}{n},2\right)$, then for all $0<\nu<2n\sqrt{K}$, all interior orbits on the Poincar\'e-Lyapunov unit disk converge to a unique interior stable limit cycle $\Gamma_{\mathrm{in}}$ as $\xi\rightarrow+\infty$. 
\end{theorem}
\begin{proof}
	Here we make use of the equivalent system~\eqref{LIenardPlane} and apply \emph{Li\'enard's Theorem}, see e.g.~\cite{Perko}. 
	The functions $g(\bar{x})$ and $F(\bar{x})=\int^{\bar{x}}_0 f(s)ds$ in~\eqref{Lienard2} are odd functions of $\bar{x}$ and, if $\nu<2n\sqrt{K}$, then $\bar{x}g(\bar{x})>0$ for $\bar{x}\neq0$. Moreover, $F(0)=0$, $F^{\prime}(0)<0$, and $F(\bar{x})$ has unique positive zero at $\bar{x}=\left(\frac{n}{\nu}\frac{3\gamma_{\mathrm{pf}}}{2(n-1)}(1+nK)\right)^{\frac{1}{n-1}}$. Furthermore for $\bar{x}\geq \left(\frac{n}{\nu}\frac{3\gamma_{\mathrm{pf}}}{2(n-1)}(1+nK)\right)^{\frac{1}{n-1}}$, $F(\bar{x})$ is monotonically increasing to infinity as $\bar{x}\rightarrow +\infty$. Therefore the system has a unique stable interior limit cycle $\Gamma_\mathrm{in}$. Since in this case $\nu<2n\sqrt{K}<2n$, then by Lemma~\ref{StabGamma} the infinity consists of an unstable limit cycle. Moreover $\mathrm{M}$ is a hyperbolic source and therefore the only possible $\omega$-limit set is the unique interior stable limit cycle $\Gamma_\mathrm{in}$.
\end{proof}
\begin{remark}
	We note that Li\'enard's theorem also gives the relative location of the interior stable limit cycle $\Gamma_\mathrm{in}$.
\end{remark}
Figure~\ref{fig:PP2} shows the Poincar\'e-Lyapunov disk for $K\in\left(0,\frac{1}{n}\right)$ and $0<\nu<2n\sqrt{K}$, where the orbits accumulate at the interior stable limit cycle $\Gamma_\mathrm{in}$.
\begin{figure}[ht!]
	\begin{center}
		\includegraphics[width=0.30\textwidth]{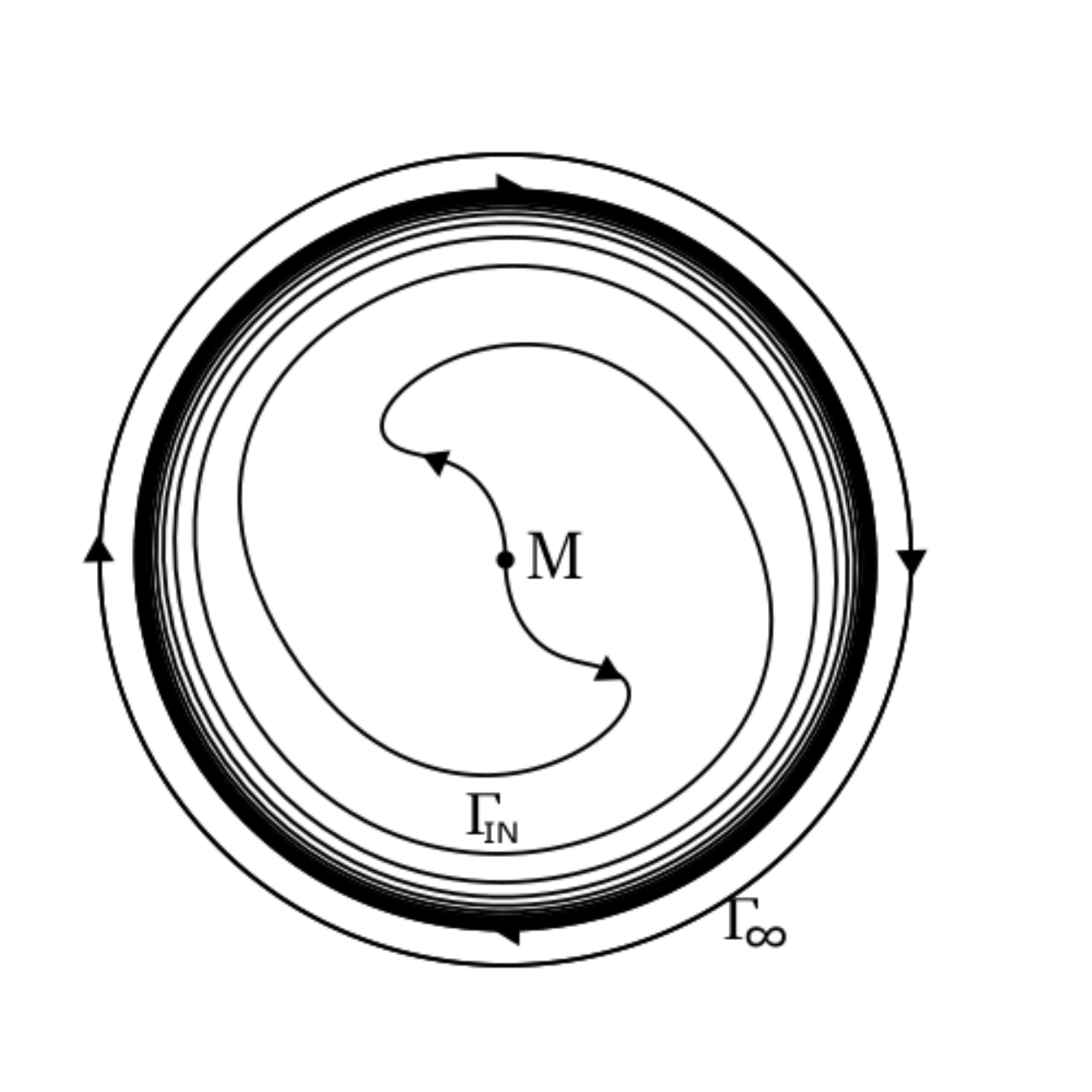}
	\end{center}
	\vspace{-0.5cm}
	\caption{Poincar\'e-Lyapunov disk when $p=\frac{1}{2}(n-1)$ and $K\in(0,\frac{1}{n})$, with $\nu\in(0,2n\sqrt{K})$.}
	\label{fig:PP2}
\end{figure}
\begin{remark}\label{Missing}
	We now briefly discuss the missing case $K\in(0,\frac{1}{n})$ with $\nu\geq2n\sqrt{K}$. Numerical results suggest that, in this case, an interior stable limit cycle exists if $\nu\leq (1+nK)\sqrt{n}$, with $\mathrm{S}^{\pm}_{+}$ being sources when $\nu< (1+nK)\sqrt{n}$ and centers with an unstable outer periodic orbit if $\nu= (1+nK)\sqrt{n}$, while no interior periodic orbits exist when $\nu>(1+nK)\sqrt{n}$. Recall that when $K>0$, the fixed point $\mathrm{M}$ is a source and $\mathrm{S}^{\pm}_{-}$ are saddles when they exist, i.e. when $\nu>2n\sqrt{K}$. If $\nu=2n\sqrt{K}$, then $\mathrm{S}^{\pm}_{+}$ and $\mathrm{S}^{\pm}_{-}$ merge into the fixed points $\mathrm{S}^{\pm}_{0}$, which are unstable strong focus. See Figure~\ref{fig:PL4}, for some representative cases.
\end{remark}
\begin{figure}[ht!]
	\begin{center}
				\subfigure[$\nu=2n\sqrt{K}$.]{\label{c0}
			\includegraphics[width=0.25\textwidth]{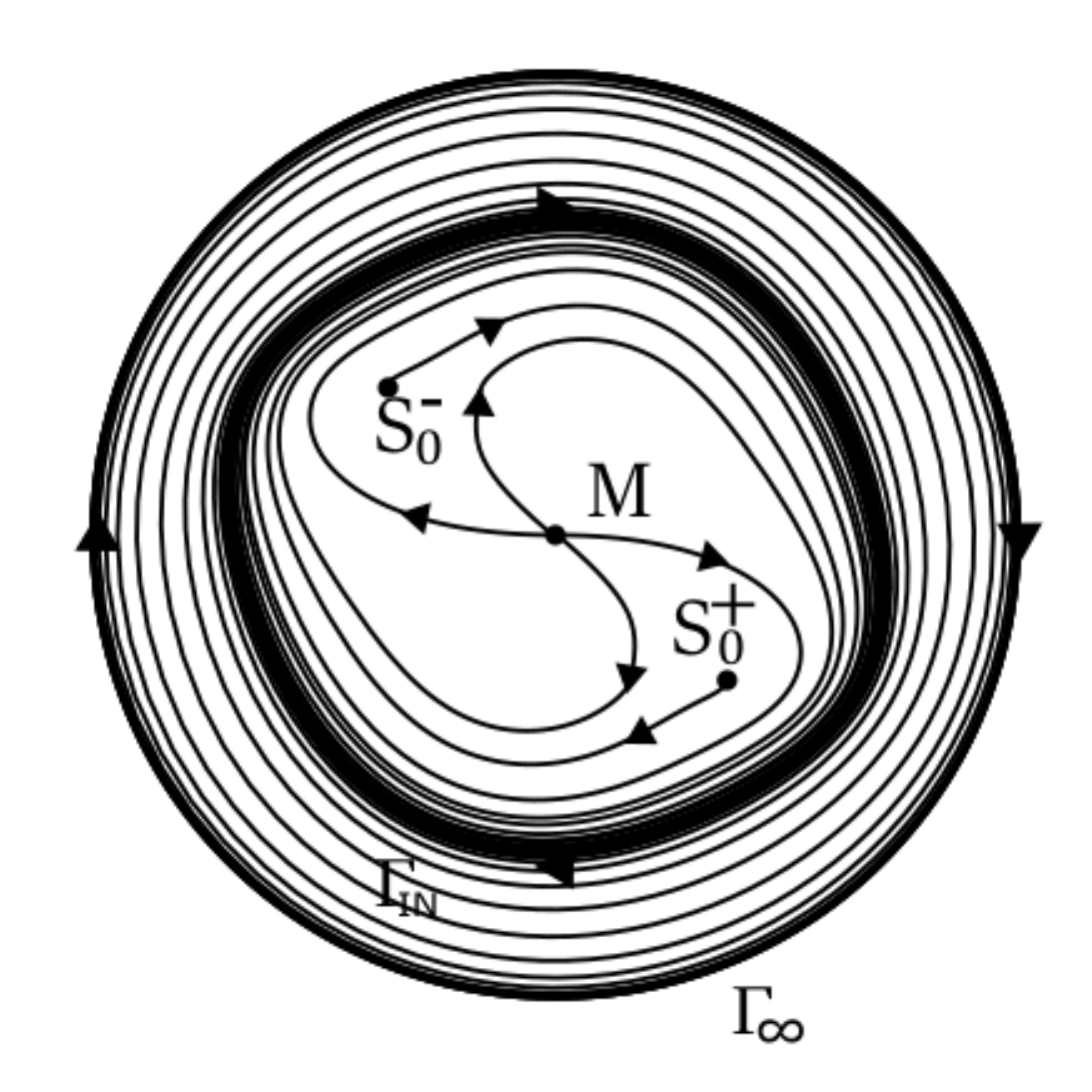}}
		\subfigure[$2n\sqrt{K}<\nu\leq f_+(n,K)$.]{\label{c1}
			\includegraphics[width=0.25\textwidth]{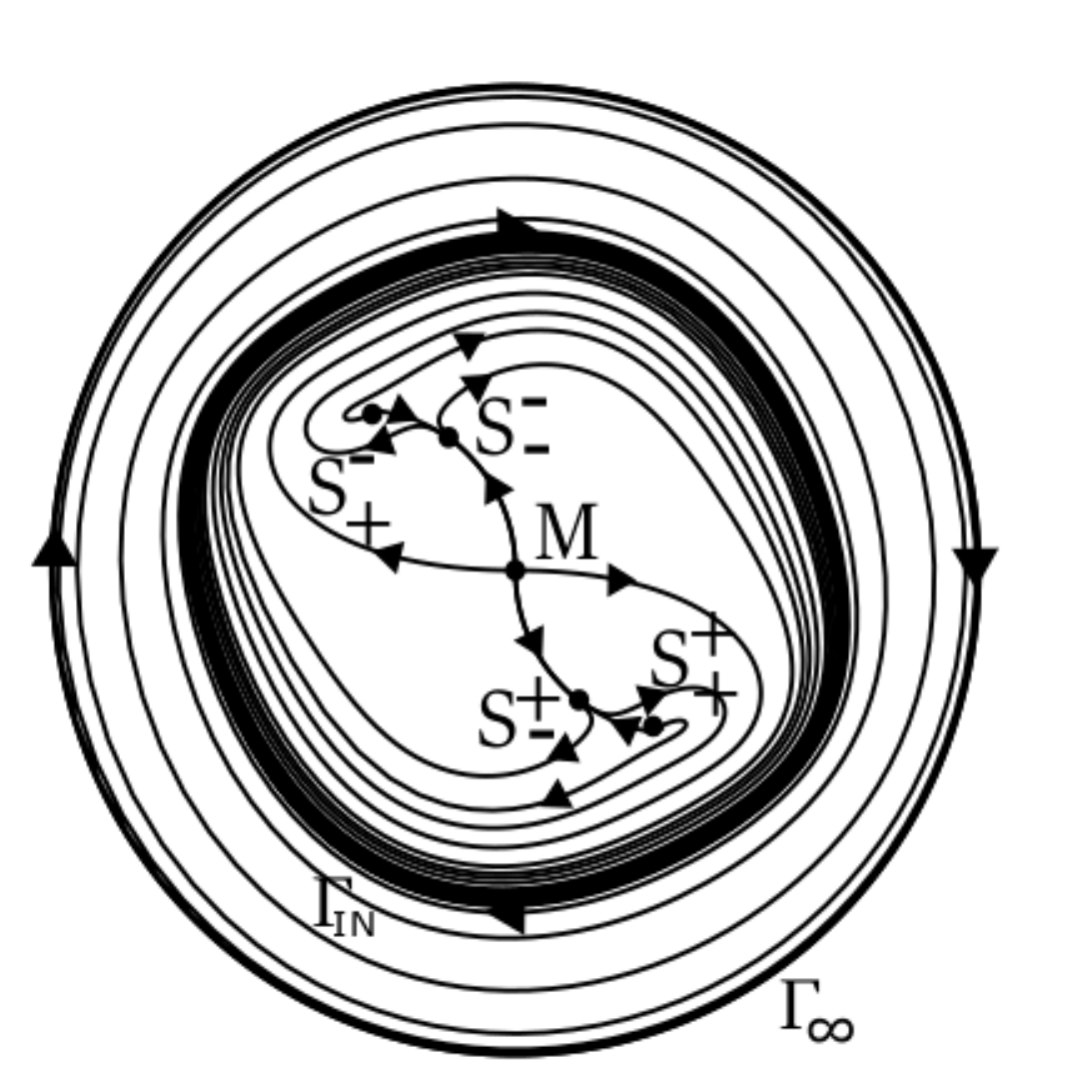}}
		\subfigure[$f_+(n,K)<\nu<(1+nK)\sqrt{n}$.]{\label{c2}
			\includegraphics[width=0.25\textwidth]{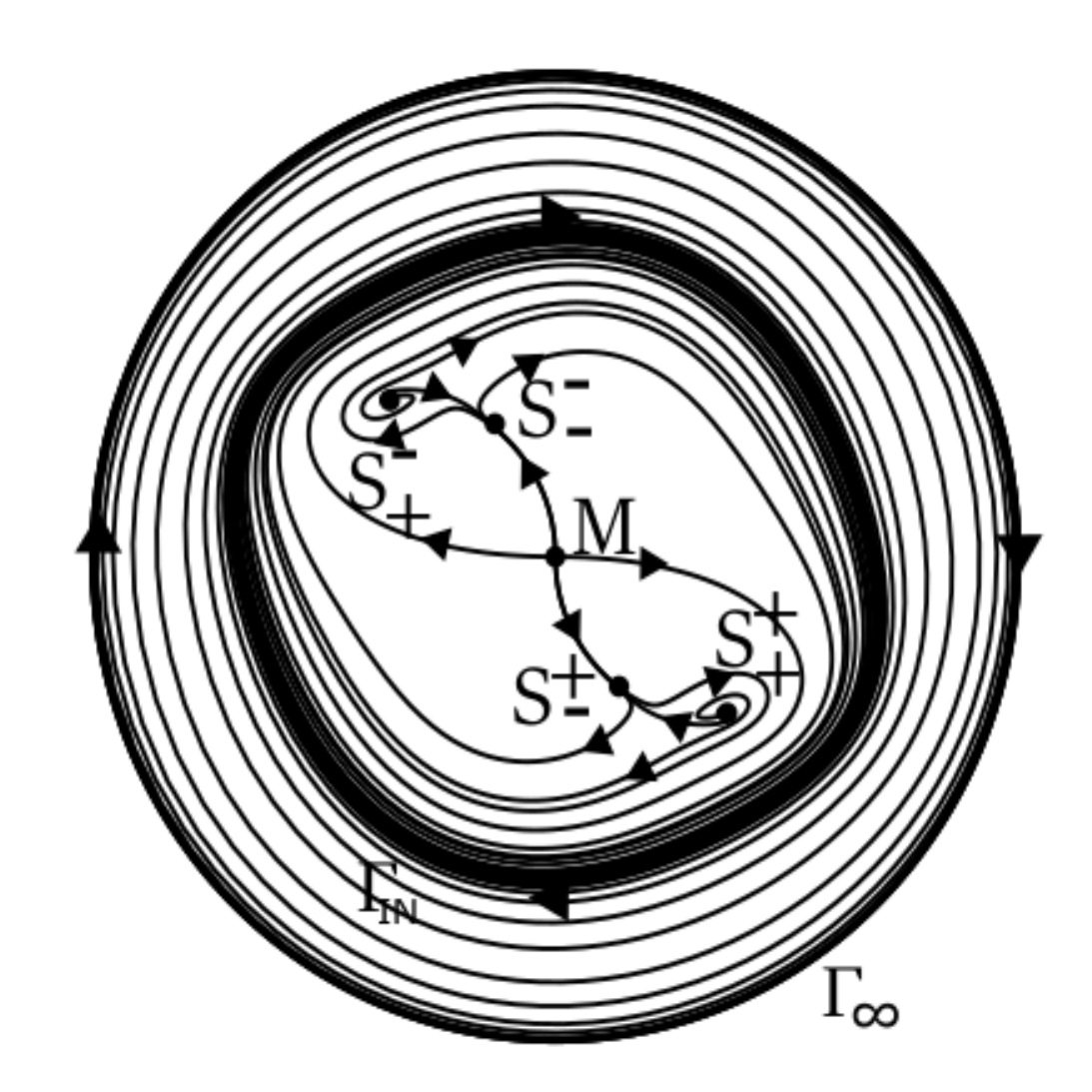}}
		\subfigure[$\nu=(1+nK)\sqrt{n}$.]{\label{c3}
			\includegraphics[width=0.25\textwidth]{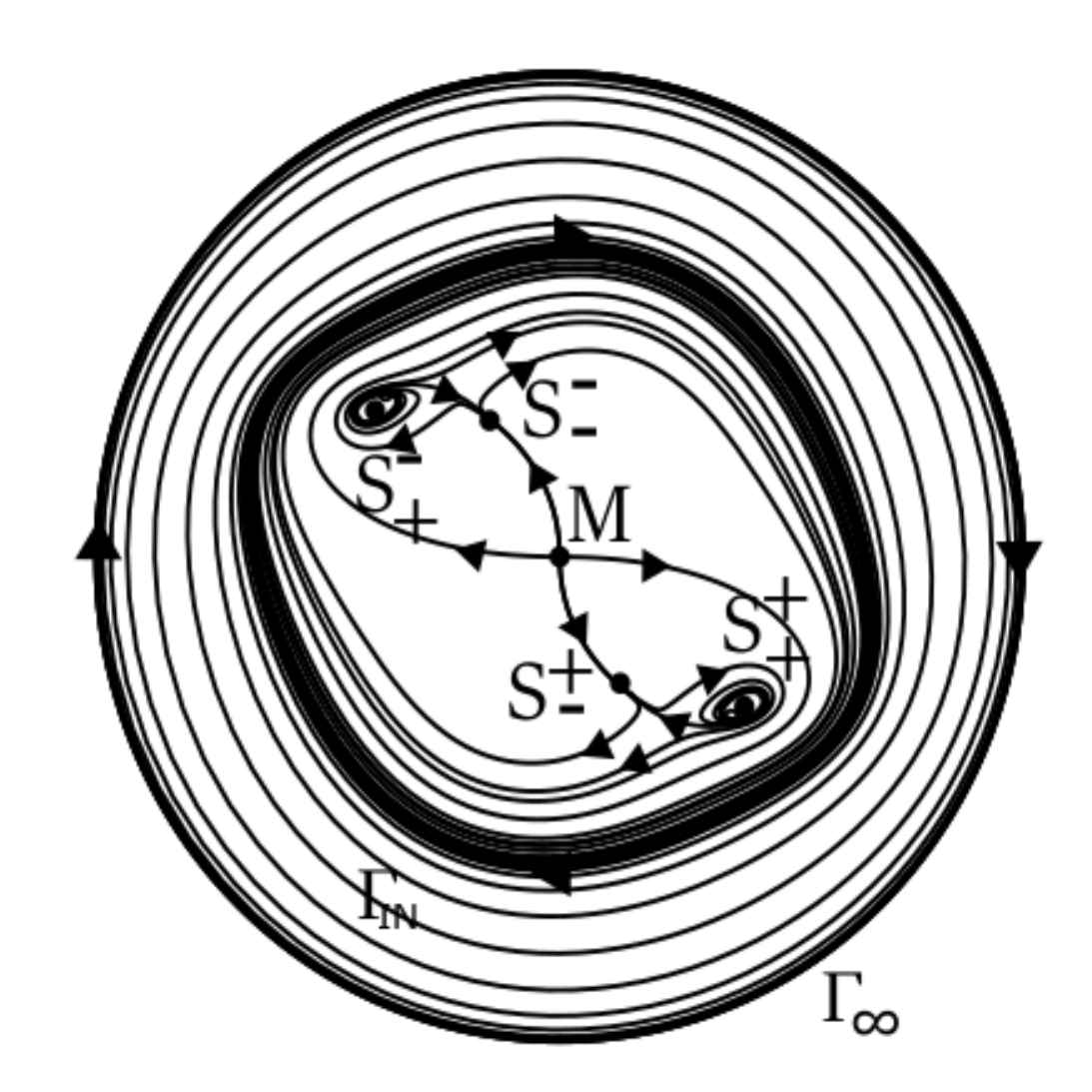}}
		\subfigure[{$(1+nK)\sqrt{n}<\nu<f_-(n,K)$}.]{\label{c4}
			\includegraphics[width=0.25\textwidth]{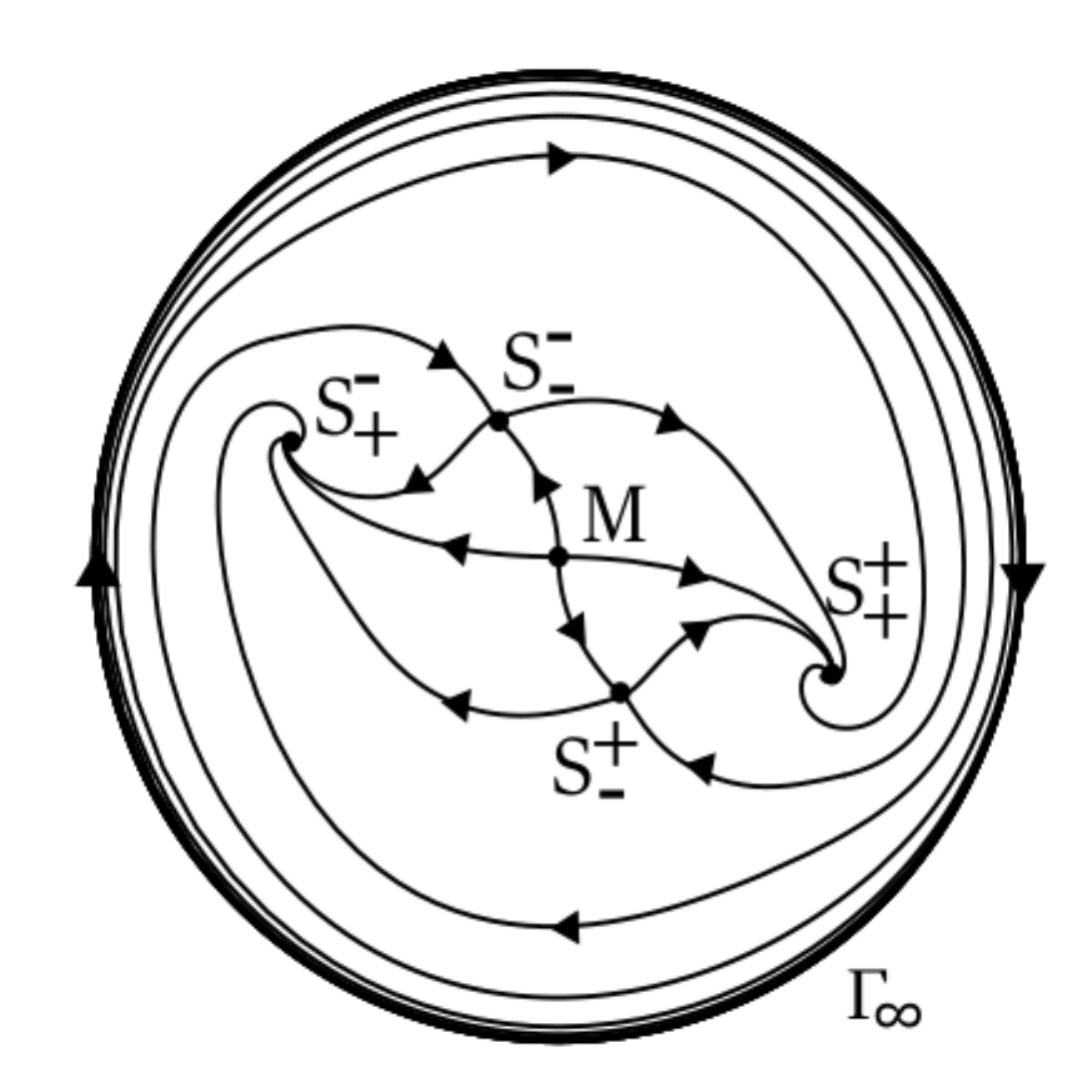}}
		\subfigure[$ f_-(n,K)\leq \nu<2n$.]{\label{c5}
			\includegraphics[width=0.25\textwidth]{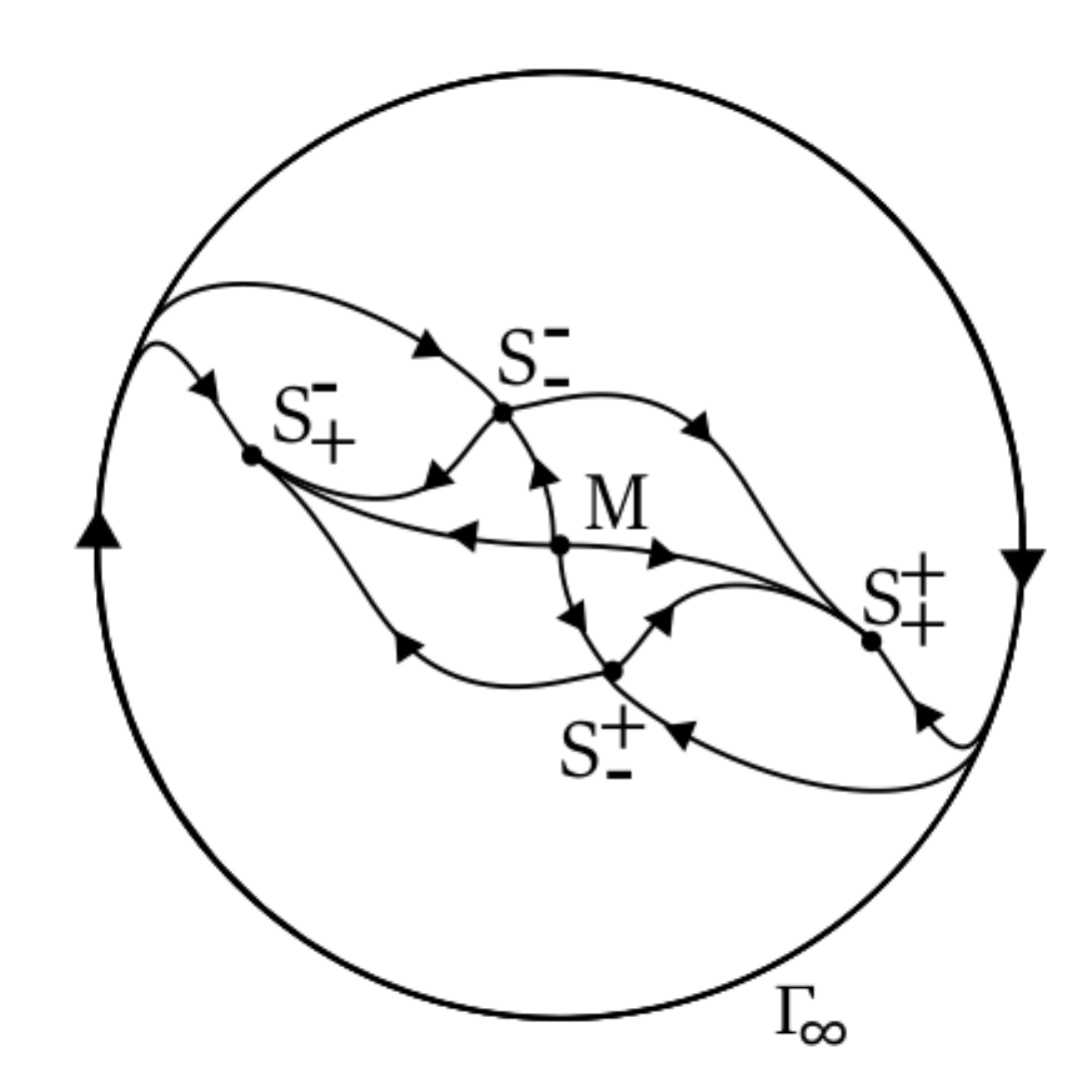}}
	\end{center}
	\vspace{-0.5cm}
	\caption{Poincar\'e-Lyapunov disks when $p=\frac{1}{2}(n-1)$ and $K\in(0,\frac{1}{n})$. 
	}
	\label{fig:PL4}
\end{figure}
Most of the results about the existence of limit cycles for the Li\'enard system rely on the strong assumption that $\bar{x}g(\bar{x})>0$ for $\bar{x}\neq0$, i.e. that the fixed point at the origin is the only fixed point of the system, see e.g.~\cite{SV11} and references therein. For recent works where such assumption is relaxed see e.g.~\cite{GG21} and references therein. Nevertheless, they do not seem enough to prove the numerical results discussed in Remark~\ref{Missing}. This is the case, for instance, of Theorem 3.4 of~\cite{GG21} with $V(x,Y)=Y^2+2G(x)-YF(x)$, for which $H(x)=-2F(x)G'(x)+4F'(x)G(x)=\frac{6(n-1)x^{n+1}}{n^2(n+1)(1-K)}\left(\left(n(1+n)(1+ n K)-\nu^2\right)x^{n-1}-3\nu\right)$.
\begin{remark}\label{ASFL1_pEQ}
	It is interesting to obtain the asymptotics for the orbits on the cylinder $\mathbf{S}$ towards $\mathrm{FL}_1$. For example when $0<\gamma_\mathrm{pf}<\frac{2n}{n+1}$, there exists a one parameter family of orbits in $\mathbf{S}$ with the asymptotics towards $\mathrm{FL}_1$ as in Remark \ref{AsFL1} with $n-2p=1$.
\end{remark}
%
\section{Dynamical systems' analysis when $p=\frac{n}{2}$}
\label{case2}
By setting $p=\frac{n}{2}$ with $n$ even, i.e. $(p,n)=(1,2),(2,4),(3,6),...$, the global dynamical system \eqref{globalsys} becomes
\begin{subequations} \label{system,n=2p}
	\begin{align}
	\frac{dX}{d\tau} &= \frac{1}{n}(1+q)(1-T)X+T\Sigma_{\mathrm{\phi}}, \\
	\frac{d\Sigma_{\mathrm{\phi}}}{d\tau} &= -(2-q)(1-T)\Sigma_{\mathrm{\phi}}-n T X^{2n-1}-\nu(1-T)X^{n}\Sigma_{\mathrm{\phi}} , \\
	\frac{dT}{d\tau} &= \frac{1}{n}(1+q)T(1-T)^2 ,
	\end{align}
\end{subequations}
where the deceleration parameter $q$ is given by $\eqref{deceleration}$. The auxiliary equation for $\Omega_\phi$ (or equivalently for $\Omega_\mathrm{pf}=1-\Omega_\phi$) takes the form
\begin{equation}
\frac{d\Omega_{\phi}}{d\tau} = -3(1-T)\left[(\gamma_\mathrm{pf}-\gamma_\phi)\Omega_\phi(1-\Omega_{\phi})+\nu \sigma_\phi \Omega^{\frac{3}{2}}_\phi\right],
\end{equation}
where we recall that the \emph{effective scalar-field equation of state} is $\gamma_\phi=2\Sigma^2_\phi/\Omega_\phi$
and we introduced the \emph{effective interaction term} $\sigma_\phi$ defined by
\begin{equation}
\sigma_\phi = \frac{2\Sigma_{\mathrm{\phi}}^2X^n}{3\Omega^{3/2}_\phi}.
\end{equation}
%
\subsection{Invariant boundary $T=0$}
%
The induced flow on the $T=0$ invariant boundary is given by
\begin{equation}\label{T=0,n=2p}
	\frac{dX}{d\tau}= \frac{1}{n}(1+q)X, \qquad
\frac{d\Sigma_{\mathrm{\phi}}}{d\tau}=-\Big[2-q+\nu X^{n}\Big]\Sigma_{\mathrm{\phi}},
\end{equation}
and from the auxiliary equation for $\Omega_\mathrm{pf}$, we have that
\begin{equation}\label{Inv}
\left.\frac{d\Omega_{\mathrm{pf}}}{d\tau}\right|_{\Omega_{\mathrm{pf}}=1}=0, \qquad\left.\frac{d\Omega_{\mathrm{pf}}}{d\tau}\right|_{\Omega_{\mathrm{pf}}=0}=2\nu X^{n}\Sigma_{\mathrm{\phi}}^2 \geq 0.
\end{equation}
Therefore on the $T=0$ invariant boundary, the system \eqref{system,n=2p} admits five fixed points. The fixed point at the origin, with $\Omega_{\mathrm{pf}}=1$, is given by 
\begin{equation}
\mathrm{FL_{0}}:\quad X=0 \quad, \quad \Sigma_{\mathrm{\phi}}=0, \quad\quad T=0,
\end{equation}
where $q=\frac{1}{2}(3\gamma_\mathrm{pf}-2)$ corresponds to the flat Friedmann-Lema\^itre solution, and whose linearisation yields the eigenvalues $\frac{3}{2n}$, $-\frac{3}{2}(2-\gamma_{\mathrm{pf}})$ and $\frac{3}{2n}$ with associated eigenvectors the canonical basis of $\mathbb{R}^3$. This fixed point has one negative and two positive real eigenvalues, being a hyperbolic saddle from which originates a 1-parameter family of Class A orbits in $\mathbf{S}$ (see Definition \ref{definition1}). 

On $T=0$, the subset $\Omega_\mathrm{pf}=0$ ($X^{2n}+\Sigma^2_\phi=1$) is not invariant (except at $\Sigma_\phi=0$ or $X=0$), but it is future-invariant as follows from~\eqref{Inv}. On this subset there are four fixed points. The first two equivalent fixed points are given by
\begin{equation}
\mathrm{K^{\pm}}:\quad X=0,\quad\quad \Sigma_{\mathrm{\phi}}=\pm 1, \quad\quad T=0,
\end{equation}
and correspond to massless scalar field states (or kinaton states) with $q=2$. The linearisation around these fixed points yields the eigenvalues $\frac{3}{n}$, $\frac{3}{n}$, and $3(2-\gamma_{\mathrm{pf}})$, with associated eigenvectors the canonical basis of $\mathbb{R}^3$. It follows that $\mathrm{K^{\mathrm{\pm}}}$ are hyperbolic sources, so that from each $\mathrm{K^{\mathrm{\pm}}}$ originates a 2-parameter family of Class A orbits in $\mathbf{S}$. The other two equivalent fixed points are
\begin{equation}
\mathrm{dS^{\pm}_0}:\quad X=\pm 1,\quad\quad \Sigma_{\mathrm{\phi}}=0,\quad\quad T=0 ,
\end{equation}
and correspond to a quasi-de-Sitter state with $q=-1$. 
The linearisation yields the eigenvalues $-3\gamma_{\mathrm{pf}}$, $-(3+\nu)$ and $0$ with eigenvectors $(1,0,0)$, $(0,1,0)$ and $(0,\mp\frac{n}{3+\nu},1)$. These fixed points have two negative real eigenvalues and a zero eigenvalue, possessing a 2-dimensional stable manifold contained in the boundary $T=0$, and a 1-dimensional center manifold (the inflationary attractor solution). Just as in section~\ref{T0,n-2p>1}, the monotonicity of $T$ implies that $\mathrm{dS}^\pm_0$ are center-saddles with a unique orbit, the center manifold orbit, entering the state-space $\mathbf{S}$. However, due to its relevant physical meaning, it is important to obtain approximations for the center manifold solution.

In order to analyse the center manifold of $\mathrm{dS}^\pm_0$  we use instead system \eqref{sistemalocal} for the unbounded  variable $\tilde{T}$ for $p=\frac{n}{2}$, and introduce the adapted variables
\begin{equation}\label{def-v}
\bar{X} = X\mp 1,\qquad \bar{\Sigma}_\phi = \Sigma_\phi\pm \frac{n}{3+\nu}\tilde{T},\qquad \bar{T}=\tilde{T},
\end{equation}
which put the fixed points $\mathrm{dS}^\pm_0$ at the origin of coordinates $(\bar{X},\bar{\Sigma}_\phi,\tilde{T})=(0,0,0)$. This leads to the system
\begin{equation}
\frac{d\bar{X}}{dN} = -3\gamma_\mathrm{pf} \bar{X}+F(\bar{X},\bar{\Sigma}_\phi,\tilde{T}),\qquad 
\frac{d\bar{\Sigma}_\phi}{dN} = -(3+\nu)\bar{\Sigma}_\phi+ G(\bar{X},\bar{\Sigma}_\phi,\tilde{T}),\qquad 
\frac{d\tilde{T}}{dN} = N(\bar{X},\bar{\Sigma}_\phi,\tilde{T}),
\end{equation}
where $F$, $G$ and $N$ are functions of higher-order terms. The 1-dimensional center manifold $W^{\mathrm{c}}$ at $\mathrm{dS}^\pm_0$ can be locally represented as the graph  $h\,: \,E^c\rightarrow E^s$, i.e. $(\bar{X},\bar{\Sigma}_\phi)=(h_1(\tilde{T}),h_2(\tilde{T}))$, satisfying the fixed point, $h_1(0)=0=h_2(0)$, and the tangency conditions $\frac{dh_1(0)}{d\tilde{T}}=0=\frac{dh_2(0)}{d\tilde{T}}$. Hence using $\tilde{T}$ as an independent variable, we get
\begin{subequations}\label{CMds}
	\begin{align}
	&\frac{1}{n}(1+q)\Big(\tilde{T}\frac{dh_1}{d\tilde{T}}-(h_1(\tilde{T})\pm 1)\Big)-\tilde{T}\Big(h_2(\tilde{T})\mp \frac{n}{3+\nu}\tilde{T}\Big)=0, \\
	&\frac{1}{n}(1+q)\tilde{T}\Big(\frac{dh_2}{d\tilde{T}}\mp \frac{n}{3+\nu}\Big)+(2-q)\Big(h_2(\tilde{T})\mp \frac{n}{3+\nu}\Big)+\\ \nonumber &\qquad\qquad+\nu \Big(h_1(\tilde{T})\pm 1\Big)^{n}\Big(h_2(\tilde{T})\mp \frac{n}{3+\nu}\tilde{T}\Big)-\tilde{T}\Big(h_1(\tilde{T})\pm1\Big)^{2n-1}=0.
	\end{align}
\end{subequations}
Finding the attractor solution amounts to solve the above system of non-linear ordinary differential equations. We can however approximate the solution by performing a formal power series expansion
\begin{equation}\label{taylor}
h_1(\tilde{T}) = \sum_{i=1}^{N}a_i \tilde{T}^{i} , \qquad
h_2(\tilde{T}) = \sum_{i=1}^{N}b_i \tilde{T}^{i}, \qquad \text{as}\quad \tilde{T}\rightarrow 0,
\end{equation}
where $a_i,b_i\in\mathbb{R}$. Inserting \eqref{taylor} into~\eqref{CMds} subject to the fixed point and tangency conditions, and solving the resulting linear system of equations for the coefficients, results in
\begin{subequations}\label{exptaylor}
	\begin{align*}
	X &= \pm 1\mp \frac{n}{6\gamma_{\mathrm{pf}}}\frac{3\gamma_{\mathrm{pf}}+2\nu}{(3+\nu)^2}\tilde{T}^2+\mathcal{O}(\tilde{T}^3), \\
	\Sigma_{\mathrm{\phi}}  &=\mp\frac{n}{3+\nu}\tilde{T}\left[1-\frac{n}{6\gamma_{\mathrm{pf}}(3+\nu)^3}\Bigg(3\gamma_{\mathrm{pf}}\Big(3+(1-n)\nu\Big)-2\nu \Big(3+(1-n)\nu-6n\Big)\Bigg)\tilde{T}^2+\mathcal{O}(\tilde{T}^3)\right], \nonumber\\
	\Omega_\mathrm{pf} &= \frac{2n^2 \nu}{3\gamma_{\mathrm{pf}}(3+\nu)^2}\tilde{T}^2+\mathcal{O}(\tilde{T}^3).
	\end{align*}
\end{subequations}
Therefore, it follows that to leading order on the center manifold
\begin{equation}
\frac{d\tilde{T}}{dN}=\frac{n}{3+\nu}\tilde{T}^3+\mathcal{O}(\tilde{T}^4),\quad\text{as}\quad\tilde{T}\rightarrow 0,
\end{equation}
which shows explicitly that $\mathrm{dS}^\pm_0$ are center saddles with a unique class A center manifold orbit originating from each fixed point into the interior of $\mathbf{S}$. 

We now show that on $T=0$ the fixed points $\mathrm{FL}_0$, $\mathrm{K}^{\pm}$ and $\mathrm{dS}^{\pm}_0$ are the only possible $\alpha$-limit sets for Class A orbits in $\mathbf{S}$, and that the orbit structure on $T=0$ is very simple:
\begin{figure}[ht!]
	\begin{center}
		\includegraphics[width=0.30\textwidth]{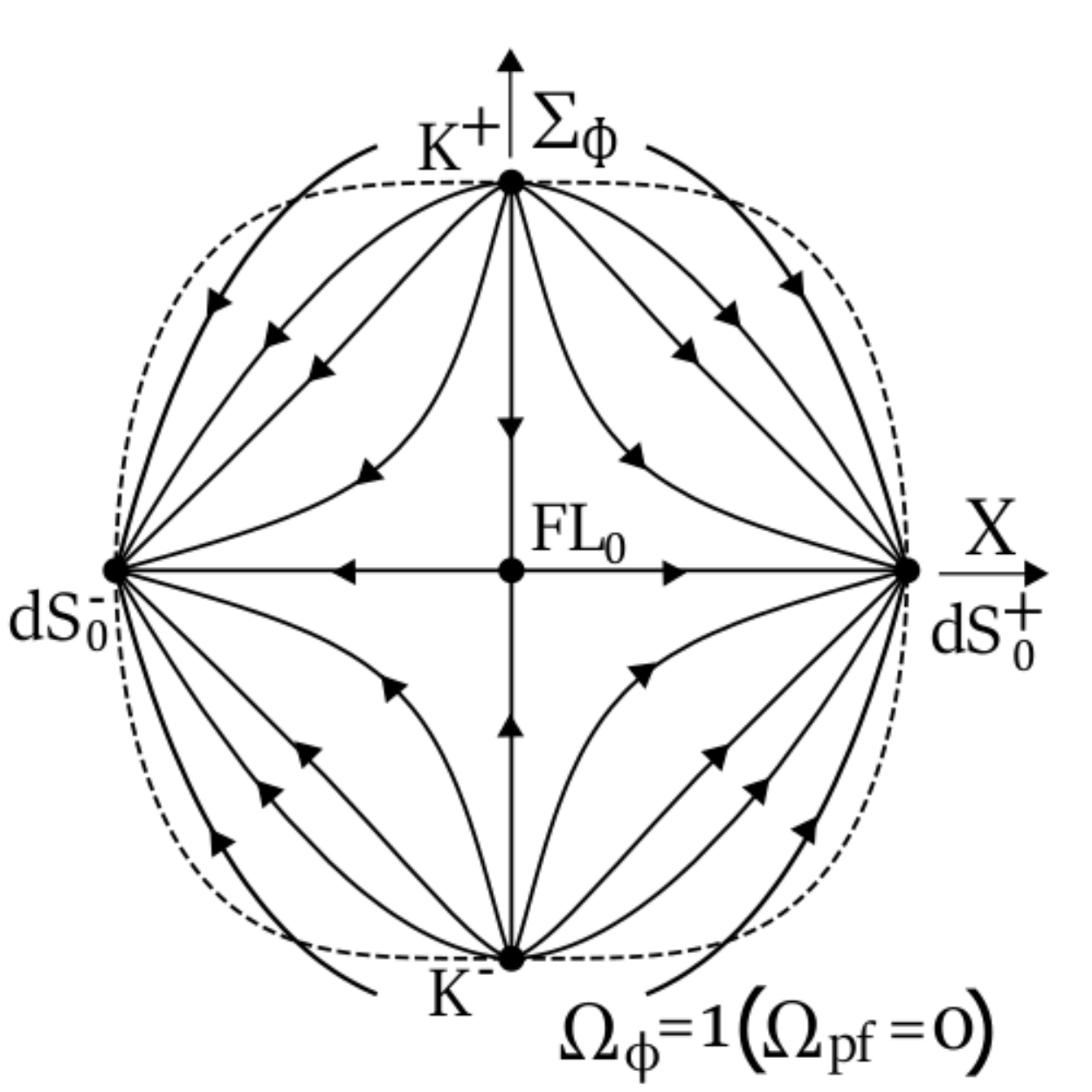}
	\end{center}
	\vspace{-0.5cm}
	\caption{Invariant boundary $\{T=0\}$ when $p=\frac{n}{2}$.}
	\label{fig:T=0,n=2p}
\end{figure}
\begin{lemma}\label{Dulac1}
	Let $p=\frac{n}{2}$. Then the $\{T=0\}$ invariant boundary consists of heteroclinic orbits connecting the fixed points and semi-orbits crossing the set $\Omega_\mathrm{pf}=0$ and converging to $\mathrm{dS}^{\pm}_0$, as depicted in Figure~\ref{fig:T=0,n=2p}. 
\end{lemma}
\begin{proof}
	 The proof is identical to the proof of Lemma~\ref{LemT0_ng2p}, although in this case there are no conserved quantities, and the set $\Omega_\mathrm{pf}=0$ is not invariant but future invariant, except at $\mathrm{K}^\pm$ and $\mathrm{dS}^\pm_0$.
\end{proof}
\begin{theorem}
	Let $p=\frac{n}{2}$. Then the $\alpha$-limit set for class $\mathrm{A}$ orbits in $\mathbf{S}$ consists of fixed points on $\{T=0\}$. In particular as $\tau\rightarrow-\infty$, a 2-parameter set of orbits converges to each $\mathrm{K}^\pm$, a 1-parameter set orbits converges to $\mathrm{FL}_0$, and a unique center manifold orbit converges to each $\mathrm{dS}^\pm_0$. 
\end{theorem}
\begin{proof}
	The proof follows from Lemma \ref{lemma1}, Lemma \ref{Dulac1} and the local stability analysis of the fixed points.
\end{proof}
\begin{remark}\label{PastAs2}
	When $p=n/2$, the asymptotics for the inflationary attractor solution originating from $\mathrm{dS}^\pm_0$ are given by 
	\begin{subequations}	
		\begin{align*}
		&n=1\,:\qquad H\sim -t\quad,\quad \phi\sim-t\quad,\quad \rho_\mathrm{pf} \sim (-t)^{-2}, \quad \quad\text{as}\quad t\rightarrow -\infty \\
		&n=2\,:\qquad H\sim e^{-\frac{2}{3+\nu}t}\quad,\quad \phi\sim e^{-\frac{t}{3+\nu}}\quad,\quad \rho_\mathrm{pf} \sim e^{\frac{4}{3+\nu}t}, \quad \quad\text{as}\quad t\rightarrow -\infty \\
		&n\geq 3\,:\qquad H\sim (-t)^{\frac{n}{n-2}}\quad,\quad \phi\sim (-t)^{\frac{2}{n-2}}\quad,\quad \rho_\mathrm{pf} \sim (-t)^{-2}, \quad \quad\text{as}\quad t\rightarrow -\infty.
		\end{align*}
	\end{subequations}
\end{remark}
%
\subsection{Invariant boundary $T=1$}
%
When $p=\frac{n}{2}$, the induced flow on the $T=1$ invariant boundary is given by
\begin{equation}\label{T1neq2p}
	\frac{dX}{d\tau} = \Sigma_{\mathrm{\phi}}, \qquad 
\frac{d\Sigma_{\mathrm{\phi}}}{d\tau} = -n X^{2n-1}.
\end{equation}
This system has only one fixed point at $\Omega_\mathrm{pf}=1$ given by
\begin{equation}
\mathrm{FL_1}: \quad X=0, \quad\quad \Sigma_{\mathrm{\phi}}=0, \quad\quad T=1.
\end{equation}
\begin{lemma}
	Let $p=\frac{n}{2}$. Then the $\{T=1\}$ invariant boundary is foliated by periodic orbits $\mathcal{P}_{\Omega_\phi}$ characterized by constant values of $\Omega_\phi$, centered at $\mathrm{FL}_1$, see Figure~\ref{fig:T=1,n=2p}.
\end{lemma}
\begin{proof}
On $\{T=1\}$ the auxiliary equation for $\Omega_\phi$ reads
\begin{equation}
\frac{d\Omega_\phi}{d\tau}=0\quad\Rightarrow \quad \Omega_\phi=X^{2n}+\Sigma_{\mathrm{\phi}}^2=\text{const.}
\end{equation}
from which the result follows.
%
\end{proof}
\begin{figure}[ht!]
	\begin{center}
		\includegraphics[width=0.30\textwidth]{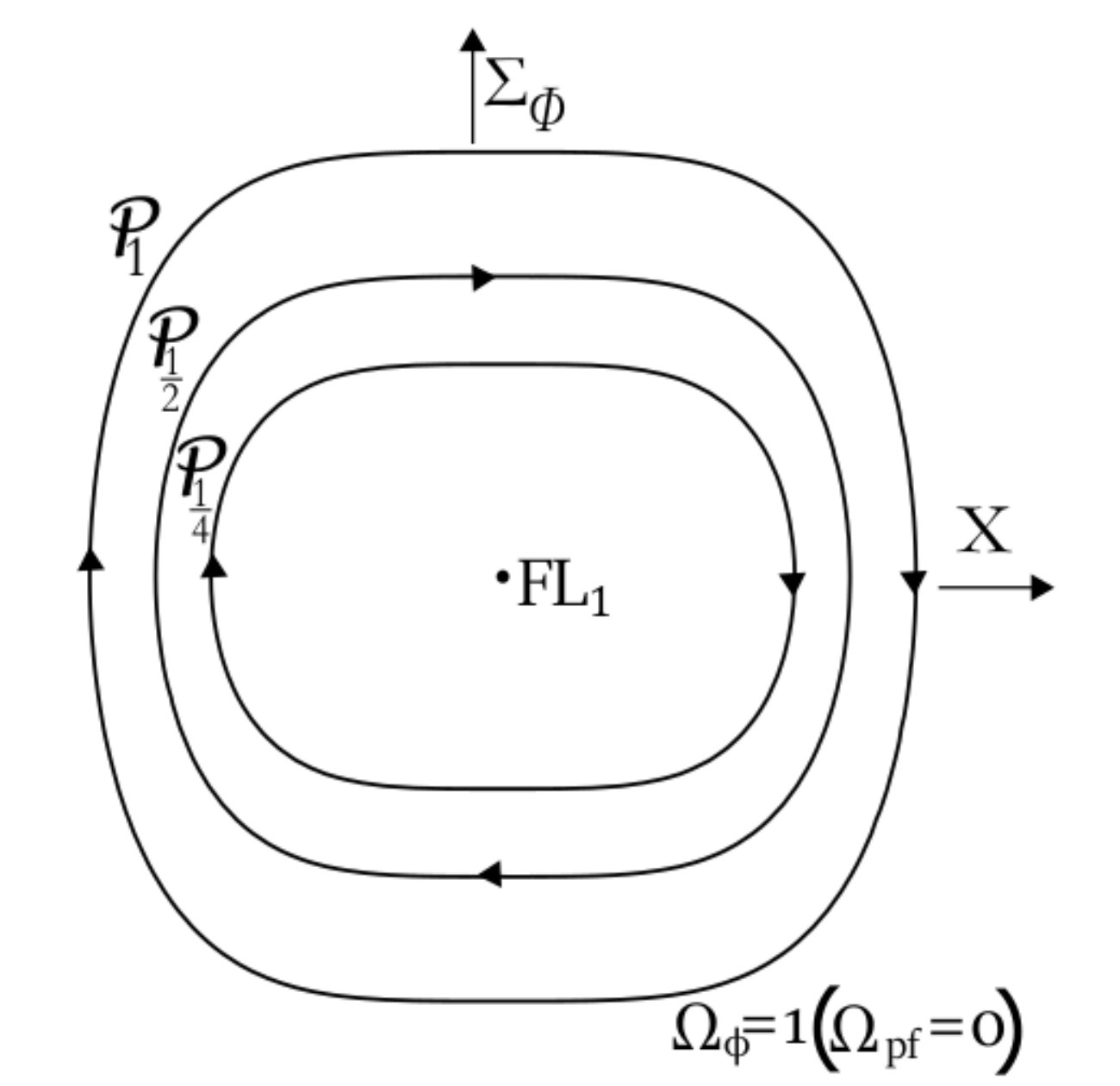}
	\end{center}
	\vspace{-0.5cm}
	\caption{Invariant boundary $\{T=1\}$ when $p=\frac{n}{2}$.}
	\label{fig:T=1,n=2p}
\end{figure}
\begin{theorem}\label{teorema1}
	Let $\nu>0$ and $p=\frac{n}{2}$. 
	\begin{itemize}
		\item[(i)] If $0<\gamma_\mathrm{pf}\leq\frac{2n}{n+1}$, then all orbits in $\mathbf{S}$ converge, for $\tau \rightarrow + \infty$, to the fixed point $\mathrm{FL}_1$ with $\Omega_\mathrm{pf}=0$ $(\Omega_\mathrm{pf}=1)$; \label{igual}
		\item[(ii)] If $\frac{2n}{n+1}<\gamma_\mathrm{pf}<2$, then all orbits in $\mathbf{S}$ converge, for $\tau \rightarrow + \infty$, to an inner periodic orbit $\mathcal{P}_{\Omega_\phi}$ with $0<\Omega_\phi=1-\Omega_\mathrm{pf}<1$.\label{maior}
	\end{itemize}
\end{theorem}
\begin{proof}
	The proof relies on Lemma \ref{lemma1} together with generalised averaging techniques based on the methods introduced in \cite{Alho2,Bessa}, see also~\cite{Faj21} and~\cite{Letal21}. Standard averaging techniques and theorems can be found in~\cite{GHbook} for the periodic case and~\cite{Sandersbook} for the general case. In these theorems a key role is played by a perturbation parameter $\epsilon$. Roughly, a differential equation of the form
	$\frac{dx}{d\tau}=\epsilon f(x,\tau,\epsilon)$ with $\epsilon>0$ is approximated by the truncated averaged equation at $\epsilon=0$ , i.e. $\frac{d\bar{y}}{d\tau}=\epsilon \langle f(x,.,0)\rangle$. In the present situation the role of   $\epsilon$-parameter is instead played by the function 
	\begin{equation}
	\label{eps-eq}
	\epsilon(\tau)=T(\tau)-1.
	\end{equation}
	Therefore, we have to prove an averaging theorem for the case where $\epsilon$ is not a constant, but a variable that slowly goes to zero. 
	
	Each periodic orbit on $T=1$ $(\epsilon=0)$ has an associated time period $P(\Omega_{\mathrm{\phi}})$, so that for a given real function $f$, its average over a time period associated with $\Omega_{\mathrm{\phi}}$ is given by
	\begin{equation}\label{average}
	\langle f \rangle = \frac{1}{P(\Omega_{\mathrm{\phi}})}\int_{\tau_0}^{\tau_0+P(\Omega_{\mathrm{\phi}})} f(\tau) d\tau.
	\end{equation}
	In what follows we will need to compute several averaged quantities such as $\langle\Sigma^2_\phi\rangle$ and $\langle X^{n} \Sigma^2_\phi\rangle$. We therefore use a different formulation which is better adapted to the problem at hand. So we introduce new polar variables $(r,\theta)$ that solve the constraint equation $\Omega_\phi=\Sigma^2_\phi+X^{2n}$, where $\Omega_\phi$ can be seen as the square of the radial coordinate, i.e. $r=\sqrt{\Omega_\phi}$, and
\begin{equation}\label{Angular}
(X,\Sigma_\phi) =(r^{\frac{1}{n}}\cos\theta ,r G(\theta)\sin\theta),\qquad G(\theta)=\sqrt{\frac{1-\cos^{2n}\theta}{1-\cos^2 \theta}}=\sqrt{\sum^{n-1}_{k=0}\cos^{2k}\theta},
\end{equation}
where $G(\theta)$ satisfies $G(\theta)\geq1$ (with  $G\equiv1$ when $n=1$) and $G(0)=\sqrt{n}$, see~\cite{Alho2,Alho3}. The resulting system of equations takes the form
\begin{subequations} \label{original}
	\begin{align}
	\frac{dr}{d\tau} &= \frac{3}{2}\epsilon\left[(\gamma_{\mathrm{pf}}-\gamma_{\mathrm{\phi}})r(1-r^2)-\nu \sigma_\phi r^2\right] :=\epsilon f(r,\tau,\epsilon), \label{omega evo}\\
	\frac{d\theta}{d\tau} &= -\frac{\epsilon}{2n}\left(3+\nu r\cos^n{\theta}\right)G^2(\theta)\sin{2\theta}+(1-\epsilon)G(\theta)r^{\frac{n-1}{n}}, \label{ThetaEv} \\
	\frac{d\epsilon}{d\tau} &= -\frac{1}{n}\epsilon^2(1-\epsilon)(1+q), \label{epsilon evo}
	\end{align}
\end{subequations}
where 
\begin{equation}
1+q=\frac{3}{2}\left((\gamma_\phi-\gamma_\mathrm{pf})r^2+\gamma_\mathrm{pf}\right)
\end{equation}
and
\begin{subequations}
	\begin{align}
	\gamma_\phi &=2G^2(\theta)\sin^2\theta, \label{EqStheta}\\
	\sigma_\phi  &=\frac{2}{3}G^2(\theta)\sin^2{\theta}\cos^{n}\theta. \label{Inttheta}
	\end{align}
\end{subequations}
At $\epsilon=0$, it follows that
\begin{equation}
\frac{d\theta}{d\tau} = G(\theta)r^{\frac{n-1}{n}}
\end{equation}
and therefore $\theta$ is strictly monotonically increasing in $r\in(0,1]$. In this case, the average of a real function $f$ over a time period, associated with $\Omega_\phi=\text{const.}$, is given by
\begin{equation}
\langle f \rangle = \frac{\Gamma[\frac{n+1}{2n}]}{4\sqrt{\pi}\Gamma[1+\frac{1}{2n}]}\int^{2\pi}_{0}\frac{f}{G(\theta)}d\theta,
\end{equation}
where $\Gamma[x]$ is the usual $\Gamma$-function. From ~\eqref{Angular} we obtain
\begin{subequations}
	\begin{align}
	\langle X^{2n} \rangle =\frac{\Omega_{\mathrm{\phi}}}{n+1}, \qquad
	\langle\Sigma_{\mathrm{\phi}}^2\rangle
	=\frac{n\Omega_{\mathrm{\phi}}}{n+1},\qquad 	\langle\Sigma_{\mathrm{\phi}}^2X^n\rangle= \frac{\Gamma^2[\frac{1}{2}+\frac{1}{2n}]\Omega^{3/2}_\phi}{2(1+2n)\Gamma^2[1+\frac{1}{2n}]},
	\end{align}
\end{subequations}
where $\Omega_\phi=r^2$. Note that the above implies  $\langle\Sigma_{\mathrm{\phi}}^2 \rangle=n\langle X^{2n}\rangle$, which is in accordance with the result in~\cite{Alho2} obtained by averaging the dynamical system. In particular, it follows that the scalar field equation of state~\eqref{EqStheta} has an average
	\begin{equation}
	\langle \gamma_\phi\rangle = \frac{2n}{n+1},
	\end{equation}
while for the interaction term~\eqref{Inttheta}, we get
\begin{equation}
\langle\sigma_\phi\rangle=\frac{\Gamma^2[\frac{1}{2}+\frac{1}{2n}]}{3(1+2n)\Gamma^2[1+\frac{1}{2n}]},
\end{equation}
both independent of $\Omega_\phi$. 

The general idea of the averaging method is to start with the near identity transformation
	\begin{equation} \label{omega avr}
	r(\tau)=y(\tau)+\epsilon(\tau)g(y,\tau,\epsilon)
	\end{equation}
	and then prove that the evolution of the variable $y$ is approximated, at first order, by the solution $\bar{y}$ of the truncated averaged equation. The evolution equation for $y$ is obtained using equations \eqref{omega evo} and \eqref{epsilon evo}, leading to
\begin{eqnarray*}
\frac{dy}{d\tau} &=& \left(1+\epsilon\frac{\partial g}{\partial y}\right)^{-1}\left[\frac{dr}{d\tau}-\left(g+\epsilon\frac{\partial g}{\partial \epsilon}\right)\frac{d\epsilon}{d\tau}-\epsilon\frac{\partial g}{\partial \tau}\right]\\
&=&\left(1+\epsilon\frac{\partial g}{\partial y}\right)^{-1}\Bigg[\frac{3}{2}\epsilon\left((\gamma_{\mathrm{pf}}-\langle\gamma_{\mathrm{\phi}}\rangle)y(1-y^2)+(\langle\gamma_{\mathrm{\phi}}\rangle-\gamma_{\mathrm{\phi}})y(1-y^2)-\nu \langle\sigma_{\mathrm{\phi}}\rangle y^2+\nu(\langle\sigma_{\mathrm{\phi}}\rangle-\sigma_{\mathrm{\phi}})y^2-\frac{2}{3}\frac{\partial g}{\partial \tau}\right)\nonumber\\
&+&\frac{3}{2}\epsilon^2g\left((\gamma_{\mathrm{pf}}-\gamma_{\mathrm{\phi}})(1-3y^2)-2\nu\sigma_{\mathrm{\phi}}+\frac{1}{n}(1-\epsilon)\left((\gamma_\phi-\gamma_\mathrm{pf})y^2+\gamma_\mathrm{pf}\right)\right)\nonumber\\&-&\frac{3}{2}\epsilon^3\left((\gamma_{\mathrm{pf}}-\gamma_{\mathrm{\phi}})(1-2y)g^2-\nu \sigma_{\mathrm{\phi}}g^2+\frac{1}{n}(1-\epsilon)\left((\gamma_\phi-\gamma_\mathrm{pf})y^2+\gamma_\mathrm{pf}\right)\frac{\partial g}{\partial \epsilon}\right)\Bigg].\nonumber
\end{eqnarray*}
	Setting 
\begin{equation}\label{gbound}
\begin{split}
\frac{\partial g}{\partial \tau}&=f(y,\tau,\epsilon)-\langle f(y,.,0)\rangle\\&=\frac{3}{2}\left(\langle \gamma_{\mathrm{\phi}}\rangle -\gamma_{\mathrm{\phi}}\right)y(1-y^2)+\frac{3}{2}\left(\langle\sigma_{\mathrm{\phi}}\rangle -\sigma_{\mathrm{\phi}}\right)y^2,
\end{split}
\end{equation}
where, for large times, the right-hand-side is almost periodic and has an average that is zero so that the variable $g$ is bounded. Then we can use the fact that $(1+\epsilon\frac{\partial g}{\partial y})^{-1}\approx 1-\epsilon\frac{\partial g}{\partial y}+\mathcal{O}(\epsilon^2)$ to get
\begin{equation}\label{y final}
\frac{dy}{d\tau}=\epsilon\langle f \rangle (y)+\epsilon^2 h(y,g,\tau,\epsilon)+\mathcal{O}(\epsilon^3),
\end{equation}
where
\begin{equation}
\langle f \rangle (y)=\frac{3}{2}\left((\gamma_{\mathrm{pf}}-\langle\gamma_{\mathrm{\phi}}\rangle)y(1-y^2)-\nu \langle\sigma_{\mathrm{\phi}}\rangle y^2\right)
\end{equation}
and
\begin{equation}
\begin{split}
	h(y,g,,\tau,\epsilon)=&\frac{3}{2}\left((\gamma_{\mathrm{pf}}-\gamma_{\mathrm{\phi}})(1-3y^2)g-2\nu \sigma_{\mathrm{\phi}}+\frac{1}{n}\left((\gamma_\phi-\gamma_\mathrm{pf})y^2+\gamma_\mathrm{pf}\right)\right) g \\
	&-\frac{3}{2}\left((\gamma_{\mathrm{pf}}-\langle\gamma_{\mathrm{\phi}}\rangle)y(1-y^2)-\nu \langle\sigma_{\mathrm{\phi}}\rangle y^2\right)\frac{\partial g}{\partial y}.
\end{split}
\end{equation}
Dropping the higher-order terms in $\epsilon$ in \eqref{y final}, we study the truncated averaged equation coupled to an evolution equation for $\epsilon$:
	\begin{subequations}
		\begin{align}
		\frac{d\bar{y}}{d\tau} &= \frac{3}{2}\epsilon\Big(\gamma_{\mathrm{pf}}-\langle\gamma_{\mathrm{\phi}}\rangle\Big)\bar{y}(1-\bar{y}^2)-\frac{3}{2}\epsilon \nu \langle \sigma_{\mathrm{\phi}} \rangle \bar{y}^2, \nonumber \\
		\frac{d\epsilon}{d\tau} &= -\frac{1}{n}\epsilon^2(1-\epsilon)(1+q) .\nonumber
		\end{align}
	\end{subequations}
	This system has a line of fixed points at $\epsilon=0$ which can be removed by introducing a new time variable 
	\begin{equation*}
	\frac{1}{\epsilon}\frac{d}{d\tau}=\frac{d}{d\bar{\tau}},
	\end{equation*}
leading to
	\begin{subequations} \label{y avr system}
		\begin{align}
		\frac{d\bar{y}}{d\bar{\tau}} &= \frac{3}{2}\Big(\gamma_{\mathrm{pf}}-\langle\gamma_{\mathrm{\phi}}\rangle\Big)\bar{y}(1-\bar{y}^2)-\frac{3}{2}\nu \langle \sigma_{\mathrm{\phi}} \rangle \bar{y}^2, \label{y avr system 1}\\
		\frac{d\epsilon}{d\bar{\tau}} &= -\frac{1}{n}\epsilon(1-\epsilon)(1+q), \label{y avr system 2}
		\end{align}
	\end{subequations}
	where now the $\epsilon=0$ invariant subset has the isolated fixed point
	\begin{equation}
	\mathrm{F_1}:\quad y=0,\quad \epsilon=0,
	\end{equation}
	whose linearisation yields the eigenvalues $3(\gamma_{\mathrm{pf}}-\langle \gamma_{\mathrm{\phi}}\rangle)$ and $-3\gamma_{\mathrm{pf}}/2$, with associated eigenvectors $(1,0)$ and $(0,1)$ respectively. When 
	$\gamma_{\mathrm{pf}}>\langle\gamma_{\mathrm{\phi}}\rangle$, and since $\langle \sigma_{\mathrm{\phi}}\rangle \in (0,1)$ and $\nu >0$, there is a second fixed point
	\begin{equation}
	\mathrm{F_2}:\quad \bar{y}=-\frac{1}{2}\frac{\nu \langle \sigma_{\mathrm{\phi}} \rangle}{\gamma_{\mathrm{pf}}-\langle \gamma_{\mathrm{\phi}}\rangle}+\sqrt{1+\frac{1}{4}\frac{\nu^2 \langle \sigma_{\mathrm{\phi}}\rangle ^2}{(\gamma_{\mathrm{pf}}-\langle \gamma_{\mathrm{\phi}}\rangle)^2}},\qquad \epsilon=0.
	\end{equation}
	The linearised system at $\mathrm{F}_2$ has eigenvalues
	\begin{subequations}
		\begin{align*}
		\lambda_1 & = -3(\gamma_{\mathrm{pf}}-\langle \gamma_{\mathrm{\phi}}\rangle)-\frac{3}{2}\nu \langle \sigma_{\phi} \rangle \Bigg(\frac{1}{2}\frac{\nu \langle\sigma_{\mathrm{\phi}}\rangle}{\gamma_{\mathrm{pf}}-\langle \gamma_{\mathrm{\phi}}\rangle}-\sqrt{1+\frac{1}{4}\frac{\nu^2 \langle \sigma_{\phi}\rangle^2}{(\gamma_{\mathrm{pf}}-\langle\gamma_{\mathrm{\phi}}\rangle)^2}}\Bigg), \\
		\lambda_2 & = -\frac{3}{4n}\Bigg(\nu \langle \sigma_{\mathrm{\phi}}\rangle+2\gamma_{\mathrm{pf}}(\langle \gamma_{\mathrm{\phi}}\rangle -\gamma_{\mathrm{pf}})\sqrt{4+\frac{\nu^2 \langle\sigma_{\mathrm{\phi}}\rangle^2}{(\gamma_{\mathrm{pf}}-\langle \gamma_{\mathrm{\phi}}\rangle)^2}}\Bigg),
		\end{align*}
	\end{subequations}
	with associated eigenvectors $v_1=(1,0)$ and $v_2=(0,1)$. Hence $\mathrm{F}_2$ is a hyperbolic sink while $\mathrm{F}_1$ is a hyperbolic saddle. Notice that in the absence of interactions, i.e. when $\nu\rightarrow0$, the fixed point $\mathrm{F}_2$ has the limit $\bar{y}=1$ as in~\cite{Alho,Alho2}. When $\gamma_{\mathrm{pf}}=\langle\gamma_{\mathrm{\phi}}\rangle$, $\mathrm{F}_2$ merges into $\mathrm{F}_1$, leading to a center manifold with flow given by $d\bar{y}/d\bar{\tau}=-3/2 \nu \langle\sigma_{\mathrm{\phi}}\rangle \bar{y}^2$, so that the solutions converge to $\mathrm{F}_1$ tangentially to the $\epsilon=0$ axis. When $\gamma_\mathrm{pf}<\langle \gamma_{\mathrm{\phi}}\rangle$, $\mathrm{F}_1$ is the only fixed point being a hyperbolic sink. 
	
	Next, we prove that the solutions $y$ of the full averaged system \eqref{y final} have the same limit as the solutions $\bar{y}$ of the truncated averaged equation when $\tau \rightarrow +\infty$ and, hence, so does $r$ and subsequently $\Omega_{\mathrm{\phi}}$.  In order to do that we define sequences $\{\tau_n\}$ and $\{\epsilon_n\}$, with $n\in \mathbb{N}$, as follows
	\begin{subequations}
		\begin{align}
		\tau_{n+1}-\tau_{n}=\frac{1}{\epsilon_n},\quad\quad\tau_0=0, \\
		\epsilon_{n+1}=\epsilon(\tau_{n+1}),\quad\quad\epsilon_0>0,
		\end{align}
	\end{subequations}
	with $\lim \tau_n = +\infty$ and $\lim \epsilon_n = 0$, since $\epsilon(\tau) \rightarrow 0$ as $\tau \rightarrow +\infty$. Notice that for $\epsilon$ small enough, $y$ is monotone and bounded and, therefore, has a limit as $\tau\rightarrow+\infty$. Then, we estimate $|\eta(\tau)|=|y(\tau)-\bar{y}(\tau)|$, where 
	$y$ and $\bar{y}$ are solution trajectories with the same initial conditions as
	\begin{eqnarray}
	|\eta(\tau)| &=& \Bigg|\int_{\tau_n}^{\tau}\Bigg(\frac{3}{2}\epsilon\Big(\gamma_{\mathrm{pf}}-\langle\gamma_{\mathrm{\phi}}\rangle\Big)y\Big(1-y^2\Big)-\frac{3}{2}\epsilon\nu \langle\sigma_{\mathrm{\phi}}\rangle y^2+\epsilon^2 h(y,g,\tau,\epsilon)\Bigg)ds\nonumber\\
	&-&\int_{\tau_n}^{\tau}\Bigg(\frac{3}{2}\epsilon\Big(\gamma_{\mathrm{pf}}-\langle\gamma_{\mathrm{\phi}}\rangle\Big)\bar{y}\Big(1-\bar{y}^2\Big)-\frac{3}{2}\epsilon \nu\langle\sigma_{\mathrm{\phi}}\rangle^2\bar{y}^2\Bigg)ds\Bigg|\nonumber \\
	&\leq& \epsilon_n\int_{\tau_n}^{\tau}\frac{3}{2}\underbrace{|\gamma_{\mathrm{pf}}-\langle\gamma_{\mathrm{\phi}}\rangle|}_{|.|\leq C}|y-\bar{y}|\underbrace{|1-(y^2+\bar{y}^2+y\bar{y})|}_{|.|\leq 1}ds+\frac{3}{2}\epsilon_n\nu\langle\sigma_{\mathrm{\phi}}\rangle^2\int_{\tau_n}^{\tau}\underbrace{|1-(y+\bar{y})|}_{|.|\leq 1}|\bar{y}-y|ds \nonumber \\
	&+&\epsilon_n^2\int_{\tau_n}^{\tau} \underbrace{|h(y,g,\tau,\epsilon)|}_{|.|\leq M}ds+\mathcal{O}(\epsilon_n^3) \nonumber \\
	&\leq&\frac{3}{2}C\epsilon_n\int_{\tau_n}^{\tau}|\eta(s)|ds+\frac{3}{2}\epsilon_n\nu\langle\sigma_{\mathrm{\phi}}\rangle\int_{\tau_n}^{\tau}|\eta(s)|ds+\epsilon_n^2 M(\tau-\tau_{n})+\mathcal{O}(\epsilon_n^3)\nonumber\\
	&=& \frac{3}{2}\epsilon_n\Bigg(C+\nu \langle\sigma_{\mathrm{\phi}}\rangle\Bigg)\int_{\tau_n}^{\tau}|\eta(s)|ds+\epsilon_n^2M(\tau-\tau_{n})+\mathcal{O}(\epsilon_n^3),\nonumber
	\end{eqnarray}
	for $\tau>\tau_{n+1}$, and where $C$ and $M$ are  positive constants. By Gronwall's inequality 
	\begin{equation}
	|\eta(\tau)|\leq \frac{\epsilon_n M}{\frac{3}{2}\Big(C+\nu \langle\sigma_{\mathrm{\phi}}\rangle\Big)}\Bigg(e^{\frac{3}{2}\epsilon_n(C+\nu \langle\sigma_{\mathrm{\phi}}\rangle)(\tau-\tau_n)}-1\Bigg)+\mathcal{O}(\epsilon^2).
	\end{equation}
	Hence for $\tau-\tau_{n} \in [0,1/\epsilon_n]$, i.e. $\tau\in[\tau_n,\tau_{n+1}]$, it follows that
	\begin{eqnarray}
	|\eta(\tau)|\leq K \epsilon_n,
	\end{eqnarray}
	with $K$ a positive constant. Letting $n\rightarrow+\infty$ implies that $\eta\rightarrow0$ as $\tau\rightarrow+\infty$. Therefore $y$ and $\bar{y}$ have the same limit as $\tau\rightarrow+\infty$, i.e. the fixed point $\mathrm{F}_1$ or $\mathrm{F}_2$. 
	Finally, from equation \eqref{omega avr}, using the triangle inequality and the fact that $\epsilon \rightarrow 0$ as $\tau \rightarrow +\infty$, it follows that $r$ (and hence also $\Omega_{\mathrm{\phi}}$) has the same limit as $\bar{y}$.
	\end{proof}
The global state-space picture when $p=n/2$ for the three different future asymptotic regimes  is shown in Figure~\ref{fig:n=2p final}. The solid numerical curves correspond to the center manifold solutions of $\mathrm{dS}^{\pm}_0$, and the dashed numerical curves to orbits originating from the source $\mathrm{K}^{-}$.
\begin{figure}[ht!]
	\begin{center}
		\subfigure[$\gamma_\mathrm{pf}-\langle\gamma_\phi\rangle>0$.]{\label{fig:1}
			\includegraphics[width=0.30\textwidth]{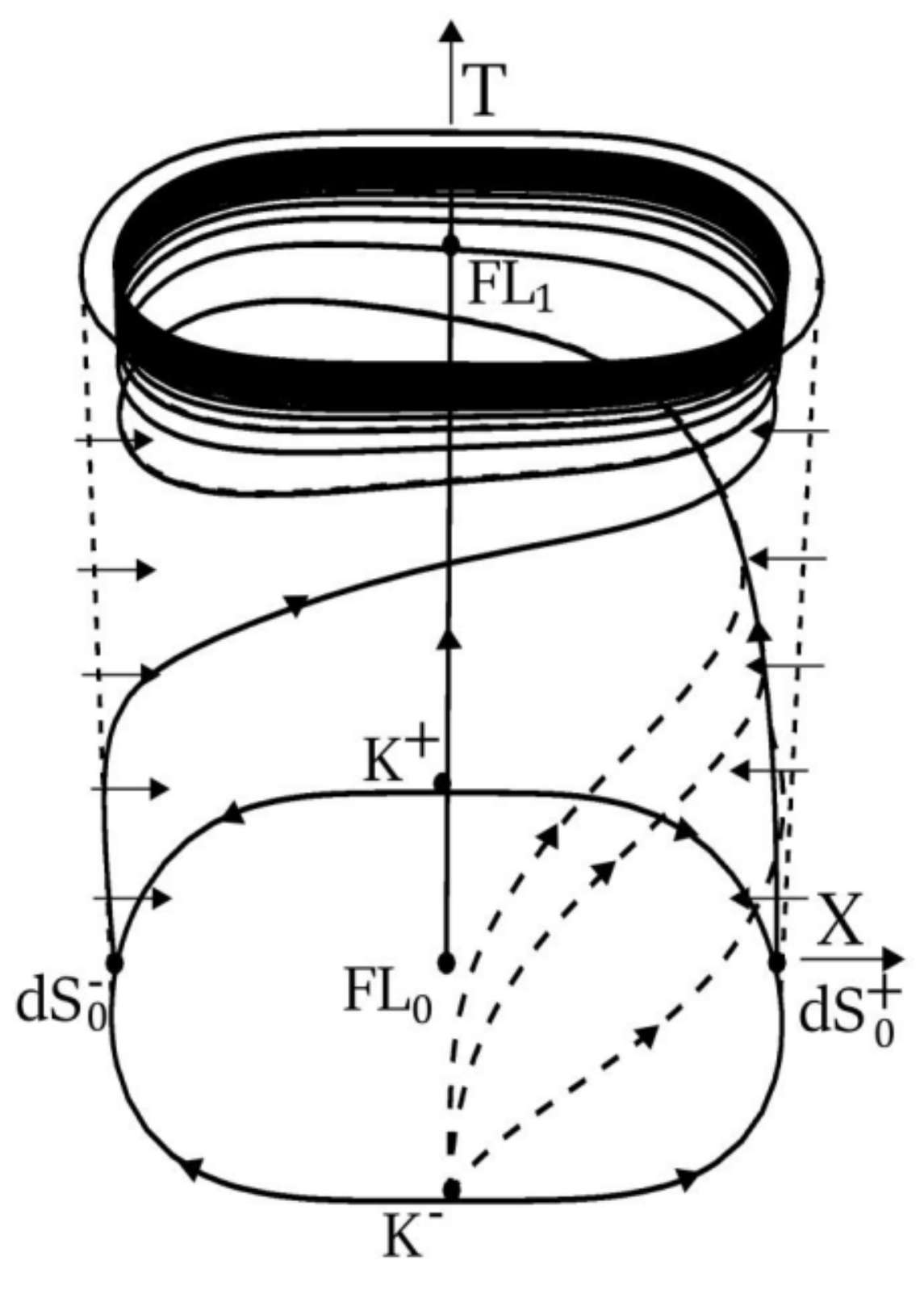}}
		\hspace{0.1 	cm}
		\subfigure[$\gamma_\mathrm{pf}-\langle\gamma_\phi\rangle=0$.]{\label{fig:1.3}
			\includegraphics[width=0.30\textwidth]{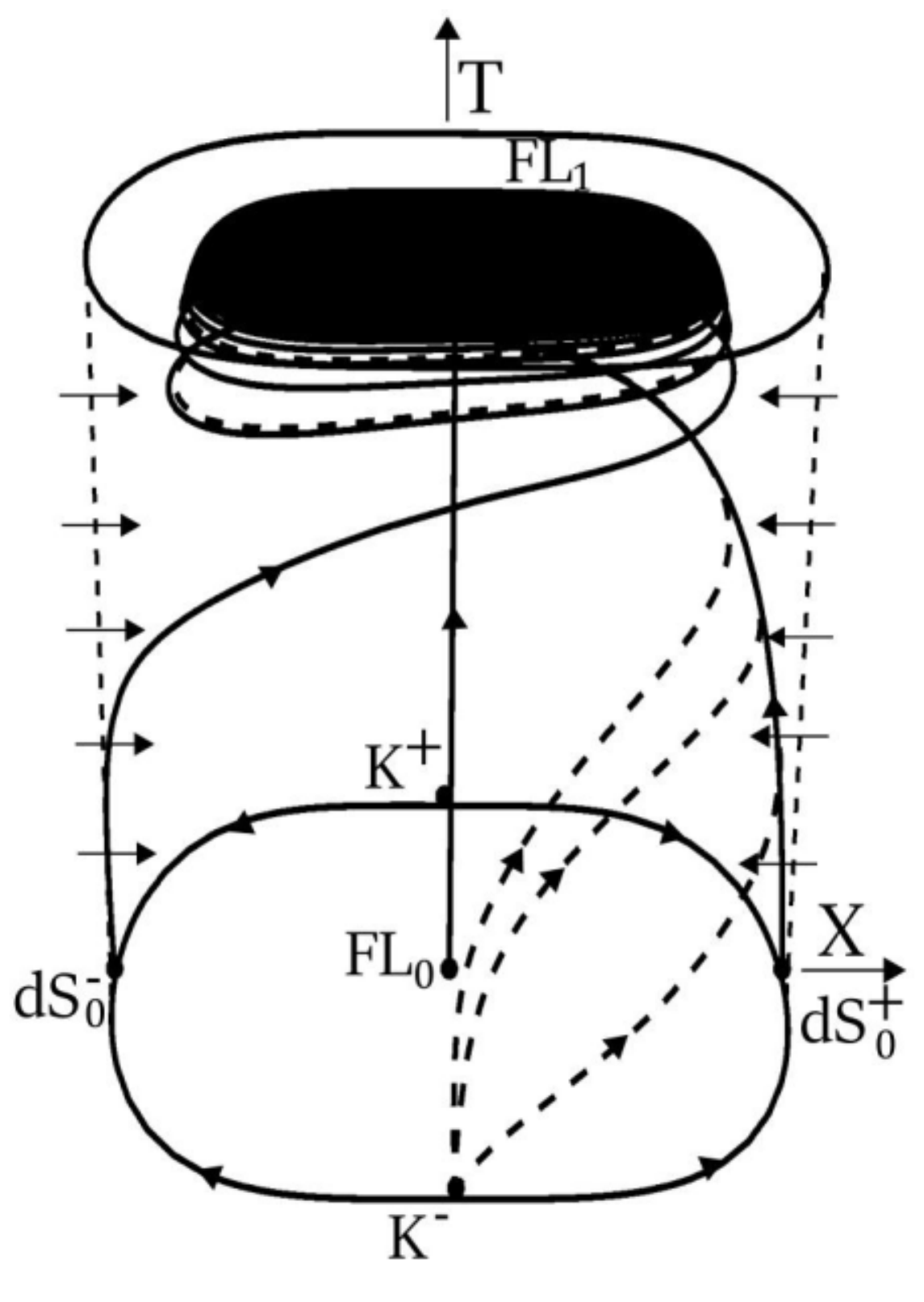}}
		\hspace{0.1cm}
		\subfigure[$\gamma_\mathrm{pf}-\langle\gamma_\phi\rangle<0$.]{\label{fig:1.5}
			\includegraphics[width=0.30\textwidth]{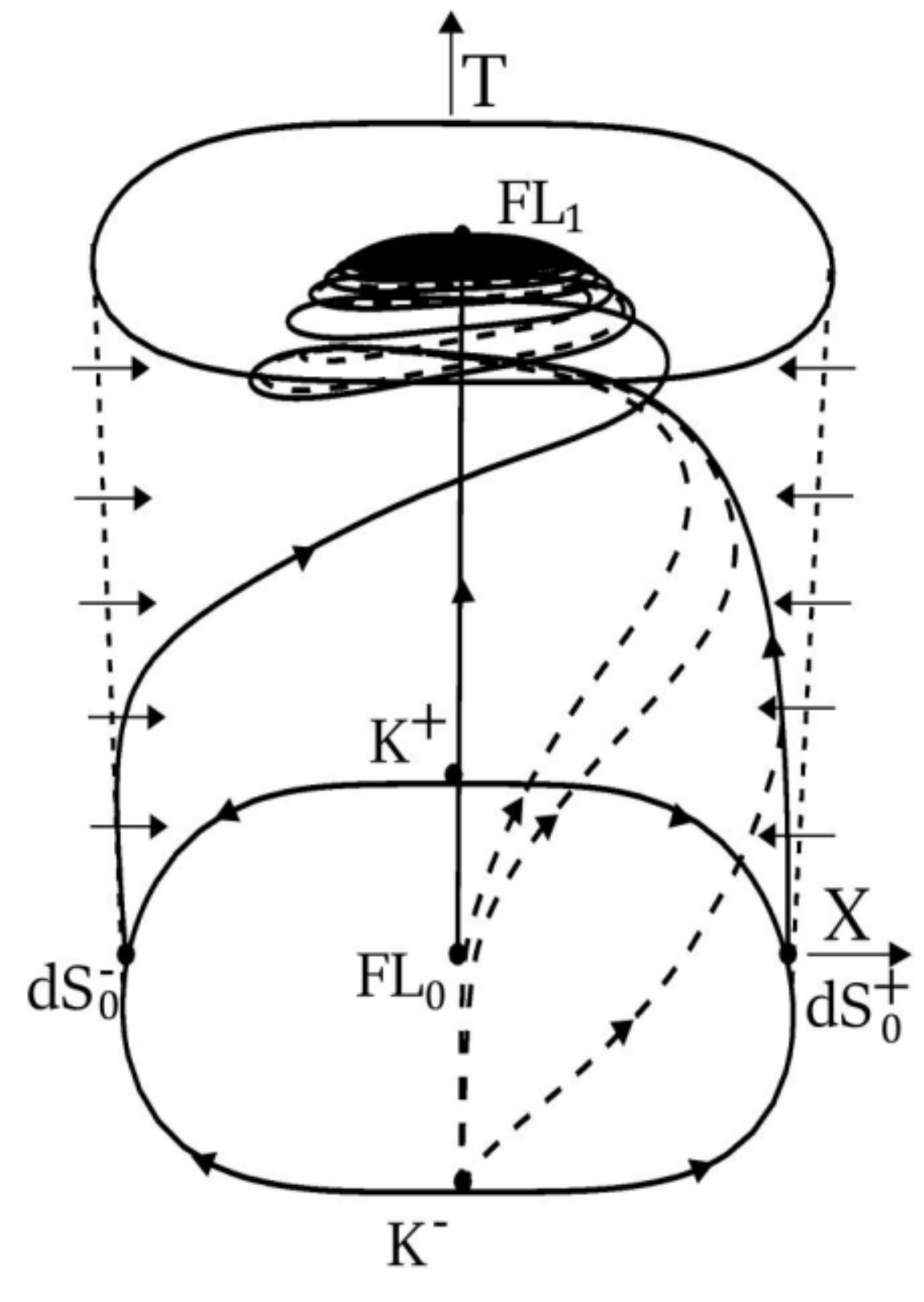}}
	\end{center}
	\vspace{-0.5cm}
	\caption{Global state space $\mathbf{S}$ when $p=\frac{n}{2}$.
		}
	\label{fig:n=2p final}
\end{figure}
%
%
\section{Dynamical systems' analysis when $p>\frac{n}{2}$}
\label{case3}
When $p>\frac{n}{2}$ the global dynamical system \eqref{globalsys} reduces to
\begin{subequations}\label{globalnleq2p}
	\begin{align}
	\frac{dX}{d\tau} &= \frac{1}{n}(1+q)T^{2p-n}(1-T)X + T^{2p-n+1}\Sigma_{\mathrm{\phi}}, \\
	\frac{d\Sigma_{\mathrm{\phi}}}{d\tau} &= -(2-q)T^{2p-n}(1-T)\Sigma_{\mathrm{\phi}}-n T^{2p-n+1}X^{2n-1} - \nu (1-T)^{2p-n+1} X^{2p}\Sigma_{\mathrm{\phi}}, \\
	\frac{dT}{d\tau} &= \frac{1}{n}(1+q)T^{2p-n+1}(1-T)^2,
	\end{align}
\end{subequations}
where we recall $q=-1+3\Sigma^2_\phi+\frac{3}{2}\gamma_\mathrm{pf}\Omega_\mathrm{pf}$ and the auxiliary equation for $\Omega_\phi$ becomes
\begin{equation}
\frac{d\Omega_{\phi}}{d\tau} = -3(1-T)\left[\left(\gamma_\phi-\gamma_{\mathrm{pf}}\right)T^{2p-n}\Omega_{\phi}(1-\Omega_{\phi}) + \nu\sigma_\phi (1-T)^{2p-n}\Omega^{1+\frac{p}{n}}_\phi\right]
\end{equation}
where now the \emph{effective interaction term} $\sigma_\phi$ is given by
\begin{equation}\label{NewInt}
\sigma_\phi := \frac{2\Sigma_{\mathrm{\phi}}^2X^{2p}}{3\Omega^{1+\frac{p}{n}}_\phi}.
\end{equation}

\subsection{Invariant boundary $T=0$}
When $p>\frac{n}{2}$ the induced flow on the $T=0$ invariant boundary reduces to
\begin{equation}\label{T0_2p>n}
\frac{dX}{d\tau}=0,\quad \frac{d\Sigma_\mathrm{\phi}}{d\tau}=-\nu X^{2p}\Sigma_\mathrm{\phi},
\end{equation}	
and $\Omega_\mathrm{pf}=1-\Sigma^2_\phi-X^{2n}$ satisfies
\begin{equation}
\frac{d\Omega_\mathrm{pf}}{d\tau}=2\nu X^{2p}\Sigma_\mathrm{\phi}^2.
\end{equation}
Thus, the subset $\Omega_\mathrm{pf}=0$ is not invariant but future invariant except at $\Sigma_\phi=0$ or $X=0$, which are the points of intersection of the subset $\Omega_\mathrm{pf}=0$ with the lines of fixed points
\begin{equation}
\mathrm{L}_1:\quad X=X_0,\quad \Sigma_\mathrm{\phi}=0,\quad T=0
\end{equation}
with $X_0 \in [-1,1]$ and 
\begin{equation}
\mathrm{L}_2:\quad X_0=0,\quad \Sigma_\mathrm{\phi}=\Sigma_\mathrm{\phi_0},\quad T=0,
\end{equation}
with $\Sigma_\mathrm{\phi_{0}}\in[-1,1]$. We shall refer to the non-isolated fixed point at the origin of the $T=0$ invariant set as $\mathrm{FL}_0=\mathrm{L}_1\cap\mathrm{L}_2$, the end points of $\mathrm{L}_1$ with $X=\pm 1$ as $\mathrm{dS}_0^{\pm}$, and the end points of $\mathrm{L}_2$ with $\Sigma_{\phi0}=\pm1$ as $\mathrm{K}^{\pm}$. The description of the induced flow on $T=0$ is given by the following lemma:
\begin{lemma}\label{T0:p>n/2}
	When $p>\frac{n}{2}$, the set $\{T=0\}\setminus \{\mathrm{L}_1\cup\mathrm{L}_2\}$ is foliated by invariant subsets $X=\text{const.}$ consisting of regular orbits which enter the region $\Omega_{\mathrm{pf}}>0$ by crossing the set $\Omega_\mathrm{pf}=0$ and converging to the line of fixed points $\mathrm{L}_\mathrm{1}$ as $\tau\rightarrow-\infty$, see Figure~\ref{figT0:n-2p<0}.
\end{lemma}

\begin{proof}
	When $p>n/2$ the system \eqref{T0_2p>n} admits the following conserved quantity
	\begin{equation}
	X=\text{const.}
	\end{equation}
	which determines the solution trajectories on the $\{T=0\}$ invariant boundary. The remaining properties of the flow follow from the fact that on the set $\{T=0\}\setminus \{\mathrm{L}_1\cup\mathrm{L}_2\}$, we have that $d\Sigma_\mathrm{\phi}/d\tau<0$ and $d\Omega_{\mathrm{pf}}/d\tau<0$.
\end{proof}
\begin{figure}[ht!]
	\begin{center}
		\includegraphics[width=0.30\textwidth]{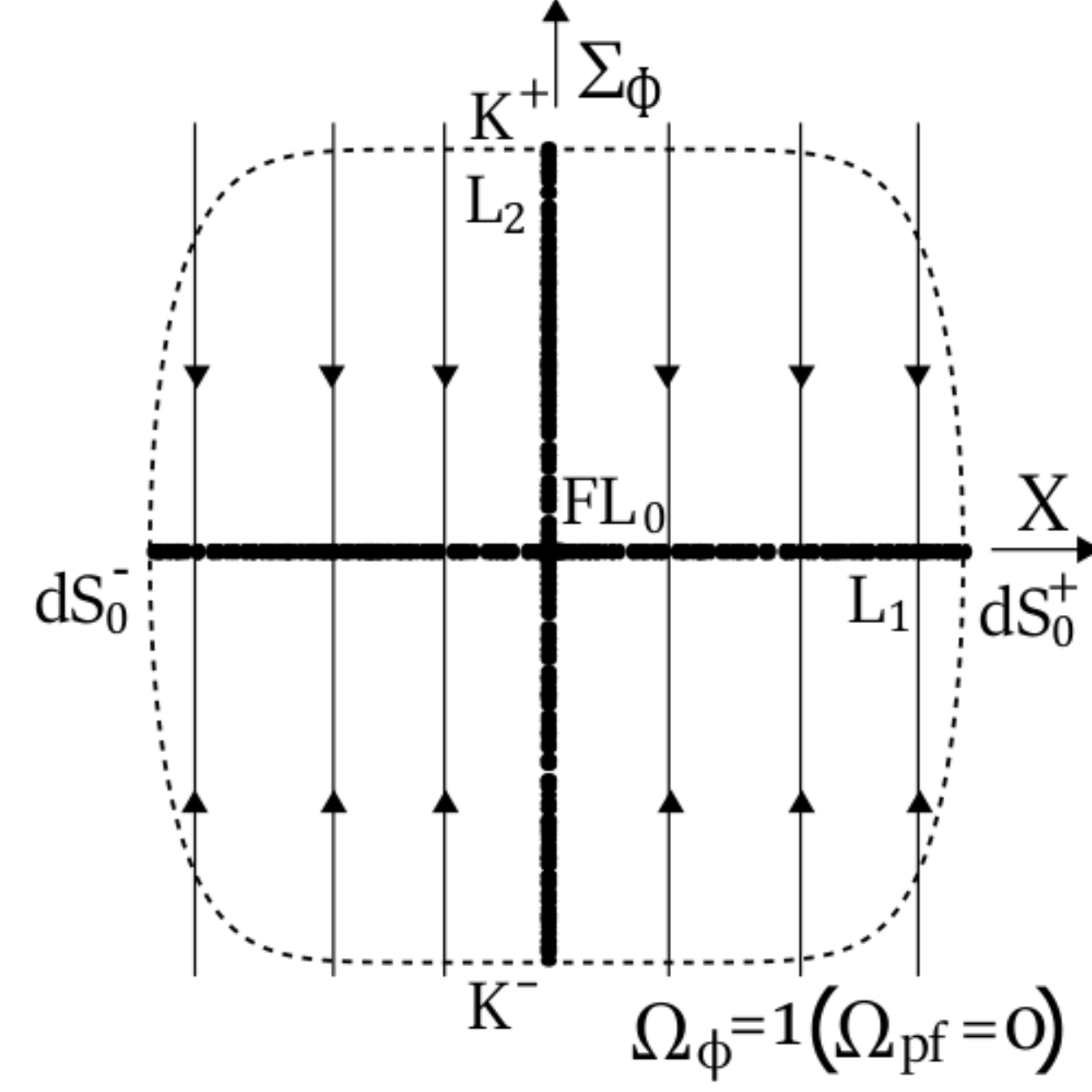}
	\end{center}
	\vspace{-0.5cm}
	\caption{Invariant boundary $\{T=0\}$ when $p>\frac{n}{2}$.}
	\label{figT0:n-2p<0}
\end{figure}

\begin{theorem}
	Let $p>\frac{n}{2}$. Then, the $\alpha$-limit set of all orbits in $\mathbf{S}$ is contained in the set  $\mathrm{dS}^{\pm}_0\cup\mathrm{FL}_0\cup\mathrm{K}^{\pm}$. In particular, as $\tau\rightarrow-\infty$, a 2-parameter set of orbits converges to each fixed point $\mathrm{K}^{\pm}$, a 1-parameter set of orbits converges to $\mathrm{FL}_0$ and a single orbit converges to each of the fixed points $\mathrm{dS}^{\pm}_0$.
\end{theorem}

\begin{proof}
	By Lemma~\ref{lemma1}, the $\alpha$-limit set of all orbits in $\mathbf{S}$ is located at $T=0$ and the description of this boundary is given in Lemma~\ref{T0:p>n/2}. We start by analysing the line of fixed points $\mathrm{L}_1$. The linearised system around $\mathrm{L}_{1}$ has eigenvalues $0$, $-\nu X_{0}^{2p}$ and $0$, with associated eigenvectors $(1,0,0)$, $(0,1,0)$ and $(\frac{3(2p-n)}{2n}X_0(1-X_0^2)\delta^{2p-n}_1,0,1)$. On the $\{T=0\}$ invariant boundary, the line of fixed points $\mathrm{L}_{1}$ is normally hyperbolic, i.e. the linearisation yields one negative eigenvalue for all $X_0\in[-1,1]$, except at $\mathrm{FL}_0$ where the two lines intersect with $X_0=0$, and one zero eigenvalue with eigenvector tangent to the line itself, see e.g.~\cite{Aul84}. On the complement of $\mathbf{S}$, the set $\mathrm{L}_1\setminus\mathrm{FL}_0$ is said to be \emph{partially hyperbolic}, see e.g.~\cite{Tak71}. Each fixed point on this set has a 1-dimensional stable manifold and a 2-dimensional center manifold, while the fixed point $\mathrm{FL}_0$ with $X_0=0$ is non-hyperbolic. To analyse the 2-dimensional center manifold of each partially hyperbolic fixed point on the line, we make the change of coordinates given by
	\begin{equation}
	\bar{X}=X-X_0+\frac{3(2p-n)\gamma_{\mathrm{pf}}}{2n}X_0(1-X_0^{2})\bar{T}\delta^{2p-n}_1,\qquad \bar{\Sigma}_\phi=\Sigma_\phi ,\qquad \bar{T}=T,
	\end{equation}
	which takes a point in the line $\mathrm{L_1}$ to the origin $(\bar{X},\bar{\Sigma}_{\phi},\bar{T})=(0,0,0)$ with $\bar{T}\geq0$. The resulting system of equations takes the form
	\begin{equation}\label{centereq3}
	\frac{d\bar{X}}{d\tau} = F(\bar{X},\bar{\Sigma}_{\mathrm{\phi}},\bar{T}),\qquad 
	\frac{d\bar{\Sigma}_{\mathrm{\phi}}}{d\tau} = -\nu X_{\mathrm{0}}^{2p}\bar{\Sigma}_{\mathrm{\phi}}+G(\bar{X},\bar{\Sigma}_{\mathrm{\phi}},\bar{T}),\qquad
	\frac{d\bar{T}}{d\tau} =  N(\bar{X},\bar{\Sigma}_{\mathrm{\phi}},\bar{T}),
	\end{equation}
	where $F$, $G$ and $N$ are functions of higher-order. The center manifold reduction theorem yields that the above system is locally topological equivalent to a decoupled system on the 2-dimensional center manifold, which can be locally represented as the graph $h:\, E^{c}\rightarrow E^{s}$, i.e. $\bar{\Sigma}_{\mathrm{\phi}}=h(\bar{X},\bar{T})$, which solves the nonlinear partial differential equation
	\begin{equation}\label{FlowCMLine2}
	F(\bar{X},h(\bar{X},\bar{T}),\bar{T})\partial_{\bar{X}}  h(\bar{X},\bar{T})+N(\bar{X},h(\bar{X},\bar{T}),\bar{T}) \partial_{\bar{T}} h(\bar{X},\bar{T})= -\nu X_0^{2p} h(\bar{X},\bar{T}) + G(\bar{X},h(\bar{X},\bar{T}),\bar{T})
	\end{equation}
	subject to the fixed point and tangency conditions $h(0,0)=0$ and $\nabla h(0,0)=0$, respectively. A quick look at the nonlinear terms suggests that we approximate the center manifold at $(\bar{X},\bar{T})=(0,0)$, by making a formal multi-power series expansion for $h$ of the form
	\begin{equation}
	h(\bar{X},\bar{T})=\bar{T}^{2p-n+1}\sum^{N}_{i,j=0} \tilde{a}_{ij}\bar{X}^{i}\bar{T}^{j}, \qquad \tilde{a}_{ij}\in\mathbb{R}.
	\end{equation}
	Solving for the coefficients of the expansion one sees that all coefficients of type $\tilde{a}_{i0}$ are identically zero, so that $h$ can be written as a series expansion in $\bar{T}$ with coefficients depending on $\bar{X}$, i.e. 	
	\begin{equation}
	h(\bar{X},\bar{T})=\bar{T}^{2p-n+1}\sum^{N}_{j=1} \bar{a}_{j}(\bar{X})\bar{T}^{j}, \qquad \bar{a}_j(X)=\sum^{N}_{i=0}a_{ij}\bar{X}^{i}, \qquad a_{ij}\in\mathbb{R},
	\end{equation}
	where for example
	\begin{subequations}
		\begin{align}
		a_{01} &= 0,\qquad a_{11}=0, \qquad	a_{02} =-\frac{n}{\nu}X_0^{2(p-n)-1},\qquad
		a_{12}=-\frac{n(2p+1)}{\nu}X_0^{2(p+1)},  \nonumber \\
		a_{03} &=\frac{1}{2\nu^2}\left(6n+3(2p+1)\gamma_{\mathrm{pf}}\nu (1-X_0^{2n})X_0^{n}\right)\delta^{n-2p}_1.
		\end{align}
	\end{subequations}
	After a change of time $d/d\tau=\bar{T}^{2p-n}d/d\bar{\tau}$, the flow on the 2-dimensional center manifold is given by 
	\begin{subequations}
		\begin{align}
		\frac{d\bar{X}}{d\bar{\tau}} &= \sum^{N}_{j=1} \bar{b}_{j}(\bar{X})\bar{T}^{j}, \qquad \bar{b}_j(\bar{X})=\sum^{N}_{i=0}b_{ij}\bar{X}^{i},\qquad b_{ij}\in\mathbb{R}, \\
		\frac{d\bar{T}}{d\bar{\tau}} &=\bar{T}\sum^{N}_{j=1} \bar{c}_{j}(\bar{X})\bar{T}^{j}\qquad \bar{c}_j(\bar{X})=\sum^{N}_{i=0}c_{ij}\bar{X}^{i}, \qquad c_{ij}\in\mathbb{R},
		\end{align}
	\end{subequations}
	with
	\begin{subequations}
		\begin{align*}
		b_{01} &=\frac{3}{2}(1-X_0^{2n})X_0, \quad 	
		b_{11} =\frac{3\gamma_{\mathrm{pf}}}{2n}(1-(1+2n)X_0^{2n}), \quad 
		b_{21} =-\frac{3\gamma_{\mathrm{pf}}}{2}(2n+1)X_0^{2n}, \\
		b_{02} &=-\frac{n}{\nu}X_0^{2(p-n)}, \quad 
		b_{12} =-3\gamma_{\mathrm{pf}}X_0^{2n},\\
		c_{01}&=\frac{3\gamma_{\mathrm{pf}}}{2n}(1-X_0^{2n}),\qquad c_{11}=0, \\
		c_{02}&=-\frac{9\gamma_{\mathrm{pf}}^2}{2n}X_0^{2n}(1-X_0^{2n})\delta^{n-2p}_1,\qquad c_{12}= -3\gamma_{\mathrm{pf}}X_0^{2n}.
		\end{align*}
	\end{subequations}
	For $X\neq0$, the coefficient $b_{01}$ only vanishes at $X_0=\pm1$, being negative for $X_0 \in(-1,0)$ and positive for $X_0\in (0,1)$. When $b_{01}\neq0$ the origin $(0,0)$ is a nilpotent singularity. Since the coefficient $c_{01}(\bar{X})\neq 0$ for all $X_0$, then the normal formal form is zero with
	\begin{equation}
	\frac{d\bar{X}_{*}}{d\bar{\tau}_{*}}=\text{sign}(b_{01})\bar{T}_{*},\qquad \frac{d\bar{T}_{*}}{d\bar{\tau}_{*}}=\bar{T}^2_{*}\Phi(\bar{X}_{*},\bar{T}_{*}),
	\end{equation}
	and $\Phi$ an analytic function. The phase-space is the flow-box multiplied by the functional $\bar{T}_{*}$, with direction given by the sign of $b_{01}$, see Figure \ref{center1T0}. When $X_0=\pm 1$ we have that $b_{11}=-3\gamma_{\mathrm{pf}}<0$, $b_{02}=0$, $c_{01}=0$, $c_{02}=0$ and $c_{12}<-3\gamma_{\mathrm{pf}}$. After changing the time variable to $d/d\tilde{\tau}=T^{-1}d/d\bar{\tau}$, we obtain
	\begin{subequations}
		\begin{align}
		\frac{d\bar{X}}{d\tilde{\tau}}&=-3\gamma_{\mathrm{pf}}\bar{X}-\frac{n}{\nu}\bar{T}-\frac{3\gamma_{\mathrm{pf}}}{2}(2n+1)\bar{X}^2-3\gamma_{\mathrm{pf}}\bar{X} \bar{T}+\mathcal{O}\left(||(\bar{X},\bar{T})||^3\right), \\
		\frac{d\bar{T}}{d\tilde{\tau}}&=-3\gamma_{\mathrm{pf}}\bar{X}\bar{T}+\mathcal{O}\left(||(\bar{X},\bar{T})||^3\right),
		\end{align}
	\end{subequations}
	so that the origin is a semi-hyperbolic fixed point with eigenvalues $-3\gamma_{\mathrm{pf}}$, $0$, and associated eigenvectors $(1,0)$, and $(-\frac{n}{3\gamma_{\mathrm{pf}}\nu},1)$. To analyse the center manifold we introduce the adapted variable $\tilde{X}=\bar{X}+\frac{n}{3\gamma_\mathrm{pf}\nu}\bar{T}$. The $1$-dimensional center manifold $W^{c}$ at $(0,0)$ can then be locally represented as the graph $h: E^{c}\rightarrow E^{s}$, i.e. $\tilde{X}=h(\bar{T})$, satisfying the fixed point $h(0)=0$ and tangency $\frac{dh(0)}{d\bar{T}}=0$ conditions, using $\bar{T}$ as an independent variable. Approximating the solution by a formal truncated power series expansion and solving for the coefficients  yields to leading order on the center manifold
	\begin{equation}
	\frac{d\bar{T}}{d\tilde{\tau}}=\frac{n}{\nu}\bar{T}^2+\mathcal{O}(\bar{T}^3),\quad \text{as}\quad \bar{T}\rightarrow 0.
	\end{equation}
	So for $X_0=\pm 1$, the origin is the $\alpha$-limit set of a single orbit, the inflationary attractor solution, see Figure~\ref{center2T0}. Therefore on the set $\mathrm{L}_1\setminus\mathrm{FL}_0$ only the fixed points $\mathrm{dS}^{\pm}_0$ with $X_0=\pm1$ are $\alpha$-limit sets for Class A interior orbits in $\mathbf{S}$, being unique center manifolds. 
	
	The conclusions about the non-hyperbolic fixed point $\mathrm{FL}_0$ can be found in Section~\ref{ApB}, where the blow-up of $\mathrm{FL}_0$ is done as well as the blow-up of the line $\mathrm{L}_2$. In particular, Lemma~\ref{L2BU-2} in Section~\ref{FPInfty}, states that no interior orbit in $\mathbf{S}$ converges to the set $\mathrm{L}_2\setminus \{\mathrm{K}^\pm\cup\mathrm{FL}_0\}$ and that the fixed points $\mathrm{K}^{\pm}$ are sources. 
\end{proof}
\begin{figure}[ht!]
	\begin{center}
		\subfigure[$b_{01}>0$. For $b_{01}<0$ the direction of the flow is reversed.]{\label{center1T0}
			\includegraphics[width=0.30\textwidth]{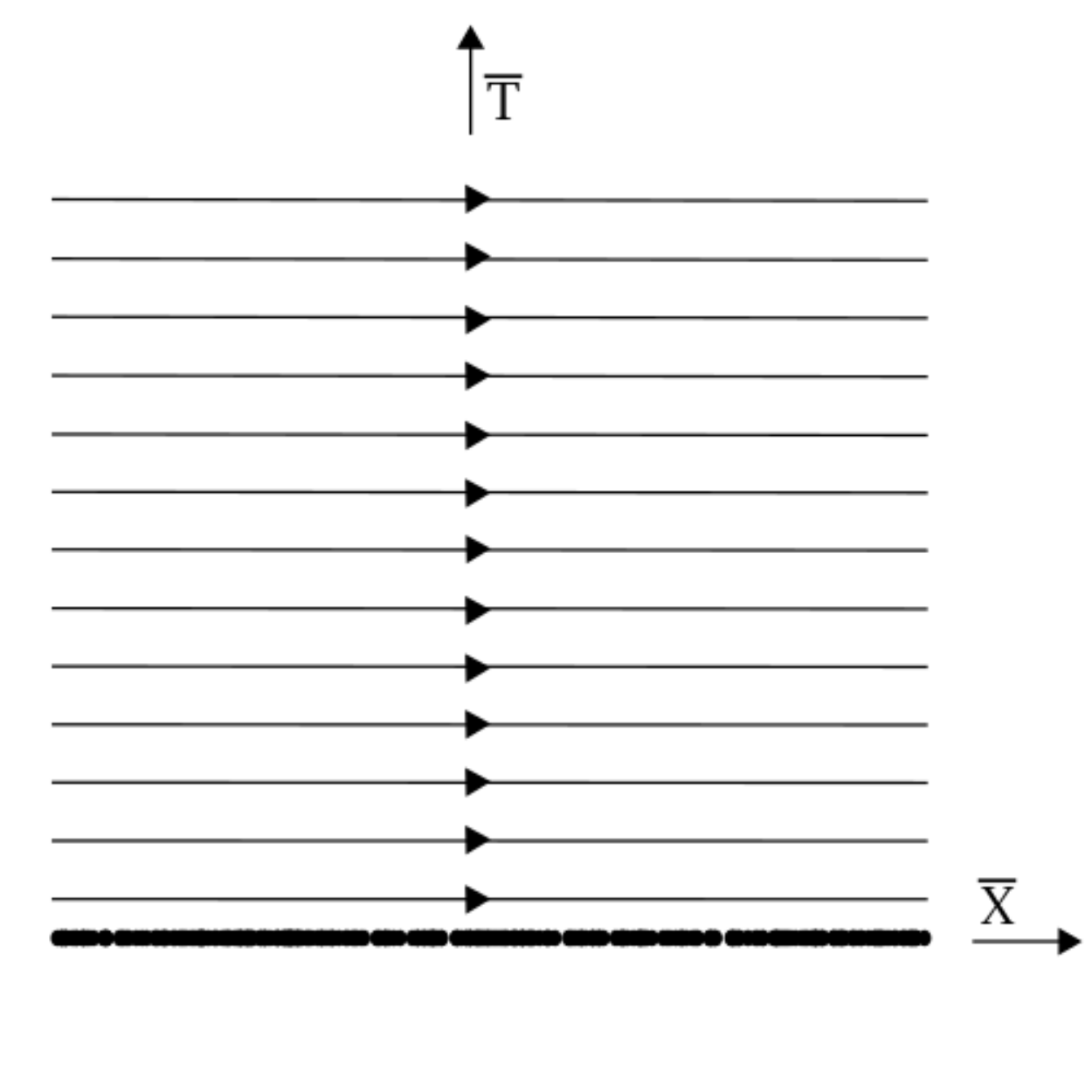}}
		\hspace{2cm}
		\subfigure[$X_0=\pm 1$.]{\label{center2T0}
			\includegraphics[width=0.30\textwidth]{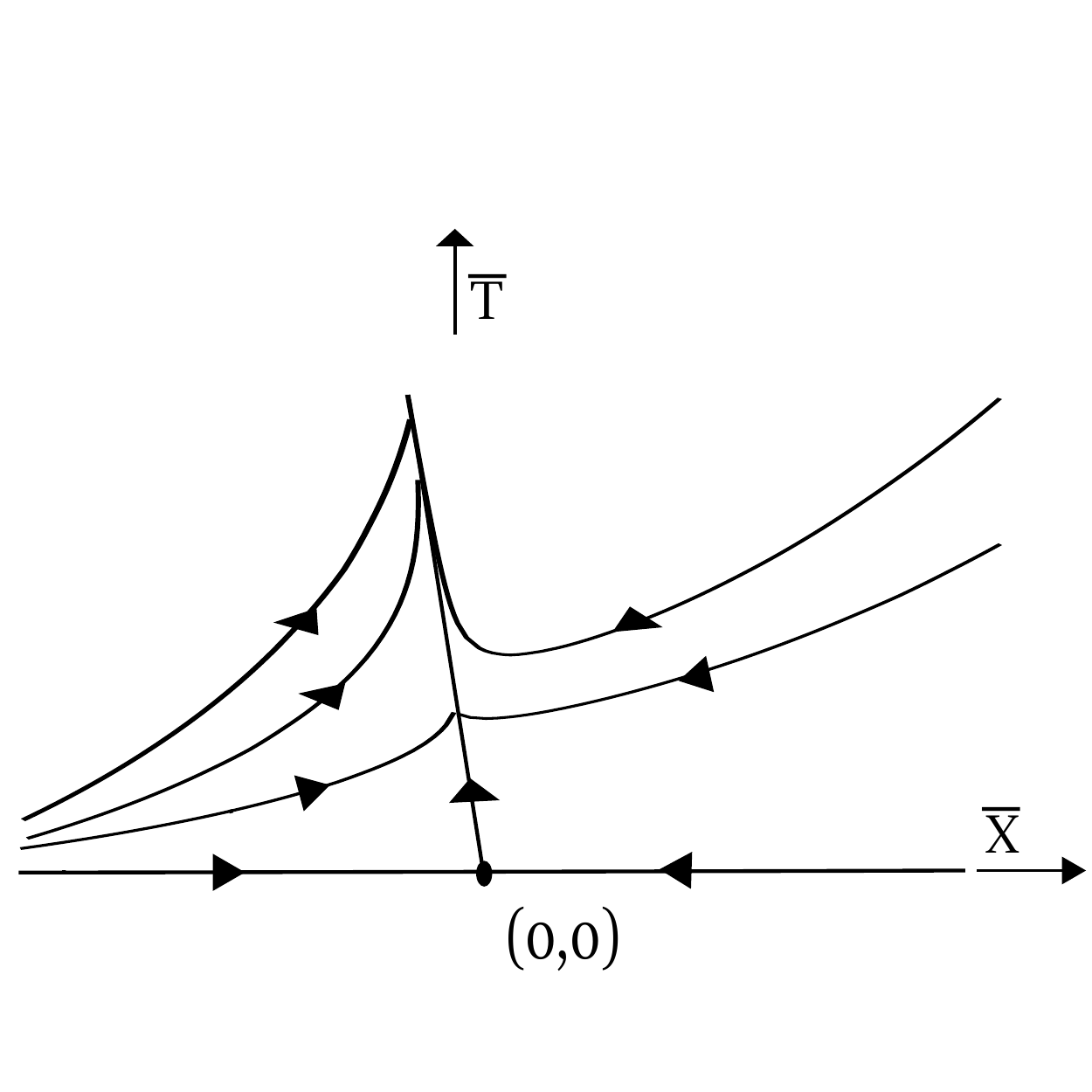}}
	\end{center}
	\vspace{-0.5cm}
	\caption{Flow on the 2-dimensional center manifold of each point on $\mathrm{L}_1\setminus\mathrm{FL}_0$.}
	\label{fig:2D_CM_L1_T0}
\end{figure}
\begin{remark}\label{PastAs3}
	When $p>n/2$, the asymptotics for the inflationary attractor solution originating from $\mathrm{dS}^\pm_0$ are given by 
	\begin{equation}
		 H\sim (-t)^{\frac{n}{2(2p-n)}}\quad,\quad \phi\sim (-t)^{\frac{1}{2(2p-n)}}\quad,\quad \rho_\mathrm{pf} \sim (-t)^{-\frac{n}{2p-n}}, \quad \quad\text{as}\quad t\rightarrow -\infty.
	\end{equation}
\end{remark}
%
%
\subsection{Blow-up of $\mathrm{FL}_0$}
\label{ApB}
To analyse the non-hyperbolic fixed point $\mathrm{FL}_0$ we use the unbounded dynamical system~\eqref{sistemalocal}, which for $p>\frac{n}{2}$ gives
\begin{subequations}
	\begin{align}
	\frac{dX}{d\tilde{N}} &= \frac{1}{n}(1+q)\tilde{T}^{2p-n}X+\tilde{T}^{2p-n+1}\Sigma_\mathrm{\phi}\\
	\frac{d\Sigma_\mathrm{\phi}}{d\tilde{N}} &= -\left[(2-q)\tilde{T}^{2p-n}+\nu X^{2p}\right]\Sigma_\mathrm{\phi}-n \tilde{T}^{2p-n+1}X^{2n-1}\\
	\frac{d\tilde{T}}{d\tilde{N}} &= \frac{1}{n}(1+q)\tilde{T}^{1+2p-n}
	\end{align}
\end{subequations}
where we recall
\begin{equation}
q=-1+3\Sigma_\mathrm{\phi}^2+\frac{3\gamma_\mathrm{pf}}{2}\left(1-X^{2n}-\Sigma_\mathrm{\phi}^2\right).
\end{equation}
In order to understand the dynamics near the origin $(X,\Sigma_\phi,\tilde{T})=(0,0,0)$, which is a non-hyperbolic fixed point we employ the spherical blow-up method. I.e. we transform the fixed point at the origin to the unit 2-sphere $\mathbb{S}^2=\{(x,y,z)\in\mathbb{R}^3:x^2+y^2+z^2=1\}$ and define the blow-up space manifold as $\mathcal{B}:=\mathbb{S}^2\times[0,u_0]$ for some fixed $u_0>0$. We further define the quasi-homogeneous blow-up map
\begin{equation}\label{BUPMAP2}
\Psi\,: \mathcal{B}\rightarrow \mathbb{R}^3,\qquad \Psi(x,y,z,u)=(u^{2p-n}x, u^{n}y, u^{2p}z),
\end{equation}
which after cancelling a common factor $u^{2p(2p-n)}$ (i.e. by changing the time variable $d/d\tilde{N}=u^{2p(2p-n)}d/d\bar{\tau}$) leads to a desingularisation of the non-hyperbolic fixed point on the blow-up locus $\{u=0\}$. Just as in the blow-up of the fixed point $\mathrm{FL}_1$ studied in section~\ref{ApA}, one can simplify the computations if, instead of standard spherical coordinates on $\mathcal{B}$, one uses different local charts $\kappa_i:\mathcal{B}\rightarrow\mathbb{R}^3$ such that $\psi_i:\Psi\circ \kappa^{-1}_i$. We then choose six charts $\kappa_i$ such that
\begin{subequations}
	\begin{align}
	\psi_{1\pm}&=(\pm u_{1\pm}^{2p-n}, u_{1\pm}^{n}y_{1\pm},u_{1\pm}^{2p}z_{1\pm}),\\
	\psi_{2\pm} &= (u_{2\pm}^{2p-n}x_{2\pm},\pm u_{2\pm}^{n},u_{2\pm}^{2p}z_{2\pm}),\\
	\psi_{3\pm} &= (u_{3\pm}^{2p-n}x_{3\pm},u_{3\pm}^n y_{3\pm},\pm u_{3\pm}^{2p}),
	\end{align}
\end{subequations}
where $\psi_{1\pm}$, $\psi_{2\pm}$ and $\psi_{3\pm}$ are called the directional blows ups in the positive/negative $x$, $y$, and $z$-directions respectively.
It is easy to check that the different charts are given explicitly by
\begin{subequations}
	\begin{align}
	\kappa_{1\pm}&:\quad(u_{1\pm},y_{1\pm},z_{1\pm})=(\pm u x^{\frac{1}{2p-n}},\pm y x^{-n},\pm z x^{-2p})\\
	\kappa_{2\pm}&:\quad(x_{2\pm},u_{2\pm},z_{2\pm})=(\pm x y^{-\frac{2p-n}{n}}, \pm u y^{\frac{1}{n}},\pm z y^{-\frac{2p}{n}})\\
	\kappa_{3\pm}&:\quad (x_{3\pm},u_{3\pm},z_{3\pm})=(\pm x z^{-\frac{2p-n}{2p}},\pm y z^{-\frac{n}{2p}},\pm u z^{\frac{1}{2p}}).
	\end{align}
\end{subequations}
The transition maps $k_{ij}=\kappa_j \circ\kappa^{-1}_i$ will allow to identify special invariant manifolds and fixed points on different local charts, and to deduce all dynamics on the blow-up space. In this case we will need some of the following transition charts:
\begin{subequations}
	\begin{align}
	\kappa_{1+2+}\quad &:\quad(x_{2+},u_{2+},z_{2+})=(y_{1+}^{-\frac{2p-n}{n}},u_{1+}y_{1+}^{\frac{1}{n}},y_{1+}^{-\frac{2p}{n}}z_{1+}), \quad y_{1+}>0;\\
	\kappa_{2+1+}\quad &: \quad(u_{1+},y_{1+},z_{1+})=(u_{2+}x_{2+}^{\frac{1}{2p-n}},x_{2+}^{-n}, z_{2+} x_{2+}^{-2p}),\quad x_{2+}>0;
	\end{align}
\end{subequations}
\begin{subequations}
	\begin{align}
	\kappa_{1+3+}\quad &: \quad (x_{3+},y_{3+},u_{3+})=(z_{1+}^{-\frac{2p-n}{2p}},y_{1+}z_{1+}^{-\frac{n}{2p}},u_{1+}z_{1+}^{\frac{1}{2p}}),\quad 	z_{1+}>0;\\
	\kappa_{3+1+}\quad &: \quad (u_{1+},y_{1+},z_{1+})=(u_{3+}x_{3+}^{\frac{1}{2p-n}},y_{3+}, y_{3+}x_{3+}^{-n},x_{3+}^{-2p}),\quad x_{3+}>0;
	\end{align}
\end{subequations}
\begin{subequations}
	\begin{align}
	\kappa_{2+3+}\quad &: \quad (x_{3+},y_{3+},u_{3+})=(x_{2+} z_{2+}^{-\frac{2p-n}{2p}},z_{2+}^{-\frac{n}{2p}}, u_{2+}z_{2+}^{\frac{1}{2p}}),\quad z_{2+}>0;\\
	\kappa_{3+2+}\quad &:\quad (x_{2+},u_{2+},z_{2+})=(x_{3+}y_{3+}^{-\frac{2p-n}{n}},u_{3+}y_{3+}^{\frac{1}{n}},y_{3+}^{-\frac{2p}{n}}),\quad y_{3+}>0.
	\end{align}
\end{subequations}
Similarly to the blow-up of $\mathrm{FL}_1$, we are only interested in the region $\{z\geq0\}$, i.e. the union of the upper hemisphere of the unit sphere $\mathbb{S}^2$ and the  equator of the sphere $\{z=0\}$ which constitutes an invariant boundary subset. This motivates that we start the analysis by using chart $\kappa_{3+}$, i.e. the directional blow-up map in the positive $z$-direction, on which the the northern hemisphere is mapped into the plane $z=1$ and the equator of the sphere is at infinity, which is better analysed using the charts $\kappa_{1+}$ and $\kappa_{2+}$. Figure~\ref{fig:BUP2p} shows the blow-up space of $\mathrm{FL}_0$ when $p>\frac{n}{2}$.
\begin{figure}[ht!]
	\begin{center}
		\includegraphics[width=0.44\textwidth]{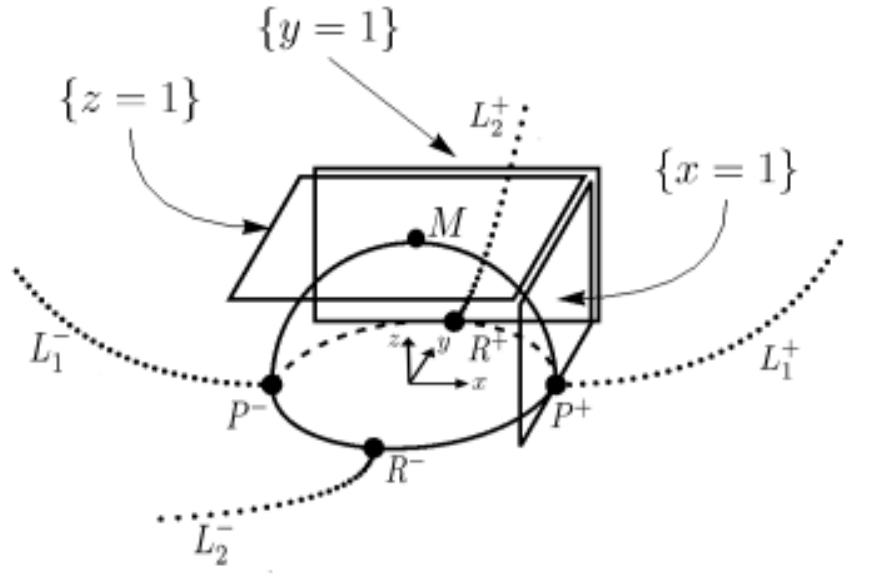}
	\end{center}
	\vspace{-0.5cm}
	\caption{Blow-up space $\mathcal{B}$ for $p>n/2$.}
	\label{fig:BUP2p}
\end{figure}
Later, instead of projecting the upper-half of the unit $2$-sphere on the $z=1$ plane, we shall project it into the open unit disk  $x^2+y^2< 1$ which can be joined with the equator (unit circle on $\{z=0\}$), thus obtaining a global understanding of the flow on the Poincar\'e-Lyapunov unit disk.
%
\subsubsection{Positive $z$-direction}
We start with the positive $z$-direction $\{z=1\}$ which after cancelling a common factor $u_3^{2p(2p-n)}$ (i.e. by changing the time variable $d/d\tilde{N}=u_3^{2p(2p-n)}d/d\bar{\tau}_{3}$) leads to the regular dynamical system
\begin{subequations}
\label{235}
	\begin{align}
	\frac{dx_3}{d\bar{\tau}_3}&= \frac{1}{2p}(1+q)x_3+u_3^{2n}y_3, \\
	\frac{dy_3}{d\bar{\tau}}&= -\left((2-q)+\frac{1+q}{2p}+\nu x_3^{2p}\right)y_3- n u_3^{2n(2p-n)}x_3^{2n-1}, \\
	\frac{du_3}{d\bar{\tau}_3}&=\frac{1}{2np}(1+q)u_3,
	\end{align}
\end{subequations}
where 
\begin{equation}
q=-1+\frac{3\gamma_\mathrm{pf}}{2}\left(1+\frac{2-\gamma_\mathrm{pf}}{\gamma_\mathrm{pf}}y_3^{2}u_3^{2n}-x_3^{2n}u_3^{2n(2p-n)}\right).
\end{equation}
All fixed points are located at the invariant subset $\{u_3=0\}$, where the induced flow is given by
\begin{equation}
\frac{dx_3}{d\bar{\tau}_3}=\frac{3}{(2p-1)(1-\tilde{K})}x_3,\qquad \frac{dy_3}{d\bar{\tau}_3}=\left(\frac{3\tilde{K}}{1-\tilde{K}}- \nu x_3^{2p}\right)y_3
\end{equation}
and where we introduced the notation
\begin{equation}
\tilde{K}=1-\frac{4p}{(2p-1)\gamma_\mathrm{pf}}<0.
\end{equation}
The system has only one fixed point at the origin
\begin{equation}
\mathrm{M}:\quad x_3=0,\quad y_3=0,
\end{equation}
whose linearisation yields the eigenvalues $\lambda_1=\frac{3}{(2p-1)(1-\tilde{K})}$, $\lambda_2=\frac{3\tilde{K}}{(1-\tilde{K})}$ and $\lambda_3=\frac{3}{n(2p-1)(1-\tilde{K})}$, with associated eigenvectors the canonical basis of $\mathbb{R}^3$. Hence, on $\{u_3=0\}$, $\mathrm{M}$ is a hyperbolic saddle.
\subsubsection{Fixed points at infinity}\label{FPInfty}
To study the points at infinity, we notice that both directional blow-ups in the positive/negative $x$ and $y$ directions already tell how such local chart must be given. Moreover, since the system~\eqref{235} is invariant under the transformation $(x_3,y_3)\rightarrow-(x_3,y_3)$, it suffices to consider the positive $x$ and $y$ directions. The analysis in the remaining directions is then easily inferred from symmetry considerations.

To study the region where $x_3$ blows up, we use the chart
\begin{equation}
\left(y_1,z_1,u_1\right)=\left(y_3x_3^{-n},x_3^{-2p},u_3 x_3^{\frac{1}{2p-n}}\right)
\end{equation}
and change the time variable to $d/d\bar{\tau}_1=z_1d/d\bar{\tau}_3$, i.e. $d/d\tilde{N}=u^{2p(2p-n)}_1d/d\bar{\tau}_1$, which leads to the regular system of equations
\begin{subequations}
	\begin{align}
	\frac{dy_1}{d\bar{\tau}_1}&= -\left(\left(2-q+\frac{1+q}{2p-n}\right)z_1^{2p-n}+\nu\right)y_1-\left(n u_1^{2n(2p-n+1)}-\frac{n^2}{2p-n}y_1^2\right)z_1^{2p-n+1}u_1^{2n} , \\
	\frac{dz_1}{d\bar{\tau}_1}&=-\frac{1+q}{2p-n}z_1^{2p-n+1}-\frac{2p}{2p-n}y_1 z_1^{2p-n+2}u_1^{2n} , \\
	\frac{du_1}{d\bar{\tau}_1}&=\frac{1}{2p-n}\left(\frac{1+q}{n} z_1^{2p-n}+y_1 z_1^{2p-n+1} u_1^{2n}\right)u_1.
	\end{align}
\end{subequations}
The flow in the invariant subset $\{u_1=0\}$ is given by
\begin{equation}
\frac{dy_1}{d\bar{\tau}_1}=\left(\frac{3}{2}\left(2-\gamma_{\mathrm{pf}}+\frac{\gamma_{\mathrm{pf}}}{2p-n}\right)z_{1}^{2p-n}-\nu\right)y_{1},\qquad \frac{dz_1}{d\bar{\tau}_1}=-\frac{3\gamma_\mathrm{pf}}{2(2p-n)}z_1^{2p-n+1},
\end{equation}
which has a single fixed point
\begin{equation}
\mathrm{P}^{+}\quad:\quad y_1=0,\quad z_1=0,\quad u_1=0,
\end{equation}
whose linearisation yields the eigenvalues $\lambda_1=-\nu$, $\lambda_2=0$ and $\lambda_3=0$, with associated eigenvectors $v_1=(1,0,0)$, $v_2=(0,1,0)$ and $v_3=(0,0,1)$. The zero eigenvalue in the $u_{1}$-direction is associated with a line of fixed points, parameterized by constant values of $u_{1}=u_{0}>0$, which corresponds to the half of the line of fixed points $\mathrm{L}_1$ with $X_0>0$ and denoted by $\mathrm{L}^{+}_1$, see Figure~\ref{fig:BUP2p}. Thus, on the  $\{u_1=0\}$ invariant set, the fixed point $\mathrm{P}^{+}$ is semi-hyperbolic, with the center manifold being the invariant subset $\{y_1=0\}$, where the flow is simply
\begin{equation}
\frac{dz_{1}}{d\bar{\tau}_{1}}=-\frac{3\gamma_{\mathrm {pf}}}{2(2p-n)}{z}^{2p-n+1}_1,\qquad\text{as}\qquad {z}_1\rightarrow 0. 
\end{equation}
It follows that on $\{u_1=0\}$, the fixed point $\mathrm{P}^+$ is the $\omega$-limit set of a 1-parameter family of orbits, which converge to $\mathrm{P}^+$ tangentially to the $y_1=0$ axis. From the symmetry of the $(x_3,y_3)$ plane it follows that the blow-up on the negative $x$-direction yields an equivalent fixed point $\mathrm{P}^{-}$.

To study the region where $y_3$ blows up, we use the chart
\begin{equation}
\left(x_2,z_2,u_2\right)=\left(x_3y_3^{\frac{n-2p}{n}},y_3^{\frac{2p}{n}},u_3 y_3^{\frac{1}{n}}\right),
\end{equation}
together with the change of time variable $d/\bar{\tau}_2=z_2d/d\bar{\tau}_3$, i.e. $d/d\tilde{N}=u^{2p(2p-n)}_2d/\bar{\tau}_2$, which leads to the regular dynamical system
\begin{subequations}
	\begin{align}
	\frac{dx_2}{d\bar{\tau}_2}&=\nu\frac{2p-n}{n}x_{2}^{2p+1}+\frac{1}{n}\left(\left(1+q+(2p-n)(2-q)\right)x_{2}+ n u_{2}^{2n}z_{2}+n(2p-n)u_{2}^{2n(2p-n)}x_{2}^{2n}z_{2}\right)z_{2}^{2p-n},\\
	\frac{dz_2}{d\bar{\tau}_2}&= \frac{1}{n}\left((1+q+2p(2-q))z_{2}^{2p-n}+2p \nu x_{2}^{2p}+2pnu_{2}^{2n(2p-n)}x_2^{2n-1}z_{2}^{2p-n+1}\right)z_{2},\\
	\frac{du_2}{d\bar{\tau}_2}&=-\frac{1}{n}\left((2-q)z_{2}^{2p-n}+n u_{2}^{2n(2p-n)}x_{2}^{2n-1}z_{2}^{2p-n+1}+\nu x_{2}^{2p}\right)u_{2}.
	\end{align}
\end{subequations}
The induced flow on the invariant subset $\{u_2=0\}$ is given by
\begin{subequations}
	\begin{align*}
	\frac{dx_2}{d\bar{\tau}_2}&= \frac{1}{2n}\left(3\left(\gamma_\mathrm{pf}+(2p-n)(2-\gamma_\mathrm{pf})\right)z_{2}^{2p-n}+2\nu(2p-n)x_2^{2p}\right)x_{2} ,\\
	\frac{dz_{2}}{d\bar{\tau}_2}&= \frac{1}{2n}\left(3\left(\gamma_\mathrm{pf}+2p(2-\gamma_\mathrm{pf})\right)z^{2p-n}_2+4p \nu x_2^{2p}\right)z_{2},
	\end{align*}
\end{subequations}
which has only one fixed point 
\begin{equation}
\mathrm{R}^{+}:\quad x_2=0,\quad z_2=0,\quad u_2=0.
\end{equation}
The linearised system at $\mathrm{R}^+$ has all eigenvalues equal to zero. One of these zero eigenvalues is due to the line of fixed points $\mathrm{L}_{2}^{+}$ in the $u_2$ direction, which corresponds to the half of the line $\mathrm{L}_2$ with $\Sigma_{\phi0}>0$, see Figures~\ref{fig:BUP2p} and~\ref{fig:T0L2}. 

To blow-up the non-isolated set of fixed points $\mathrm{L}_{2}^{+}\cup\mathrm{R}^{+}$, we perform a cylindrical blow-up, i.e.  we transform each point on this set to a circle $\mathbb{S}^{1}=\{(v,w)\in \mathbf{R}^2:v^2+w^2=1\}$. The blow-up space is $\bar{\mathcal{B}}=\mathbb{S}^1\times[0,u_{20})\times[0,s_0)$, and we further define the quasi-homogeneous blow-up map%
\begin{equation*}
\bar{\Psi}:\bar{\mathcal{B}}\rightarrow\mathbb{R}^3,\quad \bar{\Psi}(v,w,u_2,s)=(s^{2p-n}v,s^{2p}w,u_2).
\end{equation*}
We choose four charts such that
\begin{subequations}
	\begin{align}
	\bar{\psi}_{1\pm}&=\left(\pm s_{1\pm}^{2p-n},s_{1\pm}^{2p}w_{1\pm},u_{2}\right),\\
	\bar{\psi}_{2\pm} &= \left(s_{2\pm}^{2p-n}v_{2\pm},\pm s_{2\pm}^{2p},u_{2}\right),
	\end{align}
\end{subequations}
and recall that only the half circle with $w\geq0$ is of interest since $z_2\geq0$ and, therefore, only the blow-up map in the positive $w$-direction $\bar{\psi}_{2+}$ needs to be considered. 

We start with the $v_1$-direction $\{v_1=\pm 1\}$ which, after cancelling the common factor $s_{1\pm}^{2p(2p-n)}$ (i.e. by changing the time variable $d/d\bar{\tau}_2=s_{1\pm}^{2p(2p-n)}d/d\tilde{\tau}_{1\pm}$), leads to the regular dynamical system
\begin{subequations}
	\begin{align}
	\frac{ds_{1_\pm}}{d\tilde{\tau}_{1\pm}}=& \frac{s_{1\pm}}{n(2p-n)}\Big(\left(1+q+(2p-n)(2-q)\right)w_{1\pm}^{2p-n}\\
	\pm& \left(n+(2p-n)s_{1\pm}^{2(2p-n)}u_{2}^{2n(2p-n-1)}\right)s_{1\pm}^{n}u_{2}^{2n}w_{1\pm}^{2p-n+1}+\nu(2p-n)\Big),\nonumber\\
	\frac{dw_{1\pm}}{d\tilde{\tau}_{1\pm}}=&-\frac{w_{1\pm}^{2p-n+1}}{2p-n}\left(1+q\pm 2p w_{1\pm}s_{1\pm}^{n}u_{2\pm}^{2n} \right),\\
	\frac{d{u}_{2}}{d\tilde{\tau}_{1\pm}}=&-\frac{1}{n}\left((2-q)w_{1\pm}^{2p-n}\pm n s_{1\pm}^{n(2(2p-n)+1)}u_{2}^{2n(2p-n)}w_{1\pm}^{2p-n+1}+\nu\right)u_2,
	\end{align}
\end{subequations}
where
\begin{equation}
q=-1+\frac{3\gamma_\mathrm{pf}}{2}\left(1+\frac{2-\gamma_\mathrm{pf}}{\gamma_\mathrm{pf}}u_{2}^{2n}-u_{2}^{2n(2p-n)}s_{1\pm}^{2n(2p-n)}\right).
\end{equation}
In the physical region $w\geq0$, the above two systems have only one fixed point each, given by
\begin{equation}
\mathrm{T}^{+}_{\pm}:\quad s_{1\pm}=0, \quad w_{1\pm}=0, \quad   u_{2}=0,
\end{equation}
and whose linearisation gives the eigenvalues $\frac{\nu}{n}$, $0$ and $-\frac{\nu}{n}$, with associated eigenvectors the canonical basis of $\mathbb{R}^3$. 
The fixed points $\mathrm{T}^{+}_{\pm}$ have the center manifold $s_{1\pm}=0$, i.e. the $w_{1\pm}$-axis. On the center manifold, the flow is just
\begin{equation}
\frac{dw_{1\pm}}{d\tilde{\tau}_{1\pm}}=-\frac{3\gamma_\mathrm{pf}}{2(2p-n)}w_{1\pm}^{2p-n+1},\quad \text{as}\quad w_{1\pm}\rightarrow 0,
\end{equation}
so that $\mathrm{T}^{+}_{\pm}$ are center-saddles.

In the positive $w$-direction and after cancelling the common factor $s_{2+}^{2p(2p-n)}$ (i.e. by changing the time variable $d/d\bar{\tau}_2=s_{2+}^{2p(2p-n)}d/d\tilde{\tau}_{2+}$), leads to the regular dynamical system
\begin{subequations}
	\begin{align}
	\frac{dv_{2+}}{d\tilde{\tau}_{2+}}=& \frac{1+q}{2p}v_{2+}+s_{2+}^{n}u_{2}^{2n}, \\
	\frac{ds_{2+}}{d\tilde{\tau}_{2+}}=&\frac{1}{2np}\left(1+q+2p(2-q)+2p\nu v_{2+}^{2p}\right)s_{2+}+s_{2+}^{2n(2p-n)+n+1}u_{2}^{2n(2p-n)}, \\
	\frac{du_{2}}{d\tilde{\tau}_{2+}}=& -\frac{1}{n}\left(2-q+n s_{2+}^{n(2(2p-n)+1)}u_{2}^{2n(2p-n)v_{2+}^{2n-1}}-\nu v_{2+}^{2p}\right)u_{2}, 
	\end{align}
\end{subequations}
where
\begin{equation}
	q=-1+\frac{3\gamma_\mathrm{pf}}{2}\left(1+\frac{2-\gamma_\mathrm{pf}}{\gamma_\mathrm{pf}}u_{2}^{2n}-s_{2+}^{2n(2p-n)}u_{2}^{2n(2p-n)}v_{2+}^{2n}\right).
\end{equation}
The above system has only the fixed point
\begin{equation}
\mathrm{Q}^{+}:\quad v_{2+}=0,\quad s_{2+}=0,\quad u_{2}=0,
\end{equation}
whose linearisation gives the eigenvalues $\frac{3\gamma_\mathrm{pf}}{4p}$, $\frac{3(2p(2-\gamma_\mathrm{pf})+\gamma_\mathrm{pf})}{4np}$ and $-\frac{3(2-\gamma_\mathrm{pf})}{2n}$, with associated eigenvectors the canonical basis of $\mathbb{R}^3$. So, on the $\{u_2=0\}$ subset, $\mathrm{Q}^{+}$  is a hyperbolic source. 

%
Lastly in the positive $w$-direction we have one more fixed point
\begin{equation}
\mathrm{K}^{+}:\quad v_{2+}=0, \quad s_{2+}=0,\quad u_{2}=1.
\end{equation}
The eigenvalues of the linearised system around $\mathrm{K}^{+}$ are $\frac{3}{2p}$, $\frac{3}{2 n p}$ and $3(2-\gamma_{\mathrm{pf} })$, with  associated eigenvectors the canonical basis of $\mathbb{R}^3$. Since $0<\gamma_{\mathrm{pf}}<2$, all eigenvalues are real and positive so that $\mathrm{K}^{+}$ is a hyperbolic source. The blow-up of $\mathrm{L}^+_2\cup\mathrm{R}^{+}$ is shown in Figure~\ref{fig:T0L2}. Due to the symmetry of the system in the $(x_3,y_3)$ plane, the blow-up of the equivalent non-hyperbolic set $\mathrm{L}^{-}_{2}\cup\mathrm{R}^{-}$ follows by symmetry considerations yielding identical results, and therefore equivalent fixed points $\mathrm{Q}^{-}_{\pm}$ and $\mathrm{N}^{-}$. Hence we have the following result:
\begin{lemma}\label{L2BU-2}
	No interior class A orbit in $\mathbf{S}$ converges, for $\tau\rightarrow-\infty$, to the points on the set $\mathrm{L}_2\setminus\{{\mathrm{FL}}_0\cup\mathrm{K}^{\pm}\}$, while a two-parameter set of class A orbits in $\mathbf{S}$ converges to each $\mathrm{K}^{\pm}$.
\end{lemma}
\begin{figure}[ht!]
	\begin{center}
				\includegraphics[width=0.75\textwidth]{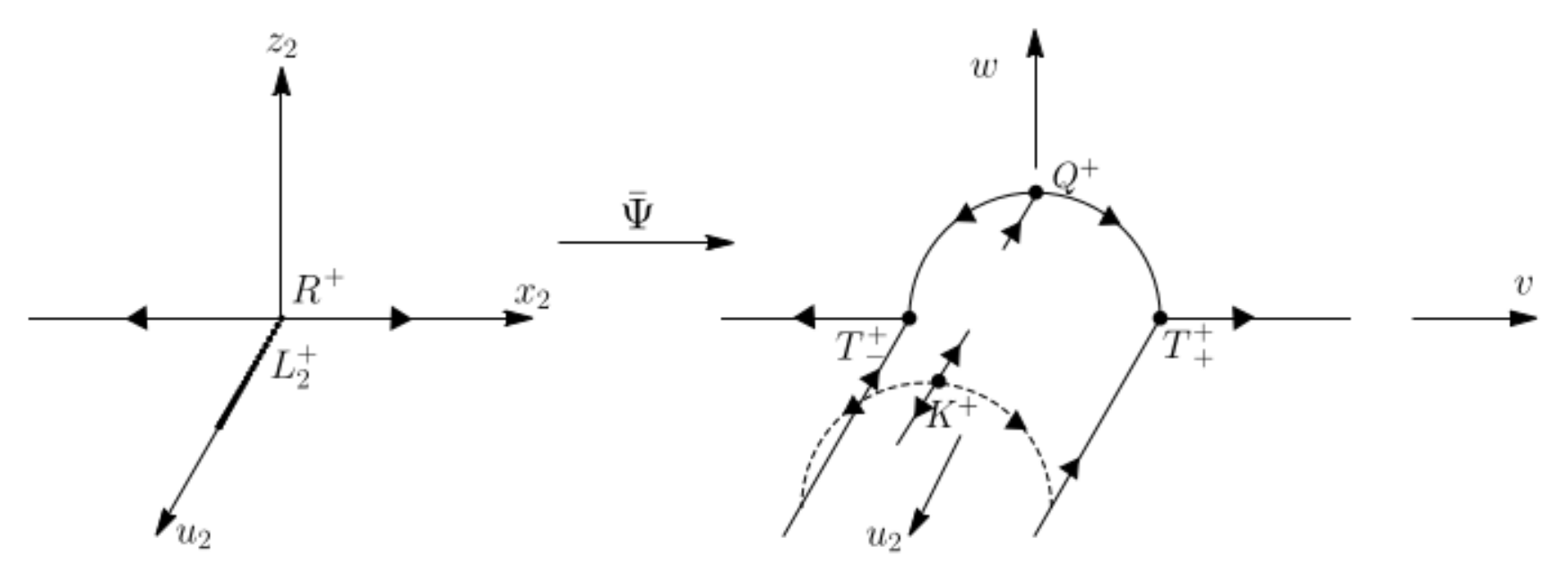}
	\end{center}
	\vspace{-0.5cm}
	\caption{Blow-up of the non-hyperbolic line of fixed points $\mathrm{L}^{+}_2\cup\mathrm{R}^{+}$.}
	\label{fig:T0L2}
\end{figure}
%
\subsubsection{Global phase-space on the Poincar\'e-Lyapunov disk}
In this section we introduce new polar variables on the $\{u_3=0\}$ subset to obtain a global phase-space picture on the unit disk $\mathbb{D}^2$. Similarly to the blow-up of $\mathrm{FL}_1$ done in section~\ref{ApA}, we introduce the variables
\begin{equation}\label{PLVAR}
\left(x_3,y_3,u_3\right)=\left(\left(\frac{r}{1-r}\right)^{\frac{1}{2p}}\cos \theta ,\left(\frac{r}{1-r}\right)^{\frac{2p+1}{2p}}F(\theta)\sin\theta,(1-r)r^{\frac{2p-n+1}{2n(2p-n)}}\bar{u}\right),
\end{equation}
where
$$
F(\theta)=\sqrt{\frac{1-\cos^{2(2p+1)}\theta}{1-\cos^2 \theta}}.
$$
The above transformation leads to
\begin{equation}
x_3^{2(2p+1)}+y_3^{2}=\left(\frac{r}{1-r}\right)^{\frac{2p+1}{p}}.
\end{equation}
By making a further change of time variable 
\begin{equation}
\frac{d}{d\xi}=(1-r)\frac{d}{d\bar{\tau}_3},
\end{equation}
we get the dynamical system on the Poincar\'e-Lyapunov cylinder
\begin{subequations}
	\begin{align}
	\frac{dr}{d\xi}=& \frac{2F(\theta)\sin \theta}{2p+1}\left(n(1-r)^{3+2n(2p-n)+\frac{n-1}{p}}r^{1+2p-n+\frac{n-1}{p}}\bar{u}^{2n(2p-n)}\cos^{2n-1}\theta+(2p+1)(1-r)^{2n-1}r^{3+\frac{1}{2p-n}}\bar{u}^{2n}	\right) \nonumber \\
	&+p r (1-r)\left(\frac{6(1-r)}{(2p-1)(1-\tilde{K})}\cos^{2(2p+1)}\theta+\frac{\nu}{2(2p+1)}rF^{2}(\theta)\cos^{2(p-1)}\sin^22\theta\right) \nonumber\\
	\frac{d\theta}{d\xi} =& -\frac{(1-r)^{2n}r^{2+\frac{1}{2p-n}}}{2p+1}F(\theta)\left((2p+1)F^{2}(\theta)\sin^{2}\theta-n (1-r)^{2n(2p-n)}r^{2p-n-2+\frac{n-1}{p}}\bar{u}^{2n(2p-n-1)}\right)\bar{u}^{2n} \nonumber \\
	&-F^{2}(\theta)\sin\theta\left(\frac{3(1-r)}{(2p-1)(1-\tilde{K})}\left(1-\tilde{K}\frac{2p-1}{2p+1}\right)-\nu \frac{r}{2p+1}\cos^{2p}\theta\right), \nonumber\\
	\frac{d\bar{u}}{d\xi}=&-\frac{1}{2nq}(1+q)\bar{u}+\frac{(2p-n+1)(1-r)-2 n p (2p-n)}{2n(2p-n)(2p+1)}\bar{u}\Bigg[ \Bigg.\nonumber \\
	&F(\theta)\left(2n p (1-r)^{1+\frac{(2p+1)(2p-n)}{p}}r \bar{u}^{2n(2p-n)}\cos^{2n-1}\theta+(1+q+2p(2-q))(1-r)F(\theta)\sin \theta\right)\sin \theta\nonumber\\
	&\Bigg.+(2p+1)\cos^{4p+1}\theta\left((1+q)r^{\frac{2p+1}{2p-n}}(1-r)\cos \theta + (1-r)^{2n}r^{\frac{(2p-n)(p+1)-1}{(2p-n)p}}\bar{u}^{2n}F(\theta)\sin \theta\right)\Bigg]\nonumber
	\end{align}
\end{subequations}
where
\begin{equation}
\begin{split}
q=&-1+\frac{6p}{(2p-1)(1-\tilde{K})}\Bigg(1-(1-r)^{2n(2p-n)}r^{\frac{2p+1}{p}}\bar{u}^{2n(2p-n)}	\cos^{2n}\theta\\
&\qquad\qquad-\frac{(2p-1)\tilde{K}-1}{2p}(1-r)^{2n-\frac{1}{p}}r^{3+\frac{1}{p}+\frac{1}{2p-n}}\bar{u}^{2n}F^{2}(\theta)\sin^2\theta\Bigg).
\end{split}
\end{equation}
The induced flow on the invariant boundary $\{\bar{u}=0\}$ can be regularly extended to the invariant boundaries $\{r=0\}$ and $\{r=1\}$, while for $\bar{u}>0$ it can be extended to these boundaries at least in a $C^1$ manner. The general structure of the Poincar\'e-Lyapunov cylinder is shown in Figure~\ref{fig:BUP_PL2}. In particular the non-hyperbolic lines $\mathrm{L}^{\pm}_2$ have been desingularised in these coordinates.
\begin{figure}[H]
	\begin{center}
		\includegraphics[width=0.35\textwidth]{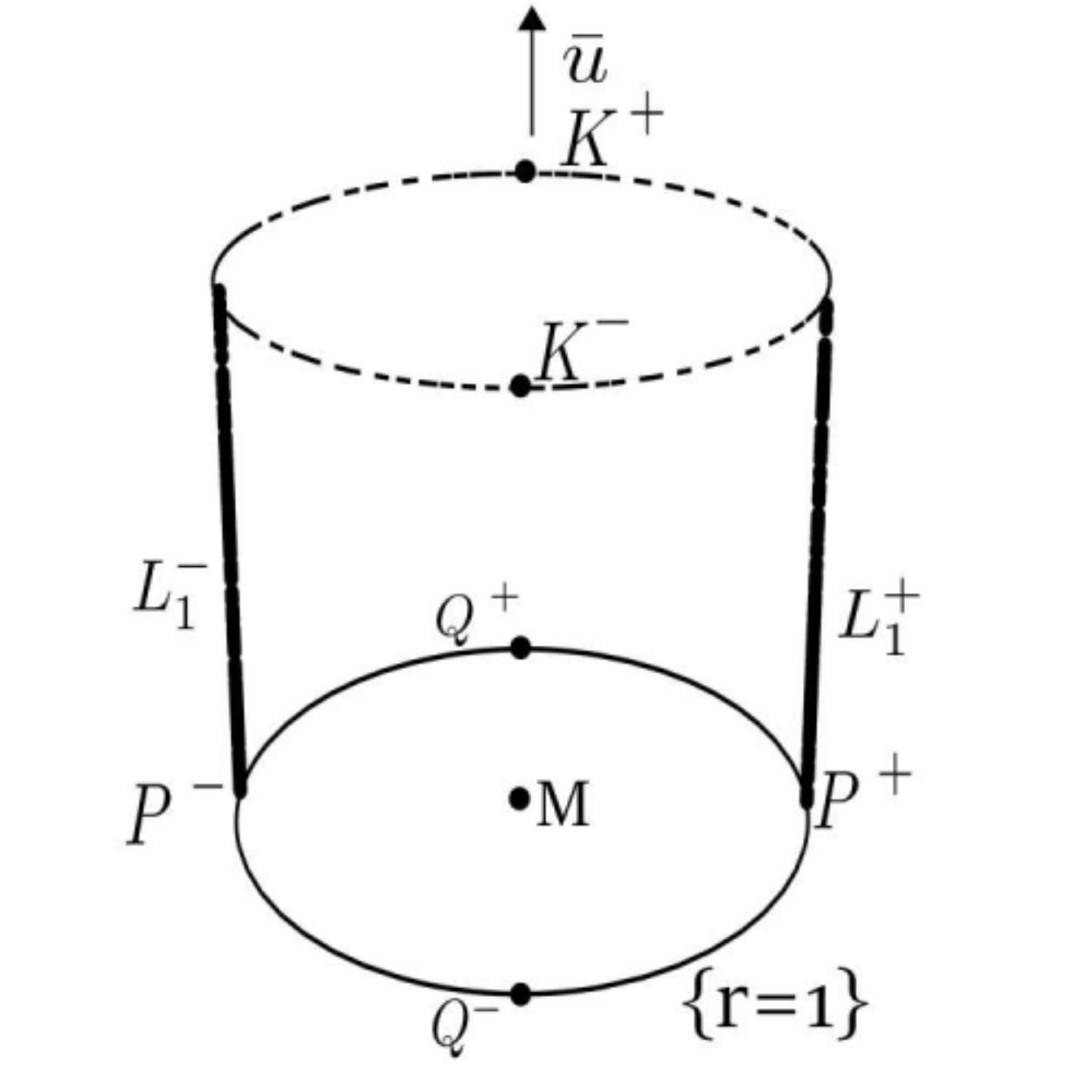}
	\end{center}
	\vspace{-0.5cm}
	\caption{Blow-up space $\mathcal{B}$ of $\mathrm{FL}_0$ on the Poincar\'e-Lyapunov cylinder.}
	\label{fig:BUP_PL2}
\end{figure}
Restricting to the invariant unit disk at $\{\bar{u}=0\}$, results in the regular dynamical system
\begin{subequations}
	\begin{align}
	\frac{dr}{d\xi}&=p r (1-r)\left(\frac{6(1-r)}{(2p-1)(1-\tilde{K})}\cos^{2(2p+1)}\theta+\frac{\nu r}{2(2p+1)}F^{2}(\theta)\cos^{2(p-1)}\theta\sin^22\theta\right),\\
	\frac{d\theta}{d\xi}&=-F^{2}(\theta)\sin\theta\left(\frac{3(1-r)}{(2p-1)(1-\tilde{K})}\left(1-\tilde{K}\frac{2p-1}{2p+1}\right)- \frac{\nu r}{2p+1}\cos^{2p}\theta\right).
	\end{align}
\end{subequations}
The boundary $\{r=0\}$ consists of a fixed point $\mathrm{M}$ which corresponds to the origin of the ($x_3,y_3$) plane, and which the previous analysis showed to be a hyperbolic saddle since $\tilde{K}<0$.
The hyperbolic fixed points $\mathrm{P}^{\pm}$ and $\mathrm{Q}^{\pm}$ at infinity in the $(x_3,y_3)$ plane are now located at the $\{r=1\}$ invariant boundary, and are given by
\begin{equation}
\theta_{\mathrm{P}^{+}}=0 ,\qquad \theta_{\mathrm{P}^{-}}=\pi,\qquad \theta_{\mathrm{Q}^{+}}=\frac{\pi}{2},\qquad  \theta_{\mathrm{Q}^{-}}=\frac{3\pi}{2}.
\end{equation}

\begin{theorem}
	Let $p>\frac{n}{2}$ . Then for all $\gamma_\mathrm{pf}\in(0,2)$ and $\nu>0$, the Poincar\'e-Lyapunov disk consists of heteroclinic orbits connecting the fixed points $\mathrm{M}$, $\mathrm{P}^\pm$ and $\mathrm{Q}^\pm$ as depicted in Figure~\ref{fig:PP3}.
\end{theorem}
\begin{proof}
	First notice that $\{y_3=0\}$ and $\{x_3=0\}$ are invariant subsets consisting of heteroclinic orbits $\mathrm{M}\rightarrow \mathrm{P}^\mathrm{\pm}$ and $\mathrm{Q}^\pm\rightarrow \mathrm{M}$ respectively. These separatrices split the phase-space into four  invariant subsets (corresponding to the four  quadrants). Since on these quadrants there are no fixed points, then there are also no periodic orbits and the result follows by the Bendixson-Poincar\'e theorem.
\end{proof}
%
%
%
\begin{figure}[H]
	\begin{center}
			\includegraphics[width=0.30\textwidth]{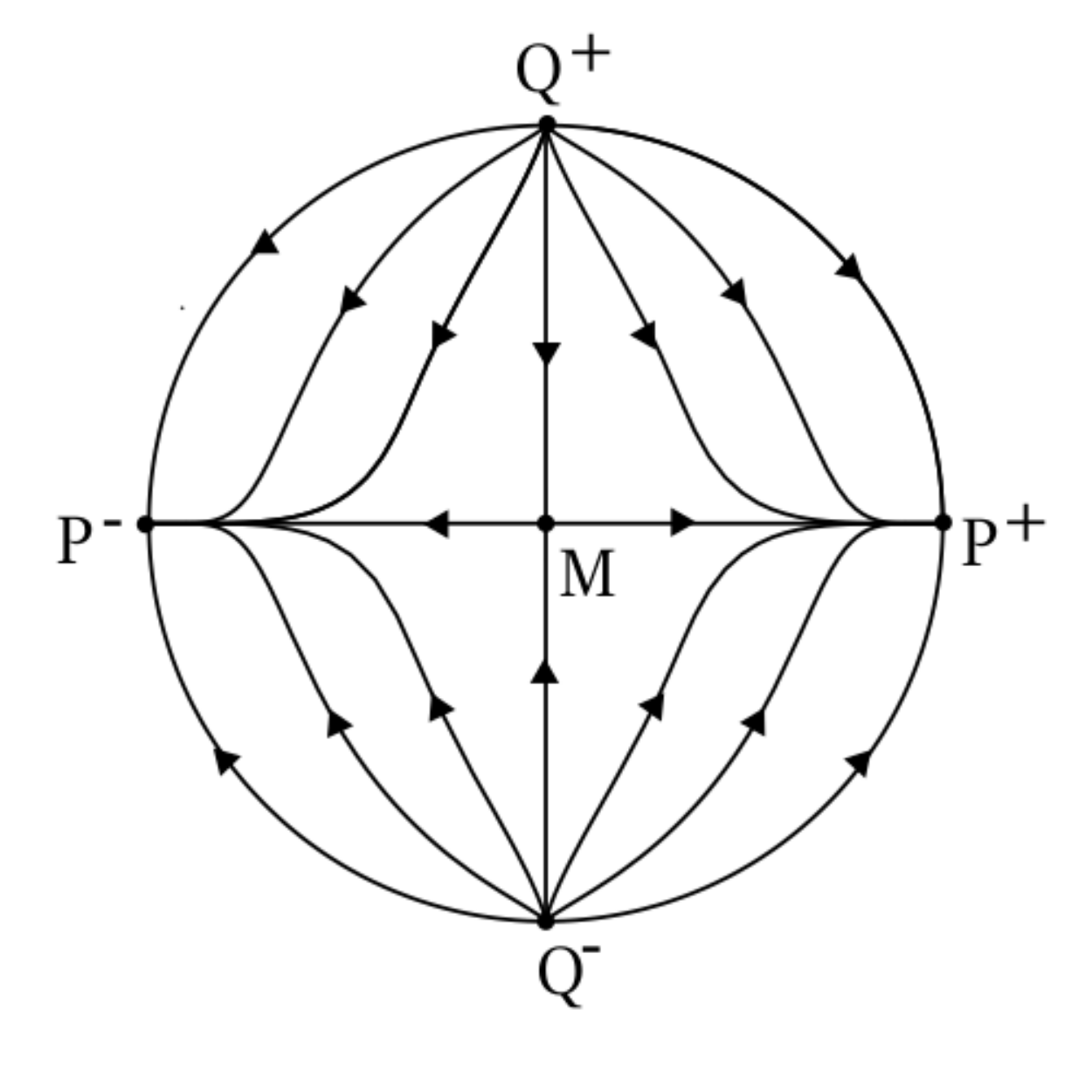}
	\end{center}
	\vspace{-0.5cm}
	\caption{Poincar\'e-Lyapunov unit disk when $p>\frac{n}{2}$.}
	\label{fig:PP3}
\end{figure}
%
%
%
\subsection{Invariant boundary $T=1$}
The induced flow on the $T=1$ invariant boundary is given by
\begin{equation}
	\frac{dX}{d\tau} = \Sigma_{\phi},\qquad\frac{d\Sigma_{\mathrm{\phi}}}{d\tau} = -n X^{2n-1}
\end{equation}
and this system has only one fixed point at $\Omega_\mathrm{pf}=1$ given by
\begin{equation}
\mathrm{FL_1}: \quad X=0, \quad\quad \Sigma_{\mathrm{\phi}}=0, \quad\quad T=1.
\end{equation}
\begin{lemma}
	Let $p>\frac{n}{2}$. Then the $\{T=1\}$ invariant boundary is foliated by periodic orbits $\mathcal{P}_{\Omega_\phi}$ characterized by constant values of $\Omega_\phi$, and centered at $\mathrm{FL}_1$, as depicted in Figure~\ref{fig:T1_2p-n}.
\end{lemma}
\begin{proof}
	On $\{T=1\}$ the auxiliary equation for $\Omega_\phi$ reads
	\begin{equation}
	\frac{d\Omega_\phi}{d\tau}=0\quad\Rightarrow \quad \Omega_\phi=X^{2n}+\Sigma_{\mathrm{\phi}}^2=\text{const.}
	\end{equation}
	from which the result follows.
	%
\end{proof}
\begin{figure}[ht!]
	\begin{center}
		\includegraphics[width=0.30\textwidth]{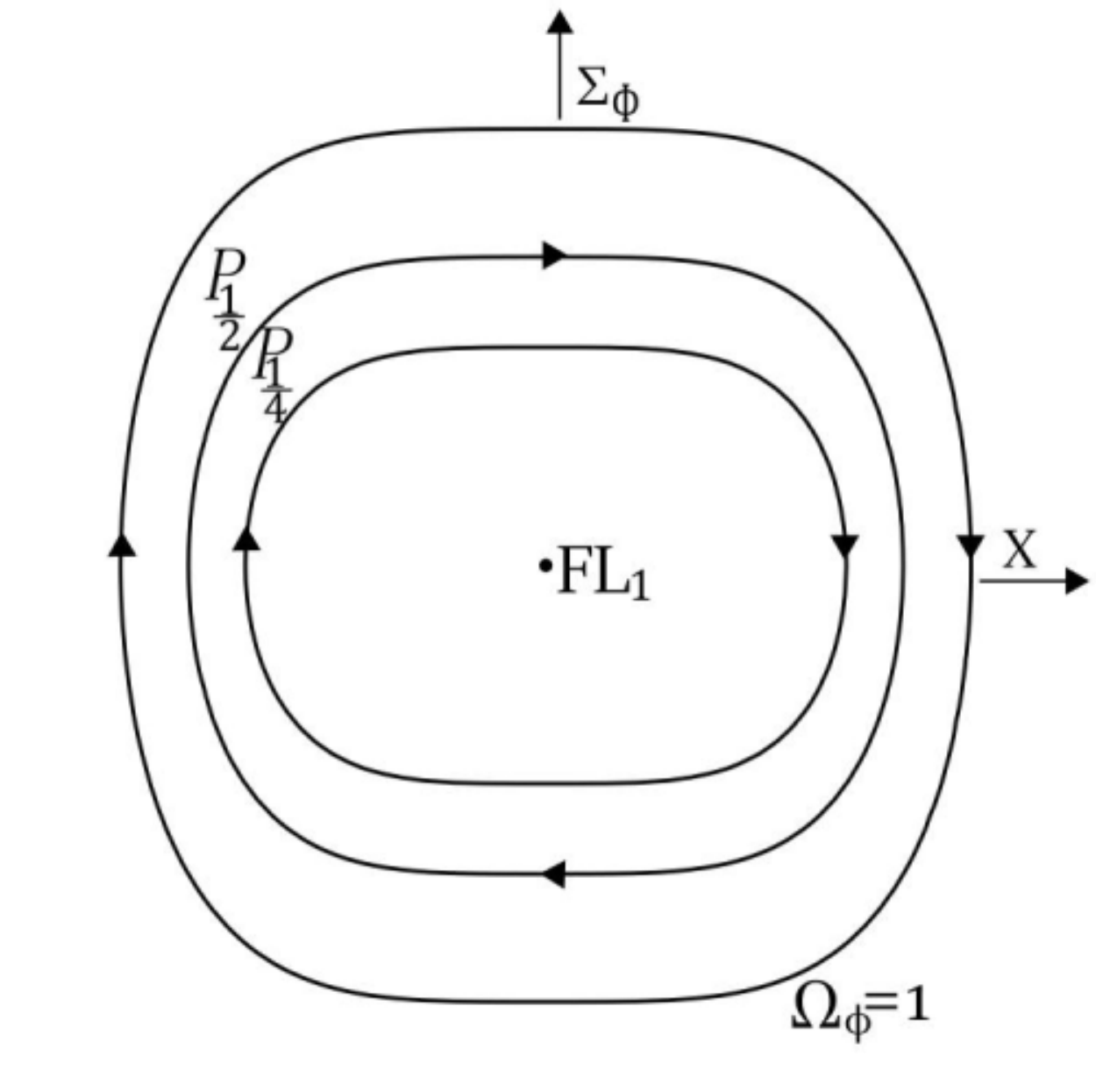}
	\end{center}
	\vspace{-0.5cm}
	\caption{Invariant boundary $\{T=1\}$ when $p>\frac{n}{2}$.}
	\label{fig:T1_2p-n}
\end{figure}

\begin{theorem} \label{teorema2}
	Consider the system \eqref{globalnleq2p} with $0<\Omega_{\phi}<1$ and $0<\gamma_{\mathrm{pf}}<2$:
	\begin{itemize}
		\item[(i)] \label{menor1} If $\gamma_\mathrm{pf} < \frac{2n}{n+1}$, then all solutions converge, for $\tau\rightarrow + \infty$, to the fixed point $\mathrm{FL_1}$ with $\Omega_{\phi}=0$ $(\Omega_{\mathrm{pf}}=1)$;
		\item[(ii)] \label{maior1} If $\gamma_\mathrm{pf}>\frac{2n}{n+1}$, then all solutions converge, for $\tau \rightarrow + \infty$, to  $\Omega_\mathrm{\phi}=1$ $(\Omega_{\mathrm{pf}}=0)$; 
		\item[(iii)] \label{igual1}  If $\gamma_\mathrm{pf}=\frac{2n}{n+1}$, then all solutions converge, for $\tau \rightarrow + \infty$, to a inner periodic orbit $P_{\Omega_\mathrm{\phi}}$. 
	\end{itemize}
\end{theorem}

\begin{proof}
	The proof uses the same methods as in the proof of Theorem~\ref{teorema1}. 
	Making use of \eqref{eps-eq} and the variable transformation \eqref{Angular} with $r=\sqrt{\Omega_\phi}$, the system~\eqref{globalnleq2p} leads to
	\begin{subequations}\label{omega evonew}
		\begin{align}
		\frac{dr}{d\tau}&=\frac{3}{2}\epsilon\left((1-\epsilon)^{2p-n}\left(\gamma_{\mathrm{pf}}-\gamma_{\mathrm{\phi}}\right)r(1-r^2)-\nu\sigma_{\mathrm{\phi}}\epsilon^{2p-n} r^{1+\frac{2p}{n}}\right):=\epsilon f(r,\tau,\epsilon) ,\label{omega evonew2}\\
		\frac{d\theta}{d\tau}&= 	-\frac{\epsilon}{2n}\left(3(1-\epsilon^{2p-n})+\epsilon^{2p-n}\nu r^{\frac{2p}{n}}\cos^{2p}{\theta}\right)G(\theta)^2\sin{2\theta}+(1-\epsilon)^{2p-n+1}r^{\frac{n-1}{n}}G(\theta),\\
		\frac{d\epsilon}{d\tau}&=-\frac{1}{n}(1+q)\epsilon^2(1-\epsilon),
		\end{align}
	\end{subequations}
	where
	\begin{equation}
	q+1=\frac{3}{2}\left(2rG(\theta)^2\sin^2\theta+\gamma_{\mathrm {pf}}(1-r^2)\right).
	\end{equation}
	Starting with the near identity transformation
	\begin{equation}
	r(\tau)=y(\tau)+\epsilon(\tau)g(y,\tau,\epsilon)
	\end{equation}
	the evolution equation for $y$ can be obtained using equations \eqref{omega evonew} and the expansions
	\begin{subequations}
		\begin{align}
		(1-\epsilon)^{2p-n} &\approx 1-(2p-n)\epsilon+\frac{1}{2}\zeta(\zeta-1)\epsilon^2-\frac{1}{6}\zeta(\zeta-1)(\zeta-2)\epsilon^3 +\mathcal{O}(\epsilon^3) \\
		\epsilon^{2p-n} &= \epsilon \delta_{1}^{2p-n}+\epsilon^2 \delta_{2}^{2p-n} +\mathcal{O}(\epsilon^3),
		\end{align}
	\end{subequations}
	where $\zeta=2p-n$, yielding
	\begin{subequations}
		\begin{align}
		\frac{dy}{d\tau} =& \left(1+\epsilon\frac{\partial g}{\partial y}\right)^{-1}\left[\frac{dr}{d\tau}-\left(g+\epsilon\frac{\partial g}{\partial \epsilon}\right)\frac{d\epsilon}{d\tau}-\epsilon\frac{\partial g}{\partial \tau}\right]\nonumber\\
		=&\left(1+\epsilon\frac{\partial g}{\partial y}\right)^{-1}\Bigg[\epsilon\left(\frac{3}{2}(\gamma_{\mathrm{pf}}-\langle\gamma_{\mathrm{\phi}}\rangle)\right)y(1-y^2)+\frac{3}{2}(\langle\gamma_{\mathrm{\phi}}\rangle-\gamma_{\mathrm{\phi}})y(1-y^2)-\frac{\partial g}{\partial \tau} \nonumber\\
		+&\frac{3}{2}\epsilon^2\left((\gamma_{\mathrm{pf}}-\gamma_{\mathrm{\phi}})[(1-3y^2)g+\zeta y(1-y^2)]- \nu\sigma_{\mathrm{\phi}} y^2\delta^{2p-n}_1 +\frac{1}{n}\left(\gamma_\mathrm{pf}+(\gamma_\phi-\gamma_\mathrm{pf})y^2\right)g\right)\Bigg]
		+\mathcal{O}(\epsilon^3)\nonumber
		\end{align}
	\end{subequations}
	Setting
	\begin{equation}\label{boundedw}
	\begin{split}
	\frac{\partial g}{\partial \tau}
	=\frac{3}{2}\left(\langle\gamma_{\mathrm{\phi}}\rangle -\gamma_{\mathrm{\phi}}\right)y(1-y^2),
	\end{split}
	\end{equation}
	where the right-hand-side of \eqref{boundedw} is, for large times, almost periodic with $\langle\gamma_{\mathrm{\phi}}\rangle-\gamma_\phi\approx0$ which implies that $g$ is bounded. Using $(1+\epsilon\frac{\partial g}{\partial y})^{-1}\approx 1-\epsilon\frac{\partial g}{\partial y}+\mathcal{O}(\epsilon^2)$ we get
	\begin{equation}
	\frac{dy}{d\tau}=\epsilon\langle f \rangle (y)+\epsilon^2 h(y,g,\tau,\epsilon)+\mathcal{O}(\epsilon^3),
	\end{equation}
	where
	\begin{subequations}
		\begin{align}
		\langle f \rangle (y)=&\langle f  (y,.,0)\rangle=\frac{3}{2}(\gamma_{\mathrm{pf}}-\langle\gamma_{\mathrm{\phi}}\rangle)y(1-y^2), \nonumber\\
		h(y,g,,\tau,\epsilon)=&\frac{3}{2}(\gamma_{\mathrm{pf}}-\langle\gamma_{\mathrm{\phi}}\rangle)y(1-y^2)\frac{\partial g}{\partial y}+\frac{3}{2n}\left(\gamma_\mathrm{pf}+(\gamma_\phi-\gamma_\mathrm{pf})y^2\right)g \nonumber \\
		+&\frac{3}{2}\left((\gamma_{\mathrm{pf}}-\gamma_{\mathrm{\phi}})(1-3y^2)g+(\gamma_{\mathrm{pf}}-\gamma_{\mathrm{\phi}})(2p-n) y(1-y^2)-\nu \sigma_{\mathrm{\phi}}y^2 \delta^{2p-n}_1 \right).\nonumber
		\end{align}
	\end{subequations}
	 Now, we study the truncated averaged system which, after changing the time variable as $d\bar{\tau}=\epsilon d\tau$, reads 
	\begin{subequations} \label{y avr system2}
		\begin{align}
		\frac{d\bar{y}}{d\bar{\tau}} &= 3\Big(\gamma_{\mathrm{pf}}-\langle\gamma_{\mathrm{\phi}}\rangle\Big)\bar{y}\Big(1-\bar{y}\Big), \label{y avr system 1.2}\\
		\frac{d\epsilon}{d\bar{\tau}} &= -\frac{1}{n}\epsilon(1-\epsilon)(1+q). \label{y avr system 2.2}
		\end{align}
	\end{subequations}
	This system is similar to the one studied in \cite{Alho2} in the absence of interactions, i.e. for $\nu=0$. For $\gamma_{\mathrm{pf}}-\langle \gamma_{\mathrm{\phi}}\rangle \neq 0$, there are two fixed points
	\begin{subequations}
		\begin{align}
		\mathrm{F_1}&: \quad \bar{y}=1,\quad\quad \epsilon=0 \\
		\mathrm{F_2}&: \quad \bar{y}=0,\quad\quad \epsilon=0.
		\end{align}
	\end{subequations}
	The linearisation around $\mathrm{F}_1$ yields the eigenvalues $-3(\gamma_{\mathrm{pf}}-\langle\gamma_{\mathrm{\phi}}\rangle )$ and $-\frac{3\gamma_{\mathrm{pf}}}{2n}$, with associated eigenvectors the canonical basis of $\mathbb{R}^2$, while for $\mathrm{F}_2$ the eigenvalues are  $3(\gamma_{\mathrm{pf}}-\langle \gamma_{\mathrm{\phi}}\rangle)$, $-\frac{3\gamma_{\mathrm{pf}}}{2n}$ and the eigenvectors are again the canonical basis of $\mathbb{R}^2$. The stability of each fixed point depends on the sign of $\gamma_{\mathrm{pf}}-\langle\gamma_{\mathrm{\phi}}\rangle$. If $\gamma_{\mathrm{pf}} >\langle\gamma_{\mathrm{\phi}}\rangle$ then $\mathrm{F_1}$ is a sink and $\mathrm{F_2}$ is a saddle, whereas if $\gamma_{\mathrm{pf}}<\langle \gamma_{\mathrm{\phi}}\rangle$ then $\mathrm{F_1}$ is a saddle and $\mathrm{F_2}$ is a sink. We can then estimate $|\eta|=|y-\bar{y}|$ for solutions $y$ and $\bar{y}$ with the same initial conditions and, using Gronwall's inequality, show that $y$ (and so also $r$ and $\Omega_{\mathrm{\phi}}$) has the same limit as the solutions $\bar{y}$ of the truncated averaged equation for $\tau \rightarrow +\infty$, proving statements $(i)$ and $(ii)$ of the Theorem. 
	
	Regading point $(iii)$, when $\gamma_{\mathrm{pf}}=\langle \gamma_{\mathrm{\phi}} \rangle=\frac{2n}{n+1}$, the equation for $y$ is given by
	\begin{equation}
	\begin{split}
	\frac{dy}{d\tau}=\frac{3}{2}\epsilon^2\Bigg((\gamma_{\mathrm{pf}}-\gamma_{\mathrm{\phi}})\left((1-3y^2)g+(2p-n) y(1-y^2)-\frac{g}{n}y^2\right)+\frac{\gamma_\mathrm{pf}}{n}g -\nu \sigma_\mathrm{\phi}y^2\delta^{2p-n}_1  \Bigg).
	\end{split}
	\end{equation}
	Taking the average of $h$ at $\epsilon=0$, we get
	\begin{equation}
	\langle h \rangle(y) =\langle h(y,.,0) \rangle=\frac{3}{n+1}\langle g \rangle (y)-\frac{3}{2}\nu\langle\sigma_{\mathrm{\phi}}\rangle y^2 \delta^{2p-n}_{1},	
	\end{equation}
	where integration by parts has been used and the average of the interaction term $\sigma_\phi$ defined in~\eqref{NewInt} is now explicitly given by
	\begin{equation}
	\langle\sigma_{\mathrm{\phi}}\rangle=\frac{\Gamma\left[\frac{1}{2}+\frac{1}{2n}\right]\Gamma\left[\frac{1}{2n}+\frac{p}{2n}\right]}{3\Gamma[1+\frac{1}{2n}]\Gamma[\frac{1}{2n}+\frac{p}{n}]}.
	\end{equation}
	After changing the time variable as $\epsilon d/d\bar{\tau}=d/d\tau$, yields the truncated averaged system
	\begin{subequations}
		\begin{align}
		\frac{d\bar{z}}{d\bar{\tau}} &=\epsilon \frac{3}{n+1}\langle g \rangle (\bar{z})-\epsilon\frac{3}{2}\langle\sigma_{\mathrm{\phi}}\rangle\nu \bar{z}^2 \delta^{2p-n}_{1}, \\
		\frac{d\epsilon}{d\bar{\tau}}&=-\frac{3}{n(n+1)}\epsilon(1-\epsilon),
		\end{align}
	\end{subequations}
	which has a line of fixed points at $\epsilon=0$  parameterized by $\bar{z}_0 \in [0,1]$. The linearisation around the line yields the eigenvalues $0$ and $-\frac{3}{n(n+1)}$ with associated eigenvectors 
	\begin{equation*}
	v_1=(1,0),\qquad v_2=\left(-n\langle g\rangle \bar{z}_0+\frac{n(n+1)}{2}\nu \langle\sigma_{\mathrm{\phi}}\rangle\bar{z}_0^2 \delta^{2p-n}_1,1\right).
	\end{equation*}
	Therefore the line is normally hyperbolic and each point on the line is exactly the $\omega$-limit point of a unique orbit with $\epsilon>0$ initially. Similarly to the proofs of cases $(i)$ and $(ii)$,  we can estimate the term $\mathcal{O}(\varepsilon^3)$ that provides bootstraping sequences. This defines a pseudo-trajectory $r^n(\bar{\tau}_n)=\bar{z}(\bar{\tau}_n)$ of the system \eqref{omega evonew}, with
	\begin{equation}
	|r^n(\bar{\tau})-\bar{z}(\bar{\tau})|\leq K\varepsilon_n^2 \,,
	\end{equation}
	where $\bar{\tau} \in [\bar{\tau}_n,\bar{\tau}_{n+1}]$ and $K$ is a positive constant. The compactness of the state space and the regularity of the flow, implies that there exists a set of initial values whose solution trajectory $r(\bar{\tau})$ shadows the pseudo-trajectory $r^n(\bar{\tau})$, in the sense that
	\begin{equation}
	\forall n\in \mathbb{N},~~\forall\bar{\tau}\in[\bar{\tau}_n,\bar{\tau}_{n+1}]:~~|r^n(\bar{\tau})-r(\bar{\tau})|\leq K\varepsilon_n^2 \,.
	\end{equation}
	Finally, using the triangle inequality, we get
	\begin{eqnarray}
	|r(\bar{\tau})-\bar{z}(\bar{\tau})|&=&|r(\bar{\tau})-r^n(\bar{\tau})+r^n(\bar{\tau})-\bar{z}(\bar{\tau})| \nonumber \\
	&\leq&\underbrace{|r^n(\bar{\tau})-r(\bar{\tau})|}_{\leq K\varepsilon_n^2}+\underbrace{|r^n(\bar{\tau})-\bar{z}(\bar{\tau})|}_{\leq K\varepsilon_n^2} \nonumber \\
	&\leq& 2K\varepsilon_n^2 \underbrace{\rightarrow}_{\bar{\tau}_n \rightarrow \infty} 0 \, ,
	\end{eqnarray}
	and, therefore, for each $\bar{z}_0 \in [0,1]$, there exists a solution trajectory $r(\bar{\tau})$ that converges to a  periodic orbit at $\varepsilon=0$ i.e. $T=1$, characterized by $r=\bar{z}_0$, which concludes the proof of $(iii)$.
\end{proof}	
The global state space $\mathbf{S}$ for $p>n/2$ can be found in Figure~\ref{fig:2p-n final}, where the solid numerical curves correspond to the center manifold of each $\mathrm{dS}^{\pm}_0$ and the dashed numerical curves to solutions originating from one of the equivalent hyperbolic sources $\mathrm{K}^{-}$. All solutions converge to the global future attractor which is either a limit cycle for $\gamma_\mathrm{pf}\geq \frac{2n}{n+1}$ or the Friedmann-Lema\^itre fixed point $\mathrm{FL}_1$ when $\gamma_\mathrm{pf}< \frac{2n}{n+1}$.
\begin{figure}[ht!]
	\begin{center}
		\subfigure[$\gamma_\mathrm{pf}-\langle\gamma_\phi\rangle>0$.]{\label{fig:2p-n_1.5}
			\includegraphics[width=0.30\textwidth]{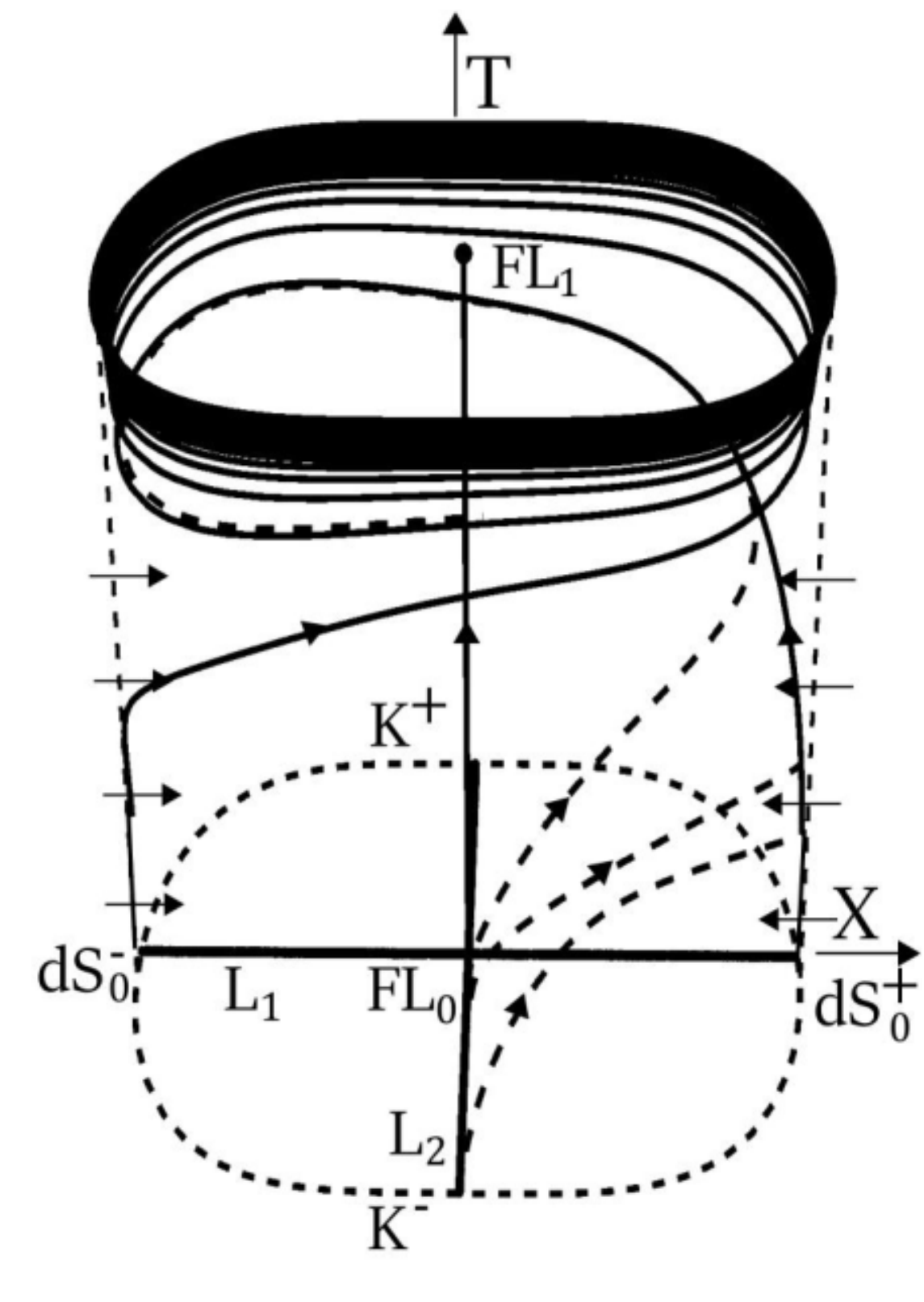}}
		\hspace{0.1 	cm}
		\subfigure[$\gamma_\mathrm{pf}-\langle\gamma_\phi\rangle=0$.]{\label{fig:2p-n_1.3}
			\includegraphics[width=0.30\textwidth]{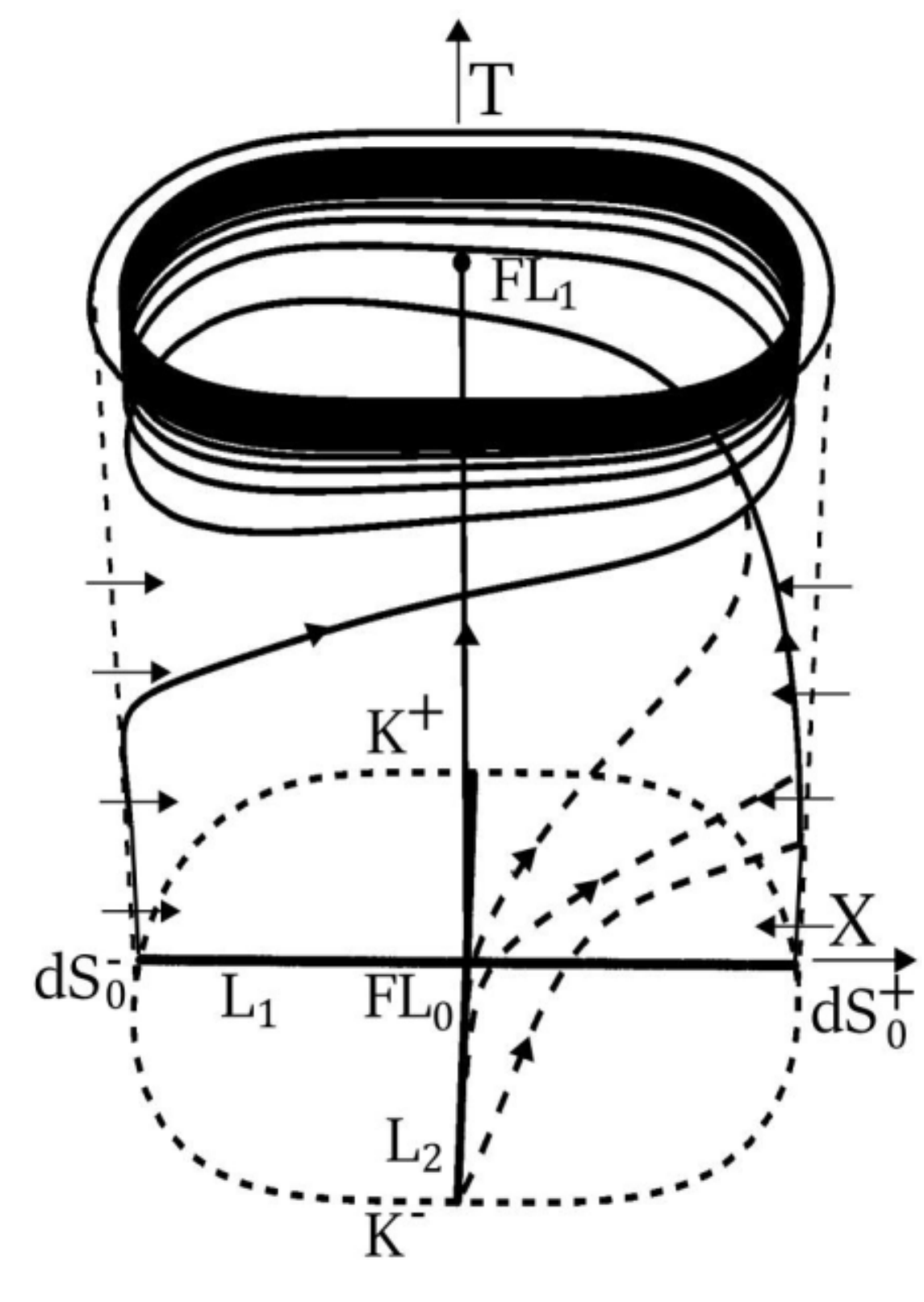}}
		\hspace{0.1 	cm}
		\subfigure[$\gamma_\mathrm{pf}-\langle\gamma_\phi\rangle<0$.]{\label{fig:2p-n_1}
			\includegraphics[width=0.30\textwidth]{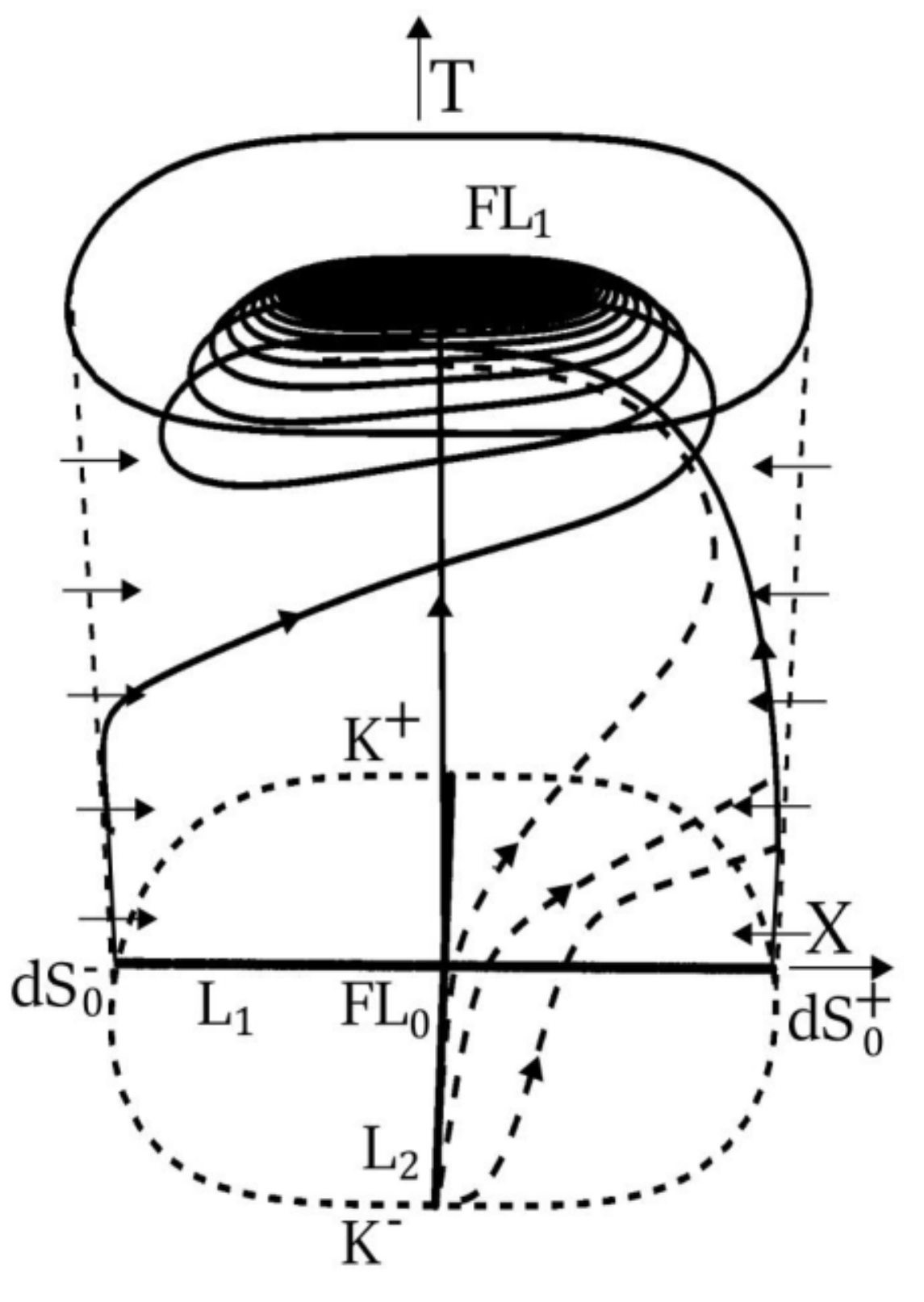}}
	\end{center}
	\vspace{-0.5cm}
	\caption{Global state space $\mathbf{S}$ when $p>\frac{n}{2}$.}
	\label{fig:2p-n final}
\end{figure}
%
\section{Concluding remarks}
\label{conclusion}
Cosmological models are commonly modeled by flat Robertson-Walker metrics which contain either a perfect fluid with linear equation of state $p_\mathrm{pf}=(\gamma_\mathrm{pf}-1)\rho_\mathrm{pf}$, $\gamma_\mathrm{pf}\in(0,2)$, or a scalar field $\phi$ with a particular form of the potential $V(\phi)$. Here, motivated by warm inflation models of the early universe, we considered that both matter models are present, interact with each other and furthermore allow for a large family of potentials $V(\phi)= \frac{(\lambda\phi)^{2n}}{2n}$, $n\in\mathbb{N}$, $\lambda>0$, with a friction-like interaction term of the form $\Gamma(\phi)=\mu\phi^{2p}$, $p\in\mathbb{N}\cup\{0\}$, $\mu>0$. This brings new mathematical challenges as new non-linearities arise in the resulting ODE system comparing to the non-interacting case with $\mu=0$, discussed in detail in~\cite{Alho2}. 

Our analysis relied on reformulating the Einstein equations as a new regular 3-dimensional dynamical system on a compact state-space. This construction showed that a significant dynamical bifurcation occurs for $p=n/2$, and that the positive dimensionless parameter
\begin{equation}
\nu=\sqrt{n 6^{2p-(n-1)}}\frac{\mu}{\lambda^{n}},
\end{equation}
with $\nu=0$ when $\mu=0$, plays a fundamental role on the qualitative description of the models. The analysis was then split into the three distinct cases $p<n/2$, $p=n/2$, and $p>n/2$. For all these models we were able to prove rigorously their main global dynamical properties which we complemented with numerical pictures.

Our results also have interesting physical consequences: First, as is common in interacting cosmologies where there is an energy exchange between scalar fields and fluids, see e.g.~\cite{BC00,Oliveira,Alho4}, there are solutions for which the perfect fluid energy density is generated via energy transfer from the scalar field, thereby starting its evolution at some finite time $t_*$, where  $\rho_\mathrm{pf}(t_*)=0$ with $H(t_*)>0$ finite, and $\rho_\mathrm{pf}(t)>0$ for all time of existence $t>t_*$. We called these solutions Class B, while Class A solutions have past asymptotics associated with the limit $H\rightarrow+\infty$.

For Class A solutions, and towards the past, we showed that irrespective of the model parameters an inflationary attractor solution always exist and that it is associated with a 1-dimensional unstable center manifold of a quasi-de-Sitter fixed point, while generically the asymptotics are governed by the massless scalar-field or kinaton solution. However, the influence of the interaction term on the inflationary leading order dynamics differs for the three different cases $p< n/2$, $p=\frac{n}{2}$ or $p>\frac{n}{2}$, see Remarks~\ref{PastAs1}, \ref{PastAs2}, and \ref{PastAs3}. 

For Class A and B solutions, and towards the future, the asymptotics are more intricate: When $p<\frac{1}{2}(n-2)$ the generic future attractor is the de-Sitter solution, a results that seems unknown in the literature and that might offer a new physical model of quintessential inflation. When $p=\frac{1}{2}(n-2)$ the generic future attractor corresponds to new scaling solutions where the deceleration parameter
\begin{equation}
q_{\mathrm{S}^{\pm}}=-1+\frac{3\gamma_\mathrm{pf}}{2}\left(1-  \left(\frac{n^2}{3\gamma_{\mathrm{pf}}\nu}\right)^2\left(-1+\sqrt{1+\left(\frac{3\gamma_{\mathrm{pf}}\nu}{n^2}\right)^2}\right)^{2}\right),
\end{equation}
satisfies $-1<q_{\mathrm{S}^{\pm}}<-1+\frac{3\gamma_\mathrm{pf}}{2}$, which interestingly, for some parameter values can also lead to late time accelerated expansion $(q_{\mathrm{S}^{\pm}}<0)$. 
When $p=\frac{1}{2}(n-1)$ the scale factor behaves as the Friedman-Lema\^itre (FL) solution as $t\rightarrow+\infty$, while the matter fields have simple asymptotics when $(p,n)=(0,1)$ depending on if $\nu<2$, $\nu=2$ or $\nu>2$, while for $p>0$ there are various kinds of asymptotic behaviour which can be obtained by the analysis of generalised Li\'enard type systems. When $p=\frac{n}{2}$ the future asymptotics are either fluid dominated if $\gamma_\mathrm{pf}\leq \frac{2n}{n+1}$ in which case the scale-factor behaves as the FL solution, or has an oscillatory behaviour, if $\gamma_\mathrm{pf}> \frac{2n}{n+1}$, where neither the fluid nor the scalar field dominates. Finally when $p>\frac{n}{2}$ the future asymptotics are similar to the non-interacting case with $\nu=0$.

Overall, the interaction term has a larger impact on the future asymptotics when $p< n/2$, and on the past asymptotics when $p>\frac{n}{2}$.
 
\section*{Acknowledgments}
The authors were supported by FCT/Portugal through CAMGSD, IST-
ID, projects UIDB/04459/2020 and UIDP/04459/2020. VB thanks FCT for the Ph.D. grant PD/BD/142891/2018. FCM and VB thank CMAT, Univ. Minho, through FCT project Est-OE/MAT/UIDB/00013/2020 and FEDER Funds COMPETE.

\appendix


\end{document}